%% file: ms.tex
\documentclass[journal]{IEEEtran}
\IEEEoverridecommandlockouts
\usepackage[utf8]{inputenc}
\usepackage{cite}
\usepackage{amsmath,amssymb,amsfonts}
\usepackage{amsthm}
\usepackage{algorithmic}
\usepackage{graphicx}
\usepackage{textcomp}
\usepackage{xcolor}
\usepackage[nolist]{acronym}
\usepackage{gensymb}

    \makeatletter
\def\maketag@@@#1{\hbox{\m@th\normalfont\normalsize#1}}
\makeatother

\usepackage{pgfplots}
  \pgfplotsset{compat=newest}
  \usetikzlibrary{plotmarks}
  \usetikzlibrary{arrows.meta}
  \usepgfplotslibrary{patchplots}
  \usepackage{grffile}

\usepackage{subcaption}
\DeclareMathSizes{10}{9}{6.3}{4.5}

\newtheorem{theorem}{Theorem}

\newtheorem{assumption}{Assumption}
\newtheorem{lemma}{Lemma}

\newcommand{\tr}[1]{\operatorname{tr}\left \lbrack #1 \right \rbrack}

\newcommand{\res}[2]{\operatorname{Res}\left\lbrack #1,\text{ } #2 \right \rbrack}

\newcommand{\rea}[1]{\operatorname{Re}\left \lbrack #1 \right \rbrack}
\newcommand{\ima}[1]{\operatorname{Im}\left \lbrack #1 \right \rbrack}

\newcounter{MYtempeqncnt1}

\newcounter{music_second_order_equation_number}

\usepackage{stfloats}

\begin{document}

%
\title{Probability of Resolution of MUSIC and g-MUSIC: An Asymptotic Approach}
%
%
%
\author{David~Schenck,
        Xavier~Mestre,
        and~Marius~Pesavento
\thanks{David Schenck and Marius Pesavento were supported by the DFG PRIDE Project PE 2080/2-1 and the German Academic Exchange Service (DAAD). Xavier Mestre was supported by the Spanish Government under grant RTI2018-099722-BI00.}
\thanks{David Schenck and Marius Pesavento are with the Communication Systems Group, Technische Universit\"at Darmstadt, Darmstadt, Germany, (e-mail: \{schenck, pesavento\}@nt.tu-darmstadt.de).}
\thanks{Xavier Mestre is with the Centre Tecnol\`{o}gic de Telecomunicacions de Catalunya, Castelldefels, Spain, (e-mail: xavier.mestre@cttc.cat).}
\thanks{Parts of this work were published at \textit{IEEE International Conference on Acoustics, Speech, and Signal Processing (ICASSP 2021)}.}%
}

\maketitle
\begin{acronym}[PR-DML]
	 \acro{ULA}{Uniform Linear Array}
	 \acro{DoA}{Direction-of-Arrival}
	 \acro{FoV}{Field of View}
	 \acro{MUSIC}{Multiple Signal Classification}
	 \acro{ESPRIT}{Estimation of Signal Parameters via Rotational Invariance Technique}
	 \acro{DML}{Deterministic Maximum Likelihood}
	 \acro{SML}{Stochastic Maximum Likelihood}
	 \acro{WSF}{Weighted Subspace Fitting}
	 \acro{CF}{Covariance Fitting}
	 \acro{FCF}{Full Covariance Fitting}
	 \acro{PR}{Partial Relaxation}
	 \acro{PR-DML}{Partially Relaxed Deterministic Maximum Likelihood}
	 \acro{PR-CF}{Partially Relaxed Covariance Fitting}
	 \acro{PR-FCF}{Partially Relaxed Full Covariance Fitting}
	 \acro{MDL}{Minimum Description Length}
	 \acro{SNR}{Signal-to-Noise Ratio} 
	 \acro{LS}{Least Square}
	 \acro{ML}{Maximum Likelihood}
	 \acro{LR}{Likelihood Ratio}
	 \acro{FWE}{Familywise Error-Rate}
	 \acro{FDR}{False Discovery Rate}
	 \acro{RMT}{Random Matrix Theory}
	 \acro{CRB}{Cramer-Rao Bound}
	 \acro{RMSE}{Root-Mean-Squared-Error}
	 \acro{DCT}{Dominated Convergence Theorem}
	 \acro{pdf}{probability density function}
	 \acro{cdf}{cumulative distribution function}
	 \acro{MSE}{Mean Square Error}
	 \acro{CLT}{Central Limit Theorem}
\end{acronym}

\thinmuskip = 0mu
\medmuskip = 0mu
\thickmuskip=0mu
\begin{abstract}
	In this article, the outlier production mechanism of the conventional \ac{MUSIC} and the g-\ac{MUSIC} \ac{DoA} estimation technique is investigated using tools from \ac{RMT}. A general \ac{CLT} is derived that allows to analyze the asymptotic stochastic behavior of eigenvector-based cost functions in the asymptotic regime where the number of snapshots and the number of antennas increase without bound at the same rate. Furthermore, this \ac{CLT} is used to provide an accurate prediction of the resolution capabilities of the \ac{MUSIC} and the g-\ac{MUSIC} \ac{DoA} estimation method. The finite dimensional distribution of the \ac{MUSIC} and the g-\ac{MUSIC} cost function is shown to be asymptotically jointly Gaussian distributed in this asymptotic regime.
\end{abstract}

\begin{IEEEkeywords}
	\ac{MUSIC}, g-\ac{MUSIC}, \acs{DoA} estimation, central limit theorem, random matrix theory, probability of resolution, performance analysis.
\end{IEEEkeywords}

\section{Introduction}		
	Due to the vast variety of use cases, \ac{DoA} estimation belongs to the most relevant research areas in signal processing. The applications range from radar and sonar to electric surveillance, seismology, astronomy and mobile communications \cite{a12, a8, a13, a14}. Multiple \ac{DoA} estimation techniques have been proposed in the literature. Among them, subspace-based \ac{DoA} estimation techniques which are known to provide a good compromise between computational complexity and \ac{DoA} estimation accuracy. This is mainly because these methods avoid multidimensional searches while providing relatively good performance. One of the most popular examples of subspace-based \ac{DoA} estimation methods is \ac{MUSIC} \cite{a1}, which exploits the orthogonality between signal and noise subspaces by finding the \acp{DoA} that achieve the highest orthogonality between the array signature and the noise subspace of the sample covariance matrix. It is well known \cite{a69} that the noise space spanned by the sample covariance matrix is not a consistent estimate of the true one when both the sample size and the number of array elements become large but are still comparable in magnitude. Therefore, it is possible to come up with a refined \ac{MUSIC} algorithm, usually referred to as g-\ac{MUSIC} \cite{a69}, which replaces the noise sample eigenvectors by a consistent estimate of the true noise subspace. This provides a significant improvement in terms of both accuracy and resolution in the low sample size scenario, whereby the number of snapshots and the number of array elements have the same order of magnitude. In any case, both \ac{MUSIC} and g-\ac{MUSIC} are subspaced-based algorithms, and as such both suffer from the so-called breakdown or threshold effect. This effect is characterized by a rapid loss of resolution capabilities due to the systematic appearance of outliers in the \ac{DoA} estimates when either the sample size (in snapshots per antennas) or the \ac{SNR} falls below a certain threshold. When this occurs, merging signal extrema in the null-spectrum of the \ac{DoA} estimator leads to the appearance of outliers in the \ac{DoA} estimates, as in this case at least one of the deepest local minima in the null-spectrum is not associated with a true source. The presence of outliers causes a severe performance breakdown in terms of \ac{DoA} estimation accuracy which has great practical implications \cite{a52}. Moreover, the threshold effect and thus the production of outliers in the \ac{DoA} estimates is not captured by standard statistical performance bounds such as the \ac{CRB} and therefore requires an estimator specific analysis. 
	
	The objective of this paper is to analytically characterize the breakdown effect (thus the outlier production mechanism) in the threshold region of \ac{MUSIC} and g-\ac{MUSIC} by studying the resolution capabilities of both \ac{DoA} estimation algorithms. More specifically, we investigate the probability that these algorithms resolve two close sources in the asymptotic regime where both the sample size and the number of array elements tend to infinity at the same rate. Up to now, the literature has mainly focused on the performance characterization of the conventional \ac{MUSIC} method in terms of both accuracy and resolution probability \cite{a70,a71,a57,a6,a48,a72,a73}. Most of the performance analyses rely on conventional large sample-size asymptotics, where the number of array elements is assumed to be fixed while the number of snapshots grows without bound. This asymptotic regime is not very suitable for characterizing the threshold performance, since loss of resolution occurs e.g. when the number of snapshots is not much larger (or even smaller) than the number of antennas. For this reason, we propose here to analyze this outlier production mechanism under the more appropriate setting where these two quantities are large but comparable in magnitude. This is indeed the setting that was considered in \cite{a74} to investigate the consistency of the \ac{DoA} estimates obtained through \ac{MUSIC} or g-\ac{MUSIC}.
	
	We extend the work in \cite{a74} and analyze the statistical fluctuations of the \ac{MUSIC} and g-\ac{MUSIC} cost functions, which are the key to understanding the outlier production mechanism that leads to loss of resolution. The approach is similar to the one followed in \cite{a75,a65,a66} to study the resolution capabilities of the recently introduced \ac{PR-DML} algorithm. It also shares some ideas with \cite{a63,a39}, where the resolution probability of the conventional \ac{DML} and the \ac{SML} method are derived through an asymptotic characterization of stochastic fluctuations of the corresponding multidimensional cost functions. In comparison to the \ac{PR-DML}, \ac{DML} and \ac{SML} cost functions, the \ac{MUSIC} and g-\ac{MUSIC} cost functions involve eigenvectors of a random matrix and therefore require a fundamentally new asymptotic analysis. Furthermore, eigenvector-based cost functions are shown to be asymptotically jointly Gaussian distributed and a general \ac{CLT} is provided that allows to characterize the asymptotic fluctuations of such cost functions. The asymptotic second order behavior is expressed as a double contour integral which is solved for both \ac{MUSIC} and g-\ac{MUSIC}. Finally, the derived asymptotic distribution of both cost functions is used to predict the probability of resolution of both subspace-based \ac{DoA} estimation methods. In comparison to \cite{a85}, we do not only analyze the asymptotic stochastic behavior of the g-\ac{MUSIC} cost function but also of the conventional \ac{MUSIC} cost function. Additionally, the detailed proofs for the second order asymptotic behavior of both cost functions are provided.
	
	The original contributions of this article can be summarized as follows:
	\begin{itemize}
		\item We derive a \ac{CLT} which states that eigenvector-based cost functions are asymptotically jointly Gaussian distributed for Gaussian distributed observations in the asymptotic regime where both sample size and array dimension go to infinity at the same rate. A general theorem is provided that completely specifies the asymptotic behavior of eigenvector-based cost functions in terms of (i) asymptotic deterministic behavior and (ii) fluctuations around this deterministic equivalent.
	    \item We determine a set of conditions that guarantee that both \ac{MUSIC} and g-\ac{MUSIC} cost functions fluctuate around their asymptotic deterministic equivalents in this asymptotic regime. 
		\item We particularize the above results to the \ac{MUSIC} and g-\ac{MUSIC} cost functions and derive a closed-form expression for the asymptotic probability of resolution of both \ac{DoA} estimation methods. These expressions can be used to determine the probability of resolving closely spaced sources for a given array geometry, a sample volume per antenna and a scenario configuration. 
	\end{itemize}
	
	The rest of the paper is organized as follows. Section \ref{sec: signal_model} introduces the signal model that is assumed in the paper, while some important \ac{RMT} fundamentals are then introduced in Section \ref{sec: rmt_preliminaries}. The conventional \ac{MUSIC} and g-\ac{MUSIC} \ac{DoA} estimation techniques are presented in Section \ref{sec: music_doa_estimation}. The asymptotic stochastic analysis of these two cost functions and the corresponding probability of resolution of the associated methods are given in Section \ref{sec: asymptotic_behavior_pr_dml_cost_function}. Section \ref{sec: proof_second_order} is devoted to the derivation of the asymptotic deterministic behavior of these two cost functions and the characterization of the corresponding fluctuations around it. Finally, these theoretical derivations are then validated by numerical experiments in Section \ref{sec: simulation_results} and Section \ref{sec: conclusion} concludes the paper.
	
	\paragraph*{Notation} Matrices are denoted by boldface uppercase letters $\boldsymbol{A}$, vectors are denoted by boldface lowercase letters $\boldsymbol{a}$, and scalars are denoted by regular letters $a$. Symbols $(\cdot)^\text{T}$, $(\cdot)^\text{H}$, $(\cdot)^{-1}$ and $(\cdot)^{1/2}$ denote the transpose, Hermitian transpose, inverse and the positive square root of the matrix argument, which is assumed positive semidefinite. The expectation operator is represented by $\mathbb{E}\lbrack \cdot \rbrack$. The trace operator is denoted by $\tr{\cdot}$, and the residue of a holomorphic complex function $f(\cdot)$ evaluated at $b$ is denoted by $\res{f(\cdot)}{b}$. $\| \cdot \|_{\text{F}}$ denotes the Frobenius norm, and $\| \cdot \|$ is the spectral norm of the matrix argument.

\section{Signal Model}
	\label{sec: signal_model}
	Consider a sensor array that is equipped with $M$ sensors and $K$ impinging narrowband signals with \acp{DoA} $\boldsymbol{\theta} = \lbrack \theta_{1}, \dots, \theta_{K}\rbrack^{\operatorname{T}}$ that lie within the field of view $\Theta$ of the array. The number of sources $K$ is assumed to be known and smaller than the number of sensors $K < M$. The full-rank steering matrix is given by $\boldsymbol{A}(\boldsymbol{\theta}) = \lbrack \boldsymbol{a}(\theta_{1}),\dots, \boldsymbol{a}(\theta_{K}) \rbrack^{\operatorname{T}} \in \mathbb{C}^{M \times K}$ where $\boldsymbol{a}(\theta_{i})\in \mathbb{C}^{M}$ denotes the steering vector associated to the $i$-th source, which is assumed to be located at $\theta_{i}$. Without loss of generality, we will assume that the array steering vector is normalized to have unit norm, that is $\|\boldsymbol{a}(\theta)\| = 1$. The received baseband signal $\boldsymbol{y}(n) = \lbrack y_{1}(n), \dots, y_{M}(n) \rbrack^{\operatorname{T}} \in \mathbb{C}^{M}$ at time instant $n$ is modeled as:
	\begin{equation}
		\boldsymbol{y}(n) = \boldsymbol{A}(\boldsymbol{\theta}) \boldsymbol{s}(n) + \boldsymbol{n}(n) \text{ for } n=1,\dots,N
		\label{eq: received_signal}
	\end{equation}
	where $\boldsymbol{s}(n)=\lbrack s_{1}(n),\dots,s_{K}(n)\rbrack^{\operatorname{T}} \in \mathbb{C}^{K}$ denotes the transmitted baseband source signal and $\boldsymbol{n}(n)$ represents the sensor noise. Assuming that both signal and noise vectors are statistically independent, zero-mean and circularly symmetric Gaussian vectors, the observation  $\boldsymbol{y}(n)$ in \eqref{eq: received_signal} can be modeled as a zero-mean circularly symmetric Gaussian vector with covariance matrix $\boldsymbol{R} \in \mathbb{C}^{M\times M}$ given by
	\begin{equation}
		\boldsymbol{R} = \mathbb{E} \left \lbrack \boldsymbol{y}(n)\boldsymbol{y}(n)^{\operatorname{H}}\right \rbrack = \boldsymbol{A}\boldsymbol{R}_{\text{s}}\boldsymbol{A}^{\operatorname{H}} + \sigma^{2}\boldsymbol{I}_{M}
		\label{eq: covariance_matrix_received_signal}
	\end{equation}
	where $\boldsymbol{R}_{\text{s}} = \mathbb{E}\lbrack \boldsymbol{s}(n)\boldsymbol{s}(n)^{\operatorname{H}}\rbrack \in \mathbb{C}^{K\times K}$ is the covariance matrix of the transmitted source signal $\boldsymbol{s}(n)$ and $\sigma^{2}\boldsymbol{I}_{M}$ denotes the noise covariance. Let us consider the eigendecomposition of this covariance matrix, which can be expressed as
	\begin{equation}
		\boldsymbol{R} = \sum_{m=1}^{\bar{M}} \gamma_{m}\boldsymbol{E}_{m}\boldsymbol{E}_{m}^{\operatorname{H}} = \boldsymbol{E} \begin{bmatrix}
			\gamma_{1}\boldsymbol{I}_{K_{1}} & & \\ & \ddots & \\ & & \gamma_{\bar{M}}\boldsymbol{I}_{K_{\bar{M}}}
		\end{bmatrix} \boldsymbol{E}^{\operatorname{H}}.
		\label{eq: eigenvalue_decomposition_distinct}
	\end{equation}
	Here, $\bar{M} \leq M$ denotes the total number of distinct eigenvalues, which are sorted in ascending order as $\gamma_{1} < \gamma_{2} < \dots < \gamma_{\bar{M}}$, and $K_{m}$ denotes the multiplicity of $\gamma_{m}$, $m=1,\dots,\bar{M}$. The eigenvectors associated to $\gamma_{m}$ are grouped into an $M \times K_m$ matrix $\boldsymbol{E}_{m}$ of orthogonal columns that span the corresponding subspace and we let $\boldsymbol{E} = \lbrack \boldsymbol{E}_{1}, \dots, \boldsymbol{E}_{\bar{M}} \rbrack \in \mathbb{C}^{M \times M}$. 
	
    We also consider here the sample covariance matrix 
	\begin{equation}
		\hat{\boldsymbol{R}} = \frac{1}{N} \sum_{n=1}^{N} \boldsymbol{y}(n)\boldsymbol{y}(n)^{\operatorname{H}} = \frac{1}{N} \boldsymbol{Y}\boldsymbol{Y}^{\operatorname{H}}
		\label{eq: sample_covariance_matrix_received_signal}
	\end{equation} 
	which has eigendecomposition given by 
	\begin{equation}
		\hat{\boldsymbol{R}} = \sum_{m=1}^{M} \hat{\lambda}_{m} \hat{\boldsymbol{e}}_{m} \hat{\boldsymbol{e}}_{m}^{\operatorname{H}} = \hat{\boldsymbol{E}} \begin{bmatrix}
			\hat{\lambda}_{1} & & \\ & \ddots & \\ & & \hat{\lambda}_{M}
		\end{bmatrix} \hat{\boldsymbol{E}}^{\operatorname{H}}
		\label{eq: sample_covariance_matrix_received_signal_eigendecomposition}
	\end{equation}
	where now $\hat{\lambda}_{1} \leq \hat{\lambda}_{2} \leq \dots \leq \hat{\lambda}_{M}$ are the sample eigenvalues, $\hat{\boldsymbol{e}}_{m}$ denotes the eigenvector associated to $\hat{\lambda}_{m}$, $m=1,\dots,M$ and $\hat{\boldsymbol{E}}= \lbrack \hat{\boldsymbol{e}}_{1},\dots,\hat{\boldsymbol{e}}_{M}\rbrack$. These sample eigenvalues are almost surely different, unless $N<M$, in which case we have a zero sample eigenvalue of multiplicity $M-N$. In this situation, $\hat{\boldsymbol{e}}_{m}$ for $m=1,\ldots,M-N$, span the subspace associated to the zero sample eigenvalue. In the following section, we provide some interesting properties of the asymptotic behavior of sample eigenvalues and eigenvectors based on \ac{RMT} results that will be extensively used throughout the paper. 

\section{\acl{RMT} Preliminaries} 
	\label{sec: rmt_preliminaries}
	Under the above statistical assumptions, the sample covariance matrix $\hat{\boldsymbol{R}}$ in \eqref{eq: sample_covariance_matrix_received_signal} is a consistent estimator of the true one  $\boldsymbol{R}$, provided that the number of antennas $M$ is kept fixed while the number of samples grows without bound, $N \to \infty$. More formally, one can easily show that $\| \hat{\boldsymbol{R}} - \boldsymbol{R}\| \rightarrow 0$ with probability one under these asymptotic assumptions. As pointed out before, this asymptotic regime is often inappropriate in the sense that both $M,N$ are typically comparable in magnitude. In this paper, we will therefore consider an asymptotic regime where these two quantities tend to infinity at the same rate.
	\begin{assumption}
		\label{as:asymptotic_growth}
		The number of samples $N$ is a function of the number of antennas $M$, that is $N = N(M)$ and $N(M) \rightarrow \infty$ as $M \rightarrow \infty$ in a way that $M/N(M) \rightarrow c$ for some constant $0 < c < \infty$.
	\end{assumption}
	It turns out that under Assumption \ref{as:asymptotic_growth} the sample covariance matrix $\hat{\boldsymbol{R}}$ in \eqref{eq: sample_covariance_matrix_received_signal} is not a consistent estimate of the true one in \eqref{eq: covariance_matrix_received_signal}, in the sense that $\|\boldsymbol{R} - \hat{\boldsymbol{R}}\| \nrightarrow 0$. In particular, this shows that the eigenvalues and eigenvectors of $\hat{\boldsymbol{R}}$ do not really converge to the eigenvalues and eigenvectors of $\boldsymbol{R}$ when the dimensions of these matrices increase without bound. However, it is well known that, under some additional assumptions, the empirical eigenvalue distribution of the sample covariance matrix $\hat{\boldsymbol{R}}$ in \eqref{eq: sample_covariance_matrix_received_signal} still shows a deterministic behavior in the large-dimensional regime. In order to formalize this observation, we need to introduce some additional technical assumptions. 
	
	\begin{assumption}
		\label{as: rewritten_sample_covariance_matrix}
		The observations $\boldsymbol{y}(n)$ in \eqref{eq: received_signal}, $n =  1,\ldots,N$ form a collection of independent circularly symmetric complex Gaussian vectors with zero-mean and covariance matrix $\boldsymbol{R}$, bounded in spectral norm. In particular, the quantities $\bar{M}$, $\gamma_ {1},\ldots,\gamma_{\bar{M}}$ and $K_{1},\ldots,K_{\bar{M}}$ corresponding to the eigendecomposition of $\boldsymbol{R}$ in \eqref{eq: eigenvalue_decomposition_distinct} may vary with $M$, but $\sup_{M} \gamma_{\bar{M}}<\infty$.
	\end{assumption}
	
	We remark that the received signal $\boldsymbol{y}(n)$ for $n=1,\dots,N$ in \eqref{eq: received_signal} and the covariance matrix $\boldsymbol{R}$ in \eqref{eq: covariance_matrix_received_signal} satisfy Assumption \ref{as: rewritten_sample_covariance_matrix} by default as long as the \ac{SNR} is bounded. Under the above technical assumptions, the eigenvalues of the sample covariance matrix $\hat{\boldsymbol{R}}$ in \eqref{eq: sample_covariance_matrix_received_signal} are asymptotically almost surely distributed as a non-random measure with density $q_{M}(x)$ \cite{a41,a42,a43,a20}. Informally stated, the histogram of the eigenvalues of the sample covariance matrix tends to be shaped around $q_{M}(x)$ as $M,N$ grow large, with probability one. This deterministic density is therefore the key to understanding the asymptotic behavior of the sample covariance matrix. In particular, one can show that, when $N > M$, the density $q_M(x)$ has compact support consisting of the union of $S$ closed intervals, namely $\mathcal{S} = \lbrack x_{1}^{-}, x_{1}^{+}\rbrack \cup \dots \cup \lbrack x_{S}^{-}, x_{S}^{+} \rbrack$ \cite{a43,a44,a45}. When $N \leq M$, the same description is valid but with the addition of the zero eigenvalue, i.e. $\{0\}$. The procedure to obtain $\mathcal{S}$ is as follows (see \cite[Proposition 1]{a20} and also \cite{a44}). Consider the following function of the true covariance matrix
	 \begin{equation}
	    \label{eq: function_support}
	     \Psi(\omega) = \frac{1}{N} \tr{\boldsymbol{R}^2 \left( \boldsymbol{R} - \omega \boldsymbol{I}_M \right)^{-2}}.
	 \end{equation}
	 The polynomial equation $ \Psi(\omega) = 1$ has $2S$ solutions counting multiplicities, which can be denoted as  $\{\omega_{1}^{-}, \omega_{1}^{+}, \ldots, \omega_{S}^{-},\omega_{S}^{+}\}$. We then define $x_{s}^{\pm} = z(\omega_{s}^{\pm})$, $s = 1, \ldots, S$, where $z(\omega)$ is the transformation
	 \begin{equation}
	 \label{eq: mapping_omega_to_z}
	     z(\omega) = \omega \left( 1- \frac{1}{N} \tr{\boldsymbol{R}\left( \boldsymbol{R} - \omega \boldsymbol{I}_M \right)^{-1}} \right).
	 \end{equation}
	 Each eigenvalue of $\boldsymbol{R}$ can be univocally associated to one of the $S$ intervals, in the sense that there exists a single interval $\lbrack \omega_s^{-}, \omega_s^{+} \rbrack$ that contains that particular eigenvalue. On the other hand, given a certain covariance matrix $\boldsymbol{R}$, the number of intervals of the support $S$ increases with increasing $N$. Furthermore, there exists a minimum number of samples per antenna that guarantees that a certain interval $\lbrack x_{s}^-, x_{s}^+ \rbrack$ is associated to a single eigenvalue of $\boldsymbol{R}$. In this paper, we will strongly rely on the assumption that the lowest eigenvalue of $\boldsymbol{R}$ is the only eigenvalue that belongs to the interval $\lbrack \omega_{1}^-, \omega_{1}^+ \rbrack$ (see Assumption \ref{as: subspace_separation} below). This will allow us to analyze the behavior of subspace \ac{DoA} detection techniques in large dimensional arrays.
	
	Having reviewed some basic notions on the asymptotic spectral behavior of the sample covariance matrix, we are now in the position of introducing the \ac{MUSIC} and g-\ac{MUSIC} subspace \ac{DoA} estimators in the large antenna regime. 

\section{\acs{MUSIC} and g-\acs{MUSIC} \acl{DoA} Estimation} 
	\label{sec: music_doa_estimation}
	
	The main idea behind the conventional \ac{MUSIC} estimator is to exploit the fact that the eigenvectors associated to the noise subspace $\boldsymbol{E}_{1}$ of the true covariance matrix $\boldsymbol{R}$ in \eqref{eq: eigenvalue_decomposition_distinct} are orthogonal to the steering vectors evaluated at the true \acp{DoA} of the received signals. Hence, we can consider a cost function 
	\begin{equation}
		\bar{\eta}_{\text{g}}(\theta) = \boldsymbol{a}(\theta)^{\operatorname{H}} \boldsymbol{E}_{1} \boldsymbol{E}_{1}^{\operatorname{H}} \boldsymbol{a}(\theta).
		\label{eq: g_music_first_order}
	\end{equation}	
	where we recall from \eqref{eq: eigenvalue_decomposition_distinct} that $\boldsymbol{E}_{1}$ contains the $M-K$ eigenvectors associated to the smallest eigenvalue of $\boldsymbol{R}$. The \acp{DoA} of the received signals can be determined as the $K$ distinct values of $\theta$ at which $\bar{\eta}_{\text{g}}(\theta) = 0$ \cite{a69}. Since the noise subspace $\boldsymbol{E}_{1}$ is unknown in practice, the conventional \ac{MUSIC} cost function is obtained by replacing the noise subspace $\boldsymbol{E}_{1}$ in \eqref{eq: g_music_first_order} with the noise eigenvectors of the sample covariance matrix, namely
	\begin{equation}
		\hat{\eta}_{\text{c}}(\theta) = \boldsymbol{a}(\theta)^{\operatorname{H}} \sum_{m=1}^{M-K} \hat{\boldsymbol{e}}_{m}\hat{\boldsymbol{e}}_{m}^{\operatorname{H}} \boldsymbol{a}(\theta)
		\label{eq: music_cost_function}
	\end{equation}
	and the \acp{DoA} are determined by the $K$ distinct values in $\theta$ where $\hat{\eta}_{\text{c}}(\theta)$ in \eqref{eq: music_cost_function} attains its $K$ deepest local minima \cite{a1}.

	Now, as a consequence of the fact that under Assumption \ref{as: rewritten_sample_covariance_matrix} $\|\boldsymbol{R} - \hat{\boldsymbol{R}}\| \nrightarrow 0$ one can generally expect that $|\hat{\eta}_{\text{c}}(\theta) - \eta_{\text{g}}(\theta)| \nrightarrow 0$. With the help of \ac{RMT} tools, one can however find a modification of the \ac{MUSIC} cost function in (\ref{eq: music_cost_function}) that is indeed consistent in this large dimensional regime. This is usually referred to as $M,N$-consistency, as opposed to the more conventional concept of $N$-consistency, which assumes a constant $M$. The modified cost function is usually referred to as g-\ac{MUSIC} \cite{a69} and can be built from a proper combination of signal and noise subspaces, that is 
	\begin{equation}
		\hat{\eta}_{\text{g}}(\theta) = \sum_{m=1}^{M} \phi(m) \boldsymbol{a}(\theta)^{\operatorname{H}}\hat{\boldsymbol{e}}_{m}\hat{\boldsymbol{e}}_{m}^{\operatorname{H}}\boldsymbol{a}(\theta)
		\label{eq: g_music_cost_function}
	\end{equation}
	with real-valued weights
	\begin{equation*}
		\phi(m) = \begin{cases}
			1+ \sum_{k=M-K+1}^{M} \left( \frac{\hat{\lambda}_{k}}{\hat{\lambda}_{m} - \hat{\lambda}_{k}} - \frac{\hat{\mu}_{k}}{\hat{\lambda}_{m}-\hat{\mu}_{k}}\right), & m \leq M-K\\
			-\sum_{k=1}^{M-K} \left( \frac{\hat{\lambda}_{k}}{\hat{\lambda}_{m} - \hat{\lambda}_{k}}-\frac{\hat{\mu}_{k}}{\hat{\lambda}_{m}-\hat{\mu}_{k}}\right), & m > M-K
		\end{cases}
	\end{equation*}
	where $\hat{\mu}_{1} \leq \hat{\mu}_{2} \leq \dots \leq \hat{\mu}_{M}$ are the real-valued solutions (counting multiplicities) to the following equation in $\hat{\mu}$ 
	\begin{equation*}
		\frac{1}{N}\sum_{k=1}^{M} \frac{\hat{\lambda}_{k}}{\hat{\lambda}_{k}-\hat{\mu}} = 1.
	\end{equation*}
	The \acp{DoA} of the g-\ac{MUSIC} estimator are determined by searching for the $K$ deepest local minima of the cost function $\hat{\eta}_{\text{g}}(\theta) $ in \eqref{eq: g_music_cost_function}. In order to justify the superiority of this \ac{DoA} estimation algorithm, we need to impose separation between the signal and noise subspaces in the asymptotic covariance $\boldsymbol{R}$. This is more formally stated in the following assumption. 
	
	\begin{assumption}
		\label{as: subspace_separation}
	We have $0 < \inf_M K_1/M \leq \sup_M K_1/M < 1 $ and $\inf_M \gamma_1 > 0 $. The eigenvalue $\gamma_{1}$ is the unique eigenvalue that is associated to the cluster with support $\lbrack x_{1}^{-}, x_{1}^{+} \rbrack$ for all $M,N$ sufficiently large. Furthermore, there exists a deterministic $\varrho$ and some small $\epsilon>0$, both independent of $M$, such that ${\sup}_{M} x_{1}^{+} +\epsilon < \varrho <{\inf}_{M} x_{2}^{-} -\epsilon$. 
	\end{assumption}
	The first part of Assumption \ref{as: subspace_separation} guarantees that the noise subspace $\boldsymbol{E}_{1}$ does not vanish in the large dimensional regime when the number of antennas grows to infinity. Hence, the number of sources $K$ is allowed to increase with $M$ but in a way that $\limsup_{M} K/M <1$. Furthermore, the noise power ($\gamma_1 = \sigma^2$) can also vary with the number of antennas, as long as it does not vanish with $M \to \infty$. The second part of the assumption ensures that the eigenvalue cluster associated with the noise eigenvalue $\gamma_{1}$, denoted by $[x_{1}^{-} , x_{1}^{+} ]$ is separated from the clusters of adjacent eigenvalues in the asymptotic eigenvalue distribution of the sample covariance matrix $\hat{\boldsymbol{R}}$. The existence of separation in the asymptotic eigenvalue distribution between the cluster associated to the smallest eigenvalue $\gamma_{1}$ and the clusters of adjacent eigenvalues depends on the \acp{DoA}, the \ac{SNR}, as well as the number of snapshots $N$ and the number of sensors $M$ and can be verified using the procedure in \cite[Section II]{a20} or \cite[Section II-A]{a43}. With the aid of this separability assumption, we are now ready to describe the asymptotic behavior of both subspace-based cost functions.
	
		\begin{theorem}
			\label{thm: first_order}
			Under Assumptions\footnote{Gaussianity is not necessary for this result, and can be replaced by a milder condition on the fourth order moments of the observations.} \ref{as:asymptotic_growth}-\ref{as: subspace_separation} and for each $\theta \in \Theta$ 
			\begin{align}
				\left | \hat{\eta}_{\operatorname{c}}(\theta) - \bar{\eta}_{\operatorname{c}}(\theta) \right | & \rightarrow 0 \\
				\left | \hat{\eta}_{\operatorname{g}}(\theta) - \bar{\eta}_{\operatorname{g}}(\theta) \right | & \rightarrow 0
			\end{align}
			almost surely, where $\bar{\eta}_{\operatorname{c}}(\theta)$ and $\bar{\eta}_{\operatorname{g}}(\theta)$ are two deterministic equivalent objective functions defined as follows. The deterministic equivalent of the \ac{MUSIC} cost function is defined as
			\begin{equation}
				\bar{\eta}_{\operatorname{c}}(\theta) = \sum_{m=1}^{\bar{M}} \psi(m) \boldsymbol{a}(\theta)^{\operatorname{H}} \boldsymbol{E}_{m} \boldsymbol{E}_{m}^{\operatorname{H}} \boldsymbol{a}(\theta)
				\label{eq: music_first_order}
			\end{equation}
			with real-valued weights
			\begin{equation*}
				\psi(m) = \begin{cases}
					1- \frac{1}{K_{1}} \sum_{r=2}^{\bar{M}} K_{r} \left( \frac{\gamma_{1}}{\gamma_{r}-\gamma_{1}} - \frac{\mu_{1}}{\gamma_{r}-\mu_{1}}\right), & m=1\\
					\frac{\gamma_{1}}{\gamma_{m}-\gamma_{1}}- \frac{\mu_{1}}{\gamma_{m}-\mu_{1}}, & m\neq 1 
				\end{cases}
			\end{equation*}
			where $\mu_{1} < \mu_{2} < \dots < \mu_{\bar{M}}$ are the real-valued solutions to the following equation in $\mu$
			\begin{equation}
				\frac{1}{N}\sum_{r=1}^{\bar{M}} \frac{ K_{r}\gamma_{r}}{\gamma_{r}-\mu} = 1.
				\label{eq: mu_values}
			\end{equation}
			The deterministic equivalent $\bar{\eta}_{\operatorname{g}}(\theta)$  of the g-\ac{MUSIC} cost function $\hat{\eta}_{\operatorname{g}}(\theta)$ in \eqref{eq: g_music_cost_function} is defined in \eqref{eq: g_music_first_order}.
		\end{theorem}
		\begin{proof}
			See \cite[Theorem 2]{a43}. A sketch of the proof is provided in Section \ref{sec: proof_second_order}.
		\end{proof}
		
	\paragraph*{Remark} The above theorem points out that only the g-\ac{MUSIC} cost function is an $M,N$-consistent estimator of the originally intended cost function in \eqref{eq: g_music_first_order}. However, this does not need to have a direct translation into the consistency of the \ac{DoA} estimates themselves.  It was shown in \cite{a74} that for widely spaced sources and a \ac{ULA}, both the conventional \ac{MUSIC} and the g-\ac{MUSIC} \ac{DoA} estimators provide $M,N$-consistent \ac{DoA} estimates in spite of the inherent inconsistency of the conventional \ac{MUSIC} cost function. Hence, under certain circumstances, the position of the local minima of the conventional \ac{MUSIC} cost function converge to the true \acp{DoA}, although the global cost function does not. Also, for closely spaced sources, both \ac{MUSIC} methods provide $M,N$-consistent \ac{DoA} estimates. However the g-\ac{MUSIC} method provides $M,N$-consistent \ac{DoA} estimates under lower asymptotic conditions on $N$, which explains the superiority in \ac{DoA} estimation accuracy of g-\ac{MUSIC} over conventional \ac{MUSIC} in scenarios with closely spaced sources and limited sample size.
	
	In order to analyze the probability of resolution of both \ac{MUSIC} methods, we next focus on the characterization of the asymptotic fluctuations of both cost functions around the asymptotic deterministic equivalents in Theorem \ref{thm: first_order}.
	
\section{Main Result: Asymptotic Fluctuations of the \acs{MUSIC} and g-\acs{MUSIC} Cost Function}
	\label{sec: asymptotic_behavior_pr_dml_cost_function}
	The objective of this section is to characterize the asymptotic fluctuations of both conventional \ac{MUSIC} and g-\ac{MUSIC} cost functions in \eqref{eq: music_cost_function} and \eqref{eq: g_music_cost_function} around their asymptotic deterministic equivalents. To that effect, we will describe the finite-dimensional asymptotic distribution of these cost functions evaluated at a constant number of $L$ points. Let us therefore consider a fixed set of $L$ directions given by $\bar{\boldsymbol{\theta}} = \lbrack \bar{\theta}_{1},\dots, \bar{\theta}_{L} \rbrack^{\operatorname{T}}$ within the field of view $\Theta$ of the sensor array and denote 
	\begin{align}
		\hat{\boldsymbol{\eta}}_{\text{c}}(\bar{\boldsymbol{\theta}}) =& \lbrack \hat{\eta}_{\text{c}}(\bar{\theta}_{1}), \dots, \hat{\eta}_{\text{c}}(\bar{\theta}_{L})\rbrack^{\operatorname{T}} \label{eq: music_cost_function_vector}\\
		\hat{\boldsymbol{\eta}}_{\text{g}}(\bar{\boldsymbol{\theta}}) =& \lbrack \hat{\eta}_{\text{g}}(\bar{\theta}_{1}), \dots, \hat{\eta}_{\text{g}}(\bar{\theta}_{L})\rbrack^{\operatorname{T}}, \label{eq: g_music_cost_function_vector} 
	\end{align}
	where $\hat{\eta}_{\text{c}}(\theta)$ and $\hat{\eta}_{\text{g}}(\theta)$ are given in \eqref{eq: music_cost_function} and \eqref{eq: g_music_cost_function}, respectively. In order to investigate the asymptotic behavior of the two vectors $\hat{\boldsymbol{\eta}}_{\text{c}}(\bar{\boldsymbol{\theta}})$ and $\hat{\boldsymbol{\eta}}_{\text{g}}(\bar{\boldsymbol{\theta}})$ we will also consider the two $L$-dimensional vectors $\bar{\boldsymbol{\eta}}_{\text{c}}(\bar{\boldsymbol{\theta}})$ and $\bar{\boldsymbol{\eta}}_{\text{g}}(\bar{\boldsymbol{\theta}})$ that contain the corresponding deterministic equivalents. These two vectors are respectively defined as \eqref{eq: music_cost_function_vector}-\eqref{eq: g_music_cost_function_vector} by replacing $\hat{\eta}_{\text{c}}(\bar{\theta}_{l})$ and $\hat{\eta}_{\text{g}}(\bar{\theta}_{l})$ with $\bar{\eta}_{\text{c}}(\bar{\theta}_{l})$ in \eqref{eq: music_first_order} and $\bar{\eta}_{\text{g}}(\bar{\theta}_{l})$ in \eqref{eq: g_music_first_order}, where $l=1,\ldots,L$. 
 
	In order to introduce the main result of this section, we need to introduce the asymptotic covariance matrices of the vectors in  \eqref{eq: music_cost_function_vector}-\eqref{eq: g_music_cost_function_vector}. Let $\omega(z)$ be defined by inverting the mapping in (\ref{eq: mapping_omega_to_z}) as follows. When $z \in \mathbb{C}^+ \doteq \{ z \in \mathbb{C}: \, \ima{z}>0 \}$, $\omega(z)$ is defined as the unique solution to 
	\begin{equation}
		\label{eq: definition_omega_z}
			z = \omega(z) \Biggl( 1- \frac{1}{N}\sum_{r=1}^{\bar{M}}\frac{K_{r}\gamma_{r}}{\gamma_{r}-\omega(z)}\Biggr) 
	\end{equation}
	in $\mathbb{C}^+$. When $z^\ast \in \mathbb{C}^+$, we take $\omega(z) = \omega^\ast(z^\ast)$. Finally, when $z \in \mathbb{R}$, we take $\omega(z)$ to be the unique solution to the above equation such that $\frac{1}{N}\sum_{r=1}^{\bar{M}} \frac{K_{r}\gamma_{r}^{2}}{|\gamma_{r}-\omega(z)|^{2}} \leq 1$. It can be shown that $\omega(z)$ is well defined on all $\mathbb{C}$ and holomorphic on $\mathbb{C}\backslash \mathcal{S}$, with derivative 
		\begin{equation}
			\omega^{\prime}(z) = \frac{\partial \omega(z)}{\partial z} = \Biggl( 1- \frac{1}{N}\sum_{r=1}^{\bar{M}} \frac{K_{r}\gamma_{r}^{2}}{(\gamma_{r}-\omega(z))^{2}}\Biggr)^{-1}.
			\label{eq: first_order_derivative_wrt_x_omega_x}
		\end{equation}
	
    Using tools from \ac{RMT} it is shown in Section \ref{sec: proof_second_order} that both the conventional \ac{MUSIC} and the g-\ac{MUSIC} cost functions fluctuate around their asymptotic equivalents as Gaussian random vectors. In order to formulate this result, we now define the asymptotic covariance matrices of the corresponding random vectors in \eqref{eq: music_cost_function_vector}-\eqref{eq: g_music_cost_function_vector} after proper centering and normalization. Regarding the conventional \ac{MUSIC} algorithm, we define 
    \begin{equation}
			\begin{aligned}
				\left \lbrack \boldsymbol{\Gamma}_{\operatorname{c}}(\bar{\boldsymbol{\theta}}) \right \rbrack_{p,q} = \sum_{r=1}^{\bar{M}} \sum_{k=1}^{\bar{M}} \Big(  \xi_{\operatorname{c}}(r,k) \boldsymbol{a}(\bar{\theta}_{p})^{\operatorname{H}} \boldsymbol{E}_{r}\boldsymbol{E}_{r}^{\operatorname{H}} \boldsymbol{a}(\bar{\theta}_{q}) &\\[-8pt]
				\times \boldsymbol{a}(\bar{\theta}_{q})^{\operatorname{H}} \boldsymbol{E}_{k}\boldsymbol{E}_{k}^{\operatorname{H}} \boldsymbol{a}(\bar{\theta}_{p})& \Big)
			\end{aligned}
			\label{eq: music_second_order}
		\end{equation}	
	for $p,q \in \lbrace 1, \dots, L\rbrace$ where the real-valued weights
		\begin{equation}
			\xi_{\operatorname{c}}(r,k) = \begin{cases}
				\gamma_{r} \gamma_{k} \tilde{\xi}_{\operatorname{c}}(r,k) & M > N \\
				\gamma_{r} \gamma_{k} \bar{\xi}_{\operatorname{c}}(r,k) & M \leq N
			\end{cases}
			\label{eq: music_second_order_xi}
		\end{equation}
		\setcounter{music_second_order_equation_number}{\value{equation}}%
	differ depending on the number of snapshots $N$ and the number of sensors $M$. In the oversampled case ($M\leq N$) the weights are given by $\bar{\xi}_{\operatorname{c}}(r,k)$ in \eqref{eq: music_second_order_xi_oversampled} whereas in the undersampled case ($M>N$) the weights are given by $\tilde{\xi}_{\operatorname{c}}(r,k)$ in \eqref{eq: music_second_order_xi_undersampled}, both at the top of the next page. The real-valued quantities $\mu_{r}$ for $r=1,\dots,\bar{M}$ are defined as in \eqref{eq: mu_values}. We remark that $\mu_{1}=0$ if $M=N$.	 
\addtocounter{equation}{2}%
%
		\begin{figure*}[!t]
			\normalsize
			\setcounter{MYtempeqncnt1}{\value{equation}}
			\setcounter{equation}{\value{music_second_order_equation_number}}
			\begin{align}	
				\bar{\xi}_{\operatorname{c}}(r,k) =& -\frac{N}{K_{1}} \frac{1}{\left(\frac{1}{N}K_{1}\gamma_{1}\right)^{2}} \Biggl( 1- \frac{1}{N} \sum_{m=2}^{\bar{M}} \frac{K_{m}\gamma_{m}}{\gamma_{m}-\gamma_{1}} \Biggr)^{2} \delta_{r=k=1} - \frac{1}{\left(\frac{1}{N} K_{1}\gamma_{1}\right)^{2}} \Biggl(1- \frac{1}{N} \sum_{m=2}^{\bar{M}} \frac{K_{m}\gamma_{m}^{2}}{\left(\gamma_{m}-\gamma_{1}\right)^{2}}\Biggr) \delta_{r=k=1} \label{eq: music_second_order_xi_oversampled} \\[-2pt]
				& - \frac{N}{K_{1}} \frac{1}{\left(\gamma_{k}-\gamma_{1}\right)^{2}} \delta_{r=1\neq k} - \frac{N}{K_{1}}\frac{1}{\left(\gamma_{r}-\gamma_{1}\right)^{2}} \delta_{k=1\neq r} + \frac{2}{\pi} \int_{x_{1}^{-}}^{x_{1}^{+}} \frac{1}{\left | 1- \frac{1}{N} \sum_{m=1}^{\bar{M}}\frac{K_{m}\gamma_{m}}{\gamma_{m}-\omega(x)} \right |^{2}} \frac{\ima{\omega(x)}}{\left| \gamma_{r}-\omega(x) \right |^{2} \left | \gamma_{k}-\omega(x) \right |^{2}} \mathrm{d}x \nonumber \\[-2pt]
				& + 2 \frac{\mu_{1}}{\gamma_{1}}\frac{1}{\left( \frac{1}{N}K_{1}\gamma_{1} \right)^{2}} \Biggl( 1- \frac{1}{N} \sum_{m=2}^{\bar{M}} \frac{K_{m}\gamma_{m}^{2}}{\left(\gamma_{m}-\gamma_{1}\right)\left(\gamma_{m}-\mu_{1}\right)}\Biggr) \delta_{r=k=1} + 2 \frac{\mu_{1}}{\gamma_{1}-\mu_{1}} \frac{1}{\left(\frac{1}{N}K_{1}\gamma_{1}\right)^{2}} \Biggl( 1- \frac{1}{N} \sum_{m=2}^{\bar{M}} \frac{K_{m}\gamma_{m}}{\gamma_{m}-\gamma_{1}}\Biggr) \delta_{r=k=1} \nonumber \\[-2pt]
				& + 2 \frac{N}{K_{1}} \frac{1}{\left(\gamma_{k}-\gamma_{1}\right)\gamma_{1}} \frac{\mu_{1}}{\gamma_{k}-\mu_{1}} \delta_{r=1\neq k} + 2 \frac{N}{K_{1}} \frac{1}{\left(\gamma_{r}-\gamma_{1}\right)\gamma_{1}} \frac{\mu_{1}}{\gamma_{r}-\mu_{1}} \delta_{k=1\neq r}  + \frac{\mu_{1}^{2}}{\left(\gamma_{r}-\mu_{1}\right)^{2}\left(\gamma_{k}-\mu_{1}\right)^{2}} \frac{1}{1- \frac{1}{N}\sum_{m=1}^{\bar{M}} \frac{K_{m}\gamma_{m}^{2}}{\left(\gamma_{m}-\mu_{1}\right)^{2}}} \nonumber  \\[-2pt]					
				\tilde{\xi}_{\operatorname{c}}(r,k) =&  \bar{\xi}_{\operatorname{c}}(r,k) -\frac{1}{\gamma_{r}\gamma_{k}}\frac{\mu_{1}}{\left(\gamma_{r}-\mu_{1}\right)\left(\gamma_{k}-\mu_{1}\right)} \frac{1}{\frac{1}{N} \sum_{n=1}^{\bar{M}} \frac{K_{n}}{\gamma_{n}-\mu_{1}}} - \frac{1}{\gamma_{r}\gamma_{k}}\frac{\mu_{1}^{2}}{\left(\gamma_{r}-\mu_{1}\right)\left(\gamma_{k}-\mu_{1}\right)} \frac{N}{N-M}  \label{eq: music_second_order_xi_undersampled}
			\end{align}
			\hrulefill
			\setcounter{equation}{\value{MYtempeqncnt1}}
		\end{figure*}
	
	Regarding the g-\ac{MUSIC} cost function, we define
		\begin{equation}
			\begin{aligned}		
				\left \lbrack \boldsymbol{\Gamma}_{\operatorname{g}}(\bar{\boldsymbol{\theta}})\right \rbrack_{p,q} = \sum_{r=1}^{\bar{M}}\sum_{k=1}^{\bar{M}}  \Big( \xi_{\operatorname{g}}(r,k) \boldsymbol{a}(\bar{\theta}_{p})^{\operatorname{H}} \boldsymbol{E}_{r} \boldsymbol{E}_{r}^{\operatorname{H}} \boldsymbol{a}(\bar{\theta}_{q}) &\\[-8pt]
				 \times \boldsymbol{a}(\bar{\theta}_{q})^{\operatorname{H}} \boldsymbol{E}_{k}\boldsymbol{E}_{k}^{\operatorname{H}} \boldsymbol{a}(\bar{\theta}_{p})& \Big)
			\end{aligned}
			\label{eq: g_music_second_order}
		\end{equation}
		for $p,q \in \lbrace 1, \dots, L\rbrace$ where the real-valued weights $\xi_{\operatorname{g}}(r,k)$ are given by
		\begin{equation}
			\xi_{\operatorname{g}}(r,k) = - \frac{N}{K_{1}}\delta_{r=k=1} + \frac{2}{\pi} \int_{x_{1}^{-}}^{x_{1}^{+}} \frac{ \gamma_{r} \gamma_{k} |\omega^{\prime}(x)|^{2} \ima{\omega(x)}}{|\gamma_{r}-\omega(x)|^{2} |\gamma_{k}-\omega(x)|^{2}}\mathrm{d}x.
			\label{eq: g_music_second_order_xi}
		\end{equation}
	
	Having introduced these two asymptotic covariance matrices, we next introduce an additional assumption that essentially guarantees that the eigenvalues of the matrices $\boldsymbol{\Gamma}_{\operatorname{c}}(\bar{\boldsymbol{\theta}})$ in \eqref{eq: music_second_order} and $\boldsymbol{\Gamma}_{\operatorname{g}}(\bar{\boldsymbol{\theta}})$ in \eqref{eq: g_music_second_order} are contained in a compact interval of the positive real axis independent of $M$. This is necessary in order to guarantee that the two random cost functions asymptotically fluctuate around the corresponding deterministic equivalents. 
		
		\begin{assumption}
			\label{as: positive_semidefinitness_covariance_matrix}
		    Let $\boldsymbol{A}(\bar{\boldsymbol{\theta}})$ denote the $M \times L$ matrix that contains, stacked side by side, the steering vectors $\boldsymbol{a}(\theta)$ evaluated at the directions $\bar{\theta}_1, \ldots, \bar{\theta}_L$. Then, if $\lambda_{\min}(\cdot)$ denotes the minimum eigenvalue of a matrix, we have 
		    \begin{equation*}
		        \inf_M
		        \sum_{m=2}^{\bar{M}} \lambda_{\min}^2\left( \boldsymbol{A}^{\operatorname{H}}(\bar{\boldsymbol{\theta}})
		        \boldsymbol{E}_m \boldsymbol{E}_m^{\operatorname{H}}
		        \boldsymbol{A}(\bar{\boldsymbol{\theta}}) \right) >0.
		    \end{equation*}
	    \end{assumption}
        This assumption is essentially pointing out that the matrix $\boldsymbol{A}(\bar{\boldsymbol{\theta}})$ must be full-rank, and its projection onto the different signal-subspaces $\boldsymbol{E}_{2},\dots,\boldsymbol{E}_{\bar{M}}$ of the true covariance matrix $\boldsymbol{R}$ cannot vanish uniformly. Having introduced this last assumption, we are now in the position to introduce the main result of this paper.  
		\begin{theorem}
			\label{thm: second_order}
			Under Assumptions \ref{as:asymptotic_growth}-\ref{as: positive_semidefinitness_covariance_matrix}, we have
			\begin{align}
				\sqrt{N}\boldsymbol{\Gamma}_{\operatorname{c}}(\bar{\boldsymbol{\theta}})^{-1/2}&\left( \hat{\boldsymbol{\eta}}_{\operatorname{c}}(\bar{\boldsymbol{\theta}}) - \bar{\boldsymbol{\eta}}_{\operatorname{c}}(\bar{\boldsymbol{\theta}}) \right) \overset{\mathcal{D}}{\rightarrow} \mathcal{N}(\boldsymbol{0}, \boldsymbol{I}_{L}) \label{eq: music_convergence_in_distribution} \\
				\sqrt{N}\boldsymbol{\Gamma}_{\operatorname{g}}(\bar{\boldsymbol{\theta}})^{-1/2}&\left( \hat{\boldsymbol{\eta}}_{\operatorname{g}}(\bar{\boldsymbol{\theta}}) - \bar{\boldsymbol{\eta}}_{\operatorname{g}}(\bar{\boldsymbol{\theta}}) \right) \overset{\mathcal{D}}{\rightarrow} \mathcal{N}(\boldsymbol{0}, \boldsymbol{I}_{L}) \label{eq: g_music_convergence_in_distribution}
			\end{align}
			where $\overset{\mathcal{D}}{\rightarrow}$ denotes convergence in distribution and where $\boldsymbol{\Gamma}_{\operatorname{c}}(\bar{\boldsymbol{\theta}}) \in \mathbb{R}^{L\times L}$ is given in \eqref{eq: music_second_order} and $\boldsymbol{\Gamma}_{\operatorname{g}}(\bar{\boldsymbol{\theta}}) \in \mathbb{R}^{L\times L}$ is given in \eqref{eq: g_music_second_order}, respectively.
		\end{theorem}
		\begin{proof}
			The proof is given in Section \ref{sec: proof_second_order}.
		\end{proof}
		\paragraph*{Remark} The order of convergence in \eqref{eq: music_convergence_in_distribution} and \eqref{eq: g_music_convergence_in_distribution} is $\mathcal{O}(N^{-1/2})$. The real-valued integrals in \eqref{eq: music_second_order_xi} and \eqref{eq: g_music_second_order_xi} can be computed using numerical integration techniques such as the Riemann sum or the Simpson's rule.
		 
		 Theorem \ref{thm: second_order} fully characterizes the asymptotic stochastic behavior of the \ac{MUSIC} as well as the g-\ac{MUSIC} cost function at any point within the \ac{FoV} of the sensor array. With the aid of Theorem \ref{thm: second_order} we can approximate the probability density function of both \ac{MUSIC} cost functions also under non-asymptotic conditions, which can be useful, e.g., for source detection or to quantify the resolution capabilities of the corresponding \ac{DoA} estimator.

\subsection{Probability of Resolution}
	\label{sec: probability_of_resolution}
	
	The practical relevance of a \ac{DoA} estimator highly depends on its computational complexity as well as its estimation accuracy. For the estimation accuracy of subspace-based \ac{DoA} estimators like \ac{MUSIC} and g-\ac{MUSIC} the so called threshold effect is of major interest. The threshold effect describes an abrupt increase in the \ac{MSE} below a certain \ac{SNR} threshold due to the production of outliers in the \ac{DoA} estimates and can be characterized by analyzing the resolution capabilities of the underlying \ac{DoA} estimator \cite{a52,a61,a8}. In the literature, several criteria to declare resolution between two neighboring sources have been introduced \cite{a87,a88,a86,a47}. In case of a one-dimensional spectral search based \ac{DoA} estimator, one may declare that two sources located at \acp{DoA} $\boldsymbol{\theta}=\lbrack \theta_{1},\theta_{2}\rbrack^{\operatorname{T}}$ are resolved if both \ac{DoA} estimation errors, namely $|\theta_{1}-\hat{\theta}_{1}|$ and $|\theta_{2}-\hat{\theta}_{2}|$ are smaller than $|\theta_{1}-\theta_{2}|/2$ \cite{a86}. A popular alternative is to evaluate the cost function at both true \acp{DoA} and to verify if the cost at the mid-angle $(\theta_{1}+\theta_{2})/2$ is larger than at the true \acp{DoA} \cite{a47}. Hence, under the latter criterion resolution is declared if $\boldsymbol{u}^{\operatorname{T}}\hat{\boldsymbol{\eta}}(\bar{\boldsymbol{\theta}})<0$ where
	\begin{equation*}
		\boldsymbol{u} = \begin{bmatrix}
			1/2 \\ 1/2 \\ -1
		\end{bmatrix}, \quad \bar{\boldsymbol{\theta}} = \begin{bmatrix}
			\theta_{1} \\ \theta_{2} \\ \frac{\theta_{1} + \theta_{2}}{2}
		\end{bmatrix}, \quad \hat{\boldsymbol{\eta}}(\bar{\boldsymbol{\theta}}) = \begin{bmatrix}
			\hat{\eta}(\theta_{1}) \\ \hat{\eta}(\theta_{2}) \\ \hat{\eta}\left( \frac{\theta_{1}+\theta_{2}}{2}\right)
		\end{bmatrix} 
	\end{equation*}	
	and $\hat{\eta}(\theta)$ denotes a generic cost function. Using the previously derived asymptotic stochastic behavior of the cost function vector $\hat{\boldsymbol{\eta}}(\bar{\boldsymbol{\theta}})$ in Theorem \ref{thm: second_order}, the probability of resolution can be expressed as the cumulative distribution function\cite{a55}
	\begin{equation}
		P_{\text{res}} = \operatorname{Pr} \left( \boldsymbol{u}^{\operatorname{T}}\hat{\boldsymbol{\eta}}(\bar{\boldsymbol{\theta}}) < 0 \right) = \int_{-\infty}^{0} f_{\boldsymbol{u}^{\operatorname{T}}\hat{\boldsymbol{\eta}}(\bar{\boldsymbol{\theta}})}(x) \mathrm{d}x
		\label{eq: predicted_probability_of_resolution}
	\end{equation}
	where $f_{\boldsymbol{u}^{\operatorname{T}}\hat{\boldsymbol{\eta}}(\bar{\boldsymbol{\theta}})}(x)$ denotes the pdf of the test quantity $\boldsymbol{u}^{\operatorname{T}}\hat{\boldsymbol{\eta}}(\bar{\boldsymbol{\theta}})$ that can be asymptotically approximated by a Gaussian distribution with law $\mathcal{N}(\boldsymbol{u}^{\operatorname{T}}\bar{\boldsymbol{\eta}}(\bar{\boldsymbol{\theta}}), N^{-1} \boldsymbol{u}^{\operatorname{T}} \boldsymbol{\Gamma}(\bar{\boldsymbol{\theta}}) \boldsymbol{u})$, where $\bar{\boldsymbol{\eta}}(\bar{\boldsymbol{\theta}})$ and $\boldsymbol{\Gamma}(\bar{\boldsymbol{\theta}})$ take the form in \eqref{eq: music_first_order}-\eqref{eq: music_second_order} or in \eqref{eq: g_music_first_order}-\eqref{eq: g_music_second_order} depending on the subspace method under evaluation \cite{a40}.
		
\section{Proof of Theorem \ref{thm: second_order}} 
	\label{sec: proof_second_order}
	We begin by reviewing some standard arguments that allow to represent the \ac{MUSIC} and g-\ac{MUSIC} cost functions as contour integrals of the resolvent of the sample covariance matrix, which is a matrix-valued function of a complex variable $z \in \mathbb{C} \backslash \mathbb{R}$ defined as $\hat{\boldsymbol{Q}}(z) = (\hat{\boldsymbol{R}} - z\boldsymbol{I}_M)^{-1}$ \cite{a20,a43}. Observe that, using the eigenvalue decomposition of $\hat{\boldsymbol{R}}$ in \eqref{eq: sample_covariance_matrix_received_signal_eigendecomposition} we are able to write 
	\begin{equation*}
		\hat{\boldsymbol{Q}}(z) = \sum_{m=1}^{M} \frac{1}{\hat{\lambda}_{m}-z} \hat{\boldsymbol{e}}_{m} \hat{\boldsymbol{e}}_{m}^{\operatorname{H}} 
	\end{equation*}
	and therefore a direct application of the Cauchy integral theorem shows that 
    \begin{equation}	
			\sum_{m=1}^{M-K} \hat{\boldsymbol{e}}_{m} \hat{\boldsymbol{e}}_{m}^{\operatorname{H}} = \frac{1}{2\pi \mathrm{j}} \oint_{\mathcal{C}_{z}} \hat{\boldsymbol{Q}}(z) \mathrm{d}z
			\label{eq: music_first_order_contour_integral}
	\end{equation}
	where $\mathcal{C}_{z}$ is a clockwise oriented simple closed contour that encloses only the $K_1 = M-K$ smallest eigenvalues of the sample covariance matrix $\hat{\boldsymbol{R}}$. 
	
	Now, it is well known that for all $M,N$ sufficiently large, the $M-K$ smallest eigenvalues of $\hat{\boldsymbol{R}}$ are located inside $[x_{1}^{},x_{1}^{+}]$ plus $\{0\}$ if $M>N$ \cite{a44,a46}. Thanks to this fact, we can always deform the contour in \eqref{eq: music_first_order_contour_integral} and make it independent of $M$. This deterministic contour may cross the real axis at the points $\varrho$ (defined in Assumption \ref{as: subspace_separation}) and any other point in the negative real axis. Consequently, one can investigate the asymptotic behavior of the eigenvectors of the sample covariance matrix by equivalently characterizing the asymptotic behavior of the resolvent $\hat{\boldsymbol{Q}}(z)$. 
	
	\subsection{Asymptotic Equivalent of \ac{MUSIC} and Derivation of g-\ac{MUSIC}}
	Under Assumptions \ref{as:asymptotic_growth}-\ref{as: rewritten_sample_covariance_matrix} and for fixed $\theta$ and $z\in\mathbb{C}\setminus\mathbb{R}$ we have  
    \begin{equation}
    	\label{eq: convergence_quadratic_forms}
        \left| \boldsymbol{a}^{\operatorname{H}}(\theta) \hat{\boldsymbol{Q}}(z)\boldsymbol{a}(\theta) - \frac{\omega(z)}{z} \boldsymbol{a}^{\operatorname{H}}(\theta) \left({\boldsymbol{R}} - \omega(z) \boldsymbol{I}_M \right)^{-1}\boldsymbol{a}(\theta) \right| \rightarrow 0
    \end{equation}
    almost surely \cite{a41,a42}, where $\omega(z)$ is defined as in \eqref{eq: definition_omega_z}. A direct application of the dominated convergence theorem therefore shows that, since
    \begin{equation*}
        \hat{\eta}_{\text{c}}(\theta) = \frac{1}{2\pi \mathrm{j}} \oint_{\mathcal{C}_{z}} \boldsymbol{a}^{\operatorname{H}}(\theta) \hat{\boldsymbol{Q}}(z) \boldsymbol{a}(\theta) \mathrm{d}z
    \end{equation*}
    we can write
    \begin{equation*}
        \left| \hat{\eta}_{\text{c}}(\theta) - \frac{1}{2\pi \mathrm{j}} \oint_{\mathcal{C}_{z}} \frac{\omega(z)}{z} \boldsymbol{a}^{\operatorname{H}}(\theta) \left({\boldsymbol{R}} - \omega(z) \boldsymbol{I}_M \right)^{-1}\boldsymbol{a}(\theta) \mathrm{d}z \right| \rightarrow 0.
    \end{equation*}
    The right hand side of the above difference can be shown to coincide with the expression in \eqref{eq: music_first_order} after applying the change of variables $z= z(\omega)$ as given in \eqref{eq: mapping_omega_to_z} and solving the corresponding integral via Cauchy integration \cite{a43}.
   	
	A very similar idea can be used to derive the g-\ac{MUSIC} cost function. We begin by expressing the cost function that we want to estimate using again the Cauchy integral, that is
	 \begin{equation*}
        \bar{\eta}_{\text{g}}(\theta) = \frac{1}{2\pi \mathrm{j}} \oint_{\mathcal{C}_{\omega}} \boldsymbol{a}^{\operatorname{H}}(\theta) \left( \boldsymbol{R} - \omega  \boldsymbol{I}_M \right)^{-1} \boldsymbol{a}(\theta) \mathrm{d}\omega
    \end{equation*}
    where now $\mathcal{C}_{\omega}$ is a clockwise oriented simple closed contour that encloses $\gamma_1$ and no other eigenvalue. In particular, it can be shown that \cite{a69} $z \mapsto \omega(z)$ in \eqref{eq: definition_omega_z} can be used as parametrization of such a contour, in the sense that we can choose $\mathcal{C}_{\omega} = \omega(\mathcal{C}_{z})$ with $\mathcal{C}_{z}$ as defined above. It can be seen that the generated contour has the properties that we are looking for \cite{a69}, and consequently we can express 
    \begin{equation}
        \bar{\eta}_{\text{g}}(\theta) =  \frac{1}{2\pi \mathrm{j}} \oint_{\mathcal{C}_{z}} \boldsymbol{a}^{\operatorname{H}}(\theta) \left( \boldsymbol{R} - \omega (z) \boldsymbol{I}_M \right)^{-1} \boldsymbol{a}(\theta) \omega^{\prime}(z) \mathrm{d}z
        \label{eq: g_music_cauchy_integral_intermediate_1}
    \end{equation}
    with $\omega^{\prime}(z) $ as in \eqref{eq: first_order_derivative_wrt_x_omega_x}. Now, observe that the above contour does not depend on $M$, so that invoking again the dominated convergence theorem we can find an $M,N$-consistent estimator of $\bar{\eta}_{\text{g}}(\theta)$ by simply replacing the integrand above with a corresponding $M,N$-consistent estimate. Using \eqref{eq: convergence_quadratic_forms} we readily see that we only need to find an $M,N$-consistent estimator of $\omega(z)$ and $\omega^{\prime}(z)$. An $M,N$-consistent estimator of $\omega(z)$ can be obtained by using the well-known fact that \cite{a41,a42}, under Assumptions \ref{as:asymptotic_growth}-\ref{as: rewritten_sample_covariance_matrix} we have 
    \begin{equation*}
        \left| \frac{1}{M}\tr{\hat{\boldsymbol{Q}}(z)} - \frac{\omega(z)}{z}  \frac{1}{M} \tr{\left({\boldsymbol{R}} - \omega(z) \boldsymbol{I}_M \right)^{-1}} \right| \rightarrow 0
    \end{equation*}
    almost surely for fixed $z \in \mathbb{C} \backslash \mathbb{R}$. Using the fact that $\omega(z)$ is a solution to the equation in \eqref{eq: definition_omega_z}, this can be reformulated as 
    \begin{equation*}
    \left| 		\omega(z) - z\left( 1- \frac{1}{N}\tr{\hat{\boldsymbol{R}}(\hat{\boldsymbol{R}}-z\boldsymbol{I}_{M})^{-1}} \right)^{-1} \right| \rightarrow 0
	\end{equation*}		
	(see \cite[Lemma 8]{a63}). Therefore, the right hand side of the above difference, which will be denoted as $\hat{\omega}(z)$, is an $M,N$-consistent estimator of $\omega(z)$. Furthermore, the $M,N$-consistent estimator of $\omega^{\prime}(z)$ can be obtained by using the fact that convergence of holomorphic functions imply the convergence of their derivatives, so that 
	\begin{equation*}
		\hat{\omega}^{\prime}(z) = \frac{\partial \hat{\omega}(z)}{\partial z} = \frac{\hat{\omega}(z)}{z} + \frac{\hat{\omega}(z)^{2}}{z}\frac{1}{N}\tr{\hat{\boldsymbol{R}}\left(\hat{\boldsymbol{R}}-z\boldsymbol{I}_{M}\right)^{-2}}    
	\end{equation*}
	is an $M,N$-consistent estimator of $\omega^{\prime}(z)$. We can therefore obtain the $M,N$-consistent estimator of the g-\ac{MUSIC} cost function by inserting these estimators into the Cauchy integral in \eqref{eq: g_music_cauchy_integral_intermediate_1}, that is 
	\begin{equation}
		\hat{\eta}_{\text{g}}({\theta}) = \frac{1}{2\pi \mathrm{j}} \oint_{\mathcal{C}_{z}}  \boldsymbol{a}^{\operatorname{H}}(\theta)
		\hat{\boldsymbol{Q}}(z)
		 \boldsymbol{a}(\theta)
		\frac{z}{\hat{\omega}(z)}\hat{\omega}^{\prime}(z) \mathrm{d}z. 
		\label{eq: g_music_cost_function_contour_integral}
	\end{equation}
	The contour integral in \eqref{eq: g_music_cost_function_contour_integral} is solved in closed-form in \cite[Theorem 2]{a69} (also see \cite[Theorem 3]{a20}) using conventional residue calculus, which yields the expression in \eqref{eq: g_music_cost_function}. 

\subsection{Asymptotic Fluctuations}
	Let us now consider again the fluctuations of these two cost functions at a set of $L$ fixed distinct angles, $\bar{\theta}_1,\ldots,\bar{\theta}_L$. Using the previously derived asymptotic equivalents of the \ac{MUSIC} cost function we can express 
	\begin{equation*}
		\sqrt{N}\left( \hat{\eta}_{\text{c}}(\bar{\theta}_{l}) - \bar{\eta}_{\text{c}}(\bar{\theta}_{l}) \right) = \frac{1}{2\pi \mathrm{j}}\oint_{\mathcal{C}_{z}} \sqrt{N} \boldsymbol{a}^{\operatorname{H}}(\bar{\theta}_l) \left( \hat{\boldsymbol{Q}}(z) - \bar{\boldsymbol{Q}}(z) \right) \boldsymbol{a}(\bar{\theta}_l) \mathrm{d}z
	\end{equation*}
	where we have introduced $\bar{\boldsymbol{Q}}(z) = \frac{\omega(z)}{z} \left(\boldsymbol{R}- \omega(z) \boldsymbol{I}_M \right)^{-1}$. Similarly, for the g-\ac{MUSIC} cost function we have
	\begin{equation*}
		\begin{aligned}
			 &\sqrt{N}\left( \hat{\eta}_{\text{g}}(\bar{\theta}_{l}) - \bar{\eta}_{\text{g}}(\bar{\theta}_{l}) \right)\\
			=& \frac{1}{2\pi \mathrm{j}}\oint_{\mathcal{C}_{z}} \sqrt{N} \boldsymbol{a}^{\operatorname{H}}(\bar{\theta}_l) \left(\frac{z \hat{\omega}^{\prime}(z)}{\hat{\omega}(z)} \hat{\boldsymbol{Q}}(z) -  \frac{z\omega^{\prime}(z)}{\omega(z)}\bar{\boldsymbol{Q}}(z) \right) \boldsymbol{a}(\bar{\theta}_l) \mathrm{d}z.
		\end{aligned}
	\end{equation*}
	It can be seen that both expressions take the form
	\begin{equation}
		\label{eq: general_form_second_order_derivation}
		\begin{aligned}
			&\sqrt{N}\left( \hat{\eta}(\bar{\theta}_{l}) - \bar{\eta}(\bar{\theta}_{l}) \right)\\
			=& \frac{1}{2\pi \mathrm{j}}\oint_{\mathcal{C}_{z}} \sqrt{N} \boldsymbol{a}^{\operatorname{H}}(\bar{\theta}_l) \left({\hat{h}(z)} \hat{\boldsymbol{Q}}(z) - \bar{h}(z) \bar{\boldsymbol{Q}}(z) \right) \boldsymbol{a}(\bar{\theta}_l) \mathrm{d}z
		\end{aligned}
	\end{equation}
	where $\hat{\eta}(\theta)$ is a generic cost function with deterministic equivalent $\bar{\eta}(\theta)$. In case of the conventional \ac{MUSIC} cost function we have $\hat{h}(z) = 1$ and $\bar{h}(z) = 1$ whereas in case of the g-\ac{MUSIC} cost function $\hat{h}(z) = \frac{z \hat{\omega}^{\prime}(z)}{\hat{\omega}(z)}$ and $\bar{h}(z) = \frac{z \omega^{\prime}(z)}{\omega(z)}$. It is well known \cite{a89,a80} that the statistic in \eqref{eq: general_form_second_order_derivation} asymptotically fluctuates as a Gaussian random variable with zero-mean and positive variance. However, we are interested in the more general case of the asymptotic joint distribution of a collection of random variables, namely 
	\begin{equation}
		\sqrt{N}\left(\hat{\boldsymbol{\eta}}(\bar{\boldsymbol{\theta}})-\bar{\boldsymbol{\eta}}(\bar{\boldsymbol{\theta}})\right) = \begin{bmatrix}
			 \sqrt{N}(\hat{\eta}(\bar{\theta}_{1})-\bar{\eta}(\bar{\theta}_{1})) \\
			\vdots \\
			 \sqrt{N}(\hat{\eta}(\bar{\theta}_{L})-\bar{\eta}(\bar{\theta}_{L}))
		\end{bmatrix}
		\label{eq: general_form_vector_second_order_derivation}
	\end{equation}
	where $\sqrt{N}(\hat{\eta}(\bar{\theta}_{l})-\bar{\eta}(\bar{\theta}_{l})$ for $l=1,\dots,L$ is defined in \eqref{eq: general_form_second_order_derivation}. The second order asymptotic behavior in Theorem \ref{thm: second_order} is derived by establishing pointwise convergence of the characteristic function of the statistic $\sqrt{N}\left(\hat{\boldsymbol{\eta}}(\bar{\boldsymbol{\theta}})-\bar{\boldsymbol{\eta}}(\bar{\boldsymbol{\theta}})\right)$ in \eqref{eq: general_form_vector_second_order_derivation} to the characteristic function of a Gaussian distributed random variable. Moreover, by the Cram\'er-Wold device \cite{a79} it is sufficient to establish that the one-dimensional projection of the statistic in \eqref{eq: general_form_vector_second_order_derivation}, namely
	\begin{equation}
		 \sum_{l=1}^{L}\sqrt{N}w_{l}\left(\hat{\eta}(\bar{\theta}_{l})-\bar{\eta}(\bar{\theta}_{l})\right) = \sqrt{N}\boldsymbol{w}^{\operatorname{T}}\left(\hat{\boldsymbol{\eta}}(\bar{\boldsymbol{\theta}})-\bar{\boldsymbol{\eta}}(\bar{\boldsymbol{\theta}})\right) 
		 \label{eq: one_dimensional_projection_vector_general}
	\end{equation}
	is asymptotically Gaussian distributed for any collection of real-valued bounded quantities $\boldsymbol{w}=\lbrack w_{1}, \dots, w_{L}\rbrack^{\operatorname{T}} \in \mathbb{R}^{L}$ to show that \eqref{eq: general_form_vector_second_order_derivation} is asymptotically jointly Gaussian distributed. Additionally, by L\'evy's continuity Theorem convergence in distribution of a set of random variables can be proven by establishing pointwise convergence of the characteristic functions. Let $\chi(r)$ be defined as $\chi(r)= \exp \left( \mathrm{j} r \sqrt{N}\boldsymbol{w}^{\operatorname{T}}\left(\hat{\boldsymbol{\eta}}(\bar{\boldsymbol{\theta}})-\bar{\boldsymbol{\eta}}(\bar{\boldsymbol{\theta}})\right) \right)$ and let $\mathbb{E}\lbrack \chi(r)\rbrack$ be the corresponding characteristic function of the one-dimensional projection in \eqref{eq: one_dimensional_projection_vector_general}. The target is to study the asymptotic behavior of the characteristic function $\mathbb{E}\lbrack \chi(r)\rbrack$ in the asymptotic regime where $M,N \rightarrow \infty$ at the same rate by establishing convergence towards the characteristic function of a Gaussian random variable with zero-mean and covariance $ \boldsymbol{w}^{\operatorname{T}} \boldsymbol{\Gamma}(\bar{\boldsymbol{\theta}}) \boldsymbol{w}$, that is $\mathbb{E}\lbrack \chi(r) \rbrack - \bar{\chi}(r) \rightarrow 0$ where
	\begin{equation} 
		\bar{\chi}(r) = \exp\left( - r^{2} \frac{\boldsymbol{w}^{\operatorname{T}} \boldsymbol{\Gamma}(\bar{\boldsymbol{\theta}}) \boldsymbol{w}}{2}\right).
		\label{eq: regularization_function_gaussian_general}
	\end{equation}
	In \eqref{eq: regularization_function_gaussian_general}, $\boldsymbol{\Gamma}(\bar{\boldsymbol{\theta}}) \in \mathbb{R}^{L \times L}$ characterizes the second order asymptotic behavior of the statistic $\sqrt{N}(\hat{\boldsymbol{\eta}}(\bar{\boldsymbol{\theta}})-\bar{\boldsymbol{\eta}}(\bar{\boldsymbol{\theta}}))$ in \eqref{eq: general_form_vector_second_order_derivation}, which is the quantity of interest. In the asymptotic analysis of the characteristic function $\mathbb{E}\lbrack \chi(r)\rbrack$ of the one-dimensional projection $\sqrt{N}\boldsymbol{w}^{\operatorname{T}}(\hat{\boldsymbol{\eta}}(\bar{\boldsymbol{\theta}})-\bar{\boldsymbol{\eta}}(\bar{\boldsymbol{\theta}}))$ we rely on the integration by parts formula \cite{a56,a77,a83} and the Nash-Poincar\'e inequality \cite{a77,a78,a83}. It can be shown that the second order asymptotic behavior of the statistic in \eqref{eq: one_dimensional_projection_vector_general} is computed as follows. 
	\begin{theorem}
		\label{thm: second_order_general_theorem}
	Consider an $L \times L$ matrix $\boldsymbol{\Gamma}(\bar{\boldsymbol{\theta}})$ the elements of which are given by 
		\begin{equation}
			\begin{aligned}
				\left \lbrack \boldsymbol{\Gamma}(\bar{\boldsymbol{\theta}})\right \rbrack_{p,q}=\frac{1}{2 \pi \mathrm{j}}\frac{1}{2 \pi \mathrm{j}} \oint_{\mathcal{C}_{\omega_{1}}}\oint_{\mathcal{C}_{\omega_{2}}} \bar{h}(z(\omega_{1})) \bar{h}(z(\omega_{2})) \frac{\partial z(\omega_{1})}{\partial \omega_{1}}\frac{\partial z(\omega_{2})}{\partial \omega_{2}}& \\[-2pt]
				\times \frac{\omega_{1}}{z(\omega_{1})} \frac{\omega_{2}}{z(\omega_{2})} \frac{\Upsilon_{p,q}(\omega_{1},\omega_{2})}{1-\Omega(\omega_{1},\omega_{2})} \mathrm{d}\omega_{2}\mathrm{d}\omega_{1}&
			\end{aligned}
			\label{eq: second_order_general}
		\end{equation}
		for $p,q \in \lbrace 1, \dots, L \rbrace$, where $\mathcal{C}_{\omega_1}$, $\mathcal{C}_{\omega_2}$ are two clockwise oriented simple closed contours enclosing only the smallest eigenvalue $\gamma_{1}$ of $\boldsymbol{R}$,
		\begin{equation}
			\begin{aligned}
				\Upsilon_{p,q}(\omega_{1},\omega_{2}) =& \boldsymbol{a}(\bar{\theta}_{p})^{\operatorname{H}} ( \boldsymbol{R} - \omega_{1}\boldsymbol{I}_{M})^{-1} \boldsymbol{R} (\boldsymbol{R} - \omega_{2}\boldsymbol{I}_{M})^{-1} \boldsymbol{a}(\bar{\theta}_{q}) \\
				\times& \boldsymbol{a}(\bar{\theta}_{q})^{\operatorname{H}}  (\boldsymbol{R} - \omega_{2}\boldsymbol{I}_{M})^{-1}\boldsymbol{R}(\boldsymbol{R}-\omega_{1}\boldsymbol{I}_{M})^{-1}\boldsymbol{a}(\bar{\theta}_{p})
			\end{aligned}
			\label{eq: upsilon_theorem}
		\end{equation}
		and where
		\begin{equation}
			\Omega(\omega_{1},\omega_{2}) = \frac{1}{N}\tr{\boldsymbol{R}(\boldsymbol{R}-\omega_{1}\boldsymbol{I}_{M})^{-1}\boldsymbol{R}(\boldsymbol{R}-\omega_{2}\boldsymbol{I}_{M})^{-1}}.	
			\label{eq: omega_captial_theorem}
		\end{equation}
		Function $z(\omega)$ is defined in \eqref{eq: mapping_omega_to_z} 		and its first order derivative can be computed as
		\begin{equation}
			\frac{\partial z(\omega)}{\partial \omega} = 1- \frac{1}{N}\sum_{r=1}^{\bar{M}} \frac{K_{r}\gamma_{r}^{2}}{(\gamma_{r}-\omega)^{2}}.
			\label{eq: derivative_z_of_omega_wrt_omega}
		\end{equation}
		Let Assumptions \ref{as:asymptotic_growth}-\ref{as: subspace_separation} hold true and assume additionally that the eigenvalues of $\boldsymbol{\Gamma}(\bar{\boldsymbol{\theta}})$ are all located in a compact interval of the positive real axis independent of $M$. Then, $\sqrt{N}\boldsymbol{\Gamma}(\bar{\boldsymbol{\theta}})^{-1/2}(\hat{\boldsymbol{\eta}}(\bar{\boldsymbol{\theta}})-\bar{\boldsymbol{\eta}}(\bar{\boldsymbol{\theta}}))$ in \eqref{eq: general_form_vector_second_order_derivation} converges in law to a standardized multivariate Gaussian distribution.
	\end{theorem}
	\begin{proof}
		A detailed proof is provided in \cite{a84}.
	\end{proof}
	
    With the help of Theorem \ref{thm: second_order_general_theorem} above, the proof of Theorem \ref{thm: second_order} follows directly once we have been able to (i) compute the integral that defines the asymptotic covariance matrix for these two cost functions and (ii) prove that the maximum (resp. minimum) eigenvalue of these two covariance matrices is bounded (resp. bounded away from zero). 

    Let us first consider the computation of the two asymptotic covariance matrices, which follows from solving the integral in \eqref{eq: second_order_general} for $\bar{h}(z(\omega)) = 1$ (conventional \ac{MUSIC}) and for $\bar{h}(z(\omega)) = \frac{z(\omega)}{\omega}\frac{\partial \omega(z)}{\partial z}$ (g-\ac{MUSIC}). To simplify the computation of the asymptotic covariance, we express $\boldsymbol{\Gamma}(\bar{\boldsymbol{\theta}})$ in \eqref{eq: second_order_general} as 
    \begin{equation}
        \label{eq: Gamma_as_Sum}
        \boldsymbol{\Gamma}(\bar{\boldsymbol{\theta}}) = \sum_{r=1}^{\bar{M}} \sum_{k=1}^{\bar{M}} \xi(r,k)  \boldsymbol{A}^{\operatorname{H}}(\bar{\boldsymbol{\theta}}) \boldsymbol{E}_{r}\boldsymbol{E}_{r}^{\operatorname{H}} \boldsymbol{A}(\bar{\boldsymbol{\theta}}) \odot \left(
        \boldsymbol{A}^{\operatorname{H}}(\bar{\boldsymbol{\theta}}) \boldsymbol{E}_{k}\boldsymbol{E}_{k}^{\operatorname{H}} \boldsymbol{A}(\bar{\boldsymbol{\theta}})\right)^{\operatorname{T}}
    \end{equation}
    where $\odot$ denotes element-wise product and where the coefficients $\xi(r,k)$ are defined as
    \begin{equation}
		\xi(r,k) = \frac{1}{2\pi \mathrm{j}} \oint_{\mathcal{C}_{\omega_{1}}} \bar{h}(z(\omega_{1})) \frac{\omega_{1}}{z(\omega_{1})}  \frac{\partial z(\omega_{1})}{\partial \omega_{1}} \frac{\gamma_r \gamma_k \mathcal{I}_{r,k}(\omega_{1})}{(\gamma_{r}-\omega_{1})(\gamma_{k}-\omega_{1})} \mathrm{d}\omega_{1},
		\label{eq: second_order_general_separated_omega_1}
	\end{equation}
	with 
	\begin{equation}
		\mathcal{I}_{r,k}(\omega_{1}) = \frac{1}{2 \pi \mathrm{j}}\oint_{\mathcal{C}_{\omega_{2}}}  \bar{h}(z(\omega_{2})) \frac{\omega_{2}}{z(\omega_{2})} \frac{\partial z(\omega_{2})}{\partial \omega_{2}} \frac{\frac{1}{\left(\gamma_{r}-\omega_{2}\right)\left(\gamma_{k}-\omega_{2}\right)}}{1-\Omega(\omega_{1},\omega_{2})}\mathrm{d}\omega_{2}.
		\label{eq: second_order_general_separated_omega_2}
	\end{equation}
	These two integrals are solved in Appendices \ref{app: second_order_music} and \ref{app: second_order_g_music} for the conventional \ac{MUSIC} and the g-\ac{MUSIC} cost functions, leading to the expressions in \eqref{eq: music_second_order} and \eqref{eq: g_music_second_order} respectively. 
	
	To obtain an upper bound on the spectral norm of $\boldsymbol{\Gamma}(\bar{\boldsymbol{\theta}})$, we consider the expression in \eqref{eq: Gamma_as_Sum} and observe that the coefficients $\xi(r,k)$ are all real-valued and positive. The fact that they are real-valued follows from the property $z(\omega^{\ast}) = z^{\ast}(\omega)$ in \eqref{eq: mapping_omega_to_z} as well as the fact that $\bar{h}(z^{\ast}) = \bar{h}^{\ast}(z)$ for both definitions of this function. The fact that the coefficients $\xi(r,k)$ in \eqref{eq: second_order_general_separated_omega_1} are non-negative can be shown by using the following results. 

	\begin{lemma}
		\label{lemma: bound_Omega}
	    Let $\Omega(\omega_1,\omega_2)$ be defined as in \eqref{eq: omega_captial_theorem}. Then,
	    \begin{equation}
	    	\sup_M \sup_{(z_1,z_2) \in \mathcal{C}_z \times \mathcal{C}_z} | \Omega(\omega(z_1), \omega(z_2)) |  < 1. 
	      	\label{eq: bound_Omega}
	    \end{equation}
	\end{lemma}
	\begin{proof}
	    By the Cauchy-Schwarz inequality, we have $| \Omega(\omega(z_1), \omega(z_2)) |^{2} \leq \tilde{\Psi}(\omega(z_1)) \tilde{\Psi}(\omega(z_1))$ with 
	    \begin{equation}
	        \label{eq: definition_tilde_Psi}  
	        \tilde{\Psi}(\omega) = \frac{1}{N} \sum_{m=1}^{\bar{M}} K_{m} \frac{\gamma^2_m}{|\gamma_m - \omega|^{2}}. 
	    \end{equation}
	    Hence, it is sufficient to prove that $\sup_M \sup_{z \in \mathcal{C}_z} |\tilde{\Psi} (\omega(z))| < 1$. 
	    Let us first consider $z \in \mathbb{C}^{+}$. Taking imaginary parts on both sides of \eqref{eq: definition_omega_z} we see that \begin{equation*}
	    \tilde{\Psi}(\omega(z)) = 1 - \frac{\ima{z}}{\ima{\omega(z)}} 
	    \end{equation*} 
	    where we recall that $\ima{\omega(z)}>0$ when $\ima{z}>0$. This implies that $\sup_M \tilde{\Psi}(\omega(z)) \leq 1$ and $\inf_M \ima{\omega(z)} > 0$. In order to see that the inequality $\sup_{M} \tilde{\Psi}(\omega(z)) \leq 1$ must be strict, we reason by contradiction. Assume that the equality holds, and consider a subsequence $M'$ such that $\tilde{\Psi}(\omega(z)) \rightarrow 1$, implying $\ima{\omega(z)} \rightarrow \infty$, as $M' \to \infty$.  Now, using the fact that $|\gamma_m - \omega(z)|^2 \geq \ima{w(z)}^2$ we see from the definition of $\tilde{\Psi}(\omega(z))$ in \eqref{eq: definition_tilde_Psi} that  $\tilde{\Psi}(\omega(z)) \leq \ima{\omega(z)}^{-2} \frac{1}{N}\sum_{m=1}^{\bar{M}} K_m \gamma_m^2 \rightarrow 0$ along that subsequence, leading to contradiction. 
	    
	    Consider now $z \in  \mathcal{C}_z \cap \mathbb{R}$ and observe that this means that either $z < 0$ or $z = \varrho$ in Assumption \ref{as: subspace_separation}, which are the two crossing points of the contour on the real axis. Assume first that $z = \varrho$ and consider $\omega_1^+$ as right hand side of the first cluster of the support $\mathcal{S}$ and $\omega_{2}^{-}$ as the left hand side of the second cluster such that $\gamma_{1}<\omega_{1}^{+}<\omega(\varrho) < \omega_{2}^{-}<\gamma_{2}$ (the case $z<0$ follows by similar arguments). It can be observed that $\tilde{\Psi}(\omega)$ in \eqref{eq: definition_tilde_Psi} is a strongly convex function since the second order derivative w.r.t. $\omega$ 
	    \begin{equation}
	        \tilde{\Psi}^{\prime \prime}(\omega) = \frac{6}{N}\sum_{r=1}^{\bar{M}} \frac{K_{r}\gamma_{r}^{2}}{(\gamma_{r}-\omega)^{4}}
	        \label{eq: definition_tilde_Psi_second_derivative}
	    \end{equation}
	    is lower bounded by a positive quantity $q$, namely
	    \begin{equation*}
	         q= \underset{M}{\inf} \underset{\omega\in(\omega_{1}^{+},\omega_{2}^{-})}{\inf} \tilde{\Psi}^{\prime \prime}(\omega) > \frac{3}{4}\frac{1}{N}\sum_{r=1}^{\bar{M}}\frac{K_{r}\gamma_{r}^{2}}{\gamma_{r}^{4}+\gamma_{2}^{4}} >0.
	    \end{equation*} 
	    By strong convexity it follows that
	    \begin{align*}
	        \tilde{\Psi}(\omega_{2}^{-}) - \tilde{\Psi}(\omega(\varrho)) \geq &  \tilde{\Psi}^{\prime}(\omega(\varrho)) (\omega_{2}^{-} - \omega(\varrho))+ \frac{q}{2}(\omega_{2}^{-}-\omega(\varrho))^{2}\\
	        \tilde{\Psi}(\omega_{1}^{+}) - \tilde{\Psi}(\omega(\varrho)) \geq & - \tilde{\Psi}^{\prime}(\omega(\varrho)) (\omega(\varrho)-\omega_{1}^{+})+ \frac{q}{2}(\omega_{1}^{+}-\omega(\varrho))^{2}.
	    \end{align*}
	    Using the fact that $ \tilde{\Psi}(\omega_{2}^{-})=1$ and $ \tilde{\Psi}(\omega_{1}^{+}) =1$ we obtain
	    \begin{equation*}
	        1-\tilde{\Psi}(\omega(\varrho)) \geq \begin{cases} 
	            \tilde{\Psi}^{\prime}(\omega(\varrho)) (\omega_{2}^{-} - \omega(\varrho))+ \frac{q}{2}(\omega_{2}^{-}-\omega(\varrho))^{2} \\
	            - \tilde{\Psi}^{\prime}(\omega(\varrho)) (\omega(\varrho)-\omega_{1}^{+})+ \frac{q}{2}(\omega_{1}^{+}-\omega(\varrho))^{2}.
	        \end{cases}
	    \end{equation*}
	    Hence, $\inf_{M}1-\tilde{\Psi}(\omega(\varrho)) > \frac{q}{2}\min \left \lbrack (\omega_{2}^{-}-\omega(\varrho))^{2}, (\omega_{1}^{+}-\omega(\varrho))^{2}\right \rbrack >0$ and therefore $\sup_{M} \tilde{\Psi}(\omega(\varrho)) < 1$. The boundedness of $\tilde{\Psi}(\omega(z))$ at the second crossing point with the real axis $z<0$ follows by similar reasoning. 
	\end{proof}

	The fact that $\Omega(\omega_1,\omega_2)$ is uniformly bounded on $\mathcal{C}_\omega \times \mathcal{C}_\omega$ implies that we can take a power series expansion of the term $(1- \Omega(\omega_1,\omega_2))^{-1}$ and write 
	\begin{equation}
	    \begin{aligned}
	         \xi(r,k) = \gamma_r \gamma_k \sum_{n \geq 0} \frac{1}{N^n}  \sum_{r_1 + \ldots + r_{\bar{M}} = n} {{n}\choose{r_1,\ldots,r_{\bar{M}}}} \prod_{m=1}^{\bar{M}} K_m^{r_m} \gamma_m^{2 r_m}& \\[-2pt]
	        \times \Biggl( \frac{1}{2\pi\mathrm{j}} \oint_{\mathcal{C}_z} \frac{\omega(z)}{z} \frac{\bar{h}(z)}{(\gamma_r - \omega(z))(\gamma_k-\omega(z))} \prod_{m=1}^{\bar{M}} \frac{1}{(\gamma_m - \omega(z))^{r_m}} \mathrm{d}z \Biggr)^2 &
	    \end{aligned}
	    \label{eq: xi_rk_power_series}
	\end{equation}
	where we employed the multinomial theorem to factor out the different powers $\Omega^n(\omega_1,\omega_2)$, $n\geq 0$. An immediate consequence of the above decomposition is the fact that the coefficients $\xi(r,k)$ are non-negative. This will be useful in order to determine the appropriate lower bound on the corresponding covariance matrix.  

	\begin{lemma}
		\label{lemma: upperbound_Gamma}
		Under Assumptions \ref{as:asymptotic_growth}-\ref{as: subspace_separation}, we have $\sup_M \| \boldsymbol{\Gamma}_{\mathrm{c}}(\bar{\boldsymbol{\theta}}) \| < + \infty$ and $\sup_M \| \boldsymbol{\Gamma}_{\mathrm{g}}(\bar{\boldsymbol{\theta}}) \| < + \infty$. 
	\end{lemma}
	\begin{proof}
	    Consider the general expression in \eqref{eq: Gamma_as_Sum} and assume that we are able to find a positive constant $\kappa$ such that $\sup_M \sup_{r,k} \xi(r,k) < \kappa$. It will immediately follow that    
	    \begin{align*}
	        \| \boldsymbol{\Gamma}(\bar{\boldsymbol{\theta}}) \| &\leq \kappa \left\|\sum_{r=1}^{\bar{M}} \sum_{k=1}^{\bar{M}}   \boldsymbol{A}^{\operatorname{H}}(\bar{\boldsymbol{\theta}}) \boldsymbol{E}_{r}\boldsymbol{E}_{r}^{\operatorname{H}} \boldsymbol{A}(\bar{\boldsymbol{\theta}}) \odot \left(
	        \boldsymbol{A}^{\operatorname{H}}(\bar{\boldsymbol{\theta}}) \boldsymbol{E}_{k}\boldsymbol{E}_{k}^{\operatorname{H}} \boldsymbol{A}(\bar{\boldsymbol{\theta}})\right)^{\operatorname{T}} \right\|  \\
	       & =\kappa \left\|\boldsymbol{A}^{\operatorname{H}}(\bar{\boldsymbol{\theta}}) \boldsymbol{A}(\bar{\boldsymbol{\theta}}) \odot \left(
	        \boldsymbol{A}^{\operatorname{H}}(\bar{\boldsymbol{\theta}}) \boldsymbol{A}(\bar{\boldsymbol{\theta}})\right)^{\operatorname{T}} \right\| \leq L \kappa
	    \end{align*}    
	    where in the last inequality we have used the fact that the steering vectors have unit norm and the spectral norm of a matrix is upper bounded by its trace. Hence, we only need to find an upper bound on the coefficients. Following the definition of these coefficients in \eqref{eq: xi_rk_power_series} we can readily see that 
	    \begin{equation}
	        \xi(r,k) \leq \frac{\gamma_r \gamma_k}{ 1 - \bar{\Omega}} \left( \frac{1}{2\pi\mathrm{j}}  \oint_{\mathcal{C}_z} |\bar{h}(z)| \frac{|\omega(z)|}{|z|} \frac{1}{|\gamma_r - \omega(z)||\gamma_k - \omega(z)|} |\mathrm{d}z| \right)^2
	        \label{eq: upper_bound_coefficients}
	    \end{equation}
	    where $\bar{\Omega} = \sup_{M,(\omega_1,\omega_2) \in \mathcal{C}_\omega \times \mathcal{C}_\omega} \Omega(\omega_1, \omega_2) < 1$. Now, from the proof of Lemma \ref{lemma: bound_Omega} it directly follows that $\inf_{M,\omega \in \mathcal{C}_\omega} |\omega(z) - \gamma_k|> 0$ for $k = 1,\ldots,\bar{M}$. On the other hand, taking real and imaginary part on both sides of \eqref{eq: definition_omega_z} it can be shown by contradiction that for bounded $|z|<\infty$ we have $\sup_{M} |\omega(z)/z|< \infty$.    

	   This directly shows that the coefficients $\xi(r,k)$ are bounded for the \ac{MUSIC} cost function, which has $\bar{h}(z) = 1$. In order to prove boundedness of the g-\ac{MUSIC} covariance $\boldsymbol{\Gamma}_{\mathrm{g}}(\bar{\boldsymbol{\theta}})$ we only need to show that $\sup_{M} |\omega^{\prime}(z)|<\infty$ since $\bar{h}(z)=\frac{z}{\omega(z)}\omega^{\prime}(z)$ and $\frac{z}{\omega(z)}$ cancels out with $\frac{\omega(z)}{z}$ in \eqref{eq: upper_bound_coefficients}. The upper bound $\sup_M |\omega^{\prime}(z) | < \infty$ follows from the fact that, by the triangular inequality, $\sup_M |\omega^{\prime}(z)| \leq \sup_M (1-|\tilde{\Psi}(\omega(z))|)^{-1} <1$. This completes the proof of Lemma \ref{lemma: upperbound_Gamma}. 
	\end{proof}

	We finally conclude the proof of Theorem \ref{thm: second_order} by showing that the smallest eigenvalue of the two asymptotic covariance matrices is bounded away from zero.
	\begin{lemma}
	    \label{lemma: lowerbound_Gamma}
	     Under Assumptions \ref{as:asymptotic_growth}-\ref{as: positive_semidefinitness_covariance_matrix}, we have $\inf_M \lambda_{\min}(\boldsymbol{\Gamma}_{\mathrm{c}}(\bar{\boldsymbol{\theta}})) > 0$ and $\inf_M \lambda_{\min}(\boldsymbol{\Gamma}_{\mathrm{g}}(\bar{\boldsymbol{\theta}})) > 0$, where $\lambda_{\min}(\cdot)$ denotes the minimum eigenvalue of a matrix. 
	\end{lemma}
	\begin{proof}
	    In order to proof this lemma, we consider again the general expression of any of these two matrices that is given in \eqref{eq: Gamma_as_Sum}. By inserting the series expansion of $\xi(r,k)\geq 0$ in \eqref{eq: xi_rk_power_series} into \eqref{eq: Gamma_as_Sum} we obtain an expression of $\boldsymbol{\Gamma}(\bar{\boldsymbol{\theta}})$ as a linear combination of positive semidefinite matrices with non-negative coefficients. Hence, the original covariance can be lower bounded (in the ordering of positive semidefinite matrices) by selecting any of the terms of this expansion. In the case of \ac{MUSIC} we can select the term $n=0$ in \eqref{eq: xi_rk_power_series} along the sum $k = r\geq 2$ in \eqref{eq: Gamma_as_Sum}, leading to the lower bound
	    \begin{align*}
	        \boldsymbol{\Gamma}_{\mathrm{c}}(\bar{\boldsymbol{\theta}})
	        & \geq \sum_{k=2}^{\bar{M}} \upsilon_{\mathrm{c}}(k) \boldsymbol{A}^{\operatorname{H}}(\bar{\boldsymbol{\theta}}) \boldsymbol{E}_{k}\boldsymbol{E}_{k}^{\operatorname{H}} \boldsymbol{A}(\bar{\boldsymbol{\theta}}) \odot \left(
	        \boldsymbol{A}^{\operatorname{H}}(\bar{\boldsymbol{\theta}}) \boldsymbol{E}_{k}\boldsymbol{E}_{k}^{\operatorname{H}} \boldsymbol{A}(\bar{\boldsymbol{\theta}})\right)^{\operatorname{T}}
	    \end{align*}
	    where  
	    \begin{equation}
	        \upsilon_{\mathrm{c}}(k) = \left( \frac{1}{2\pi\mathrm{j}} \oint_{\mathcal{C}^{-}_z}  \frac{\omega(z)}{z} \frac{\gamma_k}{(\gamma_k -\omega(z))^2} \mathrm{d}z \right)^2. 
	        \label{eq:expr_upsilon_c}
	    \end{equation}
	    Likewise, for the g-\ac{MUSIC} cost function we can select the term $n=1$ in \eqref{eq: xi_rk_power_series} along the sum $k = r\geq 2$ in \eqref{eq: Gamma_as_Sum}, so that 
	   	\begin{align*}
	        \boldsymbol{\Gamma}_{\mathrm{g}}(\bar{\boldsymbol{\theta}}) 
	        &\geq \sum_{k=2}^{\bar{M}} \upsilon_{\mathrm{g}}(k) \boldsymbol{A}^{\operatorname{H}}(\bar{\boldsymbol{\theta}}) \boldsymbol{E}_{k}\boldsymbol{E}_{k}^{\operatorname{H}} \boldsymbol{A}(\bar{\boldsymbol{\theta}}) \odot \left(
	        \boldsymbol{A}^{\operatorname{H}}(\bar{\boldsymbol{\theta}}) \boldsymbol{E}_{k}\boldsymbol{E}_{k}^{\operatorname{H}} \boldsymbol{A}(\bar{\boldsymbol{\theta}})\right)^{\operatorname{T}}
	    \end{align*}
	    where now 
	    \begin{equation}
	          \upsilon_{\mathrm{g}}(k) =\frac{1}{N} \sum_{j=1}^{\bar{M}} K_j \gamma_j^2  \left( \frac{1}{2\pi\mathrm{j}} \oint_{\mathcal{C}^{-}_\omega}  \frac{\omega}{z(\omega)} \frac{\gamma_k}{(\gamma_k -\omega)^2(\gamma_j -\omega)} \mathrm{d}\omega \right)^2.
	          \label{eq:expr_upsilon_g}
	    \end{equation}
	    Using the fact that, for any two positive semidefinite matrices $\boldsymbol{A}, \boldsymbol{B}$, we have $\lambda_{\min}(\boldsymbol{A} \odot \boldsymbol{B}) \geq \lambda_{\min}(\boldsymbol{A})\lambda_{\min}(\boldsymbol{B})$ \cite{horn13}, we see that the lemma will follow from Assumption \ref{as: positive_semidefinitness_covariance_matrix} if we are able to show that $\inf_{M,k\geq 2} \upsilon_{\mathrm{c}}(k)>0 $ and $\inf_{M,k\geq 2} \upsilon_{\mathrm{g}}(k)>0$. 
	    
	    In case of $\upsilon_{\mathrm{g}}(k)$ in \eqref{eq:expr_upsilon_g}, we have 
	    \begin{align*}
	        \upsilon_{\mathrm{g}}(k) & =  \frac{K_1}{N}\frac{\gamma_1^2 \gamma_k^2}{(\gamma_1-\gamma_k)^4} \delta_{k \geq 2} + 
	        \frac{1}{N}\sum_{j=2}^{\bar{M}} K_j \frac{\gamma_j^2 \gamma_k^2}{(\gamma_1-\gamma_j)^4} \delta_{k = 1} \\[-2pt]
	       & \geq \frac{K_1}{N}\frac{\gamma_1^4}{(2\|\boldsymbol{R}\|)^4} \delta_{k \geq 2} + 
	        \frac{M-K_1}{N} \frac{\gamma_1^4}{(2\|\boldsymbol{R}\|)^4} \delta_{k = 1}
	    \end{align*}
	    and the lower bound follows from the fact that $\inf_M \gamma_1 >0$, $\sup_M \|\boldsymbol{R}\|< \infty$ and $0 < \inf_M K_1/M \leq \sup_M K_1/M < 1 $. Furthermore, the closed-form solution for $\upsilon_{\mathrm{c}}(k)$ in \eqref{eq:expr_upsilon_c} is given by
	    \begin{equation}
	    	\label{eq:expr_upsilon_c_bound}
	        \begin{aligned}
	            \upsilon_{\mathrm{c}}(k)  =& \gamma^2_k\left(\Phi_k(\mu_1)-\Phi_k(\gamma_1)\right)^2 \delta_{k\geq 2} \\
	            & +\Biggl( \gamma_1 \sum_{m=2}^{\bar{M}} \left(\Phi_1(\mu_m) - \Phi_1(\gamma_m)\right) \Biggr)^2 \delta_{k=1}
	         \end{aligned}
	    \end{equation}
	    where we have introduced $\Phi_k(x) =  \frac{x}{(\gamma_k - x)^2}$. Consider first the lower bound of $v_{\operatorname{c}}(k)$ in \eqref{eq:expr_upsilon_c_bound} for $k\geq 2$ in the oversampled case ($M\leq N$). Since $\Phi_k(x)$ is a convex function on the region $0<x<\gamma_k$, we can upper bound $\Phi_k(\gamma_1)>\Phi_k(\mu_1)+\Phi^{\prime}_k(\mu_1)(\gamma_1 - \mu_1)$ where $\Phi^{\prime}_k(x)$ is the derivative of $\Phi_k(x)$ w.r.t. $x$. This shows that
	    \begin{equation*}
	        (\Phi_k(\gamma_1)-\Phi_k(\mu_1))^2 \geq (\Phi^{\prime}_k(\mu_1))^2(\gamma_1 - \mu_1)^2 
	        \geq \frac{(\gamma_1 - \mu_1)^2}{\gamma_k^4} \geq \frac{(\gamma_1 - \mu_1)^2}{\|\boldsymbol{R}\|^4} 
	    \end{equation*}
	    where we have used the fact that in the oversampled case $\mu_{1} \geq 0$ and therefore $\Phi^{\prime}_k(\mu_1) \geq \Phi^{\prime}_k(0) = 1 / \gamma_k ^{2}$.  On the other hand, to see that $\inf_M | \gamma_1 - \mu_1| > 0$  we simply recall from the definition of $\mu_1$ in \eqref{eq: mu_values} that we can write
	    \begin{equation*}
	        \gamma_1 - \mu_1 = \Biggl( 1- \frac{1}{N}\sum_{j=2}^{\bar{M}} K_j \frac{\gamma_j}{\gamma_j - \mu_1} \Biggr)^{-1} \frac{K_1\gamma_1}{N} > \frac{K_1\gamma_1}{N} 
	    \end{equation*}
	    where in the last equation we used the fact that $\mu_1 < \gamma_1 < \gamma_j$ for $j=2,\ldots,\bar{M}$ and therefore the second term in the denominator is always positive. The fact that $\inf_M |\mu_1 - \gamma_1| > 0$ follows the boundedness of the spectral norm of $\boldsymbol{R}$ and Assumption \ref{as:asymptotic_growth}. This proves that in the oversampled case $v_{\text{c}}(k)$ in \eqref{eq:expr_upsilon_c_bound} for $k\geq 2$ is lower bounded by a constant independent of $M$.

	    Let us now consider the lower bound of $v_{\text{c}}(k)$ in \eqref{eq:expr_upsilon_c_bound} in the undersampled case ($M>N$) for $k\geq 2$. In the undersampled case we know that $\mu_{1} < 0 < \gamma_{1}$ and therefore
	    \begin{equation*}
			\phi_{k}(\gamma_{1}) - \phi_{k}(\mu_{1}) > \frac{\gamma_{1}}{\left(\gamma_{k}-\gamma_{1}\right)^{2}} > \frac{\gamma_{1}}{\left(2 \|\boldsymbol{R}\|\right)^{2}} = \frac{\gamma_{1}}{4 \|\boldsymbol{R}\|^{2}}
		\end{equation*}
		where we have used the fact that $-\phi_{k}(\mu_{1})>0$. Consequently, the lower bound on $v_{\text{c}}(k)$ in \eqref{eq:expr_upsilon_c_bound} for $k\geq 2$ in the undersampled case directly follows from the fact that $\inf_M \gamma_1 >0$ and $\sup_M \|\boldsymbol{R}\|< \infty$. This concludes the proof of Lemma \ref{lemma: lowerbound_Gamma}.
	\end{proof}

\section{Simulation Results}
	\label{sec: simulation_results}
	In this Section the predicted probability of resolution in \eqref{eq: predicted_probability_of_resolution} is compared to the simulated one. Consider a scenario with $K=2$ sources that are located at $\boldsymbol{\theta} = \lbrack 45^{\circ}, 50^{\circ}\rbrack^{\operatorname{T}}$ and a \ac{ULA} that is equipped with $M=15$ sensors. Since we consider only non-asymptotic scenarios with a finite number of sensors, we refrain from normalizing the steering vectors in all simulations. The transmitted signals are zero-mean with unit power and the \ac{SNR} is given by $\text{SNR}=1/\sigma^{2}$. The simulations are carried out for correlated signals with correlation coefficient $\rho=0.95$ as well as uncorrelated signals. The separation boundary is defined as the lowest $\text{SNR}$ or the smallest angular separation between both sources that provides separation between the eigenvalue cluster that is associated to the noise eigenvalue $\gamma_{1}$ and remaining eigenvalue clusters.  Hence, the separation boundary is given by the smallest \ac{SNR} or the smallest angular separation between both sources that allows to differentiate between noise and signal subspace and therefore satisfies Assumption \ref{as: subspace_separation}. All simulations are conducted for $10000$ Monte-Carlo trials.
	
	In Figure \ref{fig: probability_of_resolution_vs_snr_N_10}, \ref{fig: probability_of_resolution_vs_snr_N_15} and \ref{fig: probability_of_resolution_vs_snr_N_100} the probability of resolution is depicted for different \acp{SNR} and for $N=10$ (undersampled case), $N=15$ (special case where $M=N$) and $N=100$ (oversampled case) snapshots, respectively. In all three scenarios our prediction of the probability of resolution in \eqref{eq: predicted_probability_of_resolution} is very accurate since it is very close to the simulated one. Even in case of correlated sources and limited number of snapshots the proposed forecast of the probability of resolution in \eqref{eq: predicted_probability_of_resolution} provides a remarkably accurate description of the threshold effect. Furthermore, in all three scenarios the g-\ac{MUSIC} \ac{DoA} estimation method shows superior resolution capabilities than the conventional \ac{MUSIC} technique. 
	
	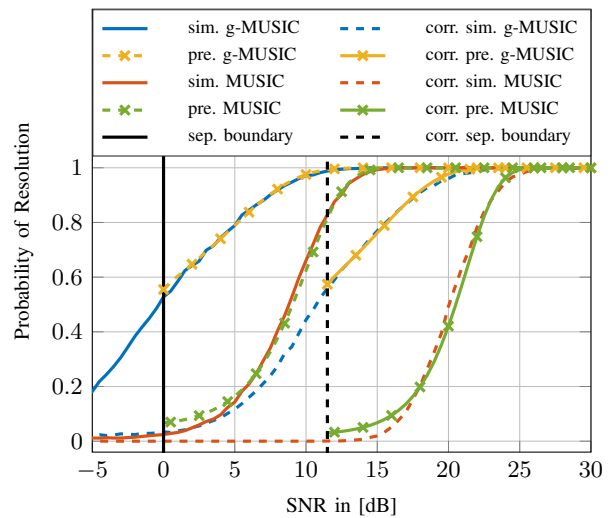
\begin{figure}
		\centering
		\input{figures/probability_of_resolution_vs_snr_M_15_N_10_K_2_doa_45_50_runs_10000.tex}
		\caption{Uncorrelated and Correlated Sources, Correlation Coefficient $\rho = 0.95$, $M=15$ Sensors, $N=10$ Snapshots}
		\label{fig: probability_of_resolution_vs_snr_N_10}
	\end{figure}
	
	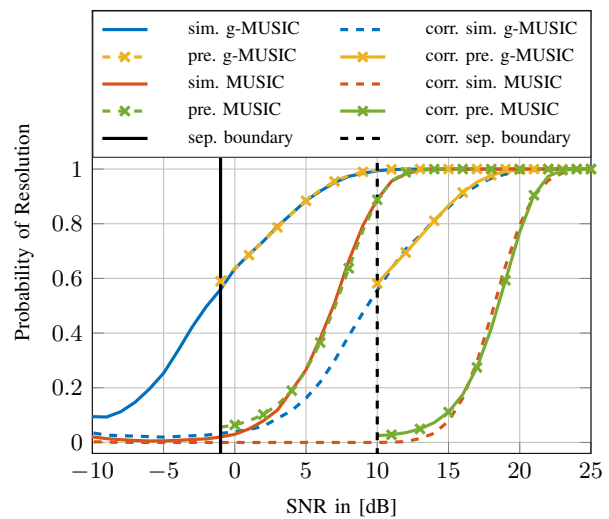
\begin{figure}
		\centering
		\input{figures/probability_of_resolution_vs_snr_M_15_N_15_doa_45_50_runs_10000.tex}
		\caption{Uncorrelated and Correlated Sources, Correlation Coefficient $\rho = 0.95$, $M=15$ Sensors, $N=15$ Snapshots}
		\label{fig: probability_of_resolution_vs_snr_N_15}
	\end{figure}
	
	\begin{figure}
		\centering
		\input{figures/probability_of_resolution_vs_snr_M_15_N_100_doa_45_50_runs_10000.tex}
		\caption{Uncorrelated and Correlated Sources, Correlation Coefficient $\rho = 0.95$, $M=15$ Sensors, $N=100$ Snapshots}
		\label{fig: probability_of_resolution_vs_snr_N_100}
	\end{figure}
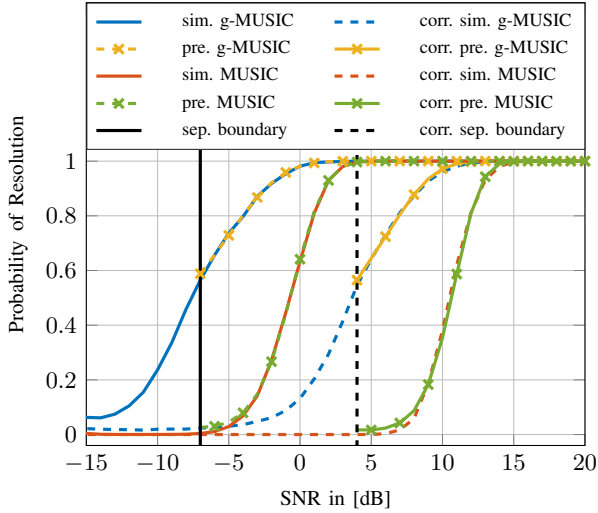
	
	In Figure \ref{fig: probability_of_resolution_vs_snapshots} the probability of resolution is depicted for different numbers of snapshots $N$. The \ac{SNR} is fixed to $\text{SNR} = 6\text{dB}$ and the source signals are uncorrelated. The proposed analytic expression for the probability of resolution very accurately describes the empirical simulated probability of resolution of \ac{MUSIC} as well as g-\ac{MUSIC}. Furthermore, g-\ac{MUSIC} outperforms conventional \ac{MUSIC} especially under harsh conditions with limited number of snapshots. 
	
	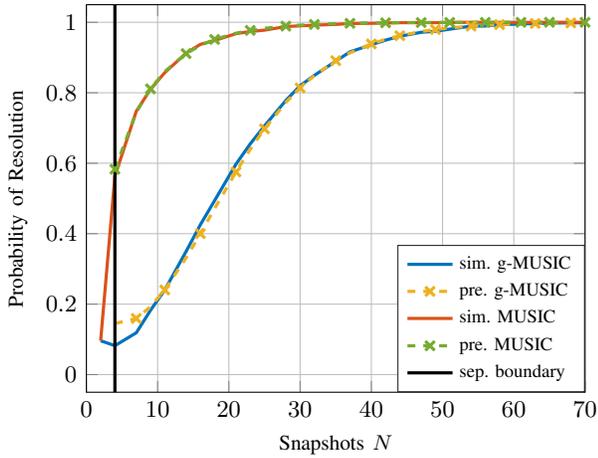
\begin{figure}
		\centering
		\input{figures/probability_of_resolution_vs_snapshots_snr_6dB_M_15_K_2_doa_45_50_rho_0_runs_10000.tex}
		\caption{Uncorrelated Sources, $M=15$ Sensors, $\text{SNR}=6\text{dB}$}
		\label{fig: probability_of_resolution_vs_snapshots}
	\end{figure}
	
	In Figure \ref{fig: probability_of_resolution_vs_angular_separation_N_15} the probability of resolution is depicted for different angular separation $\Delta\vartheta$ between two sources. The \ac{SNR} is fixed to $\text{SNR} = 2\text{dB}$ and a total of $N=15$ snapshots are considered. The transmitted signals are uncorrelated and the two sources are located at $\boldsymbol{\theta}=\lbrack 45^{\circ}, 45^{\circ}+\Delta\vartheta\rbrack^{\operatorname{T}}$ with $\Delta\vartheta \in \lbrack 0.2^{\circ}, 10^{\circ}\rbrack$. Furthermore, a \ac{ULA} with $M=15$ antennas is used. It can be observed that the proposed prediction of the probability of resolution is almost identical to the actual one. Especially in difficult scenarios with low \ac{SNR} and closely spaced sources g-\ac{MUSIC} is superior to conventional \ac{MUSIC}. 

	\begin{figure}
		\centering
		\input{figures/probability_of_resolution_vs_angular_separation_M_15_N_15_doa_45_runs_10000.tex}
		\caption{Uncorrelated Sources Located at $\boldsymbol{\theta} = \lbrack 45^{\circ}, 45^{\circ}+\Delta\vartheta \rbrack^{\operatorname{T}}$, $M=15$ Sensors, $N=15$ Snapshots, $\text{SNR}=2\text{dB}$}
		\label{fig: probability_of_resolution_vs_angular_separation_N_15}
	\end{figure}
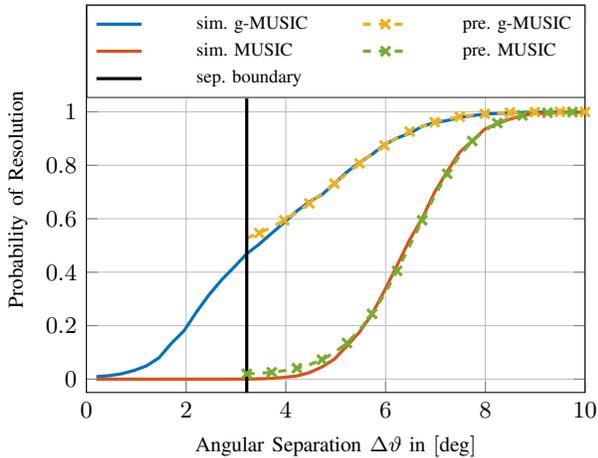

\section{Conclusion}
	\label{sec: conclusion}
	In this article the asymptotic stochastic behavior of the conventional \ac{MUSIC} and the g-\ac{MUSIC} cost function is investigated in the asymptotic regime where the number of snapshots and the number of sensors go to infinity at the same rate. Using tools from \ac{RMT} the finite dimensional distribution of the random \ac{MUSIC} and g-\ac{MUSIC} cost function is derived and shown to be asymptotically jointly Gaussian distributed. Furthermore, the resolution capabilities of both \ac{MUSIC} \ac{DoA} estimation methods is analyzed based on the asymptotic stochastic behavior of their cost functions. An analytic expression for the probability of resolution is provided that allows to predict the probability of resolution in the threshold region and thus provides an accurate description of the outlier production mechanism.

\appendices 
		\section{Determination of the Asymptotic Covariance of the MUSIC Cost Function}
			\label{app: second_order_music}
			In case of the conventional \ac{MUSIC} cost function $\bar{h}(z(\omega_{2})) = 1$. Hence, the complex contour integral in \eqref{eq: second_order_general_separated_omega_2} yields
			\begin{equation}
				\mathcal{I}_{r,k}(\omega_{1}) = \frac{1}{2 \pi \mathrm{j}}\oint_{\mathcal{C}_{\omega_{2}}} \frac{\omega_{2}}{z(\omega_{2})} \frac{\partial z(\omega_{2})}{\partial \omega_{2}} \frac{\frac{1}{\left(\gamma_{r}-\omega_{2}\right)\left(\gamma_{k}-\omega_{2}\right)}}{1-\Omega(\omega_{1},\omega_{2})}\mathrm{d}\omega_{2} 
				\label{eq: second_order_general_separated_omega_2_music}
			\end{equation}
			and can be solved in closed-form by applying conventional residue calculus \cite{a54}. A closed-form expression for $\mathcal{I}_{r,k}(\omega_{1})$ is obtained by summing the residues of the integrand in \eqref{eq: second_order_general_separated_omega_2_music} evaluated at all singularities that lie within the complex contour $\mathcal{C}_{\omega_{2}}$. To begin with, we consider the oversampled case where $M\leq N$. It can be seen, that the integrand of $\mathcal{I}_{r,k}(\omega_{1})$ in \eqref{eq: second_order_general_separated_omega_2_music} exhibits three different types of singularities that lie inside the complex contour $\mathcal{C}_{\omega_{2}}$. The first type of singularities corresponds to the poles of $\frac{\omega_{2}}{z(\omega_{2})}$, which are denoted by $\mu_{1}<\mu_{2}<\dots <\mu_{\bar{M}}$ in \eqref{eq: mu_values}. The corresponding residue w.r.t. $\omega_{2}$ evaluated at $\mu_{t}$ yields 
			\begin{equation}
				\res{\frac{\omega_{2}}{z(\omega_{2})} \frac{\partial z(\omega_{2})}{\partial \omega_{2}} \frac{\frac{1}{\left(\gamma_{r}-\omega_{2}\right)\left(\gamma_{k}-\omega_{2}\right)}}{1-\Omega(\omega_{1},\omega_{2})}}{\mu_{t}} =  \frac{\frac{\mu_{t}}{(\gamma_{r}-\mu_{t})(\gamma_{k}-\mu_{t})}}{1-\Omega(\omega_{1},\mu_{t})}
				\label{eq: music_mathcal_i_residue_a}
			\end{equation}
			where we have used $\Omega(\omega_{1},\omega_{2})$ in \eqref{eq: omega_captial_theorem}, and $\frac{\partial z(\omega)}{\partial \omega}$ in \eqref{eq: derivative_z_of_omega_wrt_omega}. However, only $\mu_{1}$ lies within the contour $\mathcal{C}_{\omega_{2}}$ which is why only the residue evaluated at $\mu_{1}$ contributes to the solution of $\mathcal{I}_{r,k} (\omega_{1})$ in \eqref{eq: second_order_general_separated_omega_2_music}. 
			
			The second type of singularities corresponds to the roots of $1-\Omega(\omega_{1},\omega_{2})$ where $\Omega(\omega_{1},\omega_{2})$ is given in \eqref{eq: omega_captial_theorem}. The complex-valued roots are denoted by $\varphi_{r}(\omega_{1})$ for $r=1,\dots,\bar{M}$ and sorted according to their real-part in ascending order $\rea{\varphi_{1}(\omega_{1})} < \dots < \rea{\varphi_{\bar{M}}(\omega_{1})}$ and given by the solutions  of the polynomial equation in $\varphi(\omega_{1})$ 
			\begin{equation}
				\frac{1}{N} \sum_{l=1}^{\bar{M}}\frac{K_{l}\gamma_{l}^{2}}{(\gamma_{l}-\omega_{1})(\gamma_{l}-\varphi(\omega_{1}))} = 1.
				\label{eq: varphi_values}
			\end{equation}
			The corresponding residue evaluated at $\varphi_{t}(\omega_{1})$ is given by 
			\begin{equation}
				\begin{aligned}
					&\res{\frac{\omega_{2}}{z(\omega_{2})} \frac{\partial z(\omega_{2})}{\partial \omega_{2}} \frac{\frac{1}{\left(\gamma_{r}-\omega_{2}\right)\left(\gamma_{k}-\omega_{2}\right)}}{1-\Omega(\omega_{1},\omega_{2})}}{\varphi_{t}(\omega_{1})} \\[-2pt]
					=& \frac{\varphi_{t}(\omega_{1})-\omega_{1}}{\left(\gamma_{r}-\varphi_{t}(\omega_{1})\right)\left(\gamma_{k}-\varphi_{t}(\omega_{1})\right)} \frac{1}{1-\frac{1}{N}\sum_{m=1}^{\bar{M}}\frac{K_{m}\gamma_{m}}{\gamma_{m}-\varphi_{t}(\omega_{1})}}.
				\end{aligned}
				\label{eq: music_mathcal_i_residue_b}
			\end{equation}
			\begin{lemma}
				Assuming that $\omega_{1}$ is located outside the contour generated by the parameterization $x \mapsto \omega_{1}(x)$ in \cite[Remark 3]{a20}, there exists exactly one solution of the equation in \eqref{eq: varphi_values} that is enclosed by the contour $\mathcal{C}_{\omega_{2}}$, namely $\varphi_{1}(\omega_{1})$.
			\end{lemma}
			\begin{proof}
			    Since $\omega_{1}$ is located outside the contour $\tilde{\Psi}(\omega_{1})<1$ where $\tilde{\Psi}(\omega)$ is given in \eqref{eq: definition_tilde_Psi}. It can be observed that for a fixed $\omega_{1}$ the $\varphi_{r}(\omega_{1})$ values are the zeros of the function $f(\varphi) = 1-\Omega(\omega_{1},\varphi)$ where $\Omega(\omega_{1},\omega_{2})$ is given in \eqref{eq: omega_captial_theorem}. By Cauchy-Schwarz inequality it follows that, if $\varphi \in \mathcal{C}_{\omega_{2}}$, we have $|1-f(\varphi)|^{2} \leq \tilde{\Psi}(\omega_{1}) \tilde{\Psi}(\varphi) < 1$ (see Lemma \ref{lemma: bound_Omega}). It can also be seen that $f(\varphi)$ has no singularities or zeros lying directly on the contour $\mathcal{C}_{\omega_{2}}$, so that by Rouch\'e's theorem $f(\varphi)$ has the same number of zeros and poles inside $\mathcal{C}_{\omega_{2}}$. The poles of $f(\varphi)$ are located at the true eigenvalues $\gamma_{r}$ for $r=1,\dots,\bar{M}$. However, only the smallest eigenvalue $\gamma_{1}$ is enclosed by $\mathcal{C}_{\omega_{2}}$. Hence, it follows that only one zero of $1-f(\varphi)$ is enclosed by the contour $\mathcal{C}_{\omega_{2}}$, namely $\varphi_{1}(\omega_{1})$.
			\end{proof}
			Consequently, only the residue evaluated at $\varphi_{1}(\omega_{1})$  contributes to the solution of $\mathcal{I}_{r,k}(\omega_{1})$ in \eqref{eq: second_order_general_separated_omega_2_music}. 
		
			The third type of singularities corresponds to the poles at the true eigenvalues $\gamma_{t}$ for $t=1,\dots,\bar{M}$. However, we have to distinguish between poles of order one and order two. The residue with respect to $\omega_{2}$ evaluated at $\gamma_{t}$ yields 
			\begin{equation}
				\begin{aligned}
					&\res{\frac{\omega_{2}}{z(\omega_{2})} \frac{\partial z(\omega_{2})}{\partial \omega_{2}} \frac{\frac{1}{(\gamma_{r}-\omega_{2})(\gamma_{k}-\omega_{2})}}{1-\Omega(\omega_{1},\omega_{2})} }{\gamma_{t}}\\
					=&\begin{cases}
						\frac{N}{K_{t}\gamma_{t}}\frac{\gamma_{t}-\omega_{1}}{\gamma_{k}-\gamma_{t}} , & \text{for } t = r \neq k\\
						\frac{N}{K_{t}\gamma_{t}}\frac{\gamma_{t}-\omega_{1}}{\gamma_{r}-\gamma_{t}}, & \text{for } t = k \neq r\\
						\alpha_{t}(\omega_{1}) , & \text{for } r=k=t
					\end{cases}
				\end{aligned}
				\label{eq: music_mathcal_i_residue_c}
			\end{equation}
			where 
			\begin{equation}
				\begin{aligned}
					\alpha_{t}(\omega_{1}) =& \frac{(\gamma_{t}-\omega_{1})}{\left(\frac{1}{N}K_{t}\gamma_{t}\right)^{2}} \Biggl( 1- \frac{1}{N}\sum_{m=1 \atop m\neq t}^{\bar{M}} \frac{K_{m}\gamma_{m}}{\gamma_{m}-\gamma_{t}}\Biggr)\\
					& +\frac{\left(\gamma_{t}-\omega_{1}\right)^{2}}{\gamma_{t}\left(\frac{1}{N} K_{t}\gamma_{t}\right)^{2}} \Biggl(1 -\frac{1}{N}\sum_{m=1\atop m\neq t}^{\bar{M}} \frac{K_{m}\gamma_{m}^{2}}{(\gamma_{m}-\omega_{1})(\gamma_{m}-\gamma_{t})}\Biggr).
				\end{aligned}
				\label{eq: zeta_music_second_order}
			\end{equation}
			It can be observed that only the noise eigenvalue $\gamma_{1}$ is enclosed by the contour $\mathcal{C}_{\omega_{2}}$. Correspondingly, the closed-form solution for the complex contour integral $\mathcal{I}_{r,k}(\omega_{1})$ in \eqref{eq: second_order_general_separated_omega_2_music} is obtained by taking the negative sum (negative sum because of the negatively orientated contour $\mathcal{C}_{\omega_{2}}$) of the residues in \eqref{eq: music_mathcal_i_residue_a}, \eqref{eq: music_mathcal_i_residue_b} and \eqref{eq: music_mathcal_i_residue_c} evaluated at all singularities that lie inside the complex contour $\mathcal{C}_{\omega_{2}}$ and yields
			\begin{align}
					\mathcal{I}_{r,k}(\omega_{1}) =& -\frac{\mu_{1}}{(\gamma_{r}-\mu_{1})(\gamma_{k}-\mu_{1})} \frac{1}{1-\Omega(\omega_{1},\mu_{1})}-\alpha_{1}(\omega_{1})\delta_{r=k=1} \nonumber\\
					&- \frac{\varphi_{1}(\omega_{1})-\omega_{1}}{(\gamma_{r}-\varphi_{1}(\omega_{1}))(\gamma_{k}-\varphi_{1}(\omega_{1}))} \frac{1}{1-\frac{1}{N}\sum_{m=1}^{\bar{M}} \frac{K_{m}\gamma_{m}}{\gamma_{m}-\varphi_{1}(\omega_{1})}}\nonumber\\
					&-\frac{N}{K_{1}\gamma_{1}}\frac{\gamma_{1}-\omega_{1}}{\gamma_{k}-\gamma_{1}} \delta_{r=1\neq k} -\frac{N}{K_{1}\gamma_{1}}\frac{\gamma_{1}-\omega_{1}}{\gamma_{r}-\gamma_{1}}\delta_{k=1\neq r}.
				\label{eq: closed_form_expression_mathcal_I_music}
			\end{align}
			In the following we substitute the closed-form expression for $\mathcal{I}_{r,k}(\omega_{1})$ in \eqref{eq: closed_form_expression_mathcal_I_music} into the general expression for the real-valued weights $\xi(r,k)$ in \eqref{eq: second_order_general_separated_omega_1} such that
			\begin{equation}
				\begin{aligned}
					 \frac{1}{\gamma_{r}\gamma_{k}}\xi_{\operatorname{c}}(r,k) = \frac{-1}{2\pi \mathrm{j}} \oint_{\mathcal{C}_{\omega_{1}}}  \frac{\omega_{1}}{z(\omega_{1})} \frac{\partial z(\omega_{1})}{\partial \omega_{1}} \frac{1}{(\gamma_{r}-\omega_{1})(\gamma_{k}-\omega_{1})} & \\[-2pt]
					\times \Biggl( \alpha_{1}(\omega_{1})\delta_{r=k=1}+\frac{\mu_{1}}{(\gamma_{r}-\mu_{1})(\gamma_{k}-\mu_{1})} \frac{1}{1-\Omega(\omega_{1},\mu_{1})}&\\[-2pt]
					+ \frac{N}{K_{1}\gamma_{1}}\frac{\gamma_{1}-\omega_{1}}{\gamma_{k}-\gamma_{1}} \delta_{r=1\neq k} +\frac{N}{K_{1}\gamma_{1}}\frac{\gamma_{1}-\omega_{1}}{\gamma_{r}-\gamma_{1}}\delta_{k=1\neq r} &\\[-2pt]
					+\frac{\varphi_{1}(\omega_{1})-\omega_{1}}{(\gamma_{r}-\varphi_{1}(\omega_{1}))(\gamma_{k}-\varphi_{1}(\omega_{1}))} \frac{1}{1-\frac{1}{N}\sum_{m=1}^{\bar{M}} \frac{K_{m}\gamma_{m}}{\gamma_{m}-\varphi_{1}(\omega_{1})}}&\Biggr) \mathrm{d}\omega_{1}
				\end{aligned}
				\label{eq: second_order_asymptotic_behavior_appendix_proof_music}
			\end{equation}
			where we have used $\bar{h}(z(\omega_{1})) =1$, $\Omega(\omega_{1},\omega_{2})$ in \eqref{eq: omega_captial_theorem} and $\alpha_{t}(\omega_{1})$ in \eqref{eq: zeta_music_second_order}. The integrand in \eqref{eq: second_order_asymptotic_behavior_appendix_proof_music} exhibits four different types of singularities. The first group of singularities corresponds to the poles that are located at the true eigenvalues $\gamma_{t}$. Hence, the residue of the first part of the integrand in \eqref{eq: second_order_asymptotic_behavior_appendix_proof_music} evaluated at $\gamma_{t}$ is given by 
			\begin{equation}
				\begin{aligned}
					&\res{\frac{\omega_{1}}{z(\omega_{1})} \frac{\partial z(\omega_{1})}{\partial \omega_{1}} \frac{\alpha_{t}(\omega_{1})}{(\gamma_{r}-\omega_{1})(\gamma_{k}-\omega_{1})}\delta_{r=k=t}}{\gamma_{t}}\\[-2pt]
					=& -\frac{N}{K_{t}} \frac{1}{\left(\frac{1}{N}K_{t}\gamma_{t}\right)^{2}} \Biggl(1-\frac{1}{N} \sum_{m=1\atop m\neq t}^{\bar{M}} \frac{K_{m}\gamma_{m}}{\gamma_{m}-\gamma_{t}}\Biggr)^{2} \delta_{r=k=t} \\[-2pt]
					& -\frac{1}{\left(\frac{1}{N}K_{t}\gamma_{t}\right)^{2}} \Biggl(1-\frac{1}{N} \sum_{m=1 \atop m \neq t}^{\bar{M}} \frac{K_{m}\gamma_{m}^{2}}{\left(\gamma_{m}-\gamma_{t}\right)^{2}}\Biggr) \delta_{r=k=t}
				\end{aligned}
				\label{eq: music_xi_residue_gamma_a}
			\end{equation}
			whereas the residue of the second part of the integrand in \eqref{eq: second_order_asymptotic_behavior_appendix_proof_music} yields 
			\begin{equation}
				\begin{aligned}
					&\res{\frac{\frac{\mu_{t}}{(\gamma_{r}-\mu_{t})(\gamma_{k}-\mu_{t})}}{1-\Omega(\omega_{1},\mu_{t})}\frac{\frac{\omega_{1}}{z(\omega_{1})} \frac{\partial z(\omega_{1})}{\partial \omega_{1}}}{(\gamma_{r}-\omega_{1})(\gamma_{k}-\omega_{1})}}{\gamma_{t}}\\[-2pt]
					=& \begin{cases}
						\frac{1}{\frac{1}{N}K_{t}\gamma_{t}} \frac{\mu_{t}}{\gamma_{k}-\mu_{t}}\frac{1}{\gamma_{k}-\gamma_{t}}, & \text{for } t=r\neq k \\
						\frac{1}{\frac{1}{N}K_{t}\gamma_{t}} \frac{\mu_{t}}{\gamma_{r}-\mu_{t}} \frac{1}{\gamma_{r}-\gamma_{t}}, & \text{for } t=k\neq r \\
						\beta_{t}, & \text{for } t=k=r
					\end{cases}
				\end{aligned}
				\label{eq: music_xi_residue_gamma_b}
			\end{equation}
			where 
			\begin{equation*}
				\begin{aligned}
					\beta_{t} =& \frac{1}{\left(\frac{1}{N} K_{t}\gamma_{t}\right)^{2}} \frac{\mu_{t}}{\gamma_{t}-\mu_{t}} \Biggl( 1- \frac{1}{N}\sum_{m=1 \atop m\neq t}^{\bar{M}} \frac{K_{m}\gamma_{m}}{\gamma_{m}-\gamma_{t}} \Biggr) \\[-2pt]
					&+ \frac{\mu_{t}}{\gamma_{t}\left(\frac{1}{N}K_{t}\gamma_{t}\right)^{2}} \Biggl(1-\frac{1}{N} \sum_{m=1 \atop m \neq t}^{\bar{M}} \frac{K_{m}\gamma_{m}^{2}}{(\gamma_{m}-\gamma_{t})(\gamma_{m}-\mu_{t})}\Biggr).
				\end{aligned}
			\end{equation*}
			The residue of the third part of the integrand in \eqref{eq: second_order_asymptotic_behavior_appendix_proof_music} evaluated at $\gamma_{t}$ is computed as follows
			\begin{equation}
				\begin{aligned}
					& \res{\frac{N}{K_{t}\gamma_{t}} \frac{\omega_{1}}{z(\omega_{1})} \frac{\partial z(\omega_{1})}{\partial \omega_{1}} \frac{\frac{\gamma_{t}-\omega_{1}}{\gamma_{k}-\gamma_{t}}\delta_{r=t\neq k} + \frac{\gamma_{t}-\omega_{1}}{\gamma_{r}-\gamma_{t}}\delta_{k=t\neq r}}{(\gamma_{r}-\omega_{1})(\gamma_{k}-\omega_{1})}  }{\gamma_{t}}\\
					&= -\frac{N}{K_{t}} \frac{1}{(\gamma_{k}-\gamma_{t})^{2}} \delta_{r=t\neq k} - \frac{N}{K_{t}} \frac{1}{(\gamma_{r}-\gamma_{t})^{2}}\delta_{k=t\neq r}.
				\end{aligned}
				\label{eq: music_xi_residue_gamma_c}
			\end{equation}
			Since only the noise eigenvalue $\gamma_{1}$ is enclosed by the contour $\mathcal{C}_{\omega_{1}}$, only the residues in \eqref{eq: music_xi_residue_gamma_a}, \eqref{eq: music_xi_residue_gamma_b} and \eqref{eq: music_xi_residue_gamma_c} evaluated at $\gamma_{1}$ contribute to the result of ${\xi}_{\operatorname{c}}(r,k)$ in \eqref{eq: second_order_asymptotic_behavior_appendix_proof_music}. 
			
			The second type of singularities belongs to the poles of $\frac{ \omega_{1}}{ z(\omega_{1})}$, which are located at the $\mu$-values that are defined in \eqref{eq: mu_values}. The corresponding residue of the first part of the integrand in \eqref{eq: second_order_asymptotic_behavior_appendix_proof_music} evaluated at $\mu_{t}$ is given by 
			\begin{equation}
				\begin{aligned}
					&\res{\frac{\omega_{1}}{z(\omega_{1})} \frac{\partial z(\omega_{1})}{\partial \omega_{1}} \frac{\alpha_{t}(\omega_{1})}{(\gamma_{r}-\omega_{1})(\gamma_{k}-\omega_{1})}\delta_{r=k=t}}{\mu_{t}}\\[-2pt]
					=& \frac{\mu_{t}}{\gamma_{t}-\mu_{t}} \frac{1}{\left(\frac{1}{N}K_{t}\gamma_{t}\right)^{2}} \Biggl( 1- \frac{1}{N} \sum_{m=1 \atop m \neq t}^{\bar{M}} \frac{K_{m}\gamma_{m}}{\gamma_{m}-\gamma_{t}} \Biggr) \delta_{r=k=t} \\[-2pt]
					&+ \frac{\mu_{t}}{\gamma_{t}\left(\frac{1}{N}K_{t}\gamma_{t}\right)^{2}} \Biggl( 1- \frac{1}{N}\sum_{m=1 \atop m\neq t}^{\bar{M}} \frac{K_{m}\gamma_{m}^{2}}{(\gamma_{m}-\mu_{t})(\gamma_{m}-\gamma_{t})}\Biggr)\delta_{r=k=t}.
				\end{aligned}
				\label{eq: music_xi_residue_mu_a}
			\end{equation}
			Furthermore, the residues of the second and third part of the integrand in \eqref{eq: second_order_asymptotic_behavior_appendix_proof_music} evaluated at $\mu_{t}$ are given by 
			\begin{equation}
				\begin{aligned}
					&\res{\frac{\frac{\mu_{t}}{(\gamma_{r}-\mu_{t})(\gamma_{k}-\mu_{t})}}{1-\Omega(\omega_{1},\mu_{t})}\frac{\frac{\omega_{1}}{z(\omega_{1})} \frac{\partial z(\omega_{1})}{\partial \omega_{1}}}{(\gamma_{r}-\omega_{1})(\gamma_{k}-\omega_{1})}}{\mu_{t}}\\[-2pt]
					=& \frac{\mu_{t}^{2}}{(\gamma_{r}-\mu_{t})^{2}(\gamma_{k}-\mu_{t})^{2} } \frac{1}{\left( 1- \frac{1}{N} \sum_{m=1}^{\bar{M}} \frac{K_{m}\gamma_{m}^{2}}{(\gamma_{m}-\mu_{t})^{2}}\right)}
				\end{aligned}
				\label{eq: music_xi_residue_mu_b}
			\end{equation}
			and
			\begin{equation}
				\begin{aligned}
					& \res{\frac{N}{K_{t}\gamma_{t}} \frac{\omega_{1}}{z(\omega_{1})} \frac{\partial z(\omega_{1})}{\partial \omega_{1}} \frac{\frac{\gamma_{t}-\omega_{1}}{\gamma_{k}-\gamma_{t}}\delta_{r=t\neq k} + \frac{\gamma_{t}-\omega_{1}}{\gamma_{r}-\gamma_{t}}\delta_{k=t\neq r}}{(\gamma_{r}-\omega_{1})(\gamma_{k}-\omega_{1})}  }{\mu_{t}}\\[-2pt]
					=& \frac{N}{K_{t}}\frac{1}{(\gamma_{k}-\gamma_{t})\gamma_{t}} \frac{\mu_{t}}{\gamma_{k}-\mu_{t}}\delta_{r=t\neq k} + \frac{N}{K_{t}}\frac{1}{(\gamma_{r}-\gamma_{t})\gamma_{t}}\frac{\mu_{t}}{\gamma_{r}-\mu_{t}}\delta_{k=t \neq r}.
				\end{aligned}
				\label{eq: music_xi_residue_mu_c}
			\end{equation}
			However, only $\mu_{1}$ is enclosed by the contour $\mathcal{C}_{\omega_{1}}$ which is why only the residues in \eqref{eq: music_xi_residue_mu_a}, \eqref{eq: music_xi_residue_mu_b} and \eqref{eq: music_xi_residue_mu_c} evaluated at $\mu_{1}$ contribute to the final result of ${\xi}_{\operatorname{c}}(r,k)$ in \eqref{eq: second_order_asymptotic_behavior_appendix_proof_music}.
			
			The third type of singularities of the integrand in \eqref{eq: second_order_asymptotic_behavior_appendix_proof_music} is given by the solutions to the polynomial equation $1=\Omega(\omega_{1},\mu_{t})$ in $\omega_{1}$ where $\Omega(\omega_{1},\omega_{2})$ is defined in \eqref{eq: omega_captial_theorem}. However, using Cauchy's argument principle it can be shown that the solutions to $1=\Omega(\omega_{1},\mu_{t})$ are located outside the contour $\mathcal{C}_{\omega_{1}}$ and therefore do not contribute to the  solution of ${\xi}_{\operatorname{c}}(r,k)$ in \eqref{eq: second_order_asymptotic_behavior_appendix_proof_music}. The following lemma is in line. 
			
			\begin{lemma}
				Function $f(\omega_{1})=1-\Omega(\omega_{1},\mu_{t})$ with $\Omega(\omega_{1},\omega_{2})$ in \eqref{eq: omega_captial_theorem} does not exhibit any zeros inside the contour $\mathcal{C}_{\omega_{1}}$.
			\end{lemma}
			\begin{proof}
				Let $P$ denote the number of poles and $Z$ the number of zeros of $f(\omega_{1})$ that are located inside the contour $\mathcal{C}_{\omega_{1}}$, then according to Cauchy's argument principle
				\begin{equation}
					-\frac{1}{2\pi \mathrm{j}} \oint_{\mathcal{C}_{\omega_{1}}} \frac{f^{\prime}(\omega_{1})}{f(\omega_{1})} \mathrm{d}\omega_{1} = Z - P,
					\label{eq: cauchy_argument_principle}
				\end{equation}
				where the first order derivative of $f(\omega_{1})$ w.r.t. $\omega_{1}$ is given by
				\begin{equation}
					f^{\prime}(\omega_{1})=\frac{\partial f(\omega_{1})}{\partial \omega_{1}} = -\frac{1}{N}\sum_{m=1}^{\bar{M}}\frac{K_{m}\gamma_{m}^{2}}{(\gamma_{m}-\omega_{1})^{2}(\gamma_{m}-\mu_{t})} .
					\label{eq: appendix_first_order_derivative_f_one_minus_omega}
				\end{equation}
				Substituting $f(\omega_{1}) = 1- \Omega(\omega_{1},\mu_{t})$ and its first order derivative in \eqref{eq: appendix_first_order_derivative_f_one_minus_omega} into the argument principle in \eqref{eq: cauchy_argument_principle} yields
				\begin{equation*}
					\begin{aligned}
						&\frac{1}{2\pi \mathrm{j}} \oint_{\mathcal{C}_{\omega_{1}}} \frac{\frac{1}{N}\sum_{m=1}^{\bar{M}} \frac{K_{m}\gamma_{m}^{2} + \gamma_{m}-\gamma_{m}}{(\gamma_{m}-\omega_{1})^{2}(\gamma_{m}-\mu_{t})}}{1-\frac{1}{N}\sum_{m=1}^{\bar{M}} \frac{K_{m}\gamma_{m}^{2}}{(\gamma_{m}-\omega_{1})(\gamma_{m}-\mu_{t})}} \frac{\mu_{t}-\omega_{1}}{\mu_{t}-\omega_{1}} \mathrm{d}\omega_{1} \\[-1pt]
						=& \frac{1}{2\pi \mathrm{j}} \oint_{\mathcal{C}_{\omega_{1}}} \frac{1}{\mu_{t}-\omega_{1}} \frac{\frac{1}{N}\sum_{m=1}^{\bar{M}} \frac{K_{m}\gamma_{m}^{2} (\mu_{t}-\gamma_{m})}{(\gamma_{m}-\omega_{1})^{2}(\gamma_{m}-\mu_{t})}}{1-\frac{1}{N}\sum_{m=1}^{\bar{M}} \frac{K_{m}\gamma_{m}^{2}}{(\gamma_{m}-\omega_{1})(\gamma_{m}-\mu_{t})}}\mathrm{d}\omega_{1} \\[-1pt]
						&+ \frac{1}{2\pi \mathrm{j}} \oint_{\mathcal{C}_{\omega_{1}}} \frac{1}{\mu_{t}-\omega_{1}}  \frac{\frac{1}{N} \sum_{m=1}^{\bar{M}} \frac{K_{m}\gamma_{m}^{2}(\gamma_{m}-\omega_{1})}{(\gamma_{m}-\omega_{1})^{2}(\gamma_{m}-\mu_{t})}}{1-\frac{1}{N}\sum_{m=1}^{\bar{M}} \frac{K_{m}\gamma_{m}^{2}}{(\gamma_{m}-\omega_{1})(\gamma_{m}-\mu_{t})}} \mathrm{d}\omega_{1}
					\end{aligned}
				\end{equation*}
				where we have added and subtracted $\gamma_{m}$ and multiplied and divided by $\mu_{t}-\omega_{1}$. Next, we add and subtract $1$ to obtain
				\begin{equation}
					\begin{aligned}
						&\frac{1}{2\pi \mathrm{j}} \oint_{\mathcal{C}_{\omega_{1}}} \frac{1}{\mu_{t}-\omega_{1}} \frac{1-\frac{1}{N} \sum_{m=1}^{\bar{M}} \frac{K_{m}\gamma_{m}^{2}}{(\gamma_{m}-\omega_{1})^{2}}}{1-\frac{1}{N}\sum_{m=1}^{\bar{M}} \frac{K_{m}\gamma_{m}^{2}}{(\gamma_{m}-\omega_{1})(\gamma_{m}-\mu_{t})}}\mathrm{d}\omega_{1}\\[-1pt]
						&+\frac{1}{2\pi \mathrm{j}} \oint_{\mathcal{C}_{\omega_{1}}} \frac{1}{\mu_{t}-\omega_{1}} \frac{\frac{1}{N} \sum_{m=1}^{\bar{M}} \frac{K_{m}\gamma_{m}^{2}}{(\gamma_{m}-\omega_{1})(\gamma_{m}-\mu_{t})}-1}{1-\frac{1}{N}\sum_{m=1}^{\bar{M}} \frac{K_{m}\gamma_{m}^{2}}{(\gamma_{m}-\omega_{1})(\gamma_{m}-\mu_{t})}} \mathrm{d}\omega_{1}\\[-1pt]
						=& \frac{1}{2\pi \mathrm{j}} \oint_{\mathcal{C}_{\omega_{1}}} \left \lbrack -\frac{1}{z_{1}} \left(1-\frac{1}{N} \sum_{m=1}^{\bar{M}} \frac{K_{m}\gamma_{m}^{2}}{(\gamma_{m}-\omega_{1})^{2}}\right)-\frac{1}{\mu_{t}-\omega_{1}} \right \rbrack \mathrm{d}\omega_{1}
					\end{aligned} 
					\label{eq: appendix_argument_principle_proof_2}
				\end{equation}
				where we have used 
				\begin{equation*}
					z_{1} = (\omega_{1}-\mu_{t}) \left \lbrack 1- \frac{1}{N}\sum_{m=1}^{\bar{M}} \frac{K_{m}\gamma_{m}^{2}}{(\gamma_{m}-\omega_{1})(\gamma_{m}-\mu_{t})}\right \rbrack
				\end{equation*}
				which follows by subtracting $0 = \mu_{t}\left(1-\frac{1}{N}\sum_{m=1}^{\bar{M}} \frac{K_{m}\gamma_{m}}{\gamma_{m}-\mu_{t}}\right)$ from the definition of $z_{1}$ in \eqref{eq: definition_omega_z}. Applying the change of variables $\mathrm{d}\omega_{1} = \frac{\partial \omega(z_{1})}{\partial z_{1}}\mathrm{d}z_{1}$ in \eqref{eq: appendix_argument_principle_proof_2} one obtains
				\begin{align*}
					&-\frac{1}{2\pi \mathrm{j}} \oint_{\mathcal{C}_{z_{1}}} \frac{1}{z_{1}} \mathrm{d}z_{1} - \frac{1}{2\pi \mathrm{j}} \oint_{\mathcal{C}_{\omega_{1}}} \frac{1}{\mu_{t}-\omega_{1}} \mathrm{d}\omega_{1}\\
					=&\res{\frac{1}{z_{1}}}{0} + \res{\frac{1}{\mu_{t}-\omega_{1}}}{\mu_{t}} = 1-1 =0
				\end{align*}	
				for \eqref{eq: cauchy_argument_principle}. Furthermore, it can be observed that in the oversampled case if $\mathcal{C}_{z_{1}}$ encloses zero, also $\mathcal{C}_{\omega_{1}}$ encloses zero. Correspondingly, $f(\omega_{1})$ exhibits a single zero at $\omega_{1}=0$ that is enclosed by the contour $\mathcal{C}_{\omega_{1}}$ such that $Z=1$ and a single pole at $\omega=\gamma_{1}$ such that $P=1$. Consequently, except of the zero at $\omega=0$ no other zero of $f(\omega_{1})$ is located inside the contour $\mathcal{C}_{\omega_{1}}$.
			\end{proof}			
			The fourth and last type of singularities of the integrand in \eqref{eq: second_order_asymptotic_behavior_appendix_proof_music} is located at zero. The corresponding residue with respect to $\omega_{1}$ evaluated at zero is given by
			\begin{align}
				& \res{\frac{1}{(\gamma_{r}-\omega_{1})(\gamma_{k}-\omega_{1})} \frac{\frac{\mu_{1}}{(\gamma_{r}-\mu_{1})(\gamma_{k}-\mu_{1})}}{1-\Omega(\omega_{1},\mu_{1})} \frac{\omega_{1}}{z(\omega_{1})}\frac{\partial z(\omega_{1})}{\partial \omega_{1}}}{0} \nonumber \\
				=& - \frac{1}{\gamma_{r}\gamma_{k}} \frac{\mu_{1}}{(\gamma_{r}-\mu_{1})(\gamma_{k}-\mu_{1})} \frac{1}{\frac{1}{N} \sum_{m=1}^{\bar{M}} \frac{K_{m}}{\gamma_{m}-\mu_{1}}}.
				\label{eq: music_xi_residue_zero}
			\end{align}
		
			It remains to compute the complex contour integral that belongs to the last part in \eqref{eq: second_order_asymptotic_behavior_appendix_proof_music}
			\begin{equation}
				 \frac{-1}{2\pi \mathrm{j}} \oint_{\mathcal{C}_{\omega_{1}}} \frac{\left(\varphi_{1}(\omega_{1})-\omega_{1}\right) \frac{\omega_{1}}{z(\omega_{1})}\frac{\partial z(\omega_{1})}{\partial \omega_{1}} \frac{1}{(\gamma_{r}-\omega_{1})(\gamma_{k}-\omega_{1})}}{(\gamma_{r}-\varphi_{1}(\omega_{1}))(\gamma_{k}-\varphi_{1}(\omega_{1}))\scriptstyle\left(1-\frac{1}{N}\sum_{m=1}^{\bar{M}} \frac{K_{m}\gamma_{m}}{\gamma_{m}-\varphi_{1}(\omega_{1})}\right)} \mathrm{d}\omega_{1}
				 \label{eq: second_order_asymptotic_behavior_appendix_proof_music_intermediate_num}
			\end{equation}
			which can be solved in part numerically and in part in closed-form. The numerical part is obtained by applying the change of variables $\mathrm{d}\omega_{1} = \frac{\partial \omega(z_{1})}{\partial z_{1}}\mathrm{d}z_{1}$ to \eqref{eq: second_order_asymptotic_behavior_appendix_proof_music_intermediate_num} and parameterizing the complex contour through $z \mapsto \omega(z)$ for $z\in \lbrack x_{1}^{-}, x_{1}^{+}\rbrack$ as proposed in \cite[Section IV]{a20} and given by
			\begin{equation}
				\begin{aligned}
					-&\frac{1}{\pi} \int_{x_{1}^{-}}^{x_{1}^{+}} \operatorname{Im}\Biggl \lbrack \frac{1}{\left(1-\frac{1}{N}\sum_{n=1}^{\bar{M}}\frac{K_{n}\gamma_{n}}{\gamma_{n}-\omega(x)}\right)\left(1-\frac{1}{N}\sum_{n=1}^{\bar{M}}\frac{K_{n}\gamma_{n}}{\gamma_{n}-\varphi_{1}(\omega(x))}\right)} \\
					& \times \frac{\varphi_{1}(\omega(x))-\omega(x)}{(\gamma_{r}-\omega(x))(\gamma_{r}-\varphi_{1}(\omega(x)))(\gamma_{k}-\omega(x))(\gamma_{k}-\varphi_{1}(\omega(x)))}\Biggr \rbrack \mathrm{d}x \\				
					=& \frac{2}{\pi} \int_{x_{1}^{-}}^{x_{1}^{+}} \frac{1}{\left | 1- \frac{1}{N}\sum_{m=1}^{\bar{M}} \frac{K_{m}\gamma_{m}}{\gamma_{m}-\omega(x)} \right |^{2}} \frac{\ima{\omega(x)}}{\left|\gamma_{r}-\omega(x)\right|^{2} \left|\gamma_{k}-\omega(x)\right|^{2}} \mathrm{d}x.
				\end{aligned}
				\label{eq: music_xi_residue_num_num}			
			\end{equation}
			The following lemma was used to simplify the real-valued integral.		
			\begin{lemma}
				\label{lemma: phi_values_converge_to_complex_conjugate_of_omega}
				By parameterizing the contour $\mathcal{C}_{\omega_{1}}$ as proposed in \cite[Section IV]{a20}
				\begin{equation}
					\frac{1}{N}\sum_{r=1}^{\bar{M}} \frac{K_{r}\gamma_{r}^{2}}{|\gamma_{r}-\omega(x)|^{2}}=1 
					\label{eq: absolute_value_omega_x_equals_one}
				\end{equation} 
				for $x\in \mathcal{S}$ and it can be shown that
				\begin{equation}
					\underset{\omega_{1} \rightarrow \omega(x)}{\lim} \varphi_{t}(\omega_{1}) = \omega(x)^{*}, \quad \text{for } x \in \lbrack x_{t}^{-}, x_{t}^{+}\rbrack. 
					\label{eq: convergence_of_varphi_omega_1}
				\end{equation}
			\end{lemma} 
			\begin{proof}
				Let $u(x)$ and $v(x)$ denote the real and imaginary parts of $\omega(x)$ in \eqref{eq: definition_omega_z}, i.e., $\omega(x) = u(x)+\mathrm{j}v(x)$. It is shown in \cite[Proposition 2]{a43} that $v(x)>0$ for $x\in \overset{\circ}{\mathcal{S}} \equiv (x_{1}^{-},x_{1}^{+}) \cup \dots \cup (x_{S}^{-},x_{S}^{+})$. Taking imaginary parts on both sides of \eqref{eq: definition_omega_z} one obtains
				\begin{equation*}
					0 = v(x) \Biggl( 1 - \frac{1}{N}\sum_{r=1}^{\bar{M}} \frac{K_{r}\gamma_{r}^{2}}{(\gamma_{r}-u(x))^{2}+v(x)^{2}}\Biggr)
				\end{equation*}
				which implies that
				\begin{equation}
					1=\frac{1}{N}\sum_{r=1}^{\bar{M}} \frac{K_{r}\gamma_{r}^{2}}{(\gamma_{r}- u(x))^{2} + v(x)^{2}} = \frac{1}{N}\sum_{r=1}^{\bar{M}}\frac{K_{r}\gamma_{r}^{2}}{|\gamma_{r}-\omega(x)|^{2}}
					\label{eq: absolute_value_omega_x_equals_one_proof}
				\end{equation}
				for $x\in \overset{\circ}{\mathcal{S}}$. From \cite[Proposition 1]{a20} it follows that \eqref{eq: absolute_value_omega_x_equals_one_proof} also holds for the boundaries of the clusters $x \in \lbrace x_{1}^{-}, x_{1}^{+}, \dots, x_{S}^{-}, x_{S}^{+} \rbrace$. Furthermore, the convergence result in  \eqref{eq: convergence_of_varphi_omega_1} follows as a direct consequence of \eqref{eq: absolute_value_omega_x_equals_one} and the definition of $\varphi_{t}(\omega_{1})$ in \eqref{eq: varphi_values}.
			\end{proof}
			
			In the undersampled case ($M>N$) it can be observed that the contour obtained through the parameterization $z \mapsto \omega(z)$ for $z \in \lbrack x_{1}^{-}, x_{1}^{+}\rbrack$ encloses zero but not $\mu_{1}$ since $\omega(0) = \mu_{1} < \omega(x_{1}^{-}) < 0 < \omega(x_{1}^{+})$. However, since $\mu_{1}$ is enclosed by $\mathcal{C}_{\omega_{1}}$ we compute the residue of the integrand in \eqref{eq: second_order_asymptotic_behavior_appendix_proof_music_intermediate_num} with respect to $\omega_{1}$ and evaluate it at $\mu_{1}$
			\begin{equation}
				\begin{aligned}
					&\res{\frac{\left(\varphi_{1}(\omega_{1})-\omega_{1}\right) \frac{\omega_{1}}{z(\omega_{1})}\frac{\partial z(\omega_{1})}{\partial \omega_{1}} \frac{1}{(\gamma_{r}-\omega_{1})(\gamma_{k}-\omega_{1})}\delta_{M>N}}{(\gamma_{r}-\varphi_{1}(\omega_{1}))(\gamma_{k}-\varphi_{1}(\omega_{1}))\left(1-\frac{1}{N}\sum_{m=1}^{\bar{M}} \frac{K_{m}\gamma_{m}}{\gamma_{m}-\varphi_{1}(\omega_{1})}\right)}}{\mu_{1}} \\
					=& - \frac{1}{\gamma_{r}\gamma_{k}} \frac{\mu_{1}^{2}}{(\gamma_{r}-\mu_{1})(\gamma_{k}-\mu_{1})} \frac{N}{N-M} \delta_{M>N}.
				\end{aligned}
				\label{eq: music_xi_residue_num_mu}	
			\end{equation}
			Furthermore, in the oversampled case ($M < N$) it can be observed that the contour obtained through the parameterization $z \mapsto \omega(z)$ for $z \in \lbrack x_{1}^{-}, x_{1}^{+}\rbrack$ encloses $\mu_{1}$, however it does not enclose zero since $0<\omega(x_{1}^{-}) < \mu_{1} < \omega(x_{1}^{+})$ although zero is enclosed by the contour. Therefore, we compute the residue of the integrand in \eqref{eq: second_order_asymptotic_behavior_appendix_proof_music_intermediate_num} with respect to $\omega_{1}$ and evaluate it at zero
			\begin{equation}
				\begin{aligned}
					&\res{\frac{\left(\varphi_{1}(\omega_{1})-\omega_{1}\right) \frac{\omega_{1}}{z(\omega_{1})}\frac{\partial z(\omega_{1})}{\partial \omega_{1}} \frac{1}{(\gamma_{r}-\omega_{1})(\gamma_{k}-\omega_{1})}\delta_{M<N}}{(\gamma_{r}-\varphi_{1}(\omega_{1}))(\gamma_{k}-\varphi_{1}(\omega_{1}))\left(1-\frac{1}{N}\sum_{m=1}^{\bar{M}} \frac{K_{m}\gamma_{m}}{\gamma_{m}-\varphi_{1}(\omega_{1})}\right)}}{0} \\
					=& \frac{1}{\gamma_{r}\gamma_{k}} \frac{\mu_{1}}{(\gamma_{r}-\mu_{1})(\gamma_{k}-\mu_{1})} \frac{1}{\frac{1}{N}\sum_{m=1}^{\bar{M}} \frac{K_{m}}{\gamma_{m}-\mu_{1}}}\delta_{M<N}.
				\end{aligned}
				\label{eq: music_xi_residue_num_zero}	
			\end{equation}
			Consequently, the solution of the contour integral in \eqref{eq: second_order_asymptotic_behavior_appendix_proof_music_intermediate_num} is obtained by summing the intermediate results in \eqref{eq: music_xi_residue_num_num}, \eqref{eq: music_xi_residue_num_mu} and \eqref{eq: music_xi_residue_num_zero} and yields  
			\begin{align}
					& \frac{-1}{2\pi \mathrm{j}} \oint_{\mathcal{C}_{\omega_{1}}} \frac{\left(\varphi_{1}(\omega_{1})-\omega_{1}\right) \frac{\omega_{1}}{z(\omega_{1})}\frac{\partial z(\omega_{1})}{\partial \omega_{1}} \frac{1}{(\gamma_{r}-\omega_{1})(\gamma_{k}-\omega_{1})}}{(\gamma_{r}-\varphi_{1}(\omega_{1}))(\gamma_{k}-\varphi_{1}(\omega_{1}))\scriptstyle\left(1-\frac{1}{N}\sum_{m=1}^{\bar{M}} \frac{K_{m}\gamma_{m}}{\gamma_{m}-\varphi_{1}(\omega_{1})}\right)} \mathrm{d}\omega_{1} \nonumber \\[-3pt]
					=& \frac{2}{\pi} \int_{x_{1}^{-}}^{x_{1}^{+}} \frac{1}{\left | 1- \frac{1}{N}\sum_{m=1}^{\bar{M}} \frac{K_{m}\gamma_{m}}{\gamma_{m}-\omega(x)} \right |^{2}} \frac{\ima{\omega(x)}}{\left|\gamma_{r}-\omega(x)\right|^{2} \left|\gamma_{k}-\omega(x)\right|^{2}} \mathrm{d}x \nonumber \\
					&- \frac{1}{\gamma_{r}\gamma_{k}} \frac{\mu_{1}^{2}}{(\gamma_{r}-\mu_{1})(\gamma_{k}-\mu_{1})} \frac{N}{N-M} \delta_{M>N} \nonumber \\
					& + \frac{1}{\gamma_{r}\gamma_{k}} \frac{\mu_{1}}{(\gamma_{r}-\mu_{1})(\gamma_{k}-\mu_{1})} \frac{1}{\frac{1}{N}\sum_{m=1}^{\bar{M}} \frac{K_{m}}{\gamma_{m}-\mu_{1}}}\delta_{M<N}.
				\label{eq: second_order_asymptotic_behavior_appendix_proof_music_numerical_part}
			\end{align}
	
			Finally, the expression of the real-valued weights ${\xi}_{\operatorname{c}}(r,k)$ in \eqref{eq: music_second_order_xi} is obtained by taking the sum of the residues in \eqref{eq: music_xi_residue_gamma_a}, \eqref{eq: music_xi_residue_gamma_b}, \eqref{eq: music_xi_residue_gamma_c}, \eqref{eq: music_xi_residue_mu_a}, \eqref{eq: music_xi_residue_mu_b}, \eqref{eq: music_xi_residue_mu_c}, \eqref{eq: music_xi_residue_zero} evaluated at all singularities that lie inside the contour $\mathcal{C}_{\omega_{1}}$ namely $\lbrace 0, \gamma_{1}, \mu_{1} \rbrace$ and adding \eqref{eq: second_order_asymptotic_behavior_appendix_proof_music_numerical_part}.		
			
		\section{Determination of the asymptotic covariance of the g-MUSIC cost function}
			\label{app: second_order_g_music}
			In case of the g-\ac{MUSIC} cost function $\bar{h}(z(\omega_{2})) = \frac{z(\omega_{2})}{\omega_{2}}\frac{\partial \omega(z_{2})}{\partial z_{2}}$. Hence the complex contour integral in \eqref{eq: second_order_general_separated_omega_2} simplifies as follows 
			\begin{equation}
				\mathcal{I}_{r,k}(\omega_{1}) = \frac{1}{2\pi \mathrm{j}}\oint_{\mathcal{C}_{\omega_{2}}}\frac{1}{(\gamma_{r}-\omega_{2})(\gamma_{k}-\omega_{2})}\frac{1}{1-\Omega(\omega_{1},\omega_{2})}\mathrm{d}\omega_{2} 
				\label{eq: second_order_general_separated_omega_2_g_music}
			\end{equation}
			which can be solved in closed-form using conventional residue calculus \cite{a54}. It can be observed that the integrand in \eqref{eq: second_order_general_separated_omega_2_g_music} exhibits two different types of singularities. The first type of singularities belongs to the roots of $1-\Omega(\omega_{1},\omega_{2})$ where $\Omega(\omega_{1},\omega_{2})$ is given in \eqref{eq: omega_captial_theorem}. The complex-valued roots are denoted by $\varphi_{t}(\omega_{1})$ for $t=1,\dots,\bar{M}$ and defined in \eqref{eq: varphi_values}. The residue w.r.t. $\omega_{2}$ evaluated at $\varphi_{t}(\omega_{1})$ is given by 
			\begin{equation}
				\res{\frac{\frac{1}{(\gamma_{r}-\omega_{2})(\gamma_{k}-\omega_{2})}}{1-\Omega(\omega_{1},\omega_{2})}}{\varphi_{t}(\omega_{1})} =\frac{\frac{-1}{(\gamma_{r}-\varphi_{t}(\omega_{1}))(\gamma_{k}-\varphi_{t}(\omega_{1}))}}{\frac{1}{N}\sum_{r=1}^{\bar{M}}\frac{K_{r}\gamma_{r}^{2}}{(\gamma_{r}-\omega_{1})(\gamma_{r}-\varphi_{t}(\omega_{1}))^{2}}}.
				\label{eq: g_music_mathcal_i_residue_a}
			\end{equation}
			However, it can be seen that only the residue evaluated at $\varphi_{1}(\omega_{1})$ contributes to the solution of $\mathcal{I}_{r,k}(\omega_{1})$ in \eqref{eq: second_order_general_separated_omega_2_g_music} as it is the only singularity of the roots of $1-\Omega(\omega_{1},\omega_{2})$ that is enclosed by the contour $\mathcal{C}_{\omega_{2}}$.
			
			The second type of singularities belongs to the poles of the integrand in \eqref{eq: second_order_general_separated_omega_2_g_music} at $\gamma_{t}$ for $t=1,\dots,\bar{M}$. The corresponding residue w.r.t. $\omega_{2}$ evaluated at $\gamma_{t}$ is given by
			\begin{equation}
				\begin{aligned}
					&\res{\frac{\frac{1}{(\gamma_{r}-\omega_{2})(\gamma_{k}-\omega_{2})}}{1-\Omega(\omega_{1},\omega_{2})}}{\gamma_{t}}= \begin{cases} \frac{\gamma_{t}-\omega_{1}}{\frac{1}{N} K_{t}\gamma_{t}^{2}} & \text{for } r=k=t\\
					0 & \text{for } r\neq k.
	 				\end{cases}
				\end{aligned}
				\label{eq: g_music_mathcal_i_residue_b}
			\end{equation}
			It can be observed that only the first eigenvalue $\gamma_{1}$ is enclosed by the contour $\mathcal{C}_{\omega_{2}}$. Hence, only the residue evaluated at $\gamma_{1}$ contributes to the closed-form expression of $\mathcal{I}_{r,k}(\omega_{1})$ in \eqref{eq: second_order_general_separated_omega_2_g_music} which is obtained by taking the negative sum (negative due to the negatively orientated contour) of the residue in \eqref{eq: g_music_mathcal_i_residue_a} evaluated at $\varphi_{1}(\omega_{1})$ and the residue in \eqref{eq: g_music_mathcal_i_residue_b} evaluated at $\gamma_{1}$		
			\begin{equation}
				\mathcal{I}_{r,k}(\omega_{1}) =\frac{\frac{1}{(\gamma_{r}-\varphi_{1}(\omega_{1}))(\gamma_{k}-\varphi_{1}(\omega_{1}))}}{\frac{1}{N}\sum_{r=1}^{\bar{M}}\frac{K_{r}\gamma_{r}^{2}}{(\gamma_{r}-\omega_{1})(\gamma_{r}-\varphi_{1}(\omega_{1}))^{2}}} - \frac{\gamma_{1}-\omega_{1}}{\frac{1}{N} K_{1}\gamma_{1}^{2}}.
				\label{eq: closed_form_expression_mathcal_I_g_music}
			\end{equation}
			Substituting the closed-form expression for $\mathcal{I}_{r,k}(\omega_{1})$ in \eqref{eq: closed_form_expression_mathcal_I_g_music} into the general expression of $\xi_{\operatorname{g}}(r,k)$ in \eqref{eq: second_order_general_separated_omega_1} and using $\bar{h}(z(\omega_{1}))=\frac{z(\omega_{1})}{\omega_{1}}\frac{\partial \omega(z_{1})}{\partial z_{1}}$ yields 
			\begin{align}
				\xi_{\operatorname{g}}(r,k) =& \frac{1}{2 \pi \mathrm{j}} \oint_{\mathcal{C}_{\omega_{1}}} \frac{\frac{\gamma_{r} \gamma_{k} }{(\gamma_{r}-\omega_{1})(\gamma_{k}-\omega_{1})(\gamma_{r}-\varphi_{1}(\omega_{1}))(\gamma_{k}-\varphi_{1}(\omega_{1}))}}{\frac{1}{N}\sum_{m=1}^{\bar{M}}\frac{K_{m}\gamma_{m}^{2}}{(\gamma_{m}-\omega_{1})(\gamma_{m}-\varphi_{1}(\omega_{1}))^{2}}}  \mathrm{d}\omega_{1} \nonumber \\
					&-\frac{1}{2\pi \mathrm{j}}\oint_{\mathcal{C}_{\omega_{1}}} \frac{1}{(\gamma_{r}-\omega_{1})(\gamma_{k}-\omega_{1})} \frac{\gamma_{1}-\omega_{1}}{\frac{1}{N}K_{1}\gamma_{1}^{2}} \mathrm{d}\omega_{1}. \label{eq: second_order_asymptotic_behavior_appendix_proof_g_music}
			\end{align}
			It can be observed that the integrand of the second contour integral in \eqref{eq: second_order_asymptotic_behavior_appendix_proof_g_music} exhibits a single singularity at $\gamma_{1}$ that lies within the contour $\mathcal{C}_{\omega_{1}}$. The corresponding residue w.r.t. $\omega_{1}$ evaluated at $\gamma_{1}$ is given by
			\begin{equation}
				\begin{aligned}
					\res{\frac{(\gamma_{t}-\omega_{1})\frac{1}{\frac{1}{N}K_{t}\gamma_{t}^{2}}}{(\gamma_{r}-\omega_{1})(\gamma_{k}-\omega_{1})} }{\gamma_{t}} = \begin{cases}
					\frac{-1}{\frac{1}{N}K_{t}\gamma_{t}^{2}},& \text{for } r=k=t\\
					0,& r\neq k.
					\end{cases}
				\end{aligned}
				\label{eq: g_music_mathcal_gamma_residue_a}
			\end{equation}
			Hence the closed-form solution for the second integral in \eqref{eq: second_order_asymptotic_behavior_appendix_proof_g_music} is given by the residue in \eqref{eq: g_music_mathcal_gamma_residue_a} evaluated at $\gamma_{1}$ for $r=k=1$ or zero otherwise. The first contour integral in \eqref{eq: second_order_asymptotic_behavior_appendix_proof_g_music} can be solved numerically by applying the change of variables $\mathrm{d}\omega=\frac{\partial \omega(x)}{\partial x} \mathrm{d}x$ and parameterizing the contour $\mathcal{C}_{\omega_{1}}$ by concatenation of $\omega(x)$ and $\omega(x)^{*}$ as proposed in \cite[Section IV]{a20} such that		
			\begin{equation}
				\begin{aligned}
					& \frac{1}{2\pi \mathrm{j}}\oint_{\mathcal{C}_{\omega_{1}}} \frac{\frac{1}{(\gamma_{r}-\omega_{1})(\gamma_{k}-\omega_{1})(\gamma_{r}-\varphi_{1}(\omega_{1}))(\gamma_{k}-\varphi_{1}(\omega_{1}))}}{\frac{1}{N}\sum_{m=1}^{\bar{M}}\frac{K_{m}\gamma_{m}^{2}}{(\gamma_{m}-\omega_{1})(\gamma_{m}-\varphi_{1}(\omega_{1}))^{2}}}  \mathrm{d}\omega_{1}\\
					=& \frac{1}{\pi}\int_{x_{1}^{-}}^{x_{1}^{+}} \ima{\frac{\frac{1}{|\gamma_{r}-\omega(x)|^{2} |\gamma_{k}-\omega(x)|^{2}}}{\frac{1}{N}\sum_{r=1}^{\bar{M}}\frac{K_{r}\gamma_{r}^{2}}{|\gamma_{r}-\omega(x)|^{2}(\gamma_{r}-\omega(x)^{*})}}\frac{\partial \omega(x)}{\partial x} }\mathrm{d}x
				\end{aligned}
				\label{eq: g_music_second_order_asymptotic_behavior_numerical_part_appendix}
			\end{equation}
			where we have used Lemma \ref{lemma: phi_values_converge_to_complex_conjugate_of_omega} to simplify the expression. The following Lemma allows to further simplify the real-valued integral in \eqref{eq: g_music_second_order_asymptotic_behavior_numerical_part_appendix}. 
			
			\begin{lemma}
				By parameterizing the contour $\mathcal{C}_{\omega}$ through concatenation of $\omega(x)$ in \eqref{eq: definition_omega_z} and $\omega(x)^{*}$ as proposed in \cite[Section IV]{a20} the following equality holds
				\begin{equation}
					\frac{1}{N}\sum_{r=1}^{\bar{M}} \frac{K_{r}\gamma_{r}^{2}}{|\gamma_{r}-\omega(x)|^{2} (\gamma_{r}-\omega(x))^{*}} = \frac{1}{2\mathrm{j}\ima{\omega(x)}\omega^{\prime}(x)^{*}} 
					\label{eq: equality_lemma}
				\end{equation}				
				where $\omega^{\prime}(x)$ is defined in \eqref{eq: first_order_derivative_wrt_x_omega_x}.
			\end{lemma}
			\begin{proof}
				Let $u(x)$ and $v(x)$ denote the real and imaginary parts of $\omega(x)$ in \eqref{eq: definition_omega_z}, i.e., $\omega(x) = u(x)+\mathrm{j}v(x)$. Using the expression in \eqref{eq: absolute_value_omega_x_equals_one} we can express the inverse of the first order derivative of $\omega(x)$ with respect to $x$ in \eqref{eq: first_order_derivative_wrt_x_omega_x} as 
				\begin{align}
					&\frac{1}{\omega^{\prime}(x)} = \frac{1}{N} \sum_{r=1}^{\bar{M}}\frac{ K_{r}\gamma_{r}^{2} \lbrack(\gamma_{r}-\omega(x))^{2}-|\gamma_{r}-\omega(x)|^{2}\rbrack}{(\gamma_{r}-\omega(x))^{2}|\gamma_{r}-\omega(x)|^{2}} \label{eq: inverse_omega_prime_x}\\[-4pt]
					&=  \frac{1}{N} \sum_{r=1}^{\bar{M}} \frac{-K_{r}\gamma_{r}^{2}\lbrack 2\mathrm{j}v(x)(\gamma_{r}-u(x))+2v(x)^{2}\rbrack}{|\gamma_{r}-\omega(x)|^{2}\lbrack (\gamma_{r}-u(x))^{2} -v(x)^{2}-2\mathrm{j}v(x)(\gamma_{r}-u(x)) \rbrack}. \nonumber 
				\end{align}
				Taking the complex conjugate of $\frac{1}{\omega^{\prime}(x)}$ in \eqref{eq: inverse_omega_prime_x} and multiplying by $\frac{1}{2\mathrm{j}v(x)}$ one obtains the following expression for $\frac{1}{2\mathrm{j}v(x)\omega^{\prime}(x)^{*}}$  
				\begin{equation*}
					\begin{aligned}
						\frac{1}{N} \sum_{r=1}^{\bar{M}} \frac{K_{r}\gamma_{r}^{2} \lbrack (\gamma_{r}-u(x)) + \mathrm{j} v(x)\rbrack}{|\gamma_{r}-\omega(x)|^{2}\lbrack (\gamma_{r}-u(x))^{2} +2 \mathrm{j} v(x)(\gamma_{r}-u(x))-v(x)^{2}\rbrack}
					\end{aligned}		
				\end{equation*}
				which is identical to the one in \eqref{eq: equality_lemma}.
			\end{proof}
			Using the equalities established in \eqref{eq: equality_lemma} the real-valued integral in \eqref{eq: g_music_second_order_asymptotic_behavior_numerical_part_appendix} can equivalently be expressed as
			\begin{equation}
				\begin{aligned}
					\frac{2}{\pi}\int_{x_{1}^{-}}^{x_{1}^{+}} \frac{|\omega^{\prime}(x)|^{2} \ima{\omega(x)} }{|\gamma_{r}-\omega(x)|^{2}|\gamma_{k}-\omega(x)|^{2}}  \mathrm{d}x.
				\end{aligned}
				\label{eq: g_music_second_order_asymptotic_behavior_numerical_part_appendix_simplified}
			\end{equation}
			Finally, the solution for $\xi_{\operatorname{g}}(r,k)$ in \eqref{eq: g_music_second_order_xi} is obtained by summing the residue in \eqref{eq: g_music_mathcal_gamma_residue_a} for $t=1$ and \eqref{eq: g_music_second_order_asymptotic_behavior_numerical_part_appendix_simplified}. 			

%
%


\bibliographystyle{IEEEtran}
\bibliography{references.bib}


\end{document}


%
\title{Supplementary Material Probability of Resolution of MUSIC and g-MUSIC: An Asymptotic Approach}
%
%
%

\author{David~Schenck,
        Xavier~Mestre,
        and~Marius~Pesavento
\thanks{David Schenck and Marius Pesavento were supported by the DFG PRIDE Project PE 2080/2-1. Xavier Mestre was supported by the Catalan Government under grant 2017-SGR-01479 and the Spanish Government under grant RTI2018-099722-BI00.}
\thanks{David Schenck and Marius Pesavento are with the Communication Systems Group, Technische Universit\"at Darmstadt, Darmstadt, Germany, (e-mail: \{schenck, pesavento\}@nt.tu-darmstadt.de).}
\thanks{Xavier Mestre is with the Centre Tecnol\`{o}gic de Telecomunicacions de Catalunya, Castelldefels, Spain, (e-mail: xavier.mestre@cttc.cat).}
\thanks{Parts of this work have been submitted to \textit{IEEE International Conference on Acoustics, Speech, and Signal Processing (ICASSP 2021)}.}%
}

\begin{acronym}[PR-DML]
	 \acro{ULA}{Uniform Linear Array}
	 \acro{DoA}{Direction-of-Arrival}
	 \acro{FoV}{Field of View}
	 \acro{MUSIC}{Multiple Signal Classification}
	 \acro{ESPRIT}{Estimation of Signal Parameters via Rotational Invariance Technique}
	 \acro{DML}{Deterministic Maximum Likelihood}
	 \acro{SML}{Stochastic Maximum Likelihood}
	 \acro{WSF}{Weighted Subspace Fitting}
	 \acro{CF}{Covariance Fitting}
	 \acro{FCF}{Full Covariance Fitting}
	 \acro{PR}{Partial Relaxation}
	 \acro{PR-DML}{Partially Relaxed Deterministic Maximum Likelihood}
	 \acro{PR-CF}{Partially Relaxed Covariance Fitting}
	 \acro{PR-FCF}{Partially Relaxed Full Covariance Fitting}
	 \acro{MDL}{Minimum Description Length}
	 \acro{SNR}{Signal-to-Noise Ratio} 
	 \acro{LS}{Least Square}
	 \acro{ML}{Maximum Likelihood}
	 \acro{LR}{Likelihood Ratio}
	 \acro{FWE}{Familywise Error-Rate}
	 \acro{FDR}{False Discovery Rate}
	 \acro{RMT}{Random Matrix Theory}
	 \acro{CRB}{Cramer-Rao Bound}
	 \acro{RMSE}{Root-Mean-Squared-Error}
	 \acro{DCT}{Dominated Convergence Theorem}
	 \acro{pdf}{probability density function}
	 \acro{cdf}{cumulative distribution function}
	 \acro{MSE}{Mean Square Error}
	 \acro{CLT}{Central Limit Theorem}
	\end{acronym}

\thinmuskip = 0mu
\medmuskip = 0mu

\setcounter{equation}{82}
\setcounter{lemma}{7}
\setcounter{proposition}{2}

\section*{Supplementary Material}
This document provides some additional proofs of results used in the article "Probability of Resolution of \acs{MUSIC} and G-\acs{MUSIC}: An Asymptotic Approach". More specifically the proof of the \acl{CLT} in Theorem 3 is provided.
	
\paragraph*{Notation} Let $x$, $x_{1}$ and $x_{2}$ be three random quantities then $\mean{x}$ denotes the expectation of $x$, $\var{x}$ denotes the variance of $x$ and $\cov{x_{1}}{x_{2}}$ denotes the covariance between $x_{1}$ and $x_{2}$. Furthermore, let $x^{(\circ)} = x - \mean{x}$.

\appendices	
\setcounter{section}{2}	
\section{Proof of Theorem 3}
	\label{sec: proof_theorem_4}
	In the following a more general form of Theorem 3 is derived. Instead of focusing on the noise subspace estimator only the asymptotic stochastic behavior of the estimator for the subspace that is associated to the eigenvalue $\gamma_{m}$ for $m=1,\dots,\bar{M}$ is derived. Let $\mathcal{K}_{m}$ denote a set of indexes
	\begin{equation*}
		\mathcal{K}_{m} = \left \lbrace \sum_{r=1}^{m-1}K_{r}+1, \sum_{r=1}^{m-1} K_{r} +2, \dots, \sum_{r=1}^{m} K_{r} \right \rbrace.
	\end{equation*}
	The cardinality of the set $\mathcal{K}_{m}$ is equal to the multiplicity of the corresponding eigenvalue $\gamma_{m}$, namely $K_{m}$. It is assumed that the cluster with support $\lbrack x_{s}^{-}, x_{s}^{+} \rbrack$ that is associated to $\gamma_{m}$ is well separated from the rest of the eigenvalue distribution and that $\gamma_{m}$ is the unique eigenvalue that is associated to this cluster (subspace separability condition holds). Let $\mathcal{C}_{z}^{(m)}$ denote a negatively (clockwise) orientated contour that encloses the sample eigenvalues $\lbrace \hat{\lambda}_{k} : k \in \mathcal{K}_{m} \rbrace$ that correspond to the true and distinct eigenvalue $\gamma_{m}$ only. Next, consider the random cost function $\hat{\eta}_{l}^{(m)}$ that belongs to the estimator of the subspace that is associated to the eigenvalue $\gamma_{m}$ and its deterministic equivalent $\bar{\eta}_{l}^{(m)}$
	\begin{equation}
		\hat{\eta}_{l}^{(m)} = \frac{1}{2 \pi \mathrm{j}} \oint_{\mathcal{C}_{z}^{(m)}}  \hat{h}(z)  \hat{m}_{l}(z) \mathrm{d}z, \qquad \bar{\eta}_{l}^{(m)} = \frac{1}{2 \pi \mathrm{j}} \oint_{\mathcal{C}_{z}^{(m)}}  \bar{h}(z)  \bar{m}_{l}(z) \mathrm{d}z,
		\label{eq: cost_function_and_deterministic_equivalent}
	\end{equation}
	where $l$ denotes the point of evaluation of the cost function. Furthermore, for the $M,N$-consistent estimator of the subspace that is associated to the $m$-th distinct eigenvalue we define $\hat{h}(z) = \frac{z}{\hat{\omega}(z)}\frac{\partial \hat{\omega}(z)}{\partial z}$ and $\bar{h}(z) = \frac{z}{\omega(z)}\frac{\partial \omega(z)}{\partial z}$ where
	\begin{equation*}
		\frac{z}{\hat{\omega}(z)} = 1- \frac{1}{N} \sum_{r=1}^{M} \frac{\hat{\lambda}_{r}}{\hat{\lambda}_{r} - z}, \qquad \frac{\partial \hat{\omega}(z)}{\partial z} =  \frac{\hat{\omega}(z)}{z} + \frac{\hat{\omega}(z)^{2}}{z} \frac{1}{N}\sum_{r=1}^{M} \frac{\hat{\lambda}_{r}}{(\hat{\lambda}_{r} -z)^{2}},
	\end{equation*}
	and
	\begin{equation*}
		\frac{z}{\omega(z)} = 1-\frac{1}{N}\sum_{r=1}^{\bar{M}} \frac{K_{r}\gamma_{r}}{\gamma_{r}-\omega(z)}, \qquad
		\frac{\partial \omega(z)}{\partial z} = \left( 1- \frac{1}{N}\sum_{r=1}^{\bar{M}} \frac{K_{r}\gamma_{r}^{2}}{(\gamma_{r}-\omega(z))^{2}}\right)^{-1}.
	\end{equation*}
	On the other hand, in case of the $N$-consistent estimator we instead have $\hat{h}(z) = 1$ and $\bar{h}(z)=1$. Furthermore, the Stieltjes transform and its deterministic equivalent are defined as
	\begin{gather}
		\hat{m}_{l}(z) = \boldsymbol{v}_{l}^{\operatorname{H}} \hat{\boldsymbol{Q}}(z) \boldsymbol{v}_{l}, \qquad \bar{m}_{l}(z) = \boldsymbol{v}_{l}^{\operatorname{H}} \bar{\boldsymbol{Q}}(z) \boldsymbol{v}_{l} \label{eq: stieltjes_transform_and_deterministic_equivalent} \\
		\hat{\boldsymbol{Q}}(z) = \left( \hat{\boldsymbol{R}} - z \boldsymbol{I}_{M}\right)^{-1}, \qquad \bar{\boldsymbol{Q}}(z) = \frac{\omega(z)}{z} \left( \boldsymbol{R}- \omega(z) \boldsymbol{I}_{M}\right)^{-1},
		\label{eq: resolvent_and_deterministic_equivalent}
	\end{gather}
	where $\omega(z)$ is given by the unique solution to the following equation
	\begin{equation}
		z = \omega(z) \left ( 1 - \frac{1}{N} \tr{\boldsymbol{R}\left(\boldsymbol{R}-\omega(z) \boldsymbol{I}_{M} \right)^{-1}}\right),
		\label{eq: definition_omega_z}
	\end{equation}
	on $\lbrace \omega(z) \in \mathbb{C}^{+} \text{ : } \imag{\omega(z)}\geq 0 \rbrace$ for $z \in \mathbb{C}^{+}$ and $\boldsymbol{v}_{l} \in \mathbb{C}^{M}$ is a deterministic $M\times 1$ vector which depends on the point of evaluation $l$ of the cost function. The cost function is evaluated at $L$ distinct points where $\hat{\boldsymbol{\eta}}^{(m)} = \lbrack \hat{\eta}_{1}^{(m)}, \dots, \hat{\eta}_{L}^{(m)} \rbrack^{\operatorname{T}} \in \mathbb{R}^{L}$ denotes the corresponding real-valued $L\times 1$ cost function vector and $\bar{\boldsymbol{\eta}}^{(m)} = \lbrack \bar{\eta}_{1}^{(m)}, \dots, \bar{\eta}_{L}^{(m)} \rbrack^{\operatorname{T}} \in \mathbb{R}^{L}$ its deterministic equivalent, respectively. Subtracting the deterministic equivalent $\bar{\eta}_{l}^{(m)}$ from the random cost function $\hat{\eta}_{l}^{(m)}$ in \eqref{eq: cost_function_and_deterministic_equivalent} yields
	\begin{equation}
		\eta_{l}^{(m)} = \sqrt{N} \left( \hat{\eta}_{l}^{(m)} - \bar{\eta}_{l}^{(m)} \right) = \dfrac{\sqrt{N}}{2 \pi \mathrm{j}} \oint_{\mathcal{C}_{z}^{(m)}}  \left( \hat{h}(z)\hat{m}_{l}(z) - \bar{h}(z)\bar{m}_{l}(z) \right) \mathrm{d}z.
		\label{eq: general_form_second_order_derivation}
	\end{equation}
	It is shown in \cite{a80} that the statistic in \eqref{eq: general_form_second_order_derivation} asymptotically fluctuates as a Gaussian random variable with zero-mean and positive variance. However, we are interested in the more general case of the asymptotic joint distribution of a collection of random variables, namely
	\begin{equation}
		\sqrt{N}\left(\hat{\boldsymbol{\eta}}^{(m)}-\bar{\boldsymbol{\eta}}^{(m)} \right)= \begin{bmatrix}
			 \sqrt{N}(\hat{\eta}_{1}^{(m)}-\bar{\eta}_{1}^{(m)}) \\
			\vdots \\
			 \sqrt{N}(\hat{\eta}_{L}^{(m)}-\bar{\eta}_{L}^{(m)})
		\end{bmatrix}
		\label{eq: general_form_vector_second_order_derivation}
	\end{equation}
	where $\sqrt{N}(\hat{\eta}(\bar{\theta}_{l})-\bar{\eta}(\bar{\theta}_{l}))$ for $l=1,\dots,L$ is defined in \eqref{eq: general_form_second_order_derivation}. In order to derive the \ac{CLT} we have to establish pointwise convergence of the characteristic function of the statistic $\sqrt{N}\left(\hat{\boldsymbol{\eta}}(\bar{\boldsymbol{\theta}})-\bar{\boldsymbol{\eta}}(\bar{\boldsymbol{\theta}})\right)$ in \eqref{eq: general_form_vector_second_order_derivation} towards the characteristic function of a Gaussian random variable. According to the Cram\'er-Wold device it is sufficient to show that the one-dimensional projection
	\begin{equation}
		\eta^{(m)} = \sum_{p=1}^{L} w_{p} \eta_{p}^{(m)} =  \sum_{p=1}^{L}\sqrt{N}w_{p}\left(\hat{\eta}_{p}^{(m)}-\bar{\eta}_{p}^{(m)}\right) = \sum_{p=1}^{L}w_{p} \frac{\sqrt{N}}{2 \pi \mathrm{j}} \oint_{\mathcal{C}_{z}^{(m)}}\left(\hat{h}(z) \hat{m}_{p}(z) - \bar{h}(z)\bar{m}_{p}(z) \right) \mathrm{d}z,
		\label{eq: one_dimensional_projection}
	\end{equation}
	is asymptotically Gaussian distributed for any collection of real-valued quantities $w_{l}$, $l=1,\dots,L$ to proof that $\sqrt{N}( \hat{\boldsymbol{\eta}}^{(m)} - \bar{\boldsymbol{\eta}}^{(m)})$ in \eqref{eq: general_form_vector_second_order_derivation} is asymptotically jointly Gaussian distributed. Let $\chi(r)$ be defined as
	\begin{equation}	
		\begin{aligned}
			\chi(r) = \exp{\mathrm{j} r \eta^{(m)}} =& \exp{\mathrm{j} r \sqrt{N} \sum_{p=1}^{L} w_{p} \left(\hat{\eta}_{p}^{(m)} - \bar{\eta}_{p}^{(m)}\right)} = \exp{\mathrm{j}r\sum_{p=1}^{L} w_{p} \frac{\sqrt{N}}{2\pi\mathrm{j}}  \oint_{\mathcal{C}_{z}^{(m)}}\left(\hat{h}(z) \hat{m}_{p}(z) - \bar{h}(z) \bar{m}_{p}(z) \right) \mathrm{d}z}\\
			=& \exp{\mathrm{j}r\sum_{p=1}^{L} w_{p} \frac{\sqrt{N}}{2\pi\mathrm{j}}  \oint_{\mathcal{C}_{z}^{(m)}} \boldsymbol{v}_{p}^{\operatorname{H}} \left( \hat{h}(z)\hat{\boldsymbol{Q}}(z) - \bar{h}(z) \bar{\boldsymbol{Q}}(z) \right) \boldsymbol{v}_{p} \mathrm{d}z}
		\end{aligned}	
		\label{eq: characteristic_function_initial_1}
	\end{equation}
	and let $\mean{\chi(r)}$ be the corresponding characteristic function of the random one-dimensional projection $\eta^{(m)}$ in \eqref{eq: one_dimensional_projection}. In the following we analyze the asymptotic behavior of the characteristic function $\mean{\chi(r)}$ in \eqref{eq: characteristic_function_initial_1} in the asymptotic regime where $M,N \rightarrow \infty$ at the same rate. In fact, we approximate the characteristic function $\mean{\chi(r)}$ in \eqref{eq: characteristic_function_initial_1} with the characteristic function $\bar{\chi}(r)$ of a normal distributed random variable 
	\begin{equation}
		\bar{\chi}(r) = \exp{\mathrm{j} r \boldsymbol{w}^{\operatorname{T}} \boldsymbol{\mu} - r^{2} \frac{\boldsymbol{w}^{\operatorname{T}} \boldsymbol{\Gamma}\boldsymbol{w}}{2}},
		\label{eq: characteristic_function_gaussian}
	\end{equation}
	with mean $\boldsymbol{w}^{\operatorname{T}}\boldsymbol{\mu}$ and covariance $\boldsymbol{w}^{\operatorname{T}}\boldsymbol{\Gamma}\boldsymbol{w}$ where $\boldsymbol{\mu}\in \mathbb{R}^{L}$, $\boldsymbol{\Gamma}\in \mathbb{R}^{L \times L}$, and $\boldsymbol{w} = \lbrack w_{1},\dots, w_{L} \rbrack^{\operatorname{T}} \in \mathbb{R}^{L}$. Hence, the target is to determine $\boldsymbol{\mu}$ and $\boldsymbol{\Gamma}$ such that $\mean{\chi(r)} - \bar{\chi}(r) \rightarrow 0$ pointwise in $r$ for $M,N \rightarrow\infty$ at the same rate. Assuming that $\boldsymbol{\Gamma}$ has bounded spectral norm and that the infimum of the smallest eigenvalue of $\boldsymbol{\Gamma}$ is bounded away from zero uniformly in $M$ we have 
	\begin{equation}
		\boldsymbol{\Gamma}^{-1/2}\left(  \sqrt{N} (\hat{\boldsymbol{\eta}}^{(m)} - \bar{\boldsymbol{\eta}}^{(m)}) - \boldsymbol{\mu}\right) \overset{\mathcal{D}}{\rightarrow} \mathcal{N}(\boldsymbol{0},\boldsymbol{I}_{L}),
		\label{eq: convergence_in_distirbution}
	\end{equation}
	for $M,N \rightarrow \infty$ at the same rate ($M/N\rightarrow c$, $0<c<\infty$) where $\overset{\mathcal{D}}{\rightarrow}$ denotes convergence in distribution. In order to derive $\boldsymbol{\mu}$ and $\boldsymbol{\Gamma}$ in \eqref{eq: convergence_in_distirbution} we choose a deterministic parameterization of the contour $\mathcal{C}_{z}^{(m)}$ that is independent on the sample eigenvalues an $M$ and $N$. Therefore, it is assumed that there exists a deterministic $\varrho^{-}$ and $\varrho^{+}$ for the cluster $\lbrack x_{s}^{-}, x_{s}^{+} \rbrack$ that is associated to $\gamma_{m}$ such that
	\begin{equation*}
		\begin{aligned}
			\varrho^{-} < 0 <\underset{M,N}{\sup} x_{s}^{+} < \varrho^{+} < \underset{M,N}{\inf} x_{s+1}^{-}, \quad &\text{for } s=1 \\
			\underset{M,N}{\sup} x_{s-1}^{+} < \varrho^{-} < \underset{M,N}{\inf} x_{s}^{-}  \leq \underset{M,N}{\sup} x_{s}^{+} < \varrho^{+} < \underset{M,N}{\inf} x_{s+1}^{-},\quad  &\text{for } 1<s<S \\
			\underset{M,N}{\sup} x_{s-1}^{+} < \varrho^{-} < \underset{M,N}{\inf}x_{s}^{-}\leq \underset{M,N}{\sup} x_{s}^{+} < \varrho^{+} , \quad  &\text{for } s=S.
		\end{aligned}
	\end{equation*}
	This allows to deform the contour such that it becomes independent on $M$ and $N$. The negatively orientated deterministic contour ${\mathcal{C}}_{z}^{(m)}$ may cross the real axis at the points $\varrho^{-}$ and $\varrho^{+}$ and can e.g. be described by the boundaries of the rectangle 
	\begin{equation}
		\mathbb{R}_{y}^{(m)} = \left \lbrace z \in \mathbb{C} \text{ : } \varrho^{-} \leq \real{z} \leq \varrho^{+}, \text{ } \abs{\imag{z}} \leq y \right \rbrace
		\label{eq: deterministic_contour}
	\end{equation}
	where $y>0$ and ${\mathcal{C}}_{z}^{(m)}=\partial \mathbb{R}_{y}^{(m)}$. Unfortunately, the random variable ${\eta}^{(m)}$ in \eqref{eq: one_dimensional_projection} does not need to have a characteristic function for all $M$ and $N$ as there might exist realizations for which the singularities of $\hat{h}(z)\hat{\boldsymbol{Q}}(z)$ become dangerously close to the contour $\mathcal{C}_{z}^{(m)}$ or even lie on $\mathcal{C}_{z}^{(m)}$. Therefore, we need to consider the points where $\hat{h}(z)\hat{\boldsymbol{Q}}(z)$ is not holomorphic, namely $\hat{\lambda}_{k}$ for $k=1,\dots,M$ and in case of the $M,N$-consistent subspace estimator (the g-MUSIC case) where
	\begin{equation}
		\hat{h}(z) = \frac{z}{\hat{\omega}(z)}\frac{\partial \hat{\omega}(z)}{\partial z} = 1+\hat{\omega}(z) \frac{1}{N}\tr{\hat{\boldsymbol{R}}\hat{\boldsymbol{Q}}^{2}(z)} = 1+\frac{z}{1-\frac{1}{N}\sum_{m=1}^{M}\frac{\hat{\lambda}_{m}}{\hat{\lambda}_{m}-z}}\frac{1}{N}\tr{\hat{\boldsymbol{R}}\hat{\boldsymbol{Q}}^{2}(z)},
		\label{eq:function_h_hat_g_music}
	\end{equation}
	also the points $\hat{\mu}_{k}$ for $k=1,\dots,M$ which are the solutions (counting multiplicities) to 
	\begin{equation}
		1=\frac{1}{N}\sum_{m=1}^{M}\frac{\hat{\lambda}_{m}}{\hat{\lambda}_{m}-z} = \frac{1}{N}\sum_{m=\lbrack M-N\rbrack^{+}+1}^{M} \frac{\hat{\lambda}_{m}}{\hat{\lambda}_{m}-z},
		\label{eq:polynomial_mu_hat_values_phi_function}
	\end{equation}
	where
	\begin{equation*}
		\left \lbrack M-N\right \rbrack^{+} = \begin{cases}
			M-N & \text{for } M>N \\
			0 & \text{for } M\leq N.
		\end{cases}
	\end{equation*}
	Furthermore, it can be observed that the $\hat{\mu}$-values satisfy the interlacing property such that $0<\hat{\mu}_{1}<\hat{\lambda}_{1}<\hat{\mu}_{2}<\dots<\hat{\lambda}_{M-1}<\hat{\mu}_{M}<\hat{\lambda}_{M}$ for $M< N$, $0=\hat{\mu}_{1}<\hat{\lambda}_{1}<\hat{\mu}_{2}<\dots<\hat{\lambda}_{M-1}<\hat{\mu}_{M}<\hat{\lambda}_{M}$ for $M=N$, and $0=\hat{\mu}_{1}=\hat{\lambda}_{1}=\dots=\hat{\mu}_{M-N}=\hat{\lambda}_{M-N}=\hat{\mu}_{M-N+1}<\hat{\lambda}_{M-N+1}<\hat{\mu}_{M-N+2}<\dots<\hat{\lambda}_{M-1}<\hat{\mu}_{M}<\hat{\lambda}_{M}$ for $M>N$. It can be observed that, for all $M$ and $N$ sufficiently large, the roots $\lbrace \hat{\mu}_{k} \text{ : } k\in\mathcal{K}_{m}\rbrace$ of the polynomial in \eqref{eq:polynomial_mu_hat_values_phi_function} and the sample eigenvalues $\lbrace \hat{\lambda}_{k} \text{ : } k\in\mathcal{K}_{m}\rbrace$ always lie inside the contour $\mathcal{C}_{z}^{(m)}$ \cite[Appendix IV]{a20}. For those realizations of $M$ and $N$ where either one of the sample eigenvalues $\hat{\lambda}_{k}$ or one of the $\hat{\mu}$-values become dangerously close to the contour $\mathcal{C}_{z}^{(m)}$ or even lie on $\mathcal{C}_{z}^{(m)}$, we will follow the approach in \cite{a77, a83, a63} and consider an equivalent representation of $\eta^{(m)}$ in \eqref{eq: one_dimensional_projection} that is guaranteed to have a characteristic function for all $M$ and $N$. For this purpose we introduce $\mathcal{S}_{\epsilon} = \lbrace x \in \mathbb{R} \text{ : } \operatorname{dist}(x,\mathcal{S})\leq \epsilon \rbrace$ for positive $\epsilon > 0$. It is assumed that $\epsilon$ is small enough such that $\mathcal{S}_{\epsilon}$ does not intersect with the deterministic contour $\partial \mathbb{R}_{y}^{(m)}$ in \eqref{eq: deterministic_contour}. %
	Furthermore, let $\phi$ denote a smooth function $\phi \text{ : } \mathbb{R} \rightarrow \lbrack 0, 1\rbrack$ such that $\phi(x) = 1$ for $x \in \mathcal{S}_{\epsilon}$ and $\phi(x) = 0$ for $x \in \mathbb{R} \setminus \mathcal{S}_{2\epsilon}$. We will write $\phi_{M} = \tilde{\phi}_{M}\check{\phi}_{M}$ where
	\begin{equation}
		\tilde{\phi}_{M} = \prod \limits_{k\in\mathcal{K}_{m}}\phi\left(\hat{\lambda}_{k}\right), \qquad \text{and} \qquad \check{\phi}_{M} =  \prod \limits_{k\in \mathcal{K}_{m}} \phi\left(\hat{\mu}_{k}\right). 
		\label{eq:definition_regularizer}
	\end{equation}
	By \cite{a44,a83} and \cite[Appendix IV]{a20} we know that $\tilde{\phi}_{M} = 1$ and $\check{\phi}_{M}=1$ with probability one for all $M$ and $N$ sufficiently large. Furthermore, $\tilde{\phi}_{M}$ and $\check{\phi}_{M}$ can almost surely be extended to a continuous function, such that the integration by parts formula is applicable \cite[Appendix IV]{a83} and ${\phi}_{M}=\tilde{\phi}_{M}\check{\phi}_{M}=1$ for all $M$ and $N$ sufficiently large. Correspondingly, we may represent $\eta^{(m)}$ in \eqref{eq: one_dimensional_projection} as 
	\begin{equation}
		\eta^{(m)} = \sum_{p=1}^{L} w_{p} \eta_{p}^{(m)}\phi_{M} =  \sum_{p=1}^{L}\sqrt{N}w_{p}\left(\hat{\eta}_{p}^{(m)}-\bar{\eta}_{p}^{(m)}\right)\phi_{M} = \sum_{p=1}^{L}w_{p} \frac{\sqrt{N}}{2 \pi \mathrm{j}} \oint_{\mathcal{C}_{z}^{(m)}}\left( \hat{h}(z)\hat{m}_{p}(z) - \bar{h}(z)\bar{m}_{p}(z) \right) \phi_{M} \mathrm{d}z,
		\label{eq: one_dimensional_projection_phi}
	\end{equation}
	almost surely for $M,N$ sufficiently large. Please note that the characteristic in \eqref{eq: one_dimensional_projection_phi} exists for every realization and every possible $M$ and $N$. From now on we will therefore consider the definition of $\eta^{(m)}$ in \eqref{eq: one_dimensional_projection_phi}.

	During this proof, we strongly rely on the integration by parts formula for Gaussian functionals as well as the Nash-Poincar\'e in \cite[Proposition 2, supplementary material]{a63}. 
	
	\paragraph*{Remark} In the following $\mathcal{O}(N^{-k})$ denotes a general bivariate complex function that is bounded by $\epsilon(z_{1},z_{2})N^{-k}$, where $\epsilon(z_{1},z_{2})$ does not depend on $N$ and is such that 
	\begin{equation}
		\underset{(z_{1},z_{2})\in \mathcal{C}_{z_{1}}^{(m)} \times \mathcal{C}_{z_{2}}^{(m)}}{\sup} \left \| \epsilon(z_{1},z_{2}) \right \|< \infty.
		\label{eq: big_o}
	\end{equation}
	The function $\epsilon(z_{1},z_{2})$ may differ from one line to another and may be matrix valued in which case \eqref{eq: big_o} is understood as the spectral norm. Furthermore, $\mathcal{O}(N^{-\mathbb{N}})$ should be understood as bivariate complex function that can be written as $\mathcal{O}(N^{-l})$ for every $l\in \mathbb{N}$ (see \cite[Remark 1, supplementary material]{a63}). 
	\begin{proposition}
		Assume that, for each fixed $z\in\mathbb{C}$, function $\Omega(\boldsymbol{X},\boldsymbol{X}^{*},z) \text{ : } \mathbb{R}^{2MN} \rightarrow \mathbb{C}$ is continuously differentiable with polynomially bounded derivatives and partial derivatives. The regularizer $\phi_{M}$ is a continuous function with polynomially bounded partial derivatives. Then, for $M,N \rightarrow \infty$ at the same rate and if 
		\begin{equation*}
			\underset{z \in \mathcal{C}_{z}^{(m)}}{\sup} \mean{\abs{\Omega(\boldsymbol{X},\boldsymbol{X}^{*},z)\phi_{M}}^{2}}<C
		\end{equation*}
		for some positive deterministic $C$ independent on $M$ we have
		\begin{equation*}
			\mean{\Omega(\boldsymbol{X},\boldsymbol{X}^{*},z)\phi_{M}^{r}} = \mean{\Omega(\boldsymbol{X},\boldsymbol{X}^{*},z)\phi_{M}} + \mathcal{O}\left(N^{-\mathbb{N}}\right),
		\end{equation*}
		for any $r \in \mathbb{N}$ and also
		\begin{equation*}
			\mean{\Omega(\boldsymbol{X},\boldsymbol{X}^{*},z)\frac{\partial \phi_{M}}{\partial X_{ij}}} = \mathcal{O}\left(N^{-\mathbb{N}}\right).
		\end{equation*}
		This allows to ignore the presence of the regularizer $\phi_{M}$ in the asymptotic analysis up to an error term of order $\mathcal{O}(N^{-\mathbb{N}})$ which will be irrelevant for our purpose. 
	\end{proposition}
	\begin{proof}
		See \cite[Proposition 2, supplementary material]{a63}.
	\end{proof}
	Moreover, it can be seen that the characteristic function of $\eta^{(m)}$ in \eqref{eq: one_dimensional_projection_phi}, namely $\mean{\chi(r)}= \mean{\exp{\mathrm{j} r \eta^{(m)}}}$ is a differentiable function of $r$ with derivative
	\begin{align}
		\frac{\partial \mean{\chi(r)}}{\partial r} = \mathrm{j} \mean{\eta^{(m)} \chi(r)\phi_{M}} =& \mathrm{j}\sum_{p=1}^{L} w_{p} \sqrt{N} \mean{\left( \hat{\eta}_{p}^{(m)} - \bar{\eta}_{p}^{(m)}\right) \chi(r)\phi_{M}}\nonumber \\
		 =& \mathrm{j} \sum_{p=1}^{L} w_{p} \frac{\sqrt{N}}{2\pi \mathrm{j}} \oint_{\mathcal{C}_{z_{1}}^{(m)}}  \mean{\left(\hat{h}(z_{1})\hat{m}_{p}(z_{1}) - \bar{h}(z_{1})\bar{m}_{p}(z_{1})\right) \chi(r)\phi_{M}}\mathrm{d}z_{1} \nonumber \\
		=& \mathrm{j} \sum_{p=1}^{L} w_{p} \frac{\sqrt{N}}{2\pi \mathrm{j}} \oint_{\mathcal{C}_{z_{1}}^{(m)}} \mean{\boldsymbol{v}_{p}^{\operatorname{H}} \left(\hat{h}(z_{1})\hat{\boldsymbol{Q}}(z_{1}) -\bar{h}(z_{1}) \bar{\boldsymbol{Q}}(z_{1})\right) \boldsymbol{v}_{p} \chi(r)\phi_{M}}\mathrm{d}z_{1}.
		\label{eq: first_derivative_characteristic_function}
	\end{align}
	In parallel we compute the first order derivative of the characteristic function $\bar{\chi}(r)$ in \eqref{eq: characteristic_function_gaussian} with respect to $r$
	\begin{equation}
		\frac{\partial \bar{\chi}(r)}{\partial r} = \left( \mathrm{j}\boldsymbol{w}^{\operatorname{T}}\boldsymbol{\mu} - r \boldsymbol{w}^{\operatorname{T}}\boldsymbol{\Gamma}\boldsymbol{w} \right) \bar{\chi}(r).
		\label{eq: first_derivative_characteristic_function_gaussian}
	\end{equation}
	Instead of approximating the characteristic function $\mean{\chi(r)}= \mean{\exp{\mathrm{j} r \eta^{(m)}}}$ with $\bar{\chi}(r)$ in \eqref{eq: characteristic_function_gaussian} such that $\mean{\chi(r)} - \bar{\chi}(r) \rightarrow 0$, one can utilize the first order derivatives in \eqref{eq: first_derivative_characteristic_function} and \eqref{eq: first_derivative_characteristic_function_gaussian} to establish that 
	\begin{equation}
		\frac{\partial \mean{\chi(r)}}{\partial r} - \frac{\partial \bar{\chi}(r)}{\partial r} \rightarrow 0,
		\label{eq: first_derivative_characteristic_function_subtraction}
	\end{equation}
	for $M,N \rightarrow \infty$ at the same rate. Noting that $\chi(0) = 1$, it can be seen that the integration constant is zero which allows to compute $\boldsymbol{\mu}$ and $\boldsymbol{\Gamma}$ from \eqref{eq: first_derivative_characteristic_function_subtraction}. The quantity inside the expectation in \eqref{eq: first_derivative_characteristic_function} can be studied using the resolvent identity
	\begin{equation}
		z \hat{\boldsymbol{Q}}(z) = \hat{\boldsymbol{Q}}(z) \hat{\boldsymbol{R}} - \boldsymbol{I}_{M},
		\label{eq: resolvent_identity}
	\end{equation}
	which immediately follows from the definition of the resolvent $\hat{\boldsymbol{Q}}(z) = (\hat{\boldsymbol{R}} - z\boldsymbol{I}_{M})^{-1}$. According to Assumption 2 it is possible to express the sample covariance matrix in terms of the true covariance matrix as $\hat{\boldsymbol{R}} = \boldsymbol{R}^{1/2} \frac{\boldsymbol{X}\boldsymbol{X}^{H}}{N}\boldsymbol{R}^{\operatorname{H}/2}$ where $\boldsymbol{X} = \lbrack \boldsymbol{x}_{1}, \dots, \boldsymbol{x}_{N}\rbrack$ denotes an $M \times N$ matrix of independent identically distributed Gaussian random variables with law $\mathcal{CN}(0,1)$. The $(i,j)$-th entry of $\boldsymbol{X}$ is denoted by $X_{ij}$ and the sample covariance matrix $\hat{\boldsymbol{R}}$ can equivalently be expressed as 
	\begin{equation}
		\hat{\boldsymbol{R}} = \sum_{i=1}^{M}\sum_{j=1}^{N} X_{ij}\boldsymbol{R}^{1/2}\frac{\boldsymbol{e}_{i}\boldsymbol{x}_{j}^{\operatorname{H}}}{N}\boldsymbol{R}^{\operatorname{H}/2},
		\label{eq: sample_covariance_matrix_sum}
	\end{equation}
	where $\boldsymbol{e}_{i}$ denotes a column vector with all zeros except the $i$-th entry which is one. Throughout the proof we will be working with with the first order derivative of the sample covariance matrix $\hat{\boldsymbol{R}}$ and the resolvent $\hat{\boldsymbol{Q}}(z)$ with respect to $X_{ij}$ and its complex conjugate $X_{ij}^{*}$:
	\begin{gather}
		\frac{\partial }{\partial X_{ij}} \boldsymbol{R}^{1/2} \frac{\boldsymbol{X}\boldsymbol{X}^{\operatorname{H}}}{N}\boldsymbol{R}^{\operatorname{H}/2} = \boldsymbol{R}^{1/2} \frac{\boldsymbol{e}_{i}\boldsymbol{x}_{j}^{\operatorname{H}}}{N}\boldsymbol{R}^{\operatorname{H}/2}, \qquad \frac{\partial }{\partial X_{ij}^{*}} \boldsymbol{R}^{1/2} \frac{\boldsymbol{X}\boldsymbol{X}^{\operatorname{H}}}{N}\boldsymbol{R}^{\operatorname{H}/2} = \boldsymbol{R}^{1/2} \frac{\boldsymbol{x}_{j}\boldsymbol{e}_{i}^{\operatorname{T}}}{N}\boldsymbol{R}^{\operatorname{H}/2},\\
		\frac{\partial \hat{\boldsymbol{Q}}(z)}{\partial X_{ij}}= - \hat{\boldsymbol{Q}}(z) \boldsymbol{R}^{1/2} \frac{\boldsymbol{e}_{i}\boldsymbol{x}_{j}^{\operatorname{H}}}{N}\boldsymbol{R}^{\operatorname{H}/2} \hat{\boldsymbol{Q}}(z), \qquad \frac{\partial \hat{\boldsymbol{Q}}(z)}{\partial X_{ij}^{*}}= - \hat{\boldsymbol{Q}}(z) \boldsymbol{R}^{1/2} \frac{\boldsymbol{x}_{j}\boldsymbol{e}_{i}^{\operatorname{T}}}{N}\boldsymbol{R}^{\operatorname{H}/2} \hat{\boldsymbol{Q}}(z). \label{eq: derivative_resolvent}
	\end{gather}
	We begin by expressing the quantity of interest $\sqrt{N} \mean{\boldsymbol{v}_{p}^{\operatorname{H}}\left( \hat{h}(z_{1})\hat{\boldsymbol{Q}}(z_{1}) - \bar{h}(z_{1}) \bar{\boldsymbol{Q}}(z_{1})\right)\boldsymbol{v}_{p} \chi(r)\phi_{M}}$ in \eqref{eq: first_derivative_characteristic_function} as
	\begin{align}
		\sqrt{N} \mean{\boldsymbol{v}_{p}^{\operatorname{H}}\left( \hat{h}(z_{1})\hat{\boldsymbol{Q}}(z_{1}) - \bar{h}(z_{1}) \bar{\boldsymbol{Q}}(z_{1})\right)\boldsymbol{v}_{p} \chi(r)\phi_{M}} =& \sqrt{N} \mean{\boldsymbol{v}_{p}^{\operatorname{H}}\left( \hat{h}(z_{1})\hat{\boldsymbol{Q}}(z_{1}) - \bar{h}(z_{1}) \bar{\boldsymbol{Q}}(z_{1}) - \bar{h}(z_{1}) \hat{\boldsymbol{Q}}(z_{1}) + \bar{h}(z_{1}) \hat{\boldsymbol{Q}}(z_{1}) \right)\boldsymbol{v}_{p} \chi(r)\phi_{M}} \nonumber \\
		=& \sqrt{N} \mean{\left(\hat{h}(z_{1})-\bar{h}(z_{1})\right)\boldsymbol{v}_{p}^{\operatorname{H}}\hat{\boldsymbol{Q}}(z_{1}) \boldsymbol{v}_{p} \chi(r)\phi_{M}} \label{eq: expectation_quantity_of_interest_a}\\
		&+ \sqrt{N} \bar{h}(z_{1}) \mean{\boldsymbol{v}_{p}^{\operatorname{H}}\left( \hat{\boldsymbol{Q}}(z_{1}) - \bar{\boldsymbol{Q}}(z_{1})\right)\boldsymbol{v}_{p} \chi(r)\phi_{M}},\label{eq: expectation_quantity_of_interest_b}
	\end{align}	
	where we have added and subtracted $\bar{h}(z_{1})\hat{\boldsymbol{Q}}(z_{1})$ and used the fact that $\bar{h}(z_{1})$ is deterministic. The expectation in \eqref{eq: expectation_quantity_of_interest_a} can be upper bounded according to the following lemma.
	\begin{lemma}
		\label{lemma:disregard_term}
		As $M,N \rightarrow \infty$ at the same rate, the quantity in \eqref{eq: expectation_quantity_of_interest_a} can be approximated by
		\begin{equation}
			\sqrt{N} \mean{\left(\hat{h}(z_{1})-\bar{h}(z_{1})\right)\boldsymbol{v}_{p}^{\operatorname{H}}\hat{\boldsymbol{Q}}(z_{1}) \boldsymbol{v}_{p} \chi(r)\phi_{M}}  = \begin{cases} 0 & \text{for MUSIC} \\ \mathcal{O}\left(N^{-1/2}\right) & \text{for G-MUSIC}.\end{cases}
			\label{eq: mean_2_unimportant_disregard}
		\end{equation} 
	\end{lemma}
	\begin{proof}
		See Appendix \ref{proof:lemma: disregard_term}.	
	\end{proof}
	Next, we analyze the quantity in \eqref{eq: expectation_quantity_of_interest_b}. 
	\begin{lemma}
		\label{lemma:main_term}
		As $M,N \rightarrow \infty$ at the same rate, the quantity in \eqref{eq: expectation_quantity_of_interest_b} can be approximated by
		\begin{equation}
			\begin{aligned}
				 &\sqrt{N}\bar{h}(z_{1})\mean{\boldsymbol{v}_{p}^{\operatorname{H}} \left( \hat{\boldsymbol{Q}}(z_{1}) - \bar{\boldsymbol{Q}}(z_{1})\right) \boldsymbol{v}_{p} \chi(r)\phi_{M}} \\
				=& \mathrm{j} r \bar{h}(z_{1})\sum_{q=1}^{L} \frac{w_{q}}{2 \pi \mathrm{j}} \oint_{\mathcal{C}_{z_{2}}^{(m)}}\frac{z_{1}}{\omega(z_{1})}\frac{z_{2}}{\omega(z_{2})} \bar{h}(z_{2}) \frac{ \boldsymbol{v}_{p}^{\operatorname{H}} \bar{\boldsymbol{Q}}(z_{1}) \boldsymbol{R} \bar{\boldsymbol{Q}}(z_{2}) \boldsymbol{v}_{q} \boldsymbol{v}_{q}^{\operatorname{H}} \bar{\boldsymbol{Q}}(z_{2}) \boldsymbol{R} \bar{\boldsymbol{Q}}(z_{1}) \boldsymbol{v}_{p}}{1-\gamma(z_{1},z_{2})}\mean{\chi(r)}\mathrm{d}z_{2} + \mathcal{O}\left(N^{-1/2}\right).
			\end{aligned}
			\label{eq: mean_1_intermediate_8}
		\end{equation}
	\end{lemma}
	\begin{proof}
		See Appendix \ref{proof:lemma: main_term}.	
	\end{proof}
	Consequently, the quantity of interest can be expressed as
	\begin{equation}
		\begin{aligned}
			&\sqrt{N} \mean{\boldsymbol{v}_{p}^{\operatorname{H}}\left( \hat{h}(z_{1})\hat{\boldsymbol{Q}}(z_{1}) - \bar{h}(z_{1}) \bar{\boldsymbol{Q}}(z_{1})\right)\boldsymbol{v}_{p} \chi(r)\phi_{M}}\\
			=& \mathrm{j} r \bar{h}(z_{1})\sum_{q=1}^{L} \frac{w_{q}}{2 \pi \mathrm{j}} \oint_{\mathcal{C}_{z_{2}}^{(m)}}\frac{z_{1}}{\omega(z_{1})}\frac{z_{2}}{\omega(z_{2})} \bar{h}(z_{2}) \frac{ \boldsymbol{v}_{p}^{\operatorname{H}} \bar{\boldsymbol{Q}}(z_{1}) \boldsymbol{R} \bar{\boldsymbol{Q}}(z_{2}) \boldsymbol{v}_{q} \boldsymbol{v}_{q}^{\operatorname{H}} \bar{\boldsymbol{Q}}(z_{2}) \boldsymbol{R} \bar{\boldsymbol{Q}}(z_{1}) \boldsymbol{v}_{p}}{1-\gamma(z_{1},z_{2})}\mean{\chi(r)}\mathrm{d}z_{2} + \mathcal{O}\left(N^{-1/2}\right).
		\end{aligned}
		\label{eq:expectation_in_derivative}
	\end{equation}
	Substituting \eqref{eq:expectation_in_derivative} into $\frac{\partial\mean{ \chi(r)}}{\partial r}$ in \eqref{eq: first_derivative_characteristic_function} yields
	\begin{equation}
		\begin{aligned}
			&\frac{\partial \mean{\chi(r)}}{\partial r} = \mathrm{j} \sum_{p=1}^{L} \frac{w_{p}}{2\pi \mathrm{j}}\oint_{\mathcal{C}_{z_{1}}^{(m)}} \sqrt{N}\mean{\boldsymbol{v}_{p}^{\operatorname{H}} \left( \hat{h}(z_{1}) \hat{\boldsymbol{Q}}(z_{1}) - \bar{h}(z_{1}) \bar{\boldsymbol{Q}}(z_{1}) \right) \boldsymbol{v}_{p} \chi(r)\phi_{M}}dz_{1}\\
			=&- r \sum_{p=1}^{L} \sum_{q=1}^{L} \frac{w_{p}}{2\pi \mathrm{j}}\frac{w_{q}}{2\pi \mathrm{j}} \oint_{\mathcal{C}_{z_{1}}^{(m)}}\oint_{\mathcal{C}_{z_{2}}^{(m)}} \frac{z_{1}}{\omega(z_{1})}\frac{z_{2}}{\omega(z_{2})} \bar{h}(z_{1})\bar{h}(z_{2}) \frac{ \boldsymbol{v}_{p}^{\operatorname{H}} \bar{\boldsymbol{Q}}(z_{1}) \boldsymbol{R} \bar{\boldsymbol{Q}}(z_{2}) \boldsymbol{v}_{q} \boldsymbol{v}_{q}^{\operatorname{H}} \bar{\boldsymbol{Q}}(z_{2}) \boldsymbol{R} \bar{\boldsymbol{Q}}(z_{1}) \boldsymbol{v}_{p}}{1-\gamma(z_{1},z_{2})}\mathrm{d}z_{2} \mathrm{d}z_{1}\mean{\chi(r)} + \mathcal{O}\left(N^{-1/2}\right).
		\end{aligned}
		\label{eq: mean_1_intermediate_9}
	\end{equation}
	Comparing \eqref{eq: mean_1_intermediate_9} with \eqref{eq: first_derivative_characteristic_function_gaussian} as proposed in \eqref{eq: first_derivative_characteristic_function_subtraction}, it can be seen that $\boldsymbol{\mu} = \boldsymbol{0}$ and 
	\begin{equation}
		\left \lbrack \boldsymbol{\Gamma} \right \rbrack_{p,q}= \frac{1}{2\pi \mathrm{j}}\frac{1}{2\pi \mathrm{j}} \oint_{\mathcal{C}_{z_{1}}^{(m)}}\oint_{\mathcal{C}_{z_{2}}^{(m)}} \frac{z_{1}}{\omega(z_{1})}\frac{z_{2}}{\omega(z_{2})} \bar{h}(z_{1})\bar{h}(z_{2}) \frac{ \boldsymbol{v}_{p}^{\operatorname{H}} \bar{\boldsymbol{Q}}(z_{1}) \boldsymbol{R} \bar{\boldsymbol{Q}}(z_{2}) \boldsymbol{v}_{q} \boldsymbol{v}_{q}^{\operatorname{H}} \bar{\boldsymbol{Q}}(z_{2}) \boldsymbol{R} \bar{\boldsymbol{Q}}(z_{1}) \boldsymbol{v}_{p}}{1-\gamma(z_{1},z_{2})}\mathrm{d}z_{2}\mathrm{d}z_{1},
		\label{eq: mean_1_intermediate_10}
	\end{equation}
	where $p,q \in \lbrace 1, \dots, L \rbrace$. Substituting $\bar{\boldsymbol{Q}}(z)$ from \eqref{eq: resolvent_and_deterministic_equivalent} into \eqref{eq: mean_1_intermediate_10} and applying the change of variables $\mathrm{d}z = \frac{\partial z(\omega)}{\partial \omega}\mathrm{d}\omega$ we obtain	
	\begin{equation}
		\begin{aligned}
			\left \lbrack \boldsymbol{\Gamma} \right \rbrack_{p,q}= \frac{1}{2\pi \mathrm{j}}\frac{1}{2\pi \mathrm{j}} \oint_{\mathcal{C}_{\omega_{1}}^{(m)}}\oint_{\mathcal{C}_{\omega_{2}}^{(m)}}  \bar{h}(z(\omega_{1}))\bar{h}(z(\omega_{2})) \frac{\omega_{1}}{z(\omega_{1})}\frac{\omega_{2}}{z(\omega_{2})} \frac{\partial z(\omega_{1})}{\partial \omega_{1}} \frac{\partial z(\omega_{2})}{\partial \omega_{2}}  \frac{\Upsilon_{p,q}(\omega_{1},\omega_{2})}{1-\Omega(\omega_{1},\omega_{2})}\mathrm{d}\omega_{2}\mathrm{d}\omega_{1}
		\end{aligned}
	\end{equation}
	where $\mathcal{C}_{\omega}^{(m)} = \omega(\mathcal{C}_{z}^{(m)})$ denotes a negatively orientated contour that encloses $\gamma_{m}$ only and 
	\begin{gather*}
		\Omega(\omega_{1},\omega_{2}) = \frac{1}{N} \tr{\boldsymbol{R}\left(\boldsymbol{R}-\omega_{1}\boldsymbol{I}_{M}\right)^{-1}\boldsymbol{R}\left(\boldsymbol{R}-\omega_{2}\boldsymbol{I}_{M}\right)^{-1}},\\
		\Upsilon_{p,q}(\omega_{1},\omega_{2})= \boldsymbol{v}_{p}^{\operatorname{H}} \left(\boldsymbol{R} - \omega_{1}\boldsymbol{I}_{M}\right)^{-1} \boldsymbol{R} \left(\boldsymbol{R}-\omega_{2}\boldsymbol{I}_{M}\right)^{-1} \boldsymbol{v}_{q} \boldsymbol{v}_{q}^{\operatorname{H}} \left( \boldsymbol{R} -\omega_{2}\boldsymbol{I}_{M}\right)^{-1}\boldsymbol{R} \left( \boldsymbol{R} - \omega_{1}\boldsymbol{I}_{M}\right)^{-1}\boldsymbol{v}_{p}.
	\end{gather*}	
	In order to obtain the expression given in Theorem 3 of the article we choose $m=1$ and $\boldsymbol{v}_{l} = \boldsymbol{a}(\bar{\theta}_{l})$ for $l=1,\dots,L$.

\section{Proof of Lemma \ref{lemma:disregard_term}}
	\label{proof:lemma: disregard_term}
	It can be observed that in case of the conventional $N$-consistent estimator (the MUSIC case) $$\sqrt{N}\mean{\left(\hat{h}(z_{1})-\bar{h}(z_{1})\right) \boldsymbol{v}_{p}^{\operatorname{H}} \hat{\boldsymbol{Q}}(z_{1})\boldsymbol{v}_{p}\chi(r)\phi_{M}}=0,$$ since $\hat{h}(z_{1})-\bar{h}(z_{1}) = 0$. However, in case of the $M,N$-consistent estimator (the g-MUSIC case) we have to study the asymptotic behavior of $(\hat{h}(z)-\bar{h}(z))$ since $\hat{h}(z_{1})-\bar{h}(z_{1}) \neq 0$. Using the Cauchy-Schwarz inequality we can upper-bound
	\begin{align}
		\abs{\mean{\left(\hat{h}(z_{1})-\bar{h}(z_{1})\right) \boldsymbol{v}_{p}^{\operatorname{H}} \hat{\boldsymbol{Q}}(z_{1})\boldsymbol{v}_{p}\chi(r)\phi_{M}}}^{2}  \leq & \mean{\abs{\hat{h}(z_{1})-\bar{h}(z_{1})}^{2}\phi_{M}} \mean{\abs{\boldsymbol{v}_{p}^{\operatorname{H}} \hat{\boldsymbol{Q}}(z_{1})\boldsymbol{v}_{p}\chi(r)}^{2}\phi_{M}}\nonumber \\
		=& \mean{\abs{\hat{h}(z_{1})-\bar{h}(z_{1})}^{2}\phi_{M}} \mean{\abs{\boldsymbol{v}_{p}^{\operatorname{H}} \hat{\boldsymbol{Q}}(z_{1})\boldsymbol{v}_{p}}^{2}\abs{\chi(r)}^{2}\phi_{M}}\nonumber \\
		=&\mean{\abs{\hat{h}(z_{1})-\bar{h}(z_{1})}^{2}\phi_{M}} \mean{\abs{\boldsymbol{v}_{p}^{\operatorname{H}} \hat{\boldsymbol{Q}}(z_{1})\boldsymbol{v}_{p}}^{2}\phi_{M}}, \label{eq:intermediate_1_disregard_term}
	\end{align}
	where we have used $\abs{\chi(r)}^{2} = 1$. Using the definition of the variance of a random variable we can express $\mean{\abs{\hat{h}(z_{1})-\bar{h}(z_{1})}^{2}\phi_{M}}$ in \eqref{eq:intermediate_1_disregard_term} as follows
	\begin{equation}
		\mean{\abs{\hat{h}(z_{1})-\bar{h}(z_{1})}^{2}\phi_{M}} = \var{(\hat{h}(z_{1})-\bar{h}(z_{1}))\phi_{M}} + \mean{(\hat{h}(z_{1})-\bar{h}(z_{1}))\phi_{M}}^{2}. 
		\label{eq:intermediate_2_disregard_term}
	\end{equation}
	The asymptotic stochastic behavior of $\hat{h}(z)$ is analyzed next. 	
	\begin{lemma}
		\label{aux:lemma5}
		As $M,N\rightarrow \infty$ at the same rate,		
		\begin{equation*}
			\mean{\hat{h}(z)\phi_{M}} = \bar{h}(z) + \mathcal{O}\left(N^{-1}\right)
		\end{equation*}
		and the variance decays with 
		\begin{equation*}
			\var{\hat{h}(z)\phi_{M}} = \mathcal{O}\left(N^{-2}\right).
		\end{equation*}
	\end{lemma}
	\begin{proof}
		See Appendix \ref{proof:aux:lemma5}
	\end{proof} 
	According to Lemma \ref{aux:lemma5} and \eqref{eq:intermediate_2_disregard_term}, $\mean{\abs{\hat{h}(z_{1})-\bar{h}(z_{1})}^{2}\phi_{M}}= \mathcal{O}\left(N^{-2}\right)$. Similarly, we can express the second term on the right hand side of \eqref{eq:intermediate_1_disregard_term} as
	\begin{equation*}
		\mean{\abs{\boldsymbol{v}_{p}^{\operatorname{H}} \hat{\boldsymbol{Q}}(z_{1})\boldsymbol{v}_{p}\phi_{M}}^{2}} = \var{\boldsymbol{v}_{p}^{\operatorname{H}} \hat{\boldsymbol{Q}}(z_{1})\boldsymbol{v}_{p}\phi_{M}} + \mean{\boldsymbol{v}_{p}^{\operatorname{H}} \hat{\boldsymbol{Q}}(z_{1})\boldsymbol{v}_{p}\phi_{M}}^{2} = \mathcal{O}\left(N^{-1}\right) + \mathcal{O}\left(1\right) = \mathcal{O}\left(1\right)
	\end{equation*}
	such that \eqref{eq:intermediate_1_disregard_term} decays with
	\begin{equation*}
		\abs{\mean{\left(\hat{h}(z_{1})-\bar{h}(z_{1})\right) \boldsymbol{v}_{p}^{\operatorname{H}} \hat{\boldsymbol{Q}}(z_{1})\boldsymbol{v}_{p}\chi(r)\phi_{M}}}^{2} \leq  \mean{\abs{\hat{h}(z_{1})-\bar{h}(z_{1})}^{2}\phi_{M}} \mean{\abs{\boldsymbol{v}_{p}^{\operatorname{H}} \hat{\boldsymbol{Q}}(z_{1})\boldsymbol{v}_{p}\chi(r)}^{2}\phi_{M}} = \mathcal{O}\left(N^{-2}\right),
	\end{equation*}
	and consequently, the quantity of interest decays with
	\begin{equation*}
		\sqrt{N}\mean{\left(\hat{h}(z_{1})-\bar{h}(z_{1})\right) \boldsymbol{v}_{p}^{\operatorname{H}} \hat{\boldsymbol{Q}}(z_{1})\boldsymbol{v}_{p}\chi(r)\phi_{M}} = \mathcal{O}\left(N^{-1/2}\right).
	\end{equation*}

	\section{Proof of Lemma \ref{aux:lemma5}}
		\label{proof:aux:lemma5}
		Computing the first order derivative of 	
		\begin{equation}
			\frac{z}{\hat{\omega}(z)} = 1- \frac{1}{N}\tr{\hat{\boldsymbol{R}}\hat{\boldsymbol{Q}}(z)} = 1- \frac{1}{N} \tr{\hat{\boldsymbol{R}}\left(\hat{\boldsymbol{R}}-z\boldsymbol{I}_{M}\right)^{-1}},
			\label{eq:proof:aux:lemma5:hat_omega_z}
		\end{equation}
		w.r.t. $z$, we can express 
		\begin{equation}
			\hat{\omega}^{\prime}(z)=\frac{\partial \hat{\omega}(z)}{\partial z} = \frac{\hat{\omega}(z)}{z} + \frac{\hat{\omega}(z)^{2}}{z} \frac{1}{N} \tr{\hat{\boldsymbol{R}}\hat{\boldsymbol{Q}}(z)^{2}} = \frac{\hat{\omega}(z)}{z} + \frac{\hat{\omega}(z)^{2}}{z} \frac{1}{N} \tr{\hat{\boldsymbol{R}}\left(\hat{\boldsymbol{R}}-z\boldsymbol{I}_{M}\right)^{-2}},
			\label{eq:proof:aux:lemma5:hat_omega_z_prime}
		\end{equation}
		such that
		\begin{equation*}
			\hat{h}(z) = \frac{z}{\hat{\omega}(z)}\frac{\partial \hat{\omega}(z)}{\partial z}= 1+\hat{\omega}(z) \frac{1}{N} \tr{\hat{\boldsymbol{R}}\hat{\boldsymbol{Q}}(z)^{2}}.
		\end{equation*}
		Similarly, we compute the first order derivative of 
		\begin{equation}
			\frac{z}{\omega(z)} = 1-\frac{z}{\omega(z)}\frac{1}{N}\tr{\boldsymbol{R}\bar{\boldsymbol{Q}}(z)} = 1- \frac{1}{N}\tr{\boldsymbol{R}\left(\boldsymbol{R}-\omega(z)\boldsymbol{I}_{M}\right)^{-1}}, 
			\label{eq:proof:aux:lemma5:omega_z}
		\end{equation}
		w.r.t. $z$ and obtain
		\begin{align}
			\omega^{\prime}(z) = \frac{\partial \omega(z)}{\partial z} =& \left( \frac{z}{\omega(z)} - \omega(z) \frac{1}{N} \tr{\boldsymbol{R}\left(\boldsymbol{R}-\omega(z)\boldsymbol{I}_{M}\right)^{-2}} \right)^{-1} \nonumber \\
			=&  \left( \frac{z}{\omega(z)} - \omega(z) \frac{1}{N} \frac{z^{2}}{\omega(z)^{2}} \tr{\boldsymbol{R}\bar{\boldsymbol{Q}}(z)^{2}} \right)^{-1}\nonumber \\
			=& \frac{\omega(z)}{z} \left(1-z\frac{1}{N}\tr{\boldsymbol{R}\bar{\boldsymbol{Q}}(z)^{2}} \right)^{-1},
			\label{eq:proof:aux:lemma5:omega_z_prime}
		\end{align}
		such that 
		\begin{equation*}
			\bar{h}(z) = \frac{z}{\omega(z)}\frac{\partial \omega(z)}{\partial z}=\left(1-z\frac{1}{N}\tr{\boldsymbol{R}\bar{\boldsymbol{Q}}(z)^{2}} \right)^{-1}.
		\end{equation*}

		\subsection{Variance}
			The variance of $\hat{h}(z)=1+\hat{\omega}(z)\frac{1}{N}\tr{\hat{\boldsymbol{R}}\hat{\boldsymbol{Q}}(z)^{2}}\phi_{M}$ is identical to $\var{\hat{\omega}(z)\frac{1}{N}\tr{\hat{\boldsymbol{R}}\hat{\boldsymbol{Q}}(z)^{2}}\phi_{M}}$ which can be upper-bounded using the common bound
			 \begin{equation}
		        \var{\hat{\omega}(z)\frac{1}{N}\tr{\hat{\boldsymbol{Q}}(z)^{2}\hat{\boldsymbol{R}}}\phi_{M}} \leq 2 \left \| \hat{\omega}(z)\phi_{M}\right\|^{2} \var{\frac{1}{N} \tr{\hat{\boldsymbol{Q}}(z)^{2}\hat{\boldsymbol{R}}}\phi_{M}} + 2 \left\|\frac{1}{N} \tr{\hat{\boldsymbol{Q}}(z)^{2}\hat{\boldsymbol{R}}}\phi_{M}\right\|^{2} \var{\hat{\omega}(z)\phi_{M}}.
		        \label{eq: variance_upper_bound_intermediate_0}
		    \end{equation}
		   According to Lemma \ref{aux:lemma1_1}, 
		    \begin{equation*}
		        \var{\frac{1}{N} \tr{\hat{\boldsymbol{Q}}(z)^{2}\hat{\boldsymbol{R}}\phi_{M}}}  = \mathcal{O}\left(N^{-2}\right).
		    \end{equation*}
		    Next, the variance of $\hat{\omega}(z)$ in \eqref{eq:proof:aux:lemma5:hat_omega_z} is upper-bounded by	    
		    \begin{equation}
		    	\begin{aligned}
		    		\var{\hat{\omega}(z)\phi_{M}} =& \mean{\abs{\hat{\omega}(z)-\mean{\hat{\omega}(z)}}^{2}\phi_{M}} = \mean{\abs{\hat{\omega}(z) - \omega(z) + \omega(z) -\mean{\hat{\omega}(z)}}^{2}\phi_{M}}\\
		        \leq & 2 \mean{\abs{\hat{\omega}(z)-\omega(z)}^{2}\phi_{M}}+2 \mean{\abs{\omega(z)-\mean{\hat{\omega}(z)}}^{2}\phi_{M}},
		    	\end{aligned}
		    	\label{eq: variance_upper_bound_intermediate_01}
		    \end{equation}
		    where we have added and subtracted $\omega(z)$ and applied Jensen's inequality. Inserting the expressions of $\omega(z)$ in \eqref{eq:proof:aux:lemma5:omega_z} and $\hat{\omega}(z)$ in \eqref{eq:proof:aux:lemma5:hat_omega_z} into \eqref{eq: variance_upper_bound_intermediate_01} yields
		    \begin{equation}
		        \begin{aligned}
		            \var{\hat{\omega}(z)\phi_{M}} \leq&  2 \mean{ \abs{\frac{z}{1-\frac{1}{N}\tr{\hat{\boldsymbol{R}}\hat{\boldsymbol{Q}}(z)}}-\frac{z}{1-\frac{1}{N}\tr{\boldsymbol{R}\boldsymbol{Q}(z)}}}^{2}\phi_{M}}\\
		            &+ 2\mean{ \abs{\frac{z}{1-\frac{1}{N}\tr{\boldsymbol{R}\boldsymbol{Q}(z)}}- \mean{\frac{z}{1-\frac{1}{N}\tr{\hat{\boldsymbol{R}}\hat{\boldsymbol{Q}}(z)}}}}^{2}\phi_{M}},
		        \end{aligned}
		        \label{eq: variance_upper_bound_intermediate_1}
		     \end{equation}
		    where $\boldsymbol{Q}(z) = (\boldsymbol{R}-\omega(z)\boldsymbol{I}_{M})^{-1}$ and $\hat{\boldsymbol{Q}}(z) = (\hat{\boldsymbol{R}}-z\boldsymbol{I}_{M})^{-1}$. Before the upper-bounds for the first and second term on the right hand side of \eqref{eq: variance_upper_bound_intermediate_1} are computed it is shown that
		    \begin{equation}
		        \frac{\phi_{M}}{\abs{1-\frac{1}{N}\tr{\hat{\boldsymbol{R}}\hat{\boldsymbol{Q}}(z)}}}
		        \label{eq: variance_upper_bound_intermediate_2}
		    \end{equation}
		    is bounded on $z\in\mathcal{C}_{z}^{(m)}$ for all $M,N$ sufficiently large which is essential in order to proof that \eqref{eq: variance_upper_bound_intermediate_1} is bounded. Using the fact that the roots $\hat{\mu}_{k}$ are the solutions to the polynomial equation $1=\frac{1}{N}\sum_{r=1}^{M}\frac{\hat{\lambda}_{r}}{\hat{\lambda}_{r}-z}=\frac{1}{N}\sum_{r=\lbrack M-N\rbrack^{+}+1}^{M}\frac{\hat{\lambda}_{r}}{\hat{\lambda}_{r}-z}$ we can establish the following identity
		    \begin{equation*}
		    	\prod_{k=\lbrack M-N \rbrack^{+}+1}^{M} (\hat{\mu}_{k}-z) = \prod_{k=\lbrack M-N \rbrack^{+} +1 }^{M} (\hat{\lambda}_{k}-z) - \frac{1}{N} \sum_{r=\lbrack M-N\rbrack^{+}+1}^{M} \hat{\lambda}_{r}\prod_{k=\lbrack M-N\rbrack^{+}+1 \atop k \neq r}^{M} (\hat{\lambda}_{k}-z).
		    \end{equation*}
		    which allows us to express $1-\frac{1}{N}\tr{\hat{\boldsymbol{R}}\hat{\boldsymbol{Q}}(z)}$ as
		    \begin{equation}
		        1-\frac{1}{N}\sum_{m=1}^{M} \frac{\hat{\lambda}_{m}}{\hat{\lambda}_{m}-z} = \frac{\prod_{k=\lbrack M-N\rbrack^{+}+1}^{M} (\hat{\mu}_{k}-z)}{\prod_{k=\lbrack M-N\rbrack^{+}+1}^{M} (\hat{\lambda}_{k} - z)}. 
		        \label{eq: variance_upper_bound_intermediate_001}
		    \end{equation}
		    Furthermore, using the representation of $1-\frac{1}{N}\tr{\hat{\boldsymbol{R}}\hat{\boldsymbol{Q}}(z)}$ in \eqref{eq: variance_upper_bound_intermediate_001} we can establish that
		    \begin{equation}
		        \abs{ 1-\frac{1}{N}\sum_{m=1}^{M} \frac{\hat{\lambda}_{m}}{\hat{\lambda}_{m}-z}}^{-1}\phi_{M} =  \frac{\prod_{k=\lbrack M-N\rbrack^{+}+1}^{M} \abs{\hat{\lambda}_{k} - z}}{\prod_{k=\lbrack M-N\rbrack^{+}+1}^{M} \abs{\hat{\mu}_{k}-z}}\phi_{M} \leq \left( \frac{\| \boldsymbol{R}\|_{\text{F}}+\abs{z}+\epsilon}{\operatorname{dist}\left(z,\mathcal{S}\right)}\right)^{\min(M,N)}.
		        \label{eq: variance_upper_bound_intermediate_002}
		    \end{equation}
		    Due to the regularizer ${\phi}_{M}$ in \eqref{eq:definition_regularizer} we have a continuous function $\phi_{M}\left(1-\frac{1}{N}\tr{\hat{\boldsymbol{R}}\hat{\boldsymbol{Q}}(z)}\right)^{-1}$ in \eqref{eq: variance_upper_bound_intermediate_002} over a compact support that always has a minimum and a maximum and is therefore bounded. In the following the upper-bound of \eqref{eq: variance_upper_bound_intermediate_2} on $z \in \mathcal{C}_{z}^{(m)}$ is denoted by $\kappa(z) = \left( \frac{\| \boldsymbol{R}\|_{\text{F}}+\abs{z}+\epsilon}{\operatorname{dist}\left(z,\mathcal{S}\right)}\right)^{\min(M,N)}$. With this the first term on the right hand side of \eqref{eq: variance_upper_bound_intermediate_1} can be upper-bounded by
		    \begin{align}
		        \mean{\abs{\hat{\omega}(z)-\omega(z)}^{2}\phi_{M}} =& \mean{ \abs{\frac{z}{1-\frac{1}{N}\tr{\hat{\boldsymbol{R}}\hat{\boldsymbol{Q}}(z)}}-\frac{z}{1-\frac{1}{N}\tr{\boldsymbol{R}\boldsymbol{Q}(z)}}}^{2}\phi_{M}}\nonumber \\
		        =& \abs{z}^{2} \mean{ \abs{ \frac{1-\frac{1}{N}\tr{\boldsymbol{R}\boldsymbol{Q}(z)} - 1+\frac{1}{N}\tr{\hat{\boldsymbol{R}}\hat{\boldsymbol{Q}}(z)} }{\left(1-\frac{1}{N}\tr{\boldsymbol{R}\boldsymbol{Q}(z)}\right)\left( 1-\frac{1}{N}\tr{\hat{\boldsymbol{R}}\hat{\boldsymbol{Q}}(z)}\right)}}^{2}\phi_{M}}\nonumber \\
		        =& \abs{z}^{2} \mean{\abs{\frac{\frac{1}{N}\tr{\hat{\boldsymbol{R}}\hat{\boldsymbol{Q}}(z)}-\frac{1}{N}\tr{\boldsymbol{R}\boldsymbol{Q}(z)} }{\left(1-\frac{1}{N}\tr{\boldsymbol{R}\boldsymbol{Q}(z)}\right)\left( 1-\frac{1}{N}\tr{\hat{\boldsymbol{R}}\hat{\boldsymbol{Q}}(z)}\right)}}^{2}\phi_{M}}\nonumber \\
		        =& \abs{\omega(z)}^{2} \mean{\abs{\frac{\frac{1}{N}\tr{\hat{\boldsymbol{R}}\hat{\boldsymbol{Q}}(z)}-\frac{1}{N}\tr{\boldsymbol{R}\boldsymbol{Q}(z)} }{ 1-\frac{1}{N}\tr{\hat{\boldsymbol{R}}\hat{\boldsymbol{Q}}(z)}}}^{2}\phi_{M}}\nonumber \\
		        \leq & \underset{z \in \mathcal{C}_{z}^{(m)}}{\sup} \abs{\omega(z)\kappa(z) }^{2}\mean{\abs{\frac{1}{N}\tr{\hat{\boldsymbol{R}}\hat{\boldsymbol{Q}}(z)}-\frac{1}{N}\tr{\boldsymbol{R}\boldsymbol{Q}(z)}}^{2}\phi_{M}} \nonumber \\
		        =& \underset{z \in \mathcal{C}_{z}^{(m)}}{\sup} \abs{\omega(z)\kappa(z) }^{2} \left(\var{\frac{1}{N}\tr{\hat{\boldsymbol{R}}\hat{\boldsymbol{Q}}(z)\phi_{M}}} + \abs{\mean{\left(\frac{1}{N}\tr{\hat{\boldsymbol{R}}\hat{\boldsymbol{Q}}(z)} -\frac{1}{N}\tr{\boldsymbol{R}\boldsymbol{Q}(z)}\right)\phi_{M} }}^{2} \right) \nonumber \\
		        =& \underset{z \in \mathcal{C}_{z}^{(m)}}{\sup} \abs{\omega(z)\kappa(z) }^{2} \left(\var{\frac{1}{N}\tr{\hat{\boldsymbol{R}}\hat{\boldsymbol{Q}}(z)\phi_{M}}} + \mathcal{O}\left(N^{-2}\right) \right) \nonumber \\
		        =& \underset{z \in \mathcal{C}_{z}^{(m)}}{\sup} \abs{\omega(z)\kappa(z) }^{2}\mathcal{O}\left(N^{-2}\right), \label{eq: variance_upper_bound_intermediate_1a}
		    \end{align}	
		    where we have used \eqref{eq:aux:lemma1: expectation2} and \eqref{eq:aux:lemma1: variance2} in Lemma \ref{aux:lemma1}. Similarly, the second term on the right hand side of \eqref{eq: variance_upper_bound_intermediate_1} can be upper-bounded by
		    \begin{align}
		        \mean{\abs{\omega(z)-\mean{\hat{\omega}(z)}}^{2}\phi_{M}} =& \mean{\abs{\frac{z}{1-\frac{1}{N}\tr{\boldsymbol{R}\boldsymbol{Q}(z)}}- \mean{\frac{z}{1-\frac{1}{N}\tr{\hat{\boldsymbol{R}}\hat{\boldsymbol{Q}}(z)}}}}^{2}\phi_{M}}\nonumber  \\
		        =& \abs{ \mean{ \left(\frac{z}{1-\frac{1}{N}\tr{\boldsymbol{R}\boldsymbol{Q}(z)}} - \frac{z}{1-\frac{1}{N}\tr{\hat{\boldsymbol{R}}\hat{\boldsymbol{Q}}(z)}} \right)\phi_{M}} }^{2}\nonumber  \\
		        =& \abs{z}^{2} \abs{ \mean{ \frac{1-\frac{1}{N}\tr{\hat{\boldsymbol{R}}\hat{\boldsymbol{Q}}(z)} - 1 + \frac{1}{N}\tr{\boldsymbol{R}\boldsymbol{Q}(z)}}{\left(1-\frac{1}{N}\tr{\boldsymbol{R}\boldsymbol{Q}(z)}\right)\left(1-\frac{1}{N}\tr{\hat{\boldsymbol{R}}\hat{\boldsymbol{Q}}(z)}\right)} \phi_{M}} }^{2}\nonumber \\
		        =& \abs{\omega(z)}^{2} \abs{\mean{\frac{\frac{1}{N}\tr{\boldsymbol{R}\boldsymbol{Q}(z)} - \frac{1}{N}\tr{\hat{\boldsymbol{R}}\hat{\boldsymbol{Q}}(z)}}{1-\frac{1}{N}\tr{\hat{\boldsymbol{R}}\hat{\boldsymbol{Q}}(z)}}\phi_{M}}}^{2}\nonumber \\ 
		         \leq & \underset{z\in \mathcal{C}_{z}^{(m)}}{\sup}  \abs{\omega(z)\kappa(z) }^{2}\abs{ \mean{  \left(\frac{1}{N}\tr{\boldsymbol{R}\boldsymbol{Q}(z)} - \frac{1}{N}\tr{\hat{\boldsymbol{R}}\hat{\boldsymbol{Q}}(z)}\right)\phi_{M}}}^{2} \nonumber \\
		        =& \underset{z\in \mathcal{C}_{z}^{(m)}}{\sup}  \abs{\omega(z)\kappa(z) }^{2} \mathcal{O}\left(N^{-2}\right)
		        \label{eq: variance_upper_bound_intermediate_1b}
		    \end{align}
		    where we have used \eqref{eq:aux:lemma1: expectation2} in Lemma \ref{aux:lemma1}. Correspondingly, \eqref{eq: variance_upper_bound_intermediate_1a} and \eqref{eq: variance_upper_bound_intermediate_1b} both decay with $\mathcal{O}\left(N^{-2}\right)$ and the variance of $\hat{\omega}(z)$ in \eqref{eq: variance_upper_bound_intermediate_1} decays with $\var{\hat{\omega}(z)\phi_{M}}=\mathcal{O}\left(N^{-2}\right)$. Furthermore, since $\hat{\omega}(z)\phi_{M}$ and $\frac{1}{N} \tr{\hat{\boldsymbol{Q}}(z)^{2}\hat{\boldsymbol{R}}}\phi_{M}$ in \eqref{eq: variance_upper_bound_intermediate_0} are bounded on $z\in\mathcal{C}_{z}^{(m)}$ by definition, the variance in \eqref{eq: variance_upper_bound_intermediate_0} decays with $\mathcal{O}\left(N^{-2}\right)$. Hence, the variance of $\hat{h}(z)$ also decays with $\var{\hat{h}(z)\phi_{M}}=\mathcal{O}\left(N^{-2}\right)$.

		\subsection{Expectation}
			Using Lemma \ref{aux:lemma1_1} and the definition of $\omega^{\prime}(z)$ in \eqref{eq:proof:aux:lemma5:omega_z_prime} it can be shown that the expectation of $\hat{h}(z)$ is given by
			\begin{align*}
		        \mean{\hat{h}(z)\phi_{M}} =& \mean{\left(1+\hat{\omega}(z) \frac{1}{N} \tr{\hat{\boldsymbol{R}}\hat{\boldsymbol{Q}}(z)^{2}}\right)\phi_{M}} \\
		        =& 1+\omega(z)\frac{z^{2}}{\omega(z)^{2}}\frac{1}{N}\tr{\bar{\boldsymbol{Q}}(z)^{2}\boldsymbol{R}}\omega^{\prime}(z) + \mathcal{O}(N^{-1})\\
		        =& \omega^{\prime}(z) \left( \frac{1}{\omega^{\prime}(z)} + \frac{z^{2}}{\omega(z)}\frac{1}{N}\tr{\bar{\boldsymbol{Q}}(z)^{2}\boldsymbol{R}} \right)+ \mathcal{O}(N^{-1})\\
		        =& \omega^{\prime}(z) \left( \frac{z}{\omega(z)} - \frac{z^{2}}{\omega(z)}\frac{1}{N}\tr{\bar{\boldsymbol{Q}}(z)^{2}\boldsymbol{R}} +\frac{z^{2}}{\omega(z)}\frac{1}{N}\tr{\bar{\boldsymbol{Q}}(z)^{2}\boldsymbol{R}} \right) + \mathcal{O}(N^{-1})\\
		        =& \frac{z}{\omega(z)} \omega^{\prime}(z) + \mathcal{O}(N^{-1})\\
		        =&\bar{h}(z) + \mathcal{O}(N^{-1}),
		    \end{align*}
	    	where $\bar{h}(z)=\frac{z}{\omega(z)}\frac{\partial \omega(z)}{\partial z}$.

\section{Proof of Lemma \ref{lemma:main_term}}
	\label{proof:lemma: main_term}
	
	To begin with, we develop the expectation $\mean{\hat{\boldsymbol{Q}}(z)\hat{\boldsymbol{R}}\chi(r)\phi_{M}}$ by using the expression of the sample covariance matrix in \eqref{eq: sample_covariance_matrix_sum} and by applying the integration by parts formula 
	\begin{equation}
		\begin{aligned}
			\mean{\hat{\boldsymbol{Q}}(z)\hat{\boldsymbol{R}}\chi(r)\phi_{M}} =&  \sum_{i=1}^{M} \sum_{j=1}^{N}  \mean{X_{ij} \hat{\boldsymbol{Q}}(z) \boldsymbol{R}^{1/2} \frac{\boldsymbol{e}_{i}\boldsymbol{x}_{j}^{\operatorname{H}}}{N}\boldsymbol{R}^{\operatorname{H}/2}\chi(r)\phi_{M}}=\sum_{i=1}^{M}\sum_{j=1}^{N} \mean{\frac{\partial}{\partial X_{ij}^{*}} \hat{\boldsymbol{Q}}(z) \boldsymbol{R}^{1/2} \frac{\boldsymbol{e}_{i}\boldsymbol{x}_{j}^{\operatorname{H}}}{N}\boldsymbol{R}^{\operatorname{H}/2}\chi(r)\phi_{M}}, 
		\end{aligned}
		\label{eq:main: intermediate_1}
	\end{equation}
	where the first order derivative with respect to $X_{ij}^{*}$ is given by
	\begin{equation}
		\begin{aligned}
			&\frac{\partial}{\partial X_{ij}^{*}} \hat{\boldsymbol{Q}}(z_{1}) \boldsymbol{R}^{1/2} \frac{\boldsymbol{e}_{i}\boldsymbol{x}_{j}^{\operatorname{H}}}{N}\boldsymbol{R}^{\operatorname{H}/2}\chi(r)\phi_{M} \\
			=&  - \hat{\boldsymbol{Q}}(z_{1}) \boldsymbol{R}^{1/2} \frac{\boldsymbol{x}_{j}\boldsymbol{e}_{i}^{\operatorname{T}}}{N}\boldsymbol{R}^{\operatorname{H}/2} \hat{\boldsymbol{Q}}(z_{1}) \boldsymbol{R}^{1/2} \frac{\boldsymbol{e}_{i}\boldsymbol{x}_{j}^{\operatorname{H}}}{N}\boldsymbol{R}^{\operatorname{H}/2} \chi(r)\phi_{M}  + \hat{\boldsymbol{Q}}(z_{1}) \boldsymbol{R}^{1/2} \frac{\boldsymbol{e}_{i}\boldsymbol{e}_{i}^{\operatorname{T}}}{N}\boldsymbol{R}^{\operatorname{H}/2}\chi(r)\phi_{M} + \mathcal{O}\left(N^{-\mathbb{N}}\right)\\
			&+\mathrm{j}r \sum_{q=1}^{L}w_{q}\frac{\sqrt{N}}{2\pi\mathrm{j}} \oint_{\mathcal{C}_{z_{2}}^{(m)}} \hat{\boldsymbol{Q}}(z_{2}) \boldsymbol{R}^{1/2} \frac{\boldsymbol{e}_{i}\boldsymbol{x}_{j}^{\operatorname{H}}}{N}\boldsymbol{R}^{\operatorname{H}/2} \boldsymbol{v}_{q}^{\operatorname{H}} \left \lbrack \left(\frac{\partial}{\partial X_{ij}^{*}} \hat{h}(z_{2}) \right) \hat{\boldsymbol{Q}}(z_{2}) + \hat{h}(z_{2}) \left(\frac{\partial}{\partial X_{ij}^{*}} \hat{\boldsymbol{Q}}(z_{2})\right)\right \rbrack \boldsymbol{v}_{q} \chi(r)\phi_{M}\mathrm{d}z_{2}\\
			=&- \hat{\boldsymbol{Q}}(z_{1}) \boldsymbol{R}^{1/2} \frac{\boldsymbol{x}_{j}\boldsymbol{e}_{i}^{\operatorname{T}}}{N}\boldsymbol{R}^{\operatorname{H}/2} \hat{\boldsymbol{Q}}(z_{1}) \boldsymbol{R}^{1/2} \frac{\boldsymbol{e}_{i}\boldsymbol{x}_{j}^{\operatorname{H}}}{N}\boldsymbol{R}^{\operatorname{H}/2} \chi(r)\phi_{M}  + \hat{\boldsymbol{Q}}(z_{1}) \boldsymbol{R}^{1/2} \frac{\boldsymbol{e}_{i}\boldsymbol{e}_{i}^{\operatorname{T}}}{N}\boldsymbol{R}^{\operatorname{H}/2}\chi(r)\phi_{M} + \mathcal{O}\left(N^{-\mathbb{N}}\right)\\
			&- \mathrm{j} r \sum_{q=1}^{L} w_{q} \frac{\sqrt{N}}{2\pi \mathrm{j}}\oint_{\mathcal{C}_{z_{2}}^{(m)}}\hat{h}(z_{2}) \hat{\boldsymbol{Q}}(z_{1}) \boldsymbol{R}^{1/2} \frac{\boldsymbol{e}_{i}\boldsymbol{x}_{j}^{\operatorname{H}}}{N}\boldsymbol{R}^{\operatorname{H}/2} \boldsymbol{v}_{q}^{\operatorname{H}} \hat{\boldsymbol{Q}}(z_{2}) \boldsymbol{R}^{1/2} \frac{\boldsymbol{x}_{j}\boldsymbol{e}_{i}^{\operatorname{T}}}{N}\boldsymbol{R}^{\operatorname{H}/2}\hat{\boldsymbol{Q}}(z_{2}) \boldsymbol{v}_{q} \chi(r) \phi_{M}\mathrm{d}z_{2}\\
			&+\mathrm{j}r \sum_{q=1}^{L}w_{q}\frac{\sqrt{N}}{2\pi\mathrm{j}} \oint_{\mathcal{C}_{z_{2}}^{(m)}} \hat{\boldsymbol{Q}}(z_{2}) \boldsymbol{R}^{1/2} \frac{\boldsymbol{e}_{i}\boldsymbol{x}_{j}^{\operatorname{H}}}{N}\boldsymbol{R}^{\operatorname{H}/2} \boldsymbol{v}_{q}^{\operatorname{H}} \left(\frac{\partial}{\partial X_{ij}^{*}} \hat{h}(z_{2})\right) \hat{\boldsymbol{Q}}(z_{2}) \boldsymbol{v}_{q} \chi(r) \phi_{M}\mathrm{d}z_{2}.
		\end{aligned}
		\label{eq:main: intermediate_2}
	\end{equation}
	The first order derivative of $\hat{h}(z_{2})$ in \eqref{eq:function_h_hat_g_music} with respect to $X_{ij}^{*}$ yields 
	\begin{equation}
		\begin{aligned}
			\frac{\partial}{\partial X_{ij}^{*}} \hat{h}(z_{2}) =& \frac{\partial}{\partial X_{ij}^{*}}\left( 1+ \hat{\omega}(z_{2}) \frac{1}{N}\tr{\hat{\boldsymbol{R}}\hat{\boldsymbol{Q}}(z_{2})^{2}}\right) = \frac{\partial}{\partial X_{ij}^{*}}\left( 1+ \frac{z_{2} \frac{1}{N}\tr{\hat{\boldsymbol{R}}\hat{\boldsymbol{Q}}(z_{2})^{2}}}{1-\frac{1}{N}\tr{\hat{\boldsymbol{R}}\hat{\boldsymbol{Q}}(z_{2})}}\right)\\
			=& \hat{\omega}(z_{2}) \frac{1}{N}\sum_{m=1}^{M} \boldsymbol{e}_{m}^{\operatorname{T}} \boldsymbol{R}^{1/2} \frac{\boldsymbol{x}_{j}\boldsymbol{e}_{i}^{\operatorname{T}}}{N} \boldsymbol{R}^{\operatorname{H}/2} \hat{\boldsymbol{Q}}(z_{2})^{2} \boldsymbol{e}_{m}\\
			&- \hat{\omega}(z_{2}) \frac{2}{N} \sum_{m=1}^{M} \boldsymbol{e}_{m}^{\operatorname{T}} \hat{\boldsymbol{R}}\hat{\boldsymbol{Q}}(z_{2})^{3} \boldsymbol{R}^{1/2} \frac{\boldsymbol{x}_{j}\boldsymbol{e}_{i}^{\operatorname{T}}}{N} \boldsymbol{R}^{\operatorname{H}/2} \boldsymbol{e}_{m} \\
			&+ \frac{\hat{\omega}(z_{2})^{2}}{z_{2}}\frac{1}{N}\tr{\hat{\boldsymbol{R}}\hat{\boldsymbol{Q}}(z_{2})^{2}} \frac{1}{N}\sum_{m=1}^{M} \boldsymbol{e}_{m}^{\operatorname{T}} \boldsymbol{R}^{1/2} \frac{\boldsymbol{x}_{j}\boldsymbol{e}_{i}^{\operatorname{T}}}{N}\boldsymbol{R}^{\operatorname{H}/2} \hat{\boldsymbol{Q}}(z_{2}) \boldsymbol{e}_{m}\\ 
			&- \frac{\hat{\omega}(z_{2})^{2}}{z_{2}}\frac{1}{N}\tr{\hat{\boldsymbol{R}}\hat{\boldsymbol{Q}}(z_{2})^{2}} \frac{1}{N} \sum_{m=1}^{M} \boldsymbol{e}_{m}^{\operatorname{T}} \hat{\boldsymbol{R}} \hat{\boldsymbol{Q}}(z_{2}) \boldsymbol{R}^{1/2} \frac{\boldsymbol{x}_{j}\boldsymbol{e}_{i}^{\operatorname{T}}}{N} \boldsymbol{R}^{\operatorname{H}/2} \hat{\boldsymbol{Q}}(z_{2}) \boldsymbol{e}_{m}.
		\end{aligned}
		\label{eq:function_h_hat_derivative_g_music}
	\end{equation}
	Substituting the first oder derivative of $\hat{h}(z_{2})$ in \eqref{eq:function_h_hat_derivative_g_music} into \eqref{eq:main: intermediate_2} one obtains
	\begin{equation}
		\begin{aligned}
			&\frac{\partial}{\partial X_{ij}^{*}} \hat{\boldsymbol{Q}}(z_{1}) \boldsymbol{R}^{1/2} \frac{\boldsymbol{e}_{i}\boldsymbol{x}_{j}^{\operatorname{H}}}{N}\boldsymbol{R}^{\operatorname{H}/2}\chi(r)\phi_{M} \\
			=&- \hat{\boldsymbol{Q}}(z_{1}) \boldsymbol{R}^{1/2} \frac{\boldsymbol{x}_{j}\boldsymbol{e}_{i}^{\operatorname{T}}}{N}\boldsymbol{R}^{\operatorname{H}/2} \hat{\boldsymbol{Q}}(z_{1}) \boldsymbol{R}^{1/2} \frac{\boldsymbol{e}_{i}\boldsymbol{x}_{j}^{\operatorname{H}}}{N}\boldsymbol{R}^{\operatorname{H}/2} \chi(r)\phi_{M}  + \hat{\boldsymbol{Q}}(z_{1}) \boldsymbol{R}^{1/2} \frac{\boldsymbol{e}_{i}\boldsymbol{e}_{i}^{\operatorname{T}}}{N}\boldsymbol{R}^{\operatorname{H}/2}\chi(r)\phi_{M} + \mathcal{O}\left(N^{-\mathbb{N}}\right)\\
			&- \mathrm{j} r \sum_{q=1}^{L} w_{q} \frac{\sqrt{N}}{2\pi \mathrm{j}}\oint_{\mathcal{C}_{z_{2}}^{(m)}}\hat{h}(z_{2}) \hat{\boldsymbol{Q}}(z_{1}) \boldsymbol{R}^{1/2} \frac{\boldsymbol{e}_{i}\boldsymbol{x}_{j}^{\operatorname{H}}}{N}\boldsymbol{R}^{\operatorname{H}/2} \boldsymbol{v}_{q}^{\operatorname{H}} \hat{\boldsymbol{Q}}(z_{2}) \boldsymbol{R}^{1/2} \frac{\boldsymbol{x}_{j}\boldsymbol{e}_{i}^{\operatorname{T}}}{N}\boldsymbol{R}^{\operatorname{H}/2}\hat{\boldsymbol{Q}}(z_{2}) \boldsymbol{v}_{q} \chi(r) \phi_{M}\mathrm{d}z_{2}\\
			&+\mathrm{j} r \sum_{q=1}^{L} w_{q} \frac{\sqrt{N}}{2\pi \mathrm{j}} \oint_{\mathcal{C}_{z_{2}}^{(m)}} \frac{1}{N} \hat{\omega}(z_{2})\hat{\boldsymbol{Q}}(z_{2}) \boldsymbol{R}^{1/2}\frac{\boldsymbol{e}_{i}\boldsymbol{e}_{i}^{\operatorname{T}}}{N}\boldsymbol{R}^{\operatorname{H}/2} \hat{\boldsymbol{Q}}(z_{2})^{2} \boldsymbol{R}^{1/2} \frac{\boldsymbol{x}_{j}\boldsymbol{x}_{j}^{\operatorname{H}}}{N}\boldsymbol{R}^{\operatorname{H}/2} \boldsymbol{v}_{q}^{\operatorname{H}}\hat{\boldsymbol{Q}}(z_{2}) \boldsymbol{v}_{q}\chi(r)\phi_{M}\mathrm{d}z_{2}\\
			&-\mathrm{j} r \sum_{q=1}^{L} w_{q} \frac{\sqrt{N}}{2\pi \mathrm{j}} \oint_{\mathcal{C}_{z_{2}}^{(m)}} \frac{2}{N} \hat{\omega}(z_{2})\hat{\boldsymbol{Q}}(z_{2}) \boldsymbol{R}^{1/2}\frac{\boldsymbol{e}_{i}\boldsymbol{e}_{i}^{\operatorname{T}}}{N}\boldsymbol{R}^{\operatorname{H}/2} \hat{\boldsymbol{R}}\hat{\boldsymbol{Q}}(z_{2})^{3}\boldsymbol{R}^{1/2} \frac{\boldsymbol{x}_{j}\boldsymbol{x}_{j}}{N} \boldsymbol{R}^{\operatorname{H}/2} \boldsymbol{v}_{q}^{\operatorname{H}} \hat{\boldsymbol{Q}}(z_{2})\boldsymbol{v}_{q}\chi(r)\phi_{M}\mathrm{d}z_{2}\\
			&+ \mathrm{j} r \sum_{q=1}^{L} w_{q} \frac{\sqrt{N}}{2\pi \mathrm{j}} \oint_{\mathcal{C}_{z_{2}}^{(m)}}\frac{1}{N} \frac{\hat{\omega}(z_{2})^{2}}{z_{2}}  \frac{1}{N}\tr{\hat{\boldsymbol{R}}\hat{\boldsymbol{Q}}(z_{2})^{2}} \hat{\boldsymbol{Q}}(z_{2}) \boldsymbol{R}^{1/2} \frac{\boldsymbol{e}_{i}\boldsymbol{e}_{i}^{\operatorname{T}}}{N} \boldsymbol{R}^{\operatorname{H}/2} \hat{\boldsymbol{Q}}(z_{2}) \boldsymbol{R}^{1/2} \frac{\boldsymbol{x}_{j}\boldsymbol{x}_{j}^{\operatorname{H}}}{N} \boldsymbol{R}^{\operatorname{H}/2} \boldsymbol{v}_{q}^{\operatorname{H}} \hat{\boldsymbol{Q}}(z_{2})\boldsymbol{v}_{q} \chi(r) \phi_{M}\mathrm{d}z_{2}\\
			&- \mathrm{j} r \sum_{q=1}^{L} w_{q} \frac{\sqrt{N}}{2\pi \mathrm{j}} \oint_{\mathcal{C}_{z_{2}}^{(m)}}\frac{1}{N} \frac{\hat{\omega}(z_{2})^{2}}{z_{2}}  \frac{1}{N}\tr{\hat{\boldsymbol{R}}\hat{\boldsymbol{Q}}(z_{2})^{2}} \hat{\boldsymbol{Q}}(z_{2}) \boldsymbol{R}^{1/2}\frac{\boldsymbol{e}_{i}\boldsymbol{e}_{i}^{\operatorname{T}}}{N}\boldsymbol{R}^{\operatorname{H}/2} \hat{\boldsymbol{Q}}(z_{2}) \hat{\boldsymbol{R}} \hat{\boldsymbol{Q}}(z_{2}) \boldsymbol{R}^{1/2} \frac{\boldsymbol{x}_{j}\boldsymbol{x}_{j}^{\operatorname{H}}}{N} \boldsymbol{R}^{\operatorname{H}/2} \boldsymbol{v}_{q}^{\operatorname{H}}\hat{\boldsymbol{Q}}(z_{2}) \boldsymbol{v}_{q}\chi(r)\phi_{M}\mathrm{d}z_{2}
		\end{aligned}
		\label{eq:main: intermediate_2_1}
	\end{equation}
	Next, substituting \eqref{eq:main: intermediate_2_1} into \eqref{eq:main: intermediate_2} and computing the double sum we can equivalently express the expectation $\mean{\hat{\boldsymbol{Q}}(z_{1})\hat{\boldsymbol{R}}\chi(r)\phi_{M}}$ as
	\begin{equation}
		\begin{aligned}
			\mean{\hat{\boldsymbol{Q}}(z_{1})\hat{\boldsymbol{R}}\chi(r)\phi_{M}} =& - \mean{\hat{\boldsymbol{Q}}(z_{1}) \hat{\boldsymbol{R}} \chi(r)\phi_{M} \frac{1}{N} \tr{\hat{\boldsymbol{Q}}(z_{1})\boldsymbol{R}}} + \mean{\hat{\boldsymbol{Q}}(z_{1})\boldsymbol{R}\chi(r)\phi_{M}}+ \mathcal{O}\left(N^{-\mathbb{N}}\right) \\
			&- \mathrm{j} r \sum_{q=1}^{L} w_{q} \frac{\sqrt{N}}{2\pi \mathrm{j}} \frac{1}{N} \oint_{\mathcal{C}_{z_{2}}^{(m)}} \mean{\hat{h}(z_{2})\hat{\boldsymbol{Q}}(z_{1}) \boldsymbol{R} \hat{\boldsymbol{Q}}(z_{2}) \boldsymbol{v}_{q}\boldsymbol{v}_{q}^{\operatorname{H}} \hat{\boldsymbol{Q}}(z_{2})\hat{\boldsymbol{R}}\chi(r)\phi_{M}}\mathrm{d}z_{2} \\
			&+\mathrm{j} r \sum_{q=1}^{L} w_{q} \frac{\sqrt{N}}{2\pi \mathrm{j}}  \oint_{\mathcal{C}_{z_{2}}^{(m)}} \frac{1}{N^{2}}\mean{\hat{\omega}(z_{2}) \hat{\boldsymbol{Q}}(z_{2}) \boldsymbol{R} \hat{\boldsymbol{Q}}(z_{2})^{2} \hat{\boldsymbol{R}} \boldsymbol{v}_{q}^{\operatorname{H}} \hat{\boldsymbol{Q}}(z_{2}) \boldsymbol{v}_{q} \chi(r)\phi_{M}}\mathrm{d}z_{2}\\
			&- \mathrm{j} r \sum_{q=1}^{L} w_{q} \frac{\sqrt{N}}{2\pi \mathrm{j}}  \oint_{\mathcal{C}_{z_{2}}^{(m)}} \frac{2}{N^{2}} \mean{\hat{\omega}(z_{2}) \hat{\boldsymbol{Q}}(z_{2}) \boldsymbol{R}\hat{\boldsymbol{R}} \hat{\boldsymbol{Q}}(z_{2})^{3} \hat{\boldsymbol{R}} \boldsymbol{v}_{q}^{\operatorname{H}} \hat{\boldsymbol{Q}}(z_{2}) \boldsymbol{v}_{q}\chi(r)\phi_{M}}\mathrm{d}z_{2}\\
			&+ \mathrm{j} r \sum_{q=1}^{L} w_{q} \frac{\sqrt{N}}{2\pi \mathrm{j}}  \oint_{\mathcal{C}_{z_{2}}^{(m)}} \frac{1}{N^{2}}\mean{\frac{\hat{\omega}(z_{2})^{2}}{z_{2}}  \frac{1}{N}\tr{\hat{\boldsymbol{R}}\hat{\boldsymbol{Q}}(z_{2})^{2}} \hat{\boldsymbol{Q}}(z_{2}) \boldsymbol{R}\hat{\boldsymbol{Q}}(z_{2})\hat{\boldsymbol{R}}\boldsymbol{v}_{q}^{\operatorname{H}} \hat{\boldsymbol{Q}}(z_{2})\boldsymbol{v}_{q}\chi(r)\phi_{M}}\mathrm{d}z_{2}\\
			&-  \mathrm{j} r \sum_{q=1}^{L} w_{q} \frac{\sqrt{N}}{2\pi \mathrm{j}}  \oint_{\mathcal{C}_{z_{2}}^{(m)}} \frac{1}{N^{2}}\mean{\frac{\hat{\omega}(z_{2})^{2}}{z_{2}}  \frac{1}{N}\tr{\hat{\boldsymbol{R}}\hat{\boldsymbol{Q}}(z_{2})^{2}} \hat{\boldsymbol{Q}}(z_{2}) \boldsymbol{R}\hat{\boldsymbol{Q}}(z_{2})\hat{\boldsymbol{R}} \hat{\boldsymbol{Q}}(z_{2}) \hat{\boldsymbol{R}} \boldsymbol{v}_{q}^{\operatorname{H}} \hat{\boldsymbol{Q}}(z_{2}) \boldsymbol{v}_{q} \chi(r) \phi_{M}}\mathrm{d}z_{2}.
		\end{aligned}
		\label{eq:main: intermediate_3}
	\end{equation}
	In \eqref{eq:main: intermediate_3} it can be seen that due to the regularizer $\phi_{M}$ all expectations will converge to bounded quantities as $M$ and $N$ increase without bound at the same rate. Moreover, the last four terms on the right hand side in \eqref{eq:main: intermediate_3} decay one order faster than the third term on the right hand side. Therefore, we neglect the last four terms and add an error term $\mathcal{O}(N^{-3/2})$ such that
	\begin{equation}
		\begin{aligned}
			\mean{\hat{\boldsymbol{Q}}(z_{1})\hat{\boldsymbol{R}}\chi(r)\phi_{M}} =& - \mean{\hat{\boldsymbol{Q}}(z_{1}) \hat{\boldsymbol{R}} \chi(r)\phi_{M} \frac{1}{N} \tr{\hat{\boldsymbol{Q}}(z_{1})\boldsymbol{R}}} + \mean{\hat{\boldsymbol{Q}}(z_{1})\boldsymbol{R}\chi(r)\phi_{M}}+ \mathcal{O}\left(N^{-3/2}\right) \\
			&- \mathrm{j} r \sum_{q=1}^{L} w_{q} \frac{\sqrt{N}}{2\pi \mathrm{j}} \frac{1}{N} \oint_{\mathcal{C}_{z_{2}}^{(m)}} \mean{\hat{h}(z_{2})\hat{\boldsymbol{Q}}(z_{1}) \boldsymbol{R} \hat{\boldsymbol{Q}}(z_{2}) \boldsymbol{v}_{q}\boldsymbol{v}_{q}^{\operatorname{H}} \hat{\boldsymbol{Q}}(z_{2})\hat{\boldsymbol{R}}\chi(r)\phi_{M}}\mathrm{d}z_{2}.
		\end{aligned}
		\label{eq:main: intermediate_3_1}
	\end{equation}
	In the following we add the quantity $\mean{\hat{\boldsymbol{Q}}(z_{1}) \hat{\boldsymbol{R}} \chi(r)\phi_{M}}\frac{1}{N}\tr{\boldsymbol{R} \bar{\boldsymbol{Q}}(z_{1})}$ on both sides of \eqref{eq:main: intermediate_3_1}. For the left hand side of \eqref{eq:main: intermediate_3_1} we obtain
	\begin{equation}
		\begin{aligned}
			&\mean{\hat{\boldsymbol{Q}}(z_{1}) \hat{\boldsymbol{R}} \chi(r)\phi_{M}} + \mean{\hat{\boldsymbol{Q}}(z_{1}) \hat{\boldsymbol{R}} \chi(r)\phi_{M}}\frac{1}{N}\tr{\boldsymbol{R} \bar{\boldsymbol{Q}}(z_{1})}\\ 
			=& \mean{\hat{\boldsymbol{Q}}(z_{1}) \hat{\boldsymbol{R}} \chi(r)\phi_{M}} \left( 1+ \frac{1}{N} \tr{\boldsymbol{R} \bar{\boldsymbol{Q}}(z_{1})}\right) = \mean{\hat{\boldsymbol{Q}}(z_{1}) \hat{\boldsymbol{R}} \chi(r)\phi_{M}}\frac{\omega(z_{1})}{z_{1}},
		\end{aligned}
		\label{eq: lhs_mean_1_intermediate_1}
	\end{equation}
	where we have used $\left( 1 + \frac{1}{N} \tr{\boldsymbol{R} \bar{\boldsymbol{Q}}(z_{1})}\right) = \frac{\omega(z_{1})}{z_{1}}$ which follows from \eqref{eq: definition_omega_z}. For the right hand side of \eqref{eq:main: intermediate_3_1} we obtain
	\begin{equation}
		\begin{aligned}
			& - \mean{\hat{\boldsymbol{Q}}(z_{1}) \hat{\boldsymbol{R}} \chi(r)\phi_{M}  \frac{1}{N} \tr{\hat{\boldsymbol{Q}}(z_{1}) \boldsymbol{R}}} + \mean{\hat{\boldsymbol{Q}}(z_{1}) \boldsymbol{R} \chi(r)\phi_{M}} + \mean{\hat{\boldsymbol{Q}}(z_{1}) \hat{\boldsymbol{R}} \chi(r)\phi_{M}}\frac{1}{N}\tr{\boldsymbol{R} \bar{\boldsymbol{Q}}(z_{1})}\\
			&- \mathrm{j} r \sum_{q=1}^{L} \frac{w_{q}}{2 \pi \mathrm{j}}\frac{1}{N}\sqrt{N} \oint_{\mathcal{C}_{z_{2}}^{(m)}}  \mean{\hat{h}(z_{2})\hat{\boldsymbol{Q}}(z_{1}) \boldsymbol{R} \hat{\boldsymbol{Q}}(z_{2}) \boldsymbol{v}_{q}\boldsymbol{v}_{q}^{\operatorname{H}} \hat{\boldsymbol{Q}}(z_{2}) \hat{\boldsymbol{R}} \chi(r)\phi_{M}}\mathrm{d}z_{2}+ \mathcal{O}\left(N^{-3/2}\right)\\
			=& - \mean{\hat{\boldsymbol{Q}}(z_{1}) \hat{\boldsymbol{R}} \chi(r) \phi_{M} \alpha(z_{1})}+ \mean{\hat{\boldsymbol{Q}}(z_{1}) \boldsymbol{R} \chi(r)\phi_{M}} \\
			& - \mathrm{j} r \sum_{q=1}^{L} \frac{w_{q}}{2 \pi \mathrm{j}}\frac{1}{N}\sqrt{N} \oint_{\mathcal{C}_{z_{2}}^{(m)}} \mean{\hat{h}(z_{2})\hat{\boldsymbol{Q}}(z_{1}) \boldsymbol{R} \hat{\boldsymbol{Q}}(z_{2}) \boldsymbol{v}_{q}\boldsymbol{v}_{q}^{\operatorname{H}} \hat{\boldsymbol{Q}}(z_{2}) \hat{\boldsymbol{R}} \chi(r)\phi_{M}}\mathrm{d}z_{2}+ \mathcal{O}\left(N^{-3/2}\right),
		\end{aligned}
		\label{eq: rhs_mean_1_intermediate_1}
	\end{equation} 
	where we have introduced
	\begin{equation}
		\alpha(z_{1}) = \frac{1}{N} \tr{\boldsymbol{R}\hat{\boldsymbol{Q}}(z_{1})}\phi_{M} - \frac{1}{N}\tr{\boldsymbol{R} \bar{\boldsymbol{Q}}(z_{1})}.
		\label{eq: alpha}
	\end{equation}
	Equating \eqref{eq: lhs_mean_1_intermediate_1} and \eqref{eq: rhs_mean_1_intermediate_1} and multiplying by $\frac{z_{1}}{\omega(z_{1})}$ on both sides yields
	\begin{equation}
		\begin{aligned}
			\mean{\hat{\boldsymbol{Q}}(z_{1}) \hat{\boldsymbol{R}} \chi(r)\phi_{M}} =&- \frac{z_{1}}{\omega(z_{1})} \mean{\hat{\boldsymbol{Q}}(z_{1}) \hat{\boldsymbol{R}} \chi(r) \phi_{M} \alpha(z_{1})} + \frac{z_{1}}{\omega(z_{1})} \mean{\hat{\boldsymbol{Q}}(z_{1}) \boldsymbol{R} \chi(r)\phi_{M}} + \mathcal{O}\left(N^{-3/2}\right) \\
			&-\frac{z_{1}}{\omega(z_{1})}\mathrm{j} r \sum_{q=1}^{L} \frac{w_{q}}{2 \pi \mathrm{j}}\frac{1}{N}\sqrt{N} \oint_{\mathcal{C}_{z_{2}}^{(m)}} \mean{\hat{h}(z_{2})\hat{\boldsymbol{Q}}(z_{1}) \boldsymbol{R} \hat{\boldsymbol{Q}}(z_{2}) \boldsymbol{v}_{q}\boldsymbol{v}_{q}^{\operatorname{H}} \hat{\boldsymbol{Q}}(z_{2}) \hat{\boldsymbol{R}} \chi(r) \phi_{M}}\mathrm{d}z_{2}.
		\end{aligned}
		\label{eq: mean_1_intermediate_1}
	\end{equation}
	In the following we apply the resolvent identity in \eqref{eq: resolvent_identity} to the left hand side of \eqref{eq: mean_1_intermediate_1} and subtract $\frac{z_{1}}{\omega(z_{1})}\mean{\hat{\boldsymbol{Q}}(z_{1})\boldsymbol{R}\chi(r)\phi_{M}}$ on both sides
	\begin{equation}
		\begin{aligned}
			&z_{1} \mean{\hat{\boldsymbol{Q}}(z_{1}) \chi(r)\phi_{M}} + \boldsymbol{I}_{M} \mean{\chi(r)\phi_{M}} - \frac{z_{1}}{\omega(z_{1})}\mean{\hat{\boldsymbol{Q}}(z_{1}) \boldsymbol{R} \chi(r)\phi_{M}}\\
			=&- \frac{z_{1}}{\omega(z_{1})} \mean{\hat{\boldsymbol{Q}}(z_{1}) \hat{\boldsymbol{R}} \chi(r) \phi_{M}\alpha(z_{1})} - \mathrm{j} r \frac{z_{1}}{\omega(z_{1})}\sum_{q=1}^{L} \frac{w_{q}}{2 \pi \mathrm{j}}\frac{1}{N} \sqrt{N} \oint_{\mathcal{C}_{z_{2}}^{(m)}}  \mean{\hat{h}(z_{2})\hat{\boldsymbol{Q}}(z_{1}) \boldsymbol{R} \hat{\boldsymbol{Q}}(z_{2}) \boldsymbol{v}_{q} \boldsymbol{v}_{q}^{\operatorname{H}} \hat{\boldsymbol{Q}}(z_{2}) \hat{\boldsymbol{R}} \chi(r)\phi_{M}}\mathrm{d}z_{2}+\mathcal{O}\left(N^{-3/2}\right).
		\end{aligned}
		\label{eq: mean_1_intermediate_2}
	\end{equation}
	Next, we multiply both sides of \eqref{eq: mean_1_intermediate_2} by $-\bar{\boldsymbol{Q}}(z_{1})$
	\begin{equation}
		\begin{aligned}
			&-z_{1} \mean{\hat{\boldsymbol{Q}}(z_{1}) \bar{\boldsymbol{Q}}(z_{1}) \chi(r)\phi_{M}} - \mean{\bar{\boldsymbol{Q}}(z_{1})\chi(r)\phi_{M}} + \frac{z_{1}}{\omega(z_{1})} \mean{\hat{\boldsymbol{Q}}(z_{1}) \boldsymbol{R}\bar{\boldsymbol{Q}}(z_{1}) \chi(r)\phi_{M}} \\
			=& \frac{z_{1}}{\omega(z_{1})}\mean{\hat{\boldsymbol{Q}}(z_{1}) \hat{\boldsymbol{R}} \bar{\boldsymbol{Q}}(z_{1}) \chi(r) \phi_{M}\alpha(z_{1})} \\
			&+\mathrm{j}r \frac{z_{1}}{\omega(z_{1})} \sum_{q=1}^{L} \frac{w_{q}}{2 \pi \mathrm{j}}\frac{1}{N}\sqrt{N} \oint_{\mathcal{C}_{z_{2}}^{(m)}} \mean{\hat{h}(z_{2}) \hat{\boldsymbol{Q}}(z_{1}) \boldsymbol{R} \hat{\boldsymbol{Q}}(z_{2}) \boldsymbol{v}_{q}\boldsymbol{v}_{q}^{\operatorname{H}} \hat{\boldsymbol{Q}}(z_{2}) \hat{\boldsymbol{R}} \bar{\boldsymbol{Q}}(z_{1}) \chi(r)\phi_{M}}\mathrm{d}z_{2} + \mathcal{O}\left(N^{-3/2}\right).
		\end{aligned}
		\label{eq: mean_1_intermediate_3}
	\end{equation}
	Using $\boldsymbol{I}_{M} = - z_{1} \bar{\boldsymbol{Q}}(z_{1}) + \frac{z_{1}}{\omega(z_{1})} \boldsymbol{R}\bar{\boldsymbol{Q}}(z_{1})$ which immediately follows from the definition of $\bar{\boldsymbol{Q}}(z_{1})$ in \eqref{eq: resolvent_and_deterministic_equivalent} we can develop
	\begin{equation}
		\begin{aligned}
			-z_{1} \mean{\hat{\boldsymbol{Q}}(z_{1}) \bar{\boldsymbol{Q}}(z_{1}) \chi(r)\phi_{M}} + \frac{z_{1}}{\omega(z_{1})} \mean{\hat{\boldsymbol{Q}}(z_{1}) \boldsymbol{R} \bar{\boldsymbol{Q}}(z_{1}) \chi(r)\phi_{M}}=&\mean{\hat{\boldsymbol{Q}}(z_{1}) \left( -z_{1} \bar{\boldsymbol{Q}}(z_{1}) + \frac{z_{1}}{\omega(z_{1})}\boldsymbol{R} \bar{\boldsymbol{Q}}(z_{1}) \right) \chi(r)\phi_{M}}\\
			 =& \mean{\hat{\boldsymbol{Q}}(z_{1}) \chi(r)\phi_{M}}.
		\end{aligned}
		\label{eq: mean_1_intermediate_4}
	\end{equation}
	Substituting \eqref{eq: mean_1_intermediate_4} into \eqref{eq: mean_1_intermediate_3} and multiplying by $\sqrt{N}\bar{h}(z_{1})\boldsymbol{v}_{p}^{\operatorname{H}}$ from the left hand side and by $\boldsymbol{v}_{p}$ from the right hand side yields 
	\begin{equation}
		\begin{aligned}
			&\sqrt{N} \bar{h}(z_{1})\mean{\boldsymbol{v}_{p}^{\operatorname{H}}\left( \hat{\boldsymbol{Q}}(z_{1}) - \bar{\boldsymbol{Q}}(z_{1})\right)\boldsymbol{v}_{p} \chi(r)\phi_{M}}\\
			 =& \sqrt{N}\bar{h}(z_{1})\frac{z_{1}}{\omega(z_{1})} \mean{\boldsymbol{v}_{p}^{\operatorname{H}}\hat{\boldsymbol{Q}}(z_{1}) \hat{\boldsymbol{R}} \bar{\boldsymbol{Q}}(z_{1})\boldsymbol{v}_{p} \chi(r) \phi_{M} \alpha(z_{1})} \\
			&+ \mathrm{j} r \bar{h}(z_{1})\frac{z_{1}}{\omega(z_{1})} \sum_{q=1}^{L} \frac{w_{q}}{2 \pi \mathrm{j}}\oint_{\mathcal{C}_{z_{2}}^{(m)}} \mean{\hat{h}(z_{2})\boldsymbol{v}_{p}^{\operatorname{H}}\hat{\boldsymbol{Q}}(z_{1}) \boldsymbol{R}\hat{\boldsymbol{Q}}(z_{2}) \boldsymbol{v}_{q}\boldsymbol{v}_{q}^{\operatorname{H}}\hat{\boldsymbol{Q}}(z_{2}) \hat{\boldsymbol{R}} \bar{\boldsymbol{Q}}(z_{1})\boldsymbol{v}_{p}\chi(r)\phi_{M}}\mathrm{d}z_{2} + \mathcal{O}\left(N^{-3/2}\right),
		\end{aligned}
		\label{eq: mean_1_intermediate_5}
	\end{equation}
	Please note that \eqref{eq: mean_1_intermediate_5} is equivalent to the first quantity of interest in \eqref{eq: expectation_quantity_of_interest_a}. Hence, it remains to study the asymptotic stochastic behavior of \eqref{eq: mean_1_intermediate_5}. We will do so by analyzing the asymptotic stochastic behavior of the first term on the right hand side of \eqref{eq: mean_1_intermediate_5} first.
	\begin{lemma}
		\label{lemma: a1}
		Let $\zeta = \boldsymbol{v}_{p}^{\operatorname{H}}\hat{\boldsymbol{Q}}(z_{1})\hat{\boldsymbol{R}}\bar{\boldsymbol{Q}}(z_{1})\boldsymbol{v}_{p}\phi_{M}$ then
		\begin{equation*}
			\sqrt{N}\mean{\zeta\chi(r)\alpha(z_{1})} = \mathcal{O}\left(N^{-1/2}\right).
		\end{equation*}
	\end{lemma}
	\begin{proof}
		See Appendix \ref{proof:lemma: a1}.
	\end{proof}
	The asymptotic stochastic behavior of the second term on the right hand side of \eqref{eq: mean_1_intermediate_5} is analyzed next. 
	\begin{lemma}
		\label{lemma: a2}
		Let $\zeta_{1} = \boldsymbol{v}_{p}^{\operatorname{H}}\hat{\boldsymbol{Q}}(z_{1})\boldsymbol{R}\hat{\boldsymbol{Q}}(z_{2}) \boldsymbol{v}_{q}\phi_{M}$ and $\zeta_{2} = \boldsymbol{v}_{q}^{\operatorname{H}}\hat{\boldsymbol{Q}}(z_{2})\hat{\boldsymbol{R}}\bar{\boldsymbol{Q}}(z_{1}) \boldsymbol{v}_{p}\phi_{M}$ then the expectation $\mean{\hat{h}(z_{2})\zeta_{1}\zeta_{2} \chi(r)}$ can be approximated by
		\begin{equation*}
			\begin{aligned}
				\mean{\hat{h}(z_{2})\zeta_{1}\zeta_{2} \chi(r)} =& \bar{h}(z_{2})\mean{\zeta_{1}} \mean{\zeta_{2}} \mean{\chi(r)} + \mathcal{O}\left(N^{-1/2}\right).
			\end{aligned}	 
		\end{equation*}
	\end{lemma}
	\begin{proof}
		See Appendix \ref{proof:lemma: a2}.	
	\end{proof}
	Using Lemmas \ref{lemma: a1} and \ref{lemma: a2} we can approximate the quantity on the right hand side of \eqref{eq: mean_1_intermediate_5} by
	\begin{equation}
		\begin{aligned}
			\sqrt{N} \bar{h}(z_{1})\mean{\boldsymbol{v}_{p}^{\operatorname{H}}\left( \hat{\boldsymbol{Q}}(z_{1}) - \bar{\boldsymbol{Q}}(z_{1})\right)\boldsymbol{v}_{p} \chi(r)\phi_{M}} = \mathrm{j} r \bar{h}(z_{1})\frac{z_{1}}{\omega(z_{1})} \sum_{q=1}^{L} \frac{w_{q}}{2 \pi \mathrm{j}}\oint_{\mathcal{C}_{z_{2}}^{(m)}} \bar{h}(z_{2}) \mean{\zeta_{1}} \mean{\zeta_{2}} \mean{\chi(r)} \mathrm{d}z_{2}+ \mathcal{O}\left(N^{-1/2}\right),
		\end{aligned}
		\label{eq: mean_1_intermediate_6}
	\end{equation}
	With the help of Lemma \ref{aux:lemma4} and Lemma \ref{aux:lemma3} we can compute the expectations $\mean{\zeta_{1}}$ and $\mean{\zeta_{2}}$
	\begin{align}
		\mean{\zeta_{1}} =& \boldsymbol{v}_{p}^{\operatorname{H}} \bar{\boldsymbol{Q}}(z_{1}) \boldsymbol{R}\bar{\boldsymbol{Q}}(z_{2}) \boldsymbol{v}_{q} \left(\frac{1}{1-\gamma(z_{1},z_{2})}\right) + \mathcal{O}\left(N^{-1}\right) \label{eq: mean_1_intermediate_7a} \\
		\mean{\zeta_{2}} =& \frac{z_{2}}{\omega(z_{2})} \boldsymbol{v}_{q}^{\operatorname{H}} \bar{\boldsymbol{Q}}(z_{2})\boldsymbol{R}\bar{\boldsymbol{Q}}(z_{1}) \boldsymbol{v}_{p} + \mathcal{O}\left(N^{-1}\right) \label{eq: mean_1_intermediate_7b},
	\end{align}
	where
	\begin{equation}
		\gamma(z_{1},z_{2}) = \frac{z_{1}}{\omega(z_{1})} \frac{z_{2}}{\omega(z_{2})} \frac{1}{N}\tr{\boldsymbol{R}\bar{\boldsymbol{Q}}(z_{1})\boldsymbol{R}\bar{\boldsymbol{Q}}(z_{2})}.
	\end{equation}
	Substituting the expressions in \eqref{eq: mean_1_intermediate_7a} and \eqref{eq: mean_1_intermediate_7b} into \eqref{eq: mean_1_intermediate_6} yields the expression in \eqref{eq: mean_1_intermediate_8}.
%
	
\section{Proof of Lemma \ref{lemma: a1}}
	\label{proof:lemma: a1}
		Using the Cauchy-Schwarz inequality we can upper-bound 
		\begin{equation*}
			\begin{aligned}
				\abs{\mean{\zeta \chi(r) \alpha(z_{1})}}^{2} &\leq \mean{\abs{\zeta \chi(r)}^{2}} \mean{\abs{\alpha(z_{1})}^{2}} = \mean{\abs{\zeta}^{2} \abs{\chi(r)}^{2}} \mean{\abs{\alpha(z_{1})}^{2}} =  \mean{\abs{\zeta}^{2}} \mean{\abs{\alpha(z_{1})}^{2}},
			\end{aligned}
		\end{equation*}
		where we have used $\abs{\chi(r)}^{2}=1$. According to Lemma \ref{aux:lemma1}, the expectation and the variance of $\alpha(z_{1})$ in \eqref{eq: alpha} decay with $\mean{\alpha(z_{1})} = \mathcal{O}(N^{-1})$ and $\var{\alpha(z_{1})} = \mathcal{O}(N^{-2})$, respectively. Furthermore, it can be seen in Lemma \ref{aux:lemma3} that the expectation $\mean{\zeta}$ converges to something deterministic and the variance decays with $\var{\zeta} = \mathcal{O}(N^{-1})$. Hence, $\mean{\abs{\alpha(z_{1})}^{2}}= \var{\alpha(z_{1})} + \mean{\alpha(z_{1})}^{2} = \mathcal{O}(N^{-2})$ and $\mean{\abs{\zeta}^{2}}= \var{\zeta} + \mean{\zeta}^{2} = \mathcal{O}(1)$ and the quantity of interest can be upper-bounded by
		\begin{equation*}
			\sqrt{N}\mean{\zeta \chi(r) \alpha(z_{1})} \leq \sqrt{N} \sqrt{\mean{\abs{\zeta}^{2}} \mean{\abs{\alpha(z_{1})}^{2}}}  = \mathcal{O}\left( \frac{\sqrt{N}}{N}\right) = \mathcal{O}\left(N^{-1/2}\right).
		\end{equation*}
	
\section{Proof of Lemma \ref{lemma: a2}}
	\label{proof:lemma: a2}
	Let us use the following notation $(\zeta_{1}\zeta_{2})^{(\circ)} = (\zeta_{1}\zeta_{2}) - \mean{(\zeta_{1}\zeta_{2})}$ such that $\mean{\left( \zeta_{1} \zeta_{2} \right)^{(\circ)} \chi(r)} =\mean{\left(\left(\zeta_{1}\zeta_{2}\right) - \mean{\zeta_{1}\zeta_{2}}\right)\chi(r)} =\mean{\zeta_{1} \zeta_{2}\chi(r)} - \mean{\zeta_{1}\zeta_{2}}\mean{\chi(r)}$. Hence, the quantity of interest can equivalently be expressed as 
	\begin{equation}
		\mean{\zeta_{1}\zeta_{2}\chi(r)} = \mean{\left( \zeta_{1}\zeta_{2}\right)^{(\circ)} \chi(r)} + \mean{\zeta_{1}\zeta_{2}}\mean{\chi(r)}.
		\label{eq:proof:lemma:a2: intermediate_0}
	\end{equation} 
	Next, we apply the Cauchy-Schwarz inequality to upper bound the expectation $ \mean{\left( \zeta_{1}\zeta_{2}\right)^{(\circ)} \chi(r)}$ according to
	\begin{equation}
		\left | \mean{\left( \zeta_{1} \zeta_{2} \right)^{(\circ)} \chi(r)} \right |^2 \leq \mean{\left | \left( \zeta_{1} \zeta_{2} \right)^{(\circ)} \right |^2 }\mean{\left | \chi(r) \right |^2} = \var{\zeta_{1}\zeta_{2}},
		\label{eq:proof:lemma:a2: intermediate_1}
	\end{equation}
	where we have used $ | \chi(r) | = 1$ and $\mean{ | ( \zeta_{1} \zeta_{2})^{(\circ)} |^2} = \mean{|\zeta_{1} \zeta_{2} - \mean{\zeta_{1} \zeta_{2}}|^{2}}= \var{\zeta_{1}\zeta_{2}}$. Using the Nash-Poincar\'e inequality allows to upper bound the variance 
	\begin{equation}
		\var{\zeta_{1}\zeta_{2}} \leq \sum_{i=1}^{M}\sum_{j=1}^{N} \left( \mean{\left| \frac{\partial \left(\zeta_{1}\zeta_{2}\right)}{\partial X_{ij}}\right|^2} + \mean{\left| \frac{\partial \left(\zeta_{1}\zeta_{2}\right)}{\partial X_{ij}^{*}}\right|^2} \right)= \mathcal{O}\left(N^{-1}\right),
		\label{eq:proof:lemma:a2: intermediate_2}
	\end{equation}
	where both derivatives can be upper bounded itself by applying Jensen's inequality
	\begin{align}
		\sum_{i=1}^{M}\sum_{j=1}^{N}  \mean{\abs{ \frac{\partial \left(\zeta_{1}\zeta_{2}\right)}{\partial X_{ij}}}^2} =& \sum_{i=1}^{M}\sum_{j=1}^{N} \mean{\abs{ \left( \frac{\partial \zeta_{1}}{\partial X_{ij}}\right) \zeta_{2} + \zeta_{1} \left( \frac{\partial \zeta_{2}}{\partial X_{ij}}\right)}^2} \leq  2 \sum_{i=1}^{M}\sum_{j=1}^{N} \mean{\abs{\left(\frac{\partial \zeta_{1}}{\partial X_{ij}}\right)   \zeta_{2}}^2} + 2  \sum_{i=1}^{M}\sum_{j=1}^{N}\mean{\abs{ \zeta_{1} \left( \frac{\partial \zeta_{2}}{\partial X_{ij}}\right)}^2}, \label{eq:proof:lemma:a2: intermediate_3a}\\
		\sum_{i=1}^{M}\sum_{j=1}^{N}  \mean{\abs{ \frac{\partial \left(\zeta_{1}\zeta_{2}\right)}{\partial X_{ij}^{*}}}^2}=& \sum_{i=1}^{M}\sum_{j=1}^{N} \mean{\abs{\left( \frac{\partial \zeta_{1}}{\partial X_{ij}^{*}}\right) \zeta_{2} + \zeta_{1} \left( \frac{\partial \zeta_{2}}{\partial X_{ij}^{*}}\right)}^2} \leq  2 \sum_{i=1}^{M}\sum_{j=1}^{N} \mean{\abs{ \left( \frac{\partial \zeta_{1}}{\partial X_{ij}^{*}}\right)\zeta_{2}}^2} + 2  \sum_{i=1}^{M}\sum_{j=1}^{N}\mean{\abs{\zeta_{1} \left( \frac{\partial \zeta_{2}}{\partial X_{ij}^{*}} \right)}^2 }. \label{eq:proof:lemma:a2: intermediate_3b}
	\end{align}
	The first order derivatives of $\zeta_{1}$ with respect to $X_{ij}$ and $X_{ij}^{*}$ are given in \eqref{eq:proof:aux:lemma4:variance1: intermediate_2a} and \eqref{eq:proof:aux:lemma4:variance1: intermediate_2b} respectively, where $\boldsymbol{s}_{1} = \boldsymbol{v}_{p}$ and $\boldsymbol{s}_{2} = \boldsymbol{v}_{q}$. Furthermore, the first order derivatives of $\zeta_{2}$ with respect to $X_{ij}$ and $X_{ij}^{*}$ are given in  \eqref{eq:proof:aux:lemma3:variance1: intermediate_2a} and \eqref{eq:proof:aux:lemma3:variance1: intermediate_2b} where $\boldsymbol{s}_{1} = \boldsymbol{v}_{q}$ and $\boldsymbol{s}_{2}= \bar{\boldsymbol{Q}}(z_{1})\boldsymbol{v}_{p}$. Substituting the derivatives into \eqref{eq:proof:lemma:a2: intermediate_3a} and \eqref{eq:proof:lemma:a2: intermediate_3b} and applying Jensen's inequality a second time it can be seen that 
	\begin{equation*}
		\sum_{i=1}^{M}\sum_{j=1}^{N}  \mean{\abs{ \frac{\partial \left(\zeta_{1}\zeta_{2}\right)}{\partial X_{ij}}}^2} \leq \mathcal{O}\left(N^{-1}\right), \qquad \sum_{i=1}^{M}\sum_{j=1}^{N}  \mean{\abs{ \frac{\partial \left(\zeta_{1}\zeta_{2}\right)}{\partial X_{ij}^{*}}}^2}\leq \mathcal{O}\left(N^{-1}\right),
	\end{equation*}
	since
	\begin{align*}
			\sum_{i=1}^{M}\sum_{j=1}^{N} \mean{\abs{\left(\frac{\partial \zeta_{1}}{\partial X_{ij}}\right)   \zeta_{2}}^2} \leq \mathcal{O}\left(N^{-1}\right), &\qquad \sum_{i=1}^{M}\sum_{j=1}^{N}\mean{\abs{ \zeta_{1} \left( \frac{\partial \zeta_{2}}{\partial X_{ij}}\right)}^2} \leq \mathcal{O}\left(N^{-1}\right), \\
			\sum_{i=1}^{M}\sum_{j=1}^{N} \mean{\abs{ \left( \frac{\partial \zeta_{1}}{\partial X_{ij}^{*}}\right)\zeta_{2}}^2} \leq \mathcal{O}\left(N^{-1}\right), &\qquad \sum_{i=1}^{M}\sum_{j=1}^{N}\mean{\abs{\zeta_{1} \left( \frac{\partial \zeta_{2}}{\partial X_{ij}^{*}} \right)}^2 } \leq \mathcal{O}\left(N^{-1}\right).
	\end{align*}
	Hence, the variance in \eqref{eq:proof:lemma:a2: intermediate_2} decays with $\var{\zeta_{1}\zeta_{2}} = \mathcal{O}(N^{-1})$ and as a direct consequence the expectation $\mean{(\zeta_{1}\zeta_{2})^{(\circ)}\chi(r)}$ decays with $\mean{(\zeta_{1}\zeta_{2})^{(\circ)}\chi(r)} \leq \mathcal{O}(N^{-1/2})$ which follows from \eqref{eq:proof:lemma:a2: intermediate_1}. With this, we can approximate the quantity of interest in \eqref{eq:proof:lemma:a2: intermediate_0} by
	\begin{equation}
		\mean{\zeta_{1}\zeta_{2}\chi(r)} = \mean{\zeta_{1}\zeta_{2}}\mean{\chi(r)} + \mathcal{O}\left(N^{-1/2}\right),
		\label{eq:proof:lemma:a2: intermediate_4}
	\end{equation}
	and it remains to study the asymptotic stochastic behavior of $\mean{\zeta_{1}\zeta_{2}}$. Introducing the covariance 
	\begin{equation*}
		\cov{\zeta_{1}}{\zeta_{2}} = \mean{\zeta_{1}^{(\circ)}\zeta_{2}^{(\circ)}} =\mean{\left(\zeta_{1}-\mean{\zeta_{1}}\right)\left(\zeta_{2}-\mean{\zeta_{2}}\right)} = \mean{\zeta_{1}\zeta_{2}} - \mean{\zeta_{1}}\mean{\zeta_{2}},
	\end{equation*}
	allows us to equivalently express $\mean{\zeta_{1}\zeta_{2}}$ as
	\begin{equation}
		\mean{\zeta_{1}\zeta_{2}} = \cov{\zeta_{1}}{\zeta_{2}}+\mean{\zeta_{1}}\mean{\zeta_{2}}.
		\label{eq:proof:lemma:a2: intermediate_5}
	\end{equation}
	The covariance $\cov{\zeta_{1}}{\zeta_{2}}$ can be upper-bounded using the Cauchy-Schwarz inequality
	\begin{equation}
		\abs{\cov{\zeta_{1}}{\zeta_{2}}}^2 = \abs{\mean{\left(\zeta_{1} - \mean{\zeta_{1}}\right)\left(\zeta_{2}-\mean{\zeta_{2}}\right)}}^2 \leq  \mean{\abs{\zeta_{1} - \mean{\zeta_{1}}}^2}\mean{\abs{\zeta_{2}-\mean{\zeta_{2}}}^2} = \var{\zeta_{1}} \var{\zeta_{2}}= \mathcal{O}\left(N^{-2}\right),
		\label{eq:proof:lemma:a2: intermediate_6}
	\end{equation}
	such that $\cov{\zeta_{1}}{\zeta_{2}} \leq \sqrt{\var{\zeta_{1}}\var{\zeta_{2}}}= \mathcal{O}(N^{-1})$ where we have used $\var{\zeta_{1}} = \mathcal{O}(N^{-1})$ and $\var{\zeta_{2}} = \mathcal{O}(N^{-1})$ which follows from Lemmas \ref{aux:lemma3} and \ref{aux:lemma4}. Substituting the square root of \eqref{eq:proof:lemma:a2: intermediate_6} into \eqref{eq:proof:lemma:a2: intermediate_5} we have successfully decorrelated $\mean{\zeta_{1}\zeta_{2}} = \mean{\zeta_{1}}\mean{\zeta_{2}} + \mathcal{O}(N^{-1})$. Inserting $\mean{\zeta_{1}\zeta_{2}} = \mean{\zeta_{1}}\mean{\zeta_{2}} + \mathcal{O}(N^{-1})$ into the expression in \eqref{eq:proof:lemma:a2: intermediate_4} we obtain
	\begin{equation}
		\begin{aligned}
			\mean{\zeta_{1}\zeta_{2}\chi(r)} =& \mean{\zeta_{1}\zeta_{2}}\mean{\chi(r)} + \mathcal{O}\left(N^{-1/2}\right) = \left( \cov{\zeta_{1}}{\zeta_{2}}+\mean{\zeta_{1}}\mean{\zeta_{2}}\right) \mean{\chi(r)} +\mathcal{O}\left(N^{-1/2}\right)\\
			=& \left( \mathcal{O}\left(N^{-1}\right)+\mean{\zeta_{1}}\mean{\zeta_{2}}\right) \mean{\chi(r)} +\mathcal{O}\left(N^{-1/2}\right)= \mean{\zeta_{1}}\mean{\zeta_{2}}\mean{\chi(r)} +\mathcal{O}\left(N^{-1/2}\right).
		\end{aligned}
		\label{eq:proof:lemma:a2: intermediate_7}
	\end{equation}
	By adding and subtracting $\bar{h}(z_{2})$ inside the expectation of the quantity of interest $\mean{\hat{h}(z_{2})\zeta_{1}\zeta_{2}\chi(r)}$ we can write
	\begin{equation}
		\begin{aligned}
			\mean{\hat{h}(z_{2})\zeta_{1}\zeta_{2}\chi(r)} =& \mean{\left(\hat{h}(z_{2}) -\bar{h}(z_{2}) + \bar{h}(z_{2}) \right) \zeta_{1}\zeta_{2}\chi(r)} = \mean{\left(\hat{h}(z_{2}) - \bar{h}(z_{2})\right) \zeta_{1}\zeta_{2} \chi(r)}+ \mean{\bar{h}(z_{2})\zeta_{1}\zeta_{2}\chi(r)}\\
			=& \mean{\left(\hat{h}(z_{2}) - \bar{h}(z_{2})\right) \zeta_{1}\zeta_{2} \chi(r)}+ \bar{h}(z_{2}) \mean{\zeta_{1}\zeta_{2}\chi(r)}.
		\end{aligned}
		\label{eq:proof:lemma:a2: intermediate_01}
	\end{equation}
	Applying the Cauchy-Schwarz inequality to the first term on the right hand side of \eqref{eq:proof:lemma:a2: intermediate_01} we can upper-bound
	\begin{equation}
		\abs{\mean{\left(\hat{h}(z_{2}) - \bar{h}(z_{2})\right) \zeta_{1}\zeta_{2} \chi(r)}}^{2} \leq \mean{\abs{\hat{h}(z_{2})-\bar{h}(z_{2})}^{2}} \mean{\abs{\zeta_{1}\zeta_{2}\chi(r)}^{2}}.
		\label{eq:proof:lemma:a2: intermediate_02}
	\end{equation}
	From Lemma \ref{aux:lemma5} it follows that
	\begin{equation*}
		 \mean{\abs{\hat{h}(z_{2})-\bar{h}(z_{2})}^{2}} = \var{\hat{h}(z_{2})-\bar{h}(z_{2})} + \mean{\hat{h}(z_{2})-\bar{h}(z_{2})}^{2} = \mathcal{O}\left(N^{-2}\right),
	\end{equation*}
	and from our previous analysis we know that
	\begin{equation*}
		\mean{\abs{\zeta_{1}\zeta_{2}\chi(r)}^{2}} = \mean{\abs{\zeta_{1}\zeta_{2}}^{2}\abs{\chi(r)}^{2}} = \mean{\abs{\zeta_{1}\zeta_{2}}^{2}} = \var{\zeta_{1}\zeta_{2}} + \mean{\zeta_{1}\zeta_{2}}^{2} = \var{\zeta_{1}\zeta_{2}} + \left( \cov{\zeta_{1}}{\zeta_{2}} + \mean{\zeta_{1}}\mean{\zeta_{2}}\right)^{2} = \mathcal{O}(1),
	\end{equation*}
	where we have used \eqref{eq:proof:lemma:a2: intermediate_2} and \eqref{eq:proof:lemma:a2: intermediate_5}. Correspondingly, the first term on the right hand side of \eqref{eq:proof:lemma:a2: intermediate_01} decays with
	\begin{equation*}
		\mean{\left(\hat{h}(z_{2}) - \bar{h}(z_{2})\right) \zeta_{1}\zeta_{2} \chi(r)} \leq \sqrt{\mean{\abs{\hat{h}(z_{2})-\bar{h}(z_{2})}^{2}} \mean{\abs{\zeta_{1}\zeta_{2}\chi(r)}^{2}}} =\mathcal{O}\left(N^{-1}\right) 
	\end{equation*}
	and therefore 
	\begin{equation}
		\mean{\hat{h}(z_{2})\zeta_{1}\zeta_{2}\chi(r)} =  \bar{h}(z_{2}) \mean{\zeta_{1}\zeta_{2}\chi(r)} + \mathcal{O}\left(N^{-1}\right).
		\label{eq:proof:lemma:a2: intermediate_8}
	\end{equation}
	Finally, substituting \eqref{eq:proof:lemma:a2: intermediate_7} into \eqref{eq:proof:lemma:a2: intermediate_8} we obtain the result of Lemma \ref{lemma: a2}
	\begin{equation*}
		\mean{\hat{h}(z_{2})\zeta_{1}\zeta_{2}\chi(r)} = \bar{h}(z_{2})\mean{\zeta_{1}}\mean{\zeta_{2}}\mean{\chi(r)} +\mathcal{O}\left(N^{-1/2}\right).
	\end{equation*}

\section{Auxiliary Lemmas}
	\label{sec: aux_lemmas}
	In this appendix, we provide some bounds on expectations and variances of different random functions of complex variables. The notation $\mathcal{O}(N^{-l})$ should be understood as a deterministic term whose magnitude is upper bounded by a quantity of the form $\varepsilon(z_{1},z_{2}) N^{-l}$, where $\varepsilon(z_{1},z_{2})$ is a bivariate real-valued positive function independent of $N$ such that $\sup_{(z_{1},z_{2}) \in \mathcal{C}_{z_{1}}^{(m)} \times \mathcal{C}_{z_{2}}^{(m)}} \varepsilon(z_{1},z_{2}) < \infty$. Further insights are provided in \cite[Remark 1, supplementary material]{a63}.

	\begin{lemma}
		\label{aux:lemma1}
		Let $\boldsymbol{B}$ denote an $M \times M$ deterministic matrix with bounded spectral norm. Then, as $M,N\rightarrow \infty$ at the same rate
		\begin{align}
			\frac{1}{N} \mean{\tr{\boldsymbol{B} \hat{\boldsymbol{Q}}(z_{1})\phi_{M}}} =& \frac{1}{N} \tr{\boldsymbol{B}\bar{\boldsymbol{Q}}(z_{1})} + \mathcal{O}\left(N^{-1}\right)\label{eq:aux:lemma1: expectation1} \\ 
			\frac{1}{N} \mean{\tr{\boldsymbol{B}\hat{\boldsymbol{Q}}(z_{1}) \boldsymbol{R}^{1/2}\frac{\boldsymbol{X}\boldsymbol{X}^{\operatorname{H}}}{N}\phi_{M}}} =& \frac{z_{1}}{\omega(z_{1})} \frac{1}{N} \tr{\boldsymbol{B}\bar{\boldsymbol{Q}}(z_{1}) \boldsymbol{R}^{1/2}} + \mathcal{O}\left(N^{-1}\right) \label{eq:aux:lemma1: expectation2}
		\end{align}
		and also
		\begin{align}
			\var{ \frac{1}{N} \tr{\boldsymbol{B}\hat{\boldsymbol{Q}}(z_{1})\phi_{M}}} =& \mathcal{O}\left(N^{-2}\right) \label{eq:aux:lemma1: variance1} \\
			\var{\frac{1}{N} \tr{\boldsymbol{B}\hat{\boldsymbol{Q}}(z_{1}) \boldsymbol{R}^{1/2}\frac{\boldsymbol{X}\boldsymbol{X}^{\operatorname{H}}}{N}\phi_{M}}} =& \mathcal{O}\left(N^{-2}\right).\label{eq:aux:lemma1: variance2}
		\end{align}
	\end{lemma}
	\begin{proof}
		See Appendix \ref{proof:aux:lemma1} and \cite[Lemma 8, supplementary material]{a63}.
	\end{proof}
	
	\begin{lemma}
		\label{aux:lemma2}
		Let $\boldsymbol{B}$ denote an $M \times M$ deterministic matrix with bounded spectral norm. Then, as $M,N\rightarrow \infty$ at the same rate
		\begin{align}
			\frac{1}{N} \mean{\tr{\boldsymbol{B} \hat{\boldsymbol{Q}}(z_{1})\boldsymbol{R} \hat{\boldsymbol{Q}}(z_{2})\phi_{M}}} =& \frac{1}{1-\gamma(z_{1},z_{2})} \frac{1}{N} \tr{\boldsymbol{B}\bar{\boldsymbol{Q}}(z_{1})\boldsymbol{R}\bar{\boldsymbol{Q}}(z_{2})} + \mathcal{O}\left(N^{-1}\right) \label{eq:aux:lemma2: expectation1} \\ 
			\frac{1}{N} \mean{\tr{\boldsymbol{B} \hat{\boldsymbol{Q}}(z_{1})\boldsymbol{R} \hat{\boldsymbol{Q}}(z_{2}) \boldsymbol{R}^{1/2}\frac{\boldsymbol{X}\boldsymbol{X}^{\operatorname{H}}}{N}\phi_{M}}} =& \frac{z_{2}}{\omega(z_{2})} \frac{1}{1-\gamma(z_{1},z_{2})} \frac{1}{N} \tr{\boldsymbol{B}\bar{\boldsymbol{Q}}(z_{1}) \boldsymbol{R} \bar{\boldsymbol{Q}}(z_{2}) \boldsymbol{R}^{1/2}} \label{eq:aux:lemma2: expectation2}  \\
			&- \frac{\gamma(z_{1},z_{2})}{1-\gamma(z_{1},z_{2})} \frac{1}{N} \tr{\boldsymbol{B} \bar{\boldsymbol{Q}}(z_{1}) \boldsymbol{R}^{1/2}}+ \mathcal{O}\left(N^{-1}\right)\nonumber
		\end{align}
		where
		\begin{equation}
			\gamma\left(z_{1}, z_{2}\right) = \frac{z_{1}}{\omega(z_{1})}\frac{z_{2}}{\omega(z_{2})}\frac{1}{N} \tr{\boldsymbol{R}\bar{\boldsymbol{Q}}(z_{1}) \boldsymbol{R} \bar{\boldsymbol{Q}}(z_{2})}
			\label{eq:aux:lemma2: gamma_definition_1}
		\end{equation}
		and also
		\begin{align}
			\var{ \frac{1}{N}\tr{\boldsymbol{B} \hat{\boldsymbol{Q}}(z_{1})\boldsymbol{R} \hat{\boldsymbol{Q}}(z_{2})\phi_{M}}} =& \mathcal{O}\left(N^{-2}\right) \label{eq:aux:lemma2: variance1}\\
			\var{\frac{1}{N} \tr{\boldsymbol{B} \hat{\boldsymbol{Q}}(z_{1})\boldsymbol{R} \hat{\boldsymbol{Q}}(z_{2})\boldsymbol{R}^{1/2}\frac{\boldsymbol{X}\boldsymbol{X}^{\operatorname{H}}}{N}\phi_{M}}} =& \mathcal{O}\left(N^{-2}\right). \label{eq:aux:lemma2: variance2}
		\end{align}
	\end{lemma}
	\begin{proof}
		See Appendix \ref{proof:aux:lemma2} and \cite[Lemma 9, supplementary material]{a63}.
	\end{proof}
	
	\begin{lemma}
		\label{aux:lemma1_1}
		As $M,N\rightarrow \infty$ at the same rate
		\begin{equation*}
			\frac{1}{N}\mean{\tr{\hat{\boldsymbol{Q}}(z_{1})\hat{\boldsymbol{R}}\hat{\boldsymbol{Q}}(z_{1})\phi_{M}}} = \frac{z_{1}^{2}}{\omega(z_{1})^{2}} \frac{1}{N} \tr{\boldsymbol{R}\bar{\boldsymbol{Q}}(z_{1})^{2}} \frac{\partial \omega(z_{1})}{\partial z_{1}} + \mathcal{O}\left(N^{-1}\right)
		\end{equation*}
		and the variance decays with
		\begin{equation*}
			\var{\frac{1}{N}\tr{\hat{\boldsymbol{Q}}(z_{1})\hat{\boldsymbol{R}}\hat{\boldsymbol{Q}}(z_{1})\phi_{M}}} = \mathcal{O}\left(N^{-2}\right).
		\end{equation*}
	\end{lemma}
	\begin{proof}
		See Appendix \ref{proof:aux:lemma1_1}.
	\end{proof}
	
	\begin{lemma}
		\label{aux:lemma3}
		Let $\boldsymbol{s}_{1}$, $\boldsymbol{s}_{2}$ denote two generic $M \times 1$ deterministic column vectors. Then, as $M,N\rightarrow \infty$ at the same rate
		\begin{equation}
			\mean{\boldsymbol{s}_{1}^{\operatorname{H}} \hat{\boldsymbol{Q}}(z_{2})\hat{\boldsymbol{R}} \boldsymbol{s}_{2}\phi_{M}} = \frac{z_{2}}{\omega(z_{2})} \boldsymbol{s}_{1}^{\operatorname{H}}\bar{\boldsymbol{Q}}(z_{2}) \boldsymbol{R} \boldsymbol{s}_{2} + \mathcal{O}\left(N^{-1}\right)
			\label{eq:aux:lemma3: expectation1}
		\end{equation}
		and also
		\begin{equation}
			\var{\boldsymbol{s}_{1}^{\operatorname{H}} \hat{\boldsymbol{Q}}(z_{2})\hat{\boldsymbol{R}} \boldsymbol{s}_{2}\phi_{M}} = \mathcal{O}\left(N^{-1}\right)	
			\label{eq:aux:lemma3: variance1}
		\end{equation}
	\end{lemma}
	\begin{proof}
		See Appendix \ref{proof:aux:lemma3}.
	\end{proof}

	\begin{lemma}
		\label{aux:lemma4}
		Let $\boldsymbol{s}_{1}$, $\boldsymbol{s}_{2}$ denote two generic $M \times 1$ deterministic column vectors. Then, as $M,N\rightarrow \infty$ at the same rate
		\begin{equation}
			\mean{\boldsymbol{s}_{1}^{\operatorname{H}} \hat{\boldsymbol{Q}}(z_{1}) \boldsymbol{R} \hat{\boldsymbol{Q}}(z_{2}) \boldsymbol{s}_{2}\phi_{M}} = \boldsymbol{s}_{1}^{\operatorname{H}} \bar{\boldsymbol{Q}}(z_{1}) \boldsymbol{R}\bar{\boldsymbol{Q}}(z_{2}) \boldsymbol{s}_{2} \left( \frac{1}{1- \gamma\left(z_{1}, z_{2}\right)}\right) + \mathcal{O}\left(N^{-1}\right)
			\label{eq:aux:lemma4: expectation1}
		\end{equation}
		and also
		\begin{equation}
			\var{\boldsymbol{s}_{1}^{\operatorname{H}} \hat{\boldsymbol{Q}}(z_{1}) \boldsymbol{R} \hat{\boldsymbol{Q}}(z_{2}) \boldsymbol{s}_{2}\phi_{M}} = \mathcal{O}\left(N^{-1}\right)
			\label{eq:aux:lemma4: variance1}	
		\end{equation}
		where
		\begin{equation}
			\gamma\left(z_{1}, z_{2}\right) = \frac{z_{1}}{\omega(z_{1})}\frac{z_{2}}{\omega(z_{2})}\frac{1}{N} \tr{\boldsymbol{R}\bar{\boldsymbol{Q}}(z_{1}) \boldsymbol{R} \bar{\boldsymbol{Q}}(z_{2})}
			\label{eq:aux:lemma4: gamma_definition_2}
		\end{equation}
	\end{lemma}
	\begin{proof}
		See Appendix \ref{proof:aux:lemma4}.
	\end{proof}


\section{Proof of Lemma \ref{aux:lemma1}}
	\label{proof:aux:lemma1}	
	
	\subsubsection{Variances}
		We start by upper-bounding the variance in \eqref{eq:aux:lemma1: variance1} by applying the Nash-Poincar\'e inequality to
		\begin{equation}
			\begin{aligned}
				\var{\tr{\boldsymbol{B}\hat{\boldsymbol{Q}}(z_{1})\phi_{M}}} \leq \sum_{i=1}^{M}\sum_{j=1}^{N} \left( \mean{\abs{\frac{\partial}{\partial X_{ij}}\tr{\boldsymbol{B}\hat{\boldsymbol{Q}}(z_{1})\phi_{M}}}^{2}}+\mean{\abs{\frac{\partial}{\partial X_{ij}^{*}}\tr{\boldsymbol{B}\hat{\boldsymbol{Q}}(z_{1})\phi_{M}}}^{2}}\right),
			\end{aligned}
			\label{eq:proof:aux:lemma1:variance1: intermediate_1}
		\end{equation}
		where
		\begin{align}
			\frac{\partial}{\partial X_{ij}}\tr{\boldsymbol{B}\hat{\boldsymbol{Q}}(z_{1})\phi_{M}} =& \sum_{m=1}^{M}\boldsymbol{e}_{m}^{\operatorname{T}} \boldsymbol{B} \left( \frac{\partial}{\partial X_{ij}}\hat{\boldsymbol{Q}}(z_{1})\right)\phi_{M} \boldsymbol{e}_{m} = - \sum_{m=1}^{M} \boldsymbol{e}_{m}^{\operatorname{T}} \boldsymbol{B} \hat{\boldsymbol{Q}}(z_{1}) \boldsymbol{R}^{1/2} \frac{\boldsymbol{e}_{i}\boldsymbol{x}_{j}^{\operatorname{H}}}{N}\boldsymbol{R}^{\operatorname{H}/2} \hat{\boldsymbol{Q}}(z_{1}) \phi_{M}\boldsymbol{e}_{m} + \mathcal{O}\left(N^{-\mathbb{N}}\right), \label{eq:proof:aux:lemma1:variance1: intermediate_2a} \\
			\frac{\partial}{\partial X_{ij}^{*}}\tr{\boldsymbol{B}\hat{\boldsymbol{Q}}(z_{1})\phi_{M}} =&  \sum_{m=1}^{M}\boldsymbol{e}_{m}^{\operatorname{T}} \boldsymbol{B} \left( \frac{\partial}{\partial X_{ij}^{*}} \hat{\boldsymbol{Q}}(z_{1})\right) \phi_{M} \boldsymbol{e}_{m} = -\sum_{m=1}^{M} \boldsymbol{e}_{m}^{\operatorname{T}} \boldsymbol{B} \hat{\boldsymbol{Q}}(z_{1}) \boldsymbol{R}^{1/2} \frac{\boldsymbol{x}_{j}\boldsymbol{e}_{i}^{\operatorname{T}}}{N}\boldsymbol{R}^{\operatorname{H}/2} \hat{\boldsymbol{Q}}(z_{1}) \phi_{M}\boldsymbol{e}_{m}+ \mathcal{O}\left(N^{-\mathbb{N}}\right). \label{eq:proof:aux:lemma1:variance1: intermediate_2b}
		\end{align}
		It follows from \eqref{eq:proof:aux:lemma1:variance1: intermediate_2a} and \eqref{eq:proof:aux:lemma1:variance1: intermediate_2b} that
		\begin{equation*}
			\sum_{i=1}^{M}\sum_{j=1}^{N} \mean{\abs{\frac{\partial}{\partial X_{ij}}\tr{\boldsymbol{B}\hat{\boldsymbol{Q}}(z_{1})\phi_{M}}}^{2}} = \mathcal{O}\left(\frac{M}{N}\right), \qquad
			\sum_{i=1}^{M}\sum_{j=1}^{N} \mean{\abs{\frac{\partial}{\partial X_{ij}^{*}}\tr{\boldsymbol{B}\hat{\boldsymbol{Q}}(z_{1})\phi_{M}}}^{2}} = \mathcal{O}\left(\frac{M}{N}\right),
		\end{equation*}
		and consequently the variance in \eqref{eq:proof:aux:lemma1:variance1: intermediate_1} is upper-bounded by $\var{\tr{\boldsymbol{B}\hat{\boldsymbol{Q}}(z_{1})}} \leq \mathcal{O}(M/N)$. Hence, the variance in \eqref{eq:aux:lemma1: variance1} decays with
		\begin{equation*}
			\begin{aligned}
				\var{\frac{1}{N}\tr{\boldsymbol{B}\hat{\boldsymbol{Q}}(z_{1})\phi_{M}}} = \frac{1}{N^{2}} \var{\tr{\boldsymbol{B}\hat{\boldsymbol{Q}}(z_{1})\phi_{M}}} \leq \frac{1}{N^{2}} \mathcal{O}\left(\frac{M}{N}\right) = \mathcal{O}\left(N^{-2}\right).
			\end{aligned}
		\end{equation*}
		
		Next, we analyze the variance in \eqref{eq:aux:lemma1: variance2}. Therefore, we use the Nash-Poincar\'e inequality to upper-bound
		\begin{equation}
			\begin{aligned}
				\var{\tr{\boldsymbol{B}\hat{\boldsymbol{Q}}(z_{1})\boldsymbol{R}^{1/2}\frac{\boldsymbol{X}\boldsymbol{X}^{\operatorname{H}}}{N}\phi_{M}}}\leq \sum_{i=1}^{M}\sum_{j=1}^{N} \left( \mean{\abs{\frac{\partial}{\partial X_{ij}} \tr{\boldsymbol{B}\hat{\boldsymbol{Q}}(z_{1})\boldsymbol{R}^{1/2}\frac{\boldsymbol{X}\boldsymbol{X}^{\operatorname{H}}}{N}\phi_{M}}}^{2}} +\mean{\abs{\frac{\partial}{\partial X_{ij}^{*}} \tr{\boldsymbol{B}\hat{\boldsymbol{Q}}(z_{1})\boldsymbol{R}^{1/2}\frac{\boldsymbol{X}\boldsymbol{X}^{\operatorname{H}}}{N}\phi_{M}}}^{2}}\right),
			\end{aligned}
			\label{eq:proof:aux:lemma1:variance2: intermediate_1} 
		\end{equation}
		where
		\begin{equation}
			\begin{aligned}
				\frac{\partial}{\partial X_{ij}} \tr{\boldsymbol{B}\hat{\boldsymbol{Q}}(z_{1})\boldsymbol{R}^{1/2}\frac{\boldsymbol{X}\boldsymbol{X}^{\operatorname{H}}}{N}\phi_{M}} =& -\sum_{m=1}^{M} \boldsymbol{e}_{m}^{\operatorname{T}}\boldsymbol{B}\hat{\boldsymbol{Q}}(z_{1})\boldsymbol{R}^{1/2} \frac{\boldsymbol{e}_{i}\boldsymbol{x}_{j}^{\operatorname{H}}}{N} \boldsymbol{R}^{\operatorname{H}/2} \hat{\boldsymbol{Q}}(z_{1}) \boldsymbol{R}^{1/2}\frac{\boldsymbol{X}\boldsymbol{X}^{\operatorname{H}}}{N} \boldsymbol{e}_{m}\phi_{M}\\
				 &+ \sum_{m=1}^{M} \boldsymbol{e}_{m}^{\operatorname{T}}\boldsymbol{B}\hat{\boldsymbol{Q}}(z_{1})\boldsymbol{R}^{1/2}\frac{\boldsymbol{e}_{i}\boldsymbol{x}_{j}^{\operatorname{H}}}{N} \boldsymbol{e}_{m}\phi_{M} + \mathcal{O}\left(N^{-\mathbb{N}}\right),
			\end{aligned}
			\label{eq:proof:aux:lemma1:variance2: intermediate_2a}
		\end{equation}
		\begin{equation}
			\begin{aligned}
				\frac{\partial}{\partial X_{ij}^{*}} \tr{\boldsymbol{B}\hat{\boldsymbol{Q}}(z_{1})\boldsymbol{R}^{1/2}\frac{\boldsymbol{X}\boldsymbol{X}^{\operatorname{H}}}{N}\phi_{M}} =& -\sum_{m=1}^{M} \boldsymbol{e}_{m}^{\operatorname{T}}\boldsymbol{B}\hat{\boldsymbol{Q}}(z_{1})\boldsymbol{R}^{1/2} \frac{\boldsymbol{x}_{j}\boldsymbol{e}_{i}^{\operatorname{T}}}{N} \boldsymbol{R}^{\operatorname{H}/2} \hat{\boldsymbol{Q}}(z_{1}) \boldsymbol{R}^{1/2}\frac{\boldsymbol{X}\boldsymbol{X}^{\operatorname{H}}}{N} \boldsymbol{e}_{m}\phi_{M}\\
				&+ \sum_{m=1}^{M} \boldsymbol{e}_{m}^{\operatorname{T}}\boldsymbol{B}\hat{\boldsymbol{Q}}(z_{1})\boldsymbol{R}^{1/2}\frac{\boldsymbol{x}_{j}\boldsymbol{e}_{i}^{\operatorname{T}}}{N} \boldsymbol{e}_{m}\phi_{M}+ \mathcal{O}\left(N^{-\mathbb{N}}\right).
			\end{aligned}	
			\label{eq:proof:aux:lemma1:variance2: intermediate_2b} 
		\end{equation}
		Using Jensen's inequality we can upper-bound the squared absolute value of the two derivatives in \eqref{eq:proof:aux:lemma1:variance2: intermediate_2a} and \eqref{eq:proof:aux:lemma1:variance2: intermediate_2b} by  
		\begin{equation}
			\begin{aligned}
				\abs{\frac{\partial}{\partial X_{ij}} \tr{\boldsymbol{B}\hat{\boldsymbol{Q}}(z_{1})\boldsymbol{R}^{1/2}\frac{\boldsymbol{X}\boldsymbol{X}^{\operatorname{H}}}{N}\phi_{M}} }^{2} \leq & 2 \abs{\sum_{m=1}^{M} \boldsymbol{e}_{m}^{\operatorname{T}}\boldsymbol{B}\hat{\boldsymbol{Q}}(z_{1})\boldsymbol{R}^{1/2} \frac{\boldsymbol{e}_{i}\boldsymbol{x}_{j}^{\operatorname{H}}}{N} \boldsymbol{R}^{\operatorname{H}/2} \hat{\boldsymbol{Q}}(z_{1}) \boldsymbol{R}^{1/2}\frac{\boldsymbol{X}\boldsymbol{X}^{\operatorname{H}}}{N} \boldsymbol{e}_{m}\phi_{M}}^{2}\\
				& + 2 \abs{\sum_{m=1}^{M} \boldsymbol{e}_{m}^{\operatorname{T}}\boldsymbol{B}\hat{\boldsymbol{Q}}(z_{1})\boldsymbol{R}^{1/2}\frac{\boldsymbol{e}_{i}\boldsymbol{x}_{j}^{\operatorname{H}}}{N} \boldsymbol{e}_{m}\phi_{M}}^{2}+ \mathcal{O}\left(N^{-\mathbb{N}}\right),
			\end{aligned}
			\label{eq:proof:aux:lemma1:variance2: intermediate_3a} 
		\end{equation}
		and
		\begin{equation}
			\begin{aligned}
				\abs{\frac{\partial}{\partial X_{ij}^{*}} \tr{\boldsymbol{B}\hat{\boldsymbol{Q}}(z_{1})\boldsymbol{R}^{1/2}\frac{\boldsymbol{X}\boldsymbol{X}^{\operatorname{H}}}{N}\phi_{M}}}^{2} \leq & 2 \abs{\sum_{m=1}^{M} \boldsymbol{e}_{m}^{\operatorname{T}}\boldsymbol{B}\hat{\boldsymbol{Q}}(z_{1})\boldsymbol{R}^{1/2} \frac{\boldsymbol{x}_{j}\boldsymbol{e}_{i}^{\operatorname{T}}}{N} \boldsymbol{R}^{\operatorname{H}/2} \hat{\boldsymbol{Q}}(z_{1}) \boldsymbol{R}^{1/2}\frac{\boldsymbol{X}\boldsymbol{X}^{\operatorname{H}}}{N} \boldsymbol{e}_{m}\phi_{M}}^{2} \\
				&+ 2 \abs{\sum_{m=1}^{M} \boldsymbol{e}_{m}^{\operatorname{T}}\boldsymbol{B}\hat{\boldsymbol{Q}}(z_{1})\boldsymbol{R}^{1/2}\frac{\boldsymbol{x}_{j}\boldsymbol{e}_{i}^{\operatorname{T}}}{N} \boldsymbol{e}_{m}\phi_{M}}^{2}+ \mathcal{O}\left(N^{-\mathbb{N}}\right).
			\end{aligned}
			\label{eq:proof:aux:lemma1:variance2: intermediate_3b} 
		\end{equation}
		It follows from \eqref{eq:proof:aux:lemma1:variance2: intermediate_3a} and \eqref{eq:proof:aux:lemma1:variance2: intermediate_3b} that 
		\begin{equation*}
			\sum_{i=1}^{M}\sum_{j=1}^{N}\mean{\abs{\frac{\partial}{\partial X_{ij}} \tr{\boldsymbol{B}\hat{\boldsymbol{Q}}(z_{1})\boldsymbol{R}^{1/2}\frac{\boldsymbol{X}\boldsymbol{X}^{\operatorname{H}}}{N}\phi_{M}}}^{2}}  = \mathcal{O}\left(\frac{M}{N}\right), \qquad
			\sum_{i=1}^{M}\sum_{j=1}^{N}\mean{\abs{\frac{\partial}{\partial X_{ij}^{*}} \tr{\boldsymbol{B}\hat{\boldsymbol{Q}}(z_{1})\boldsymbol{R}^{1/2}\frac{\boldsymbol{X}\boldsymbol{X}^{\operatorname{H}}}{N}\phi_{M}}}^{2}} = \mathcal{O}\left(\frac{M}{N}\right),
		\end{equation*}
		and the variance in \eqref{eq:proof:aux:lemma1:variance2: intermediate_1} is upper bounded by $\var{\tr{\boldsymbol{B}\hat{\boldsymbol{Q}}(z_{1})\boldsymbol{R}^{1/2}\frac{\boldsymbol{X}\boldsymbol{X}^{\operatorname{H}}}{N}}} \leq \mathcal{O}(M/N)$. As a direct consequence the variance in \eqref{eq:aux:lemma1: variance2} decays with
		\begin{equation*}
			\var{\frac{1}{N}\tr{\boldsymbol{B}\hat{\boldsymbol{Q}}(z_{1})\boldsymbol{R}^{1/2}\frac{\boldsymbol{X}\boldsymbol{X}^{\operatorname{H}}}{N}\phi_{M}}} = \frac{1}{N^{2}}\var{\tr{\boldsymbol{B}\hat{\boldsymbol{Q}}(z_{1})\boldsymbol{R}^{1/2}\frac{\boldsymbol{X}\boldsymbol{X}^{\operatorname{H}}}{N}\phi_{M}}} \leq \frac{1}{N^{2}} \mathcal{O}\left(\frac{M}{N}\right) =  \mathcal{O}\left(N^{-2}\right).
		\end{equation*}
	
	\subsubsection{Expectations}

	To derive the expectation in \eqref{eq:aux:lemma1: expectation1} we express the trace as sum 
	\begin{equation}
		\frac{1}{N}\mean{\tr{\boldsymbol{B}\hat{\boldsymbol{Q}}(z_{1})\phi_{M}}} = \frac{1}{N} \sum_{m=1}^{M} \mean{\boldsymbol{e}_{m}	^{\operatorname{T}} \boldsymbol{B} \hat{\boldsymbol{Q}}(z_{1})\boldsymbol{e}_{m}\phi_{M}} = \frac{1}{N} \sum_{m=1}^{M} \boldsymbol{e}_{m}^{\operatorname{T}} \boldsymbol{B} \mean{\hat{\boldsymbol{Q}}(z_{1})\phi_{M}}\boldsymbol{e}_{m}.
	\end{equation}
	By replacing the random resolvent $\hat{\boldsymbol{Q}}(z_{1})$ with its asymptotic deterministic equivalent $\bar{\boldsymbol{Q}}(z_{1})$ in \eqref{eq: resolvent_and_deterministic_equivalent} we obtain the expectation in \eqref{eq:aux:lemma1: expectation1} in Lemma \ref{aux:lemma1}
	\begin{equation}
		\frac{1}{N}\mean{\tr{\boldsymbol{B}\hat{\boldsymbol{Q}}(z_{1})\phi_{M}}}= \frac{1}{N}\sum_{m=1}^{M} \boldsymbol{e}_{m}^{\operatorname{T}}\boldsymbol{B}\bar{\boldsymbol{Q}}(z_{1})\boldsymbol{e}_{m} + \mathcal{O}\left(N^{-1}\right) = \frac{1}{N}\tr{\boldsymbol{B}\bar{\boldsymbol{Q}}(z_{1})} + \mathcal{O}\left(N^{-1}\right),
	\end{equation}
	where we have used $\mean{\hat{\boldsymbol{Q}}(z)\phi_{M}} = \bar{\boldsymbol{Q}}(z) + \mathcal{O}(N^{-1})$.

	To derive the expectation in \eqref{eq:aux:lemma1: expectation2} we express the trace as sum and the sample covariance matrix $\hat{\boldsymbol{R}}$ as shown in \eqref{eq: sample_covariance_matrix_sum} and apply the integration by parts formula
	\begin{equation}
		\begin{aligned}
			\frac{1}{N}\mean{\tr{\boldsymbol{B}\hat{\boldsymbol{Q}}(z_{1})\boldsymbol{R}^{1/2}\frac{\boldsymbol{X}\boldsymbol{X}^{\operatorname{H}}}{N}\phi_{M}}} =& \frac{1}{N}\sum_{i=1}^{M}\sum_{j=1}^{N} \sum_{m=1}^{M} \mean{X_{ij}\boldsymbol{e}_{m}^{\operatorname{T}}\boldsymbol{B}\hat{\boldsymbol{Q}}(z_{1})\boldsymbol{R}^{1/2} \frac{\boldsymbol{e}_{i}\boldsymbol{x}_{j}^{\operatorname{H}}}{N}\boldsymbol{e}_{m}\phi_{M}}\\
			=&\frac{1}{N} \sum_{i=1}^{M}\sum_{j=1}^{N} \sum_{m=1}^{M}  \mean{\frac{\partial }{\partial X_{ij}^{*}} \left(\boldsymbol{e}_{m}^{\operatorname{T}}\boldsymbol{B}\hat{\boldsymbol{Q}}(z_{1})\boldsymbol{R}^{1/2} \frac{\boldsymbol{e}_{i}\boldsymbol{x}_{j}^{\operatorname{H}}}{N}\boldsymbol{e}_{m}\phi_{M}\right)},
		\end{aligned}
		\label{eq:proof:aux:lemma1:expectation2: intermediate_1}
	\end{equation}
	where
	\begin{equation}
		\begin{aligned}
			\frac{\partial}{\partial X_{ij}^{*}}\left( \boldsymbol{e}_{m}^{\operatorname{T}}\boldsymbol{B}\hat{\boldsymbol{Q}}(z_{1})\boldsymbol{R}^{1/2} \frac{\boldsymbol{e}_{i}\boldsymbol{x}_{j}^{\operatorname{H}}}{N}\boldsymbol{e}_{m}\phi_{M}\right) =& - \boldsymbol{e}_{m}^{\operatorname{T}}\boldsymbol{B}\hat{\boldsymbol{Q}}(z_{1})\boldsymbol{R}^{1/2} \frac{\boldsymbol{x}_{j}\boldsymbol{e}_{i}^{\operatorname{T}}}{N} \boldsymbol{R}^{\operatorname{H}/2}\hat{\boldsymbol{Q}}(z_{1})\boldsymbol{R}^{1/2} \frac{\boldsymbol{e}_{i}\boldsymbol{x}_{j}^{\operatorname{H}}}{N}\boldsymbol{e}_{m}\phi_{M}\\
			&+\boldsymbol{e}_{m}^{\operatorname{T}}\boldsymbol{B}\hat{\boldsymbol{Q}}(z_{1})\boldsymbol{R}^{1/2} \frac{\boldsymbol{e}_{i}\boldsymbol{e}_{i}^{\operatorname{T}}}{N}\boldsymbol{e}_{m}\phi_{M} + \mathcal{O}\left(N^{-\mathbb{N}}\right).
		\end{aligned}
		\label{eq:proof:aux:lemma1:expectation2: intermediate_2}
	\end{equation}
	Substituting the derivative in \eqref{eq:proof:aux:lemma1:expectation2: intermediate_2} into \eqref{eq:proof:aux:lemma1:expectation2: intermediate_1} and computing the triple sum yields
	\begin{equation}
		\begin{aligned}
			\frac{1}{N}\mean{\tr{\boldsymbol{B}\hat{\boldsymbol{Q}}(z_{1})\boldsymbol{R}^{1/2}\frac{\boldsymbol{X}\boldsymbol{X}^{\operatorname{H}}}{N}\phi_{M}}} =& -\frac{1}{N} \mean{\tr{\boldsymbol{B}\hat{\boldsymbol{Q}}(z_{1})\boldsymbol{R}^{1/2}\frac{\boldsymbol{X}\boldsymbol{X}^{\operatorname{H}}}{N}}\frac{1}{N}\tr{\hat{\boldsymbol{Q}}(z_{1})\boldsymbol{R}}\phi_{M}}\\ 
			&+ \frac{1}{N}\mean{\tr{\boldsymbol{B}\hat{\boldsymbol{Q}}(z_{1})\boldsymbol{R}^{1/2}\phi_{M}}} + \mathcal{O}\left(N^{-\mathbb{N}}\right).
		\end{aligned}
		\label{eq:proof:aux:lemma1:expectation2: intermediate_3}
	\end{equation}
	Next, we decorrelate the first term on the right hand side of \eqref{eq:proof:aux:lemma1:expectation2: intermediate_3}. Therefore, we express the covariance as
	\begin{equation}
		\begin{aligned}
			&\cov{\frac{1}{N}\tr{\boldsymbol{B}\hat{\boldsymbol{Q}}(z_{1})\boldsymbol{R}^{1/2}\frac{\boldsymbol{X}\boldsymbol{X}^{\operatorname{H}}}{N}\phi_{M}}}{\frac{1}{N}\tr{\hat{\boldsymbol{Q}}(z_{1})\boldsymbol{R}\phi_{M}}}\\
			 = & \mean{\left( \frac{1}{N}\tr{\boldsymbol{B}\hat{\boldsymbol{Q}}(z_{1})\boldsymbol{R}^{1/2}\frac{\boldsymbol{X}\boldsymbol{X}^{\operatorname{H}}}{N}\phi_{M}}\right)^{(\circ)}\left( \frac{1}{N}\tr{\hat{\boldsymbol{Q}}(z_{1})\boldsymbol{R}\phi_{M}}\right)^{(\circ)}}\\
			=&  \mean{\frac{1}{N}\tr{\boldsymbol{B}\hat{\boldsymbol{Q}}(z_{1})\boldsymbol{R}^{1/2}\frac{\boldsymbol{X}\boldsymbol{X}^{\operatorname{H}}}{N}}\frac{1}{N}\tr{\hat{\boldsymbol{Q}}(z_{1})\boldsymbol{R}}\phi_{M}}  - \mean{\frac{1}{N}\tr{\boldsymbol{B}\hat{\boldsymbol{Q}}(z_{1})\boldsymbol{R}^{1/2}\frac{\boldsymbol{X}\boldsymbol{X}^{\operatorname{H}}}{N}\phi_{M}}} \mean{\frac{1}{N}\tr{\hat{\boldsymbol{Q}}(z_{1})\boldsymbol{R}\phi_{M}}},
		\end{aligned}
		\label{eq:proof:aux:lemma1:expectation2: intermediate_4}
	\end{equation}
 	Afterwards, we apply the Cauchy-Schwarz inequality to upper-bound the covariance 
	\begin{equation}
		\begin{aligned}
			&\abs{\cov{\frac{1}{N}\tr{\boldsymbol{B}\hat{\boldsymbol{Q}}(z_{1})\boldsymbol{R}^{1/2}\frac{\boldsymbol{X}\boldsymbol{X}^{\operatorname{H}}}{N}\phi_{M}}}{\frac{1}{N}\tr{\hat{\boldsymbol{Q}}(z_{1})\boldsymbol{R}\phi_{M}}}}^{2} \\
			\leq & \mean{\abs{\left( \frac{1}{N}\tr{\boldsymbol{B}\hat{\boldsymbol{Q}}(z_{1})\boldsymbol{R}^{1/2}\frac{\boldsymbol{X}\boldsymbol{X}^{\operatorname{H}}}{N}\phi_{M}}\right)^{(\circ)}}^{2}} \mean{\abs{\left( \frac{1}{N}\tr{\hat{\boldsymbol{Q}}(z_{1})\boldsymbol{R}\phi_{M}}\right)^{(\circ)}}^{2}}\\
			=& \var{\frac{1}{N}\tr{\boldsymbol{B}\hat{\boldsymbol{Q}}(z_{1})\boldsymbol{R}^{1/2}\frac{\boldsymbol{X}\boldsymbol{X}^{\operatorname{H}}}{N}\phi_{M}}} \var{\frac{1}{N}\tr{\hat{\boldsymbol{Q}}(z_{1})\boldsymbol{R}\phi_{M}}}.
		\end{aligned}
		\label{eq:proof:aux:lemma1:expectation2: intermediate_5}	
	\end{equation}
	Using the previously derived variances in \eqref{eq:aux:lemma1: variance1} and \eqref{eq:aux:lemma1: variance2} we know that
	\begin{equation*}
		\var{\frac{1}{N}\tr{\boldsymbol{B}\hat{\boldsymbol{Q}}(z_{1})\boldsymbol{R}^{1/2}\frac{\boldsymbol{X}\boldsymbol{X}^{\operatorname{H}}}{N}\phi_{M}}} = \mathcal{O}\left(N^{-2}\right),  \qquad
		\var{\frac{1}{N}\tr{\hat{\boldsymbol{Q}}(z_{1})\boldsymbol{R}\phi_{M}}} = \mathcal{O}\left(N^{-2}\right),
	\end{equation*}
	and therefore the covariance in \eqref{eq:proof:aux:lemma1:expectation2: intermediate_4} decays with
	\begin{equation*}
		\cov{\frac{1}{N}\tr{\boldsymbol{B}\hat{\boldsymbol{Q}}(z_{1})\boldsymbol{R}^{1/2}\frac{\boldsymbol{X}\boldsymbol{X}^{\operatorname{H}}}{N}\phi_{M}}}{\frac{1}{N}\tr{\hat{\boldsymbol{Q}}(z_{1})\boldsymbol{R}\phi_{M}}} \leq \mathcal{O}\left(N^{-2}\right).
	\end{equation*}
	This allows to decorrelate the first term on the right hand side of \eqref{eq:proof:aux:lemma1:expectation2: intermediate_3} as follows
	\begin{equation}
		\begin{aligned}
			&\frac{1}{N} \mean{\tr{\boldsymbol{B}\hat{\boldsymbol{Q}}(z_{1})\boldsymbol{R}^{1/2}\frac{\boldsymbol{X}\boldsymbol{X}^{\operatorname{H}}}{N}}\frac{1}{N}\tr{\hat{\boldsymbol{Q}}(z_{1})\boldsymbol{R}}\phi_{M}} \\
			=& \mean{\frac{1}{N}\tr{\boldsymbol{B}\hat{\boldsymbol{Q}}(z_{1})\boldsymbol{R}^{1/2}\frac{\boldsymbol{X}\boldsymbol{X}^{\operatorname{H}}}{N}\phi_{M}}} \mean{\frac{1}{N}\tr{\hat{\boldsymbol{Q}}(z_{1})\boldsymbol{R}\phi_{M}}} + \cov{\frac{1}{N}\tr{\boldsymbol{B}\hat{\boldsymbol{Q}}(z_{1})\boldsymbol{R}^{1/2}\frac{\boldsymbol{X}\boldsymbol{X}^{\operatorname{H}}}{N}\phi_{M}}}{\frac{1}{N}\tr{\hat{\boldsymbol{Q}}(z_{1})\boldsymbol{R}\phi_{M}}}\\
			=&\mean{\frac{1}{N}\tr{\boldsymbol{B}\hat{\boldsymbol{Q}}(z_{1})\boldsymbol{R}^{1/2}\frac{\boldsymbol{X}\boldsymbol{X}^{\operatorname{H}}}{N}\phi_{M}}} \mean{\frac{1}{N}\tr{\hat{\boldsymbol{Q}}(z_{1})\boldsymbol{R}\phi_{M}}} + \mathcal{O}\left(N^{-2}\right).
		\end{aligned}
		\label{eq:proof:aux:lemma1:expectation2: intermediate_6}
	\end{equation}
	Substituting the decorrelated expression in \eqref{eq:proof:aux:lemma1:expectation2: intermediate_6} into \eqref{eq:proof:aux:lemma1:expectation2: intermediate_3} yields
	\begin{equation*}
		\begin{aligned}
			\frac{1}{N} \mean{\tr{\boldsymbol{B}\hat{\boldsymbol{Q}}(z_{1})\boldsymbol{R}^{1/2}\frac{\boldsymbol{X}\boldsymbol{X}^{\operatorname{H}}}{N}\phi_{M}}} =& \frac{1}{N}\mean{\tr{\boldsymbol{B}\hat{\boldsymbol{Q}}(z_{1})\boldsymbol{R}^{1/2}\phi_{M}}}\\
			&-\mean{\frac{1}{N}\tr{\boldsymbol{B}\hat{\boldsymbol{Q}}(z_{1})\boldsymbol{R}^{1/2}\frac{\boldsymbol{X}\boldsymbol{X}^{\operatorname{H}}}{N}\phi_{M}}} \mean{\frac{1}{N}\tr{\hat{\boldsymbol{Q}}(z_{1})\boldsymbol{R}\phi_{M}}} +\mathcal{O}\left(N^{-2}\right),
		\end{aligned}	
	\end{equation*}
	which can be rearranged as follows
	\begin{equation}
		\frac{1}{N} \mean{\tr{\boldsymbol{B}\hat{\boldsymbol{Q}}(z_{1})\boldsymbol{R}^{1/2}\frac{\boldsymbol{X}\boldsymbol{X}^{\operatorname{H}}}{N}\phi_{M}}} \left( 1+ \frac{1}{N}\mean{\tr{\hat{\boldsymbol{Q}}(z_{1})\boldsymbol{R}\phi_{M}}}\right) = \frac{1}{N}\mean{\tr{\boldsymbol{B}\hat{\boldsymbol{Q}}(z_{1})\boldsymbol{R}^{1/2}\phi_{M}}} +\mathcal{O}\left(N^{-2}\right).
		\label{eq:proof:aux:lemma1:expectation2: intermediate_7}
	\end{equation}
	Using the previously derived expectation in \eqref{eq:aux:lemma1: expectation1} we know that \eqref{eq:proof:aux:lemma1:expectation2: intermediate_7} simplifies to
	\begin{equation}
		\frac{1}{N} \mean{\tr{\boldsymbol{B}\hat{\boldsymbol{Q}}(z_{1})\boldsymbol{R}^{1/2}\frac{\boldsymbol{X}\boldsymbol{X}^{\operatorname{H}}}{N}\phi_{M}}} \left( 1+ \frac{1}{N}{\tr{\bar{\boldsymbol{Q}}(z_{1})\boldsymbol{R}}}\right) = \frac{1}{N}{\tr{\boldsymbol{B}\bar{\boldsymbol{Q}}(z_{1})\boldsymbol{R}^{1/2}}} +\mathcal{O}\left(N^{-1}\right).
		\label{eq:proof:aux:lemma1:expectation2: intermediate_8}
	\end{equation}
	Using $\frac{\omega(z_{1})}{z_{1}} = (1+\frac{1}{N} \tr{\boldsymbol{R}\bar{\boldsymbol{Q}}(z_{1})})$ which immediately follows from \eqref{eq: definition_omega_z} and multiplying \eqref{eq:proof:aux:lemma1:expectation2: intermediate_8} on both sides with $\frac{z_{1}}{\omega(z_{1})}$ we obtain \eqref{eq:aux:lemma1: expectation2}.

\section{Proof of Lemma \ref{aux:lemma2}}
	\label{proof:aux:lemma2}
	
	\subsubsection{Variances}
		In order to upper-bound the covariance in \eqref{eq:aux:lemma2: variance1} we apply the Nash-Poincar\'e inequality 
		\begin{equation}
			\var{\tr{\boldsymbol{B}\hat{\boldsymbol{Q}}(z_{1})\boldsymbol{R}\hat{\boldsymbol{Q}}(z_{2})\phi_{M}}} \leq \sum_{i=1}^{M}\sum_{j=1}^{N} \left( \mean{\abs{\frac{\partial}{\partial X_{ij}}\tr{\boldsymbol{B}\hat{\boldsymbol{Q}}(z_{1})\boldsymbol{R}\hat{\boldsymbol{Q}}(z_{2})\phi_{M}}}^{2}} + \mean{\abs{\frac{\partial}{\partial X_{ij}^{*}}\tr{\boldsymbol{B}\hat{\boldsymbol{Q}}(z_{1})\boldsymbol{R}\hat{\boldsymbol{Q}}(z_{2})\phi_{M}}}^{2}}\right),
			\label{eq:proof:aux:lemma2:variance1: intermediate_1}
		\end{equation}
		where the first order derivative with respect to $X_{ij}$ and $X_{ij}^{*}$ are given by
		\begin{equation}
			\begin{aligned}
				\frac{\partial}{\partial X_{ij}}\tr{\boldsymbol{B}\hat{\boldsymbol{Q}}(z_{1})\boldsymbol{R}\hat{\boldsymbol{Q}}(z_{2})\phi_{M}} =&- \sum_{m=1}^{M} \boldsymbol{e}_{m}^{\operatorname{T}}\boldsymbol{B}\hat{\boldsymbol{Q}}(z_{1})\boldsymbol{R}^{1/2} \frac{\boldsymbol{e}_{i}\boldsymbol{x}_{j}^{\operatorname{H}}}{N}\boldsymbol{R}^{\operatorname{H}/2} \hat{\boldsymbol{Q}}(z_{1})\boldsymbol{R}\hat{\boldsymbol{Q}}(z_{2})\boldsymbol{e}_{m}\phi_{M}\\
				& -\sum_{m=1}^{M} \boldsymbol{e}_{m}^{\operatorname{T}}\boldsymbol{B}\hat{\boldsymbol{Q}}(z_{1}) \boldsymbol{R} \hat{\boldsymbol{Q}}(z_{2})\boldsymbol{R}^{1/2} \frac{\boldsymbol{e}_{i}\boldsymbol{x}_{j}^{\operatorname{H}}}{N}\boldsymbol{R}^{\operatorname{H}/2}\hat{\boldsymbol{Q}}(z_{2})\boldsymbol{e}_{m}\phi_{M} + \mathcal{O}\left(N^{-\mathbb{N}}\right), 
			\end{aligned}
			\label{eq:proof:aux:lemma2:variance1: intermediate_2a0}
		\end{equation}
		and
		\begin{equation}
			\begin{aligned}
				\frac{\partial}{\partial X_{ij}^{*}}\tr{\boldsymbol{B}\hat{\boldsymbol{Q}}(z_{1})\boldsymbol{R}\hat{\boldsymbol{Q}}(z_{2})\phi_{M}} =& - \sum_{m=1}^{M} \boldsymbol{e}_{m}^{\operatorname{T}}\boldsymbol{B}\hat{\boldsymbol{Q}}(z_{1})\boldsymbol{R}^{1/2} \frac{\boldsymbol{x}_{j}\boldsymbol{e}_{i}^{\operatorname{T}}}{N}\boldsymbol{R}^{\operatorname{H}/2} \hat{\boldsymbol{Q}}(z_{1})\boldsymbol{R}\hat{\boldsymbol{Q}}(z_{2})\boldsymbol{e}_{m}\phi_{M}\\
				& -\sum_{m=1}^{M} \boldsymbol{e}_{m}^{\operatorname{T}}\boldsymbol{B}\hat{\boldsymbol{Q}}(z_{1}) \boldsymbol{R} \hat{\boldsymbol{Q}}(z_{2})\boldsymbol{R}^{1/2} \frac{\boldsymbol{x}_{j}\boldsymbol{e}_{i}^{\operatorname{T}}}{N}\boldsymbol{R}^{\operatorname{H}/2}\hat{\boldsymbol{Q}}(z_{2})\boldsymbol{e}_{m}\phi_{M} + \mathcal{O}\left(N^{-\mathbb{N}}\right). 
			\end{aligned}
			\label{eq:proof:aux:lemma2:variance1: intermediate_2b0}
		\end{equation}
		Using Jensen's inequality we can upper bound the squared absolute value of the first order derivative with respect to $X_{ij}$ in \eqref{eq:proof:aux:lemma2:variance1: intermediate_2a0} 
		\begin{equation}
			\begin{aligned}
				\abs{\frac{\partial}{\partial X_{ij}}\tr{\boldsymbol{B}\hat{\boldsymbol{Q}}(z_{1})\boldsymbol{R}\hat{\boldsymbol{Q}}(z_{2})\phi_{M}}}^{2} \leq & 2 \abs{\sum_{m=1}^{M} \boldsymbol{e}_{m}^{\operatorname{T}}\boldsymbol{B}\hat{\boldsymbol{Q}}(z_{1})\boldsymbol{R}^{1/2} \frac{\boldsymbol{e}_{i}\boldsymbol{x}_{j}^{\operatorname{H}}}{N}\boldsymbol{R}^{\operatorname{H}/2} \hat{\boldsymbol{Q}}(z_{1})\boldsymbol{R}\hat{\boldsymbol{Q}}(z_{2})\boldsymbol{e}_{m}\phi_{M}}^{2}\\
				& + 2\abs{\sum_{m=1}^{M} \boldsymbol{e}_{m}^{\operatorname{T}}\boldsymbol{B}\hat{\boldsymbol{Q}}(z_{1}) \boldsymbol{R} \hat{\boldsymbol{Q}}(z_{2})\boldsymbol{R}^{1/2} \frac{\boldsymbol{e}_{i}\boldsymbol{x}_{j}^{\operatorname{H}}}{N}\boldsymbol{R}^{\operatorname{H}/2}\hat{\boldsymbol{Q}}(z_{2})\boldsymbol{e}_{m}\phi_{M}}^{2} + \mathcal{O}\left(N^{-\mathbb{N}}\right)
			\end{aligned}
			\label{eq:proof:aux:lemma2:variance1: intermediate_2a1}
		\end{equation}
		and the first order derivative with respect to $X_{ij}^{*}$ in \eqref{eq:proof:aux:lemma2:variance1: intermediate_2b0}
		\begin{equation}
			\begin{aligned}
				\abs{\frac{\partial}{\partial X_{ij}^{*}}\tr{\boldsymbol{B}\hat{\boldsymbol{Q}}(z_{1})\boldsymbol{R}\hat{\boldsymbol{Q}}(z_{2})\phi_{M}}}^{2} \leq & 2 \abs{\sum_{m=1}^{M} \boldsymbol{e}_{m}^{\operatorname{T}}\boldsymbol{B}\hat{\boldsymbol{Q}}(z_{1})\boldsymbol{R}^{1/2} \frac{\boldsymbol{x}_{j}\boldsymbol{e}_{i}^{\operatorname{T}}}{N}\boldsymbol{R}^{\operatorname{H}/2} \hat{\boldsymbol{Q}}(z_{1})\boldsymbol{R}\hat{\boldsymbol{Q}}(z_{2})\boldsymbol{e}_{m}\phi_{M}}^{2} \\
				&+2 \abs{\sum_{m=1}^{M} \boldsymbol{e}_{m}^{\operatorname{T}}\boldsymbol{B}\hat{\boldsymbol{Q}}(z_{1}) \boldsymbol{R} \hat{\boldsymbol{Q}}(z_{2})\boldsymbol{R}^{1/2} \frac{\boldsymbol{x}_{j}\boldsymbol{e}_{i}^{\operatorname{T}}}{N}\boldsymbol{R}^{\operatorname{H}/2}\hat{\boldsymbol{Q}}(z_{2})\boldsymbol{e}_{m}\phi_{M}}^{2} + \mathcal{O}\left(N^{-\mathbb{N}}\right)
			\end{aligned}
			\label{eq:proof:aux:lemma2:variance1: intermediate_2b1}
		\end{equation}
		Furthermore, it can be observed in \eqref{eq:proof:aux:lemma2:variance1: intermediate_2a1} and \eqref{eq:proof:aux:lemma2:variance1: intermediate_2b1} that
		\begin{equation*}
			\sum_{i=1}^{M}\sum_{j=1}^{N}\mean{\abs{\frac{\partial}{\partial X_{ij}}\tr{\boldsymbol{B}\hat{\boldsymbol{Q}}(z_{1})\boldsymbol{R}\hat{\boldsymbol{Q}}(z_{2})\phi_{M}}}^{2}} = \mathcal{O}\left(\frac{M}{N}\right), \qquad
			\sum_{i=1}^{M}\sum_{j=1}^{N} \mean{\abs{\frac{\partial}{\partial X_{ij}^{*}}\tr{\boldsymbol{B}\hat{\boldsymbol{Q}}(z_{1})\boldsymbol{R}\hat{\boldsymbol{Q}}(z_{2})\phi_{M}}}^{2}} = \mathcal{O}\left(\frac{M}{N}\right).
		\end{equation*}
		Consequently, the variance in \eqref{eq:proof:aux:lemma2:variance1: intermediate_1} is upper-bounded by $\var{\tr{\boldsymbol{B}\hat{\boldsymbol{Q}}(z_{1})\boldsymbol{R}\hat{\boldsymbol{Q}}(z_{2})\phi_{M}}}\leq \mathcal{O}(M/N)$ and the variance in \eqref{eq:aux:lemma2: variance1} decays with
		\begin{equation*}
			\var{\frac{1}{N}\tr{\boldsymbol{B}\hat{\boldsymbol{Q}}(z_{1})\boldsymbol{R}\hat{\boldsymbol{Q}}(z_{2})\phi_{M}}} = \frac{1}{N^{2}} \var{\tr{\boldsymbol{B}\hat{\boldsymbol{Q}}(z_{1})\boldsymbol{R}\hat{\boldsymbol{Q}}(z_{2})\phi_{M}}} \leq \frac{1}{N^{2}} \mathcal{O}\left(\frac{M}{N}\right) = \mathcal{O}\left(N^{-2}\right).
		\end{equation*}
		
		In the following we upper-bound the variance in \eqref{eq:aux:lemma2: variance2} by using the Nash-Poincar\'e inequality
		\begin{equation}
			\begin{aligned}
				\var{\tr{\boldsymbol{B}\hat{\boldsymbol{Q}}(z_{1})\boldsymbol{R}\hat{\boldsymbol{Q}}(z_{2})\boldsymbol{R}^{1/2}\frac{\boldsymbol{X}\boldsymbol{X}^{\operatorname{H}}}{N}\phi_{M}}} \leq \sum_{i=1}^{M}\sum_{j=1}^{N} \Biggl ( \mean{\abs{\frac{\partial}{\partial X_{ij}}\tr{\boldsymbol{B}\hat{\boldsymbol{Q}}(z_{1})\boldsymbol{R}\hat{\boldsymbol{Q}}(z_{2})\boldsymbol{R}^{1/2}\frac{\boldsymbol{X}\boldsymbol{X}^{\operatorname{H}}}{N}\phi_{M}}}^{2}}&\\
				 + \mean{\abs{\frac{\partial}{\partial X_{ij}^{*}}\tr{\boldsymbol{B}\hat{\boldsymbol{Q}}(z_{1})\boldsymbol{R}\hat{\boldsymbol{Q}}(z_{2})\boldsymbol{R}^{1/2}\frac{\boldsymbol{X}\boldsymbol{X}^{\operatorname{H}}}{N}\phi_{M}}}^{2}}\Biggr)&,
			\end{aligned}
			\label{eq:proof:aux:lemma2:variance2: intermediate_1}
		\end{equation}
		where the first order derivative with respect to $X_{ij}$ is given by
		\begin{equation}
			\begin{aligned}
				\frac{\partial}{\partial X_{ij}}\tr{\boldsymbol{B}\hat{\boldsymbol{Q}}(z_{1})\boldsymbol{R}\hat{\boldsymbol{Q}}(z_{2})\boldsymbol{R}^{1/2}\frac{\boldsymbol{X}\boldsymbol{X}^{\operatorname{H}}}{N}\phi_{M}} =& - \sum_{m=1}^{M} \boldsymbol{e}_{m}^{\operatorname{T}} \boldsymbol{B} \hat{\boldsymbol{Q}}(z_{1}) \boldsymbol{R}^{1/2}\frac{\boldsymbol{e}_{i}\boldsymbol{x}_{j}^{\operatorname{H}}}{N} \boldsymbol{R}^{\operatorname{H}/2} \hat{\boldsymbol{Q}}(z_{1}) \boldsymbol{R} \hat{\boldsymbol{Q}}(z_{2}) \frac{\boldsymbol{X}\boldsymbol{X}^{\operatorname{H}}}{N}\boldsymbol{e}_{m}\phi_{M} \\
				&- \sum_{m=1}^{M} \boldsymbol{e}_{m}^{\operatorname{T}} \boldsymbol{B} \hat{\boldsymbol{Q}}(z_{1}) \boldsymbol{R} \hat{\boldsymbol{Q}}(z_{2}) \boldsymbol{R}^{1/2}\frac{\boldsymbol{e}_{i}\boldsymbol{x}_{j}^{\operatorname{H}}}{N} \boldsymbol{R}^{\operatorname{H}/2} \hat{\boldsymbol{Q}}(z_{2}) \frac{\boldsymbol{X}\boldsymbol{X}^{\operatorname{H}}}{N}\boldsymbol{e}_{m}\phi_{M}\\
				&+ \sum_{m=1}^{M} \boldsymbol{e}_{m}^{\operatorname{T}} \boldsymbol{B}\hat{\boldsymbol{Q}}(z_{1})\boldsymbol{R}\hat{\boldsymbol{Q}}(z_{2}) \boldsymbol{R}^{1/2} \frac{\boldsymbol{e}_{i}\boldsymbol{x}_{j}^{\operatorname{H}}}{N}\boldsymbol{e}_{m}\phi_{M} + \mathcal{O}\left(N^{-\mathbb{N}}\right), 
			\end{aligned}
			\label{eq:proof:aux:lemma2:variance2: intermediate_2a0}
		\end{equation}
		and the first order derivative with respect to $X_{ij}$ is given by
		\begin{equation}
			\begin{aligned}
				\frac{\partial}{\partial X_{ij}^{*}}\tr{\boldsymbol{B}\hat{\boldsymbol{Q}}(z_{1})\boldsymbol{R}\hat{\boldsymbol{Q}}(z_{2})\boldsymbol{R}^{1/2}\frac{\boldsymbol{X}\boldsymbol{X}^{\operatorname{H}}}{N}\phi_{M}} =& - \sum_{m=1}^{M} \boldsymbol{e}_{m}^{\operatorname{T}} \boldsymbol{B} \hat{\boldsymbol{Q}}(z_{1}) \boldsymbol{R}^{1/2}\frac{\boldsymbol{x}_{j}\boldsymbol{e}_{i}^{\operatorname{T}}}{N} \boldsymbol{R}^{\operatorname{H}/2} \hat{\boldsymbol{Q}}(z_{1}) \boldsymbol{R} \hat{\boldsymbol{Q}}(z_{2}) \frac{\boldsymbol{X}\boldsymbol{X}^{\operatorname{H}}}{N}\boldsymbol{e}_{m}\phi_{M} \\
				&- \sum_{m=1}^{M} \boldsymbol{e}_{m}^{\operatorname{T}} \boldsymbol{B} \hat{\boldsymbol{Q}}(z_{1}) \boldsymbol{R} \hat{\boldsymbol{Q}}(z_{2}) \boldsymbol{R}^{1/2}\frac{\boldsymbol{x}_{j}\boldsymbol{e}_{i}^{\operatorname{T}}}{N} \boldsymbol{R}^{\operatorname{H}/2} \hat{\boldsymbol{Q}}(z_{2}) \frac{\boldsymbol{X}\boldsymbol{X}^{\operatorname{H}}}{N}\boldsymbol{e}_{m}\phi_{M}\\
				&+ \sum_{m=1}^{M} \boldsymbol{e}_{m}^{\operatorname{T}} \boldsymbol{B}\hat{\boldsymbol{Q}}(z_{1})\boldsymbol{R}\hat{\boldsymbol{Q}}(z_{2}) \boldsymbol{R}^{1/2} \frac{\boldsymbol{x}_{j}\boldsymbol{e}_{i}^{\operatorname{T}}}{N}\boldsymbol{e}_{m}\phi_{M} + \mathcal{O}\left(N^{-\mathbb{N}}\right).
			\end{aligned}
			\label{eq:proof:aux:lemma2:variance2: intermediate_2b0}
		\end{equation}
		Using Jensen's inequality we can upper-bound the absolute value squared of the first order derivative in \eqref{eq:proof:aux:lemma2:variance2: intermediate_2a0} by
		\begin{equation}
			\begin{aligned}
				\abs{\frac{\partial}{\partial X_{ij}}\tr{\boldsymbol{B}\hat{\boldsymbol{Q}}(z_{1})\boldsymbol{R}\hat{\boldsymbol{Q}}(z_{2})\boldsymbol{R}^{1/2}\frac{\boldsymbol{X}\boldsymbol{X}^{\operatorname{H}}}{N}\phi_{M}}}^{2} \leq & 3\abs{\sum_{m=1}^{M} \boldsymbol{e}_{m}^{\operatorname{T}} \boldsymbol{B} \hat{\boldsymbol{Q}}(z_{1}) \boldsymbol{R}^{1/2}\frac{\boldsymbol{e}_{i}\boldsymbol{x}_{j}^{\operatorname{H}}}{N} \boldsymbol{R}^{\operatorname{H}/2} \hat{\boldsymbol{Q}}(z_{1}) \boldsymbol{R} \hat{\boldsymbol{Q}}(z_{2}) \frac{\boldsymbol{X}\boldsymbol{X}^{\operatorname{H}}}{N}\boldsymbol{e}_{m}\phi_{M}}^{2}\\
				&+3 \abs{\sum_{m=1}^{M} \boldsymbol{e}_{m}^{\operatorname{T}} \boldsymbol{B} \hat{\boldsymbol{Q}}(z_{1}) \boldsymbol{R} \hat{\boldsymbol{Q}}(z_{2}) \boldsymbol{R}^{1/2}\frac{\boldsymbol{e}_{i}\boldsymbol{x}_{j}^{\operatorname{H}}}{N} \boldsymbol{R}^{\operatorname{H}/2} \hat{\boldsymbol{Q}}(z_{2}) \frac{\boldsymbol{X}\boldsymbol{X}^{\operatorname{H}}}{N}\boldsymbol{e}_{m}\phi_{M}}^{2}\\
				&+ 3 \abs{\sum_{m=1}^{M} \boldsymbol{e}_{m}^{\operatorname{T}} \boldsymbol{B}\hat{\boldsymbol{Q}}(z_{1})\boldsymbol{R}\hat{\boldsymbol{Q}}(z_{2}) \boldsymbol{R}^{1/2} \frac{\boldsymbol{e}_{i}\boldsymbol{x}_{j}^{\operatorname{H}}}{N}\boldsymbol{e}_{m}\phi_{M}}^{2} + \mathcal{O}\left(N^{-\mathbb{N}}\right)
			\end{aligned}
			\label{eq:proof:aux:lemma2:variance2: intermediate_2a1}
		\end{equation}
		as well as the absolute value squared of the first order derivative in \eqref{eq:proof:aux:lemma2:variance2: intermediate_2b0}
		\begin{equation}
			\begin{aligned}
				\abs{\frac{\partial}{\partial X_{ij}^{*}}\tr{\boldsymbol{B}\hat{\boldsymbol{Q}}(z_{1})\boldsymbol{R}\hat{\boldsymbol{Q}}(z_{2})\boldsymbol{R}^{1/2}\frac{\boldsymbol{X}\boldsymbol{X}^{\operatorname{H}}}{N}\phi_{M}}}^{2} \leq & 3 \abs{\sum_{m=1}^{M} \boldsymbol{e}_{m}^{\operatorname{T}} \boldsymbol{B} \hat{\boldsymbol{Q}}(z_{1}) \boldsymbol{R}^{1/2}\frac{\boldsymbol{x}_{j}\boldsymbol{e}_{i}^{\operatorname{T}}}{N} \boldsymbol{R}^{\operatorname{H}/2} \hat{\boldsymbol{Q}}(z_{1}) \boldsymbol{R} \hat{\boldsymbol{Q}}(z_{2}) \frac{\boldsymbol{X}\boldsymbol{X}^{\operatorname{H}}}{N}\boldsymbol{e}_{m}\phi_{M}}^{2}\\
				&+3 \abs{\sum_{m=1}^{M} \boldsymbol{e}_{m}^{\operatorname{T}} \boldsymbol{B} \hat{\boldsymbol{Q}}(z_{1}) \boldsymbol{R} \hat{\boldsymbol{Q}}(z_{2}) \boldsymbol{R}^{1/2}\frac{\boldsymbol{x}_{j}\boldsymbol{e}_{i}^{\operatorname{T}}}{N} \boldsymbol{R}^{\operatorname{H}/2} \hat{\boldsymbol{Q}}(z_{2}) \frac{\boldsymbol{X}\boldsymbol{X}^{\operatorname{H}}}{N}\boldsymbol{e}_{m}\phi_{M}}^{2}\\
				&+3 \abs{ \sum_{m=1}^{M} \boldsymbol{e}_{m}^{\operatorname{T}} \boldsymbol{B}\hat{\boldsymbol{Q}}(z_{1})\boldsymbol{R}\hat{\boldsymbol{Q}}(z_{2}) \boldsymbol{R}^{1/2} \frac{\boldsymbol{x}_{j}\boldsymbol{e}_{i}^{\operatorname{T}}}{N}\boldsymbol{e}_{m}\phi_{M}}^{2} + \mathcal{O}\left(N^{-\mathbb{N}}\right).
			\end{aligned}
			\label{eq:proof:aux:lemma2:variance2: intermediate_2b1}
		\end{equation}
		From \eqref{eq:proof:aux:lemma2:variance2: intermediate_2a1} and \eqref{eq:proof:aux:lemma2:variance2: intermediate_2b1} it follows that
		\begin{align*}
		 	\sum_{i=1}^{M}\sum_{j=1}^{N}	\mean{\abs{\frac{\partial}{\partial X_{ij}}\tr{\boldsymbol{B}\hat{\boldsymbol{Q}}(z_{1})\boldsymbol{R}\hat{\boldsymbol{Q}}(z_{2})\boldsymbol{R}^{1/2}\frac{\boldsymbol{X}\boldsymbol{X}^{\operatorname{H}}}{N}\phi_{M}}}^{2}} =& \mathcal{O}\left(\frac{M}{N}\right), \\
		 	\sum_{i=1}^{M}\sum_{j=1}^{N}	\mean{\abs{\frac{\partial}{\partial X_{ij}^{*}}\tr{\boldsymbol{B}\hat{\boldsymbol{Q}}(z_{1})\boldsymbol{R}\hat{\boldsymbol{Q}}(z_{2})\boldsymbol{R}^{1/2}\frac{\boldsymbol{X}\boldsymbol{X}^{\operatorname{H}}}{N}\phi_{M}}}^{2}} =&  \mathcal{O}\left(\frac{M}{N}\right),
		\end{align*}
		and consequently the variance in \eqref{eq:proof:aux:lemma2:variance2: intermediate_1} $\var{\tr{\boldsymbol{B}\hat{\boldsymbol{Q}}(z_{1})\boldsymbol{R}\hat{\boldsymbol{Q}}(z_{2})\boldsymbol{R}^{1/2}\frac{\boldsymbol{X}\boldsymbol{X}^{\operatorname{H}}}{N}\phi_{M}}} = \mathcal{O}(M/N)$. Hence the variance in \eqref{eq:aux:lemma2: variance2} decays with
		\begin{equation*}
			\var{\frac{1}{N}\tr{\boldsymbol{B}\hat{\boldsymbol{Q}}(z_{1})\boldsymbol{R}\hat{\boldsymbol{Q}}(z_{2})\boldsymbol{R}^{1/2}\frac{\boldsymbol{X}\boldsymbol{X}^{\operatorname{H}}}{N}\phi_{M}}} = \frac{1}{N^{2}}\var{\tr{\boldsymbol{B}\hat{\boldsymbol{Q}}(z_{1})\boldsymbol{R}\hat{\boldsymbol{Q}}(z_{2})\boldsymbol{R}^{1/2}\frac{\boldsymbol{X}\boldsymbol{X}^{\operatorname{H}}}{N}\phi_{M}}} \leq \frac{1}{N^{2}} \mathcal{O}\left(\frac{M}{N}\right) = \mathcal{O}\left(N^{-2}\right).
		\end{equation*}

	\subsubsection{Expectations}	
		In the following we derive the two expectations in \eqref{eq:aux:lemma2: expectation1} and \eqref{eq:aux:lemma2: expectation2}. Using the resolvent identity $z_{2}\hat{\boldsymbol{Q}}(z_{2}) = \hat{\boldsymbol{Q}}(z_{2})\hat{\boldsymbol{R}} - \boldsymbol{I}_{M}$ in \eqref{eq: resolvent_identity} we can express the expectation in \eqref{eq:aux:lemma2: expectation1} as
		\begin{equation}
			\frac{1}{N}\mean{\tr{\boldsymbol{B}\hat{\boldsymbol{Q}}(z_{1})\boldsymbol{R}\hat{\boldsymbol{Q}}(z_{2})\phi_{M}}} = z_{2}^{-1} \frac{1}{N} \mean{\tr{\boldsymbol{B}\hat{\boldsymbol{Q}}(z_{1})\boldsymbol{R}\hat{\boldsymbol{Q}}(z_{2})\hat{\boldsymbol{R}}\phi_{M}}} - z_{2}^{-1}\frac{1}{N} \mean{\tr{\boldsymbol{B}\hat{\boldsymbol{Q}}(z_{1})\boldsymbol{R}\phi_{M}}}.
			\label{eq:proof:aux:lemma2:expectation1: intermediate_1}
		\end{equation}
		Next, we develop the first term on the right hand side of \eqref{eq:proof:aux:lemma2:expectation1: intermediate_1} by rewriting the trace as sum and the sample covariance $\hat{\boldsymbol{R}}$ as shown in \eqref{eq: sample_covariance_matrix_sum} and apply the integration by parts formula
		\begin{equation}
			\begin{aligned}
				\frac{1}{N} \mean{\tr{\boldsymbol{B}\hat{\boldsymbol{Q}}(z_{1})\boldsymbol{R}\hat{\boldsymbol{Q}}(z_{2})\hat{\boldsymbol{R}}\phi_{M}}} =&\frac{1}{N} \sum_{i=1}^{M}\sum_{j=1}^{N}\sum_{m=1}^{M} \mean{ X_{ij} \boldsymbol{e}_{m}^{\operatorname{T}} \boldsymbol{B} \hat{\boldsymbol{Q}}(z_{1}) \boldsymbol{R} \hat{\boldsymbol{Q}}(z_{2}) \boldsymbol{R}^{1/2} \frac{\boldsymbol{e}_{i}\boldsymbol{x}_{j}^{\operatorname{H}}}{N}\boldsymbol{R}^{\operatorname{H}/2}\boldsymbol{e}_{m}\phi_{M}}\\
				=&\frac{1}{N} \sum_{i=1}^{M}\sum_{j=1}^{N}\sum_{m=1}^{M} \mean{\frac{\partial}{\partial X_{ij}^{*}} \left(\boldsymbol{e}_{m}^{\operatorname{T}} \boldsymbol{B} \hat{\boldsymbol{Q}}(z_{1}) \boldsymbol{R} \hat{\boldsymbol{Q}}(z_{2}) \boldsymbol{R}^{1/2} \frac{\boldsymbol{e}_{i}\boldsymbol{x}_{j}^{\operatorname{H}}}{N}\boldsymbol{R}^{\operatorname{H}/2}\boldsymbol{e}_{m}\phi_{M}\right)},
			\end{aligned}
			\label{eq:proof:aux:lemma2:expectation1: intermediate_2}
		\end{equation}
		where the first order derivative with respect to $X_{ij}^{*}$ is given by
		\begin{equation}
			\begin{aligned}
				\frac{\partial}{\partial X_{ij}^{*}}\left( \boldsymbol{e}_{m}^{\operatorname{T}} \boldsymbol{B} \hat{\boldsymbol{Q}}(z_{1}) \boldsymbol{R} \hat{\boldsymbol{Q}}(z_{2}) \boldsymbol{R}^{1/2} \frac{\boldsymbol{e}_{i}\boldsymbol{x}_{j}^{\operatorname{H}}}{N}\boldsymbol{R}^{\operatorname{H}/2}\boldsymbol{e}_{m}\phi_{M}\right) =& \boldsymbol{e}_{m}^{\operatorname{T}} \boldsymbol{B} \hat{\boldsymbol{Q}}(z_{1}) \boldsymbol{R} \hat{\boldsymbol{Q}}(z_{2}) \boldsymbol{R}^{1/2} \frac{\boldsymbol{e}_{i}\boldsymbol{e}_{i}^{\operatorname{T}}}{N}\boldsymbol{R}^{\operatorname{H}/2}\boldsymbol{e}_{m}\phi_{M}  + \mathcal{O}\left(N^{-\mathbb{N}}\right)\\
				&- \boldsymbol{e}_{m}^{\operatorname{T}} \boldsymbol{B} \hat{\boldsymbol{Q}}(z_{1}) \boldsymbol{R}^{1/2} \frac{\boldsymbol{x}_{j}\boldsymbol{e}_{i}^{\operatorname{T}}}{N}\boldsymbol{R}^{\operatorname{H}/2} \hat{\boldsymbol{Q}}(z_{1})  \boldsymbol{R} \hat{\boldsymbol{Q}}(z_{2}) \boldsymbol{R}^{1/2} \frac{\boldsymbol{e}_{i}\boldsymbol{x}_{j}^{\operatorname{H}}}{N}\boldsymbol{R}^{\operatorname{H}/2}\boldsymbol{e}_{m}\phi_{M} \\
				&- \boldsymbol{e}_{m}^{\operatorname{T}} \boldsymbol{B} \hat{\boldsymbol{Q}}(z_{1}) \boldsymbol{R} \hat{\boldsymbol{Q}}(z_{2}) \boldsymbol{R}^{1/2} \frac{\boldsymbol{x}_{j}\boldsymbol{e}_{i}^{\operatorname{T}}}{N} \boldsymbol{R}^{\operatorname{H}/2} \hat{\boldsymbol{Q}}(z_{2}) \boldsymbol{R}^{1/2} \frac{\boldsymbol{e}_{i}\boldsymbol{x}_{j}^{\operatorname{H}}}{N}\boldsymbol{R}^{\operatorname{H}/2}\boldsymbol{e}_{m}\phi_{M}.
			\end{aligned}
			\label{eq:proof:aux:lemma2:expectation1: intermediate_3}
		\end{equation}
		Substituting \eqref{eq:proof:aux:lemma2:expectation1: intermediate_3} into \eqref{eq:proof:aux:lemma2:expectation1: intermediate_2} and computing the triple sum we can equivalently express the first term on the right hand side of \eqref{eq:proof:aux:lemma2:expectation1: intermediate_1} as
		\begin{equation}
			\begin{aligned}
				\frac{1}{N} \mean{\tr{\boldsymbol{B}\hat{\boldsymbol{Q}}(z_{1})\boldsymbol{R}\hat{\boldsymbol{Q}}(z_{2})\hat{\boldsymbol{R}}\phi_{M}}}  =&  \frac{1}{N}\mean{\tr{\boldsymbol{B} \hat{\boldsymbol{Q}}(z_{1}) \boldsymbol{R}\hat{\boldsymbol{Q}}(z_{2}) \boldsymbol{R}\phi_{M}}} -\frac{1}{N} \mean{\tr{ \boldsymbol{B} \hat{\boldsymbol{Q}}(z_{1}) \hat{\boldsymbol{R}}}\frac{1}{N}\tr{\boldsymbol{R} \hat{\boldsymbol{Q}}(z_{1})\boldsymbol{R} \hat{\boldsymbol{Q}}(z_{2}) }\phi_{M}}\\
				&-\frac{1}{N} \mean{\tr{ \boldsymbol{B} \hat{\boldsymbol{Q}}(z_{1}) \boldsymbol{R} \hat{\boldsymbol{Q}}(z_{2}) \hat{\boldsymbol{R}} } \frac{1}{N}\tr{\boldsymbol{R} \hat{\boldsymbol{Q}}(z_{2})}\phi_{M}} + \mathcal{O}\left(N^{-\mathbb{N}}\right).
			\end{aligned}
			\label{eq:proof:aux:lemma2:expectation1: intermediate_4}
		\end{equation}
		Inserting \eqref{eq:proof:aux:lemma2:expectation1: intermediate_4} into \eqref{eq:proof:aux:lemma2:expectation1: intermediate_1} yields
		\begin{equation}
			\begin{aligned}
				\frac{1}{N}\mean{\tr{\boldsymbol{B}\hat{\boldsymbol{Q}}(z_{1})\boldsymbol{R}\hat{\boldsymbol{Q}}(z_{2})\phi_{M}}} =& z_{2}^{-1}\frac{1}{N}\mean{\tr{\boldsymbol{B} \hat{\boldsymbol{Q}}(z_{1}) \boldsymbol{R}\hat{\boldsymbol{Q}}(z_{2}) \boldsymbol{R}\phi_{M}}} - z_{2}^{-1}\frac{1}{N} \mean{\tr{ \boldsymbol{B} \hat{\boldsymbol{Q}}(z_{1}) \hat{\boldsymbol{R}}}\frac{1}{N}\tr{\boldsymbol{R} \hat{\boldsymbol{Q}}(z_{1})\boldsymbol{R} \hat{\boldsymbol{Q}}(z_{2}) }\phi_{M}}\\
				&- z_{2}^{-1}\frac{1}{N} \mean{\tr{ \boldsymbol{B} \hat{\boldsymbol{Q}}(z_{1}) \boldsymbol{R} \hat{\boldsymbol{Q}}(z_{2}) \hat{\boldsymbol{R}} } \frac{1}{N}\tr{\boldsymbol{R} \hat{\boldsymbol{Q}}(z_{2})}\phi_{M}} - z_{2}^{-1}\frac{1}{N} \mean{\tr{\boldsymbol{B}\hat{\boldsymbol{Q}}(z_{1})\boldsymbol{R}\phi_{M}}}+ \mathcal{O}\left(N^{-\mathbb{N}}\right).
			\end{aligned}
			\label{eq:proof:aux:lemma2:expectation1: intermediate_5}
		\end{equation}
		Next, we need to decorrelate the second and third term on the right hand side of \eqref{eq:proof:aux:lemma2:expectation1: intermediate_5}. We begin with the second term on the right hand side of \eqref{eq:proof:aux:lemma2:expectation1: intermediate_5} by analyzing the following covariance
		\begin{equation}
			\begin{aligned}
				\cov{\frac{1}{N}\tr{\boldsymbol{B}\hat{\boldsymbol{Q}}(z_{1})\hat{\boldsymbol{R}}\phi_{M}}}{\frac{1}{N}\tr{\boldsymbol{R}\hat{\boldsymbol{Q}}(z_{1})\boldsymbol{R}\hat{\boldsymbol{Q}}(z_{2})\phi_{M}}}=& \mean{\left(\frac{1}{N}\tr{\boldsymbol{B}\hat{\boldsymbol{Q}}(z_{1})\hat{\boldsymbol{R}}\phi_{M}}\right)^{(\circ)}\left(\frac{1}{N}\tr{\boldsymbol{R}\hat{\boldsymbol{Q}}(z_{1})\boldsymbol{R}\hat{\boldsymbol{Q}}(z_{2})\phi_{M}}\right)^{(\circ)}}\\
				=& \mean{\frac{1}{N}\tr{\boldsymbol{B}\hat{\boldsymbol{Q}}(z_{1})\hat{\boldsymbol{R}}}\frac{1}{N}\tr{\boldsymbol{R}\hat{\boldsymbol{Q}}(z_{1})\boldsymbol{R}\hat{\boldsymbol{Q}}(z_{2})}\phi_{M}} \\
				&- \mean{\frac{1}{N}\tr{\boldsymbol{B}\hat{\boldsymbol{Q}}(z_{1})\hat{\boldsymbol{R}}\phi_{M}}}\mean{\frac{1}{N}\tr{\boldsymbol{R}\hat{\boldsymbol{Q}}(z_{1})\boldsymbol{R}\hat{\boldsymbol{Q}}(z_{2})\phi_{M}}}.
			\end{aligned}
			\label{eq:proof:aux:lemma2:expectation1: intermediate_6}
		\end{equation}
		Applying the Cauchy-Schwarz inequality to \eqref{eq:proof:aux:lemma2:expectation1: intermediate_6} yields
		\begin{equation}
			\begin{aligned}
				&\abs{\cov{\frac{1}{N}\tr{\boldsymbol{B}\hat{\boldsymbol{Q}}(z_{1})\hat{\boldsymbol{R}}\phi_{M}}}{\frac{1}{N}\tr{\boldsymbol{R}\hat{\boldsymbol{Q}}(z_{1})\boldsymbol{R}\hat{\boldsymbol{Q}}(z_{2})\phi_{M}}}}^{2}\\ 
				\leq & \mean{\abs{\left(\frac{1}{N}\tr{\boldsymbol{B}\hat{\boldsymbol{Q}}(z_{1})\hat{\boldsymbol{R}}\phi_{M}}\right)^{(\circ)}}^{2}}\mean{\abs{\left(\frac{1}{N}\tr{\boldsymbol{R}\hat{\boldsymbol{Q}}(z_{1})\boldsymbol{R}\hat{\boldsymbol{Q}}(z_{2})\phi_{M}}\right)^{(\circ)}}^{2}}\\
				=& \var{\frac{1}{N}\tr{\boldsymbol{B}\hat{\boldsymbol{Q}}(z_{1})\hat{\boldsymbol{R}}\phi_{M}}}\var{\frac{1}{N}\tr{\boldsymbol{R}\hat{\boldsymbol{Q}}(z_{1})\boldsymbol{R}\hat{\boldsymbol{Q}}(z_{2})\phi_{M}}} = \mathcal{O}\left(N^{-4}\right), 
			\end{aligned}
			\label{eq:proof:aux:lemma2:expectation1: intermediate_7}
		\end{equation}
		where we have used \eqref{eq:aux:lemma1: variance2} in Lemma \ref{aux:lemma1} and \eqref{eq:aux:lemma2: variance1} in Lemma \ref{aux:lemma2} to obtain
		\begin{equation*}
			\var{\frac{1}{N}\tr{\boldsymbol{B}\hat{\boldsymbol{Q}}(z_{1})\hat{\boldsymbol{R}}\phi_{M}}} = \mathcal{O}\left(N^{-2}\right), \qquad \var{\frac{1}{N}\tr{\boldsymbol{R}\hat{\boldsymbol{Q}}(z_{1})\boldsymbol{R}\hat{\boldsymbol{Q}}(z_{2})\phi_{M}}} = \mathcal{O}\left(N^{-2}\right).
		\end{equation*}
		It follows from \eqref{eq:proof:aux:lemma2:expectation1: intermediate_7} that the covariance in \eqref{eq:proof:aux:lemma2:expectation1: intermediate_6} decays with
		\begin{equation}
			\cov{\frac{1}{N}\tr{\boldsymbol{B}\hat{\boldsymbol{Q}}(z_{1})\hat{\boldsymbol{R}}\phi_{M}}}{\frac{1}{N}\tr{\boldsymbol{R}\hat{\boldsymbol{Q}}(z_{1})\boldsymbol{R}\hat{\boldsymbol{Q}}(z_{2})\phi_{M}}} = \mathcal{O}\left(N^{-2}\right).
		\end{equation}
		Hence, we can decorrelate the second quantity on the right hand side of \eqref{eq:proof:aux:lemma2:expectation1: intermediate_5} as		
		\begin{equation}
			\begin{aligned}
				\mean{\frac{1}{N}\tr{\boldsymbol{B}\hat{\boldsymbol{Q}}(z_{1})\hat{\boldsymbol{R}}}\frac{1}{N}\tr{\boldsymbol{R}\hat{\boldsymbol{Q}}(z_{1})\boldsymbol{R}\hat{\boldsymbol{Q}}(z_{2})}\phi_{M}} =&\frac{1}{N}\mean{\tr{\boldsymbol{B}\hat{\boldsymbol{Q}}(z_{1})\hat{\boldsymbol{R}}\phi_{M}}} \frac{1}{N}\mean{\tr{\boldsymbol{R}\hat{\boldsymbol{Q}}(z_{1})\boldsymbol{R}\hat{\boldsymbol{Q}}(z_{2})\phi_{M}}} \\
				&+ \cov{\frac{1}{N}\tr{\boldsymbol{B}\hat{\boldsymbol{Q}}(z_{1})\hat{\boldsymbol{R}}\phi_{M}}}{\frac{1}{N}\tr{\boldsymbol{R}\hat{\boldsymbol{Q}}(z_{1})\boldsymbol{R}\hat{\boldsymbol{Q}}(z_{2})\phi_{M}}}\\
				=& \frac{1}{N}\mean{\tr{\boldsymbol{B}\hat{\boldsymbol{Q}}(z_{1})\hat{\boldsymbol{R}\phi_{M}}}} \frac{1}{N}\mean{\tr{\boldsymbol{R}\hat{\boldsymbol{Q}}(z_{1})\boldsymbol{R}\hat{\boldsymbol{Q}}(z_{2})\phi_{M}}} + \mathcal{O}\left(N^{-2}\right).
			\end{aligned}
			\label{eq:proof:aux:lemma2:expectation1: intermediate_8}
		\end{equation}
		Let us now decorrelate the third quantity on the right hand side of \eqref{eq:proof:aux:lemma2:expectation1: intermediate_5} by applying the same technique. We start by analyzing the following covariance
		\begin{equation}
			\begin{aligned}
				\cov{\frac{1}{N}\tr{\boldsymbol{B}\hat{\boldsymbol{Q}}(z_{1})\boldsymbol{R}\hat{\boldsymbol{Q}}(z_{2})\hat{\boldsymbol{R}}\phi_{M}}}{\frac{1}{N}\tr{\boldsymbol{R}\hat{\boldsymbol{Q}}(z_{2})\phi_{M}}}=&\mean{\left(\frac{1}{N}\tr{\boldsymbol{B}\hat{\boldsymbol{Q}}(z_{1})\boldsymbol{R}\hat{\boldsymbol{Q}}(z_{2})\hat{\boldsymbol{R}}\phi_{M}}\right)^{(\circ)}\left(\frac{1}{N}\tr{\boldsymbol{R}\hat{\boldsymbol{Q}}(z_{2})\phi_{M}}\right)^{(\circ)}}\\
				=& \mean{\frac{1}{N}\tr{\boldsymbol{B}\hat{\boldsymbol{Q}}(z_{1})\boldsymbol{R}\hat{\boldsymbol{Q}}(z_{2})\hat{\boldsymbol{R}}}\frac{1}{N}\tr{\boldsymbol{R}\hat{\boldsymbol{Q}}(z_{2})}\phi_{M}}\\
				&- \mean{\frac{1}{N}\tr{\boldsymbol{B}\hat{\boldsymbol{Q}}(z_{1})\boldsymbol{R}\hat{\boldsymbol{Q}}(z_{2})\hat{\boldsymbol{R}}\phi_{M}}}\mean{\frac{1}{N}\tr{\boldsymbol{R}\hat{\boldsymbol{Q}}(z_{2})\phi_{M}}}.
			\end{aligned}
			\label{eq:proof:aux:lemma2:expectation1: intermediate_9}
		\end{equation} 
		Applying the Cauchy Schwarz inequality to the covariance in \eqref{eq:proof:aux:lemma2:expectation1: intermediate_9} yields
		\begin{equation}
			\begin{aligned}
				&\abs{\cov{\frac{1}{N}\tr{\boldsymbol{B}\hat{\boldsymbol{Q}}(z_{1})\boldsymbol{R}\hat{\boldsymbol{Q}}(z_{2})\hat{\boldsymbol{R}}\phi_{M}}}{\frac{1}{N}\tr{\boldsymbol{R}\hat{\boldsymbol{Q}}(z_{2})\phi_{M}}}}^{2} \\
				\leq & \mean{\abs{\left(\frac{1}{N}\tr{\boldsymbol{B}\hat{\boldsymbol{Q}}(z_{1})\boldsymbol{R}\hat{\boldsymbol{Q}}(z_{2})\hat{\boldsymbol{R}}\phi_{M}}\right)^{(\circ)}}^{2}}\mean{\abs{\left(\frac{1}{N}\tr{\boldsymbol{R}\hat{\boldsymbol{Q}}(z_{2})\phi_{M}}\right)^{(\circ)}}^{2}}\\
				=& \var{\frac{1}{N}\tr{\boldsymbol{B}\hat{\boldsymbol{Q}}(z_{1})\boldsymbol{R}\hat{\boldsymbol{Q}}(z_{2})\hat{\boldsymbol{R}}\phi_{M}}} \var{\frac{1}{N}\tr{\boldsymbol{R}\hat{\boldsymbol{Q}}(z_{2})\phi_{M}}} = \mathcal{O}\left(N^{-4}\right),
			\end{aligned}
			\label{eq:proof:aux:lemma2:expectation1: intermediate_10}
		\end{equation}
		where we have used \eqref{eq:aux:lemma2: variance2} in Lemma \ref{aux:lemma2} and \eqref{eq:aux:lemma1: variance1} in Lemma \ref{aux:lemma1} to obtain 
		\begin{equation*}
			\var{\frac{1}{N}\tr{\boldsymbol{B}\hat{\boldsymbol{Q}}(z_{1})\boldsymbol{R}\hat{\boldsymbol{Q}}(z_{2})\hat{\boldsymbol{R}}\phi_{M}}} = \mathcal{O}\left(N^{-2}\right), \qquad \var{\frac{1}{N}\tr{\boldsymbol{R}\hat{\boldsymbol{Q}}(z_{2})\phi_{M}}} = \mathcal{O}\left(N^{-2}\right).
		\end{equation*}
		From \eqref{eq:proof:aux:lemma2:expectation1: intermediate_10} it follows that the covariance in \eqref{eq:proof:aux:lemma2:expectation1: intermediate_9} decays with
		\begin{equation*}
			\cov{\frac{1}{N}\tr{\boldsymbol{B}\hat{\boldsymbol{Q}}(z_{1})\boldsymbol{R}\hat{\boldsymbol{Q}}(z_{2})\hat{\boldsymbol{R}}\phi_{M}}}{\frac{1}{N}\tr{\boldsymbol{R}\hat{\boldsymbol{Q}}(z_{2})\phi_{M}}} = \mathcal{O}(N^{-2}).
		\end{equation*}
		Using \eqref{eq:proof:aux:lemma2:expectation1: intermediate_9} we can decorrelated the third term on the right hand side of \eqref{eq:proof:aux:lemma2:expectation1: intermediate_5} as
		\begin{equation}
			\begin{aligned}
				\mean{\frac{1}{N}\tr{\boldsymbol{B}\hat{\boldsymbol{Q}}(z_{1})\boldsymbol{R}\hat{\boldsymbol{Q}}(z_{2})\hat{\boldsymbol{R}}}\frac{1}{N}\tr{\boldsymbol{R}\hat{\boldsymbol{Q}}(z_{2})}\phi_{M}}=& \mean{\frac{1}{N}\tr{\boldsymbol{B}\hat{\boldsymbol{Q}}(z_{1})\boldsymbol{R}\hat{\boldsymbol{Q}}(z_{2})\hat{\boldsymbol{R}}\phi_{M}}}\mean{\frac{1}{N}\tr{\boldsymbol{R}\hat{\boldsymbol{Q}}(z_{2})\phi_{M}}} \\
				 &+ \cov{\frac{1}{N}\tr{\boldsymbol{B}\hat{\boldsymbol{Q}}(z_{1})\boldsymbol{R}\hat{\boldsymbol{Q}}(z_{2})\hat{\boldsymbol{R}}\phi_{M}}}{\frac{1}{N}\tr{\boldsymbol{R}\hat{\boldsymbol{Q}}(z_{2})\phi_{M}}}\\
				 =&\mean{\frac{1}{N}\tr{\boldsymbol{B}\hat{\boldsymbol{Q}}(z_{1})\boldsymbol{R}\hat{\boldsymbol{Q}}(z_{2})\hat{\boldsymbol{R}}\phi_{M}}}\mean{\frac{1}{N}\tr{\boldsymbol{R}\hat{\boldsymbol{Q}}(z_{2})\phi_{M}}} + \mathcal{O}\left(N^{-2}\right).
			\end{aligned}
			\label{eq:proof:aux:lemma2:expectation1: intermediate_11}
		\end{equation}
		Substituting the two decorrelated terms in \eqref{eq:proof:aux:lemma2:expectation1: intermediate_8} and \eqref{eq:proof:aux:lemma2:expectation1: intermediate_11} into \eqref{eq:proof:aux:lemma2:expectation1: intermediate_5} yields
		\begin{equation}
			\begin{aligned}
				\frac{1}{N}\mean{\tr{\boldsymbol{B}\hat{\boldsymbol{Q}}(z_{1})\boldsymbol{R}\hat{\boldsymbol{Q}}(z_{2})\phi_{M}}}=& z_{2}^{-1}\frac{1}{N}\mean{\tr{\boldsymbol{B} \hat{\boldsymbol{Q}}(z_{1}) \boldsymbol{R}\hat{\boldsymbol{Q}}(z_{2}) \boldsymbol{R}\phi_{M}}}\\				
				&- z_{2}^{-1}\frac{1}{N}\mean{\tr{\boldsymbol{B}\hat{\boldsymbol{Q}}(z_{1})\hat{\boldsymbol{R}}\phi_{M}}} \frac{1}{N}\mean{\tr{\boldsymbol{R}\hat{\boldsymbol{Q}}(z_{1})\boldsymbol{R}\hat{\boldsymbol{Q}}(z_{2})\phi_{M}}} \\
				&- z_{2}^{-1}\frac{1}{N}\mean{\tr{\boldsymbol{B}\hat{\boldsymbol{Q}}(z_{1})\boldsymbol{R}\hat{\boldsymbol{Q}}(z_{2})\hat{\boldsymbol{R}}\phi_{M}}}\frac{1}{N}\mean{\tr{\boldsymbol{R}\hat{\boldsymbol{Q}}(z_{2})\phi_{M}}}\\
				&- z_{2}^{-1}\frac{1}{N} \mean{\tr{\boldsymbol{B}\hat{\boldsymbol{Q}}(z_{1})\boldsymbol{R}\phi_{M}}}+ \mathcal{O}\left(N^{-2}\right).
			\end{aligned}
			\label{eq:proof:aux:lemma2:expectation1: intermediate_12}
		\end{equation}
		Using the previously derived expectation in \eqref{eq:aux:lemma1: expectation1} in Lemma \ref{aux:lemma1} we can equivalently express \eqref{eq:proof:aux:lemma2:expectation1: intermediate_12} as
		\begin{equation}
			\begin{aligned}
				\frac{1}{N}\mean{\tr{\boldsymbol{B}\hat{\boldsymbol{Q}}(z_{1})\boldsymbol{R}\hat{\boldsymbol{Q}}(z_{2})\phi_{M}}} =& z_{2}^{-1}\frac{1}{N}\mean{\tr{\boldsymbol{B} \hat{\boldsymbol{Q}}(z_{1}) \boldsymbol{R}\hat{\boldsymbol{Q}}(z_{2}) \boldsymbol{R}\phi_{M}}}\\
				&-z_{2}^{-1} \frac{z_{1}}{\omega(z_{1})}\frac{1}{N} \tr{\boldsymbol{B}\bar{\boldsymbol{Q}}(z_{1})\boldsymbol{R}} \frac{1}{N}\mean{\tr{\boldsymbol{R}\hat{\boldsymbol{Q}}(z_{1})\boldsymbol{R}\hat{\boldsymbol{Q}}(z_{2})\phi_{M}}}\\
				&-z_{2}^{-1}\frac{1}{N}\mean{\tr{\boldsymbol{B}\hat{\boldsymbol{Q}}(z_{1})\boldsymbol{R}\hat{\boldsymbol{Q}}(z_{2})\hat{\boldsymbol{R}}\phi_{M}}}\frac{1}{N}\tr{\boldsymbol{R}\bar{\boldsymbol{Q}}(z_{2})}\\
				&-z_{2}^{-1} \frac{1}{N} \tr{\boldsymbol{B}\bar{\boldsymbol{Q}}(z_{1})\boldsymbol{R}} + \mathcal{O}\left(N^{-1}\right).
			\end{aligned}
			\label{eq:proof:aux:lemma2:expectation1: intermediate_13}
		\end{equation}
		
		In order to derive the second expectation in \eqref{eq:aux:lemma2: expectation2} we express the sample covariance matrix $\hat{\boldsymbol{R}}$ as shown in \eqref{eq: sample_covariance_matrix_sum} and rewrite the trace as sum to apply the integration by parts formula
		\begin{equation}
			\begin{aligned}
				\frac{1}{N}\mean{\tr{\boldsymbol{B}\hat{\boldsymbol{Q}}(z_{1})\boldsymbol{R}\hat{\boldsymbol{Q}}(z_{2})\boldsymbol{R}^{1/2}\frac{\boldsymbol{X}\boldsymbol{X}^{\operatorname{H}}}{N}\phi_{M}}} =& \frac{1}{N}\sum_{i=1}^{M}\sum_{j=1}^{N}\sum_{m=1}^{M}\mean{X_{ij}\boldsymbol{e}_{m}^{\operatorname{T}}\boldsymbol{B}\hat{\boldsymbol{Q}}(z_{1})\boldsymbol{R}\hat{\boldsymbol{Q}}(z_{2})\boldsymbol{R}^{1/2}\frac{\boldsymbol{e}_{i}\boldsymbol{x}_{j}^{\operatorname{H}}}{N}\boldsymbol{e}_{m}\phi_{M}}\\
				=&\frac{1}{N}\sum_{i=1}^{M}\sum_{j=1}^{N}\sum_{m=1}^{M}\mean{\frac{\partial}{\partial X_{ij}^{*}}\left( \boldsymbol{e}_{m}^{\operatorname{T}}\boldsymbol{B}\hat{\boldsymbol{Q}}(z_{1})\boldsymbol{R}\hat{\boldsymbol{Q}}(z_{2})\boldsymbol{R}^{1/2}\frac{\boldsymbol{e}_{i}\boldsymbol{x}_{j}^{\operatorname{H}}}{N}\boldsymbol{e}_{m} \phi_{M} \right)},
			\end{aligned}
			\label{eq:proof:aux:lemma2:expectation2: intermediate_1}
		\end{equation}
		where the first order derivative with respect to $X_{ij}^{*}$ is given by
		\begin{equation}
			\begin{aligned}
				\frac{\partial}{\partial X_{ij}^{*}}\left( \boldsymbol{e}_{m}^{\operatorname{T}}\boldsymbol{B}\hat{\boldsymbol{Q}}(z_{1})\boldsymbol{R}\hat{\boldsymbol{Q}}(z_{2})\boldsymbol{R}^{1/2}\frac{\boldsymbol{e}_{i}\boldsymbol{x}_{j}^{\operatorname{H}}}{N}\boldsymbol{e}_{m}\phi_{M} \right) =&- \boldsymbol{e}_{m}^{\operatorname{T}}\boldsymbol{B} \hat{\boldsymbol{Q}}(z_{1}) \boldsymbol{R}^{1/2} \frac{\boldsymbol{x}_{j}\boldsymbol{e}_{i}^{\operatorname{T}}}{N}\boldsymbol{R}^{\operatorname{H}/2}\hat{\boldsymbol{Q}}(z_{1})\boldsymbol{R}\hat{\boldsymbol{Q}}(z_{2})\boldsymbol{R}^{1/2}\frac{\boldsymbol{e}_{i}\boldsymbol{x}_{j}^{\operatorname{H}}}{N}\boldsymbol{e}_{m}\phi_{M}\\
			&-\boldsymbol{e}_{m}^{\operatorname{T}} \boldsymbol{B}\hat{\boldsymbol{Q}}(z_{1})\boldsymbol{R} \hat{\boldsymbol{Q}}(z_{2}) \boldsymbol{R}^{1/2} \frac{\boldsymbol{x}_{j}\boldsymbol{e}_{i}^{\operatorname{T}}}{N}\boldsymbol{R}^{\operatorname{H}/2}\hat{\boldsymbol{Q}}(z_{2})\boldsymbol{R}^{1/2}\frac{\boldsymbol{e}_{i}\boldsymbol{x}_{j}^{\operatorname{H}}}{N}\boldsymbol{e}_{m}\phi_{M}\\
			& + \boldsymbol{e}_{m}^{\operatorname{T}}\boldsymbol{B}\hat{\boldsymbol{Q}}(z_{1})\boldsymbol{R}\hat{\boldsymbol{Q}}(z_{2})\boldsymbol{R}^{1/2}\frac{\boldsymbol{e}_{i}\boldsymbol{e}_{i}^{\operatorname{T}}}{N}\boldsymbol{e}_{m}\phi_{M} + \mathcal{O}\left(N^{-\mathbb{N}}\right).
			\end{aligned}
			\label{eq:proof:aux:lemma2:expectation2: intermediate_2}
		\end{equation}
		Substituting the derivative in \eqref{eq:proof:aux:lemma2:expectation2: intermediate_2} into \eqref{eq:proof:aux:lemma2:expectation2: intermediate_1} and computing the triple sum yields
		\begin{equation}
			\begin{aligned}
				\frac{1}{N}\mean{\tr{\boldsymbol{B}\hat{\boldsymbol{Q}}(z_{1})\boldsymbol{R}\hat{\boldsymbol{Q}}(z_{2})\boldsymbol{R}^{1/2}\frac{\boldsymbol{X}\boldsymbol{X}^{\operatorname{H}}}{N}\phi_{M}}} =& -\frac{1}{N}\mean{\tr{\boldsymbol{B}\hat{\boldsymbol{Q}}(z_{1}) \boldsymbol{R}^{1/2} \frac{\boldsymbol{X}\boldsymbol{X}^{\operatorname{H}}}{N}} \frac{1}{N} \tr{\boldsymbol{R}\hat{\boldsymbol{Q}}(z_{1})\boldsymbol{R}\hat{\boldsymbol{Q}}(z_{2})}\phi_{M}} \\
				&-\frac{1}{N}\mean{\tr{\boldsymbol{B}\hat{\boldsymbol{Q}}(z_{1})\boldsymbol{R} \hat{\boldsymbol{Q}}(z_{2}) \boldsymbol{R}^{1/2} \frac{\boldsymbol{X}\boldsymbol{X}^{\operatorname{H}}}{N}}\frac{1}{N}\tr{\boldsymbol{R}\hat{\boldsymbol{Q}}(z_{2})}\phi_{M}}\\
				&+\frac{1}{N}\mean{\tr{\boldsymbol{B}\hat{\boldsymbol{Q}}(z_{1})\boldsymbol{R}\hat{\boldsymbol{Q}}(z_{2})\boldsymbol{R}^{1/2}\phi_{M}}}+ \mathcal{O}\left(N^{-\mathbb{N}}\right).
			\end{aligned}
			\label{eq:proof:aux:lemma2:expectation2: intermediate_3}
		\end{equation}
		In the following we decorrelate the first and second term on the right hand side of \eqref{eq:proof:aux:lemma2:expectation2: intermediate_3}. We start with the first term on the right hand side of \eqref{eq:proof:aux:lemma2:expectation2: intermediate_3} by analyzing the covariance 
		\begin{equation}
			\begin{aligned}
				&\cov{\frac{1}{N}\tr{\boldsymbol{B}\hat{\boldsymbol{Q}}(z_{1}) \boldsymbol{R}^{1/2} \frac{\boldsymbol{X}\boldsymbol{X}^{\operatorname{H}}}{N}\phi_{M}}}{\frac{1}{N} \tr{\boldsymbol{R}\hat{\boldsymbol{Q}}(z_{1})\boldsymbol{R}\hat{\boldsymbol{Q}}(z_{2})\phi_{M}}}\\
				 =& \mean{\left(\frac{1}{N}\tr{\boldsymbol{B}\hat{\boldsymbol{Q}}(z_{1}) \boldsymbol{R}^{1/2} \frac{\boldsymbol{X}\boldsymbol{X}^{\operatorname{H}}}{N}\phi_{M}}\right)^{(\circ)}\left(\frac{1}{N} \tr{\boldsymbol{R}\hat{\boldsymbol{Q}}(z_{1})\boldsymbol{R}\hat{\boldsymbol{Q}}(z_{2})\phi_{M}} \right)^{(\circ)}}\\
				=& \frac{1}{N}\mean{\tr{\boldsymbol{B}\hat{\boldsymbol{Q}}(z_{1}) \boldsymbol{R}^{1/2} \frac{\boldsymbol{X}\boldsymbol{X}^{\operatorname{H}}}{N}} \frac{1}{N} \tr{\boldsymbol{R}\hat{\boldsymbol{Q}}(z_{1})\boldsymbol{R}\hat{\boldsymbol{Q}}(z_{2})}\phi_{M}} \\
				&-\mean{\frac{1}{N}\tr{\boldsymbol{B}\hat{\boldsymbol{Q}}(z_{1}) \boldsymbol{R}^{1/2} \frac{\boldsymbol{X}\boldsymbol{X}^{\operatorname{H}}}{N}\phi_{M}}}\mean{\frac{1}{N} \tr{\boldsymbol{R}\hat{\boldsymbol{Q}}(z_{1})\boldsymbol{R}\hat{\boldsymbol{Q}}(z_{2})\phi_{M}}}.
			\end{aligned}
			\label{eq:proof:aux:lemma2:expectation2: intermediate_4}
		\end{equation}
		Using the Cauchy-Schwarz inequality we can upper-bound \eqref{eq:proof:aux:lemma2:expectation2: intermediate_4} by
		\begin{equation}
			\begin{aligned}
				& \abs{\cov{\frac{1}{N}\tr{\boldsymbol{B}\hat{\boldsymbol{Q}}(z_{1}) \boldsymbol{R}^{1/2} \frac{\boldsymbol{X}\boldsymbol{X}^{\operatorname{H}}}{N}\phi_{M}}}{\frac{1}{N} \tr{\boldsymbol{R}\hat{\boldsymbol{Q}}(z_{1})\boldsymbol{R}\hat{\boldsymbol{Q}}(z_{2})\phi_{M}}}}^{2}\\
				\leq & \mean{\abs{\left(\frac{1}{N}\tr{\boldsymbol{B}\hat{\boldsymbol{Q}}(z_{1}) \boldsymbol{R}^{1/2} \frac{\boldsymbol{X}\boldsymbol{X}^{\operatorname{H}}}{N}\phi_{M}}\right)^{(\circ)}}^{2}}\mean{\abs{\left(\frac{1}{N} \tr{\boldsymbol{R}\hat{\boldsymbol{Q}}(z_{1})\boldsymbol{R}\hat{\boldsymbol{Q}}(z_{2})\phi_{M}} \right)^{(\circ)}}^{2}}\\
				=& \var{\frac{1}{N}\tr{\boldsymbol{B}\hat{\boldsymbol{Q}}(z_{1}) \boldsymbol{R}^{1/2} \frac{\boldsymbol{X}\boldsymbol{X}^{\operatorname{H}}}{N}\phi_{M}}} \var{\frac{1}{N} \tr{\boldsymbol{R}\hat{\boldsymbol{Q}}(z_{1})\boldsymbol{R}\hat{\boldsymbol{Q}}(z_{2})\phi_{M}}} = \mathcal{O}\left(N^{-4}\right),
			\end{aligned}
			\label{eq:proof:aux:lemma2:expectation2: intermediate_5}
		\end{equation}
		where we have used \eqref{eq:aux:lemma1: variance2} in Lemma \ref{aux:lemma1} and \eqref{eq:aux:lemma2: variance1} in Lemma \ref{aux:lemma2} to obtain
		\begin{equation*}
			 \var{\frac{1}{N}\tr{\boldsymbol{B}\hat{\boldsymbol{Q}}(z_{1}) \boldsymbol{R}^{1/2} \frac{\boldsymbol{X}\boldsymbol{X}^{\operatorname{H}}}{N}\phi_{M}}} = \mathcal{O}\left(N^{-2}\right), \qquad \var{\frac{1}{N} \tr{\boldsymbol{R}\hat{\boldsymbol{Q}}(z_{1})\boldsymbol{R}\hat{\boldsymbol{Q}}(z_{2})\phi_{M}}} = \mathcal{O}\left(N^{-2}\right).
		\end{equation*}
		From \eqref{eq:proof:aux:lemma2:expectation2: intermediate_5} it follows that the covariance in \eqref{eq:proof:aux:lemma2:expectation2: intermediate_4} decays with 
		\begin{equation*}
			\cov{\frac{1}{N}\tr{\boldsymbol{B}\hat{\boldsymbol{Q}}(z_{1}) \boldsymbol{R}^{1/2} \frac{\boldsymbol{X}\boldsymbol{X}^{\operatorname{H}}}{N}\phi_{M}}}{\frac{1}{N} \tr{\boldsymbol{R}\hat{\boldsymbol{Q}}(z_{1})\boldsymbol{R}\hat{\boldsymbol{Q}}(z_{2})\phi_{M}}} = \mathcal{O}\left(N^{-2}\right).
		\end{equation*}
		With this we have decorrelated the first term on the right hand side of \eqref{eq:proof:aux:lemma2:expectation2: intermediate_3} as
		\begin{equation}
			\begin{aligned}
				&\frac{1}{N}\mean{\tr{\boldsymbol{B}\hat{\boldsymbol{Q}}(z_{1}) \boldsymbol{R}^{1/2} \frac{\boldsymbol{X}\boldsymbol{X}^{\operatorname{H}}}{N}} \frac{1}{N} \tr{\boldsymbol{R}\hat{\boldsymbol{Q}}(z_{1})\boldsymbol{R}\hat{\boldsymbol{Q}}(z_{2})}\phi_{M}}\\
			=& \frac{1}{N}\mean{\tr{\boldsymbol{B}\hat{\boldsymbol{Q}}(z_{1}) \boldsymbol{R}^{1/2} \frac{\boldsymbol{X}\boldsymbol{X}^{\operatorname{H}}}{N}\phi_{M}}} \frac{1}{N} \mean{\tr{\boldsymbol{R}\hat{\boldsymbol{Q}}(z_{1})\boldsymbol{R}\hat{\boldsymbol{Q}}(z_{2})\phi_{M}}} +  \mathcal{O}\left(N^{-2}\right),
			\end{aligned}
			\label{eq:proof:aux:lemma2:expectation2: intermediate_6}
		\end{equation}
		which follows from \eqref{eq:proof:aux:lemma2:expectation2: intermediate_4}. In order to decorrelate the second term on the right hand side of \eqref{eq:proof:aux:lemma2:expectation2: intermediate_3} we analyze the covariance
		\begin{equation}
			\begin{aligned}
				&\cov{\frac{1}{N}\tr{\boldsymbol{B}\hat{\boldsymbol{Q}}(z_{1})\boldsymbol{R} \hat{\boldsymbol{Q}}(z_{2}) \boldsymbol{R}^{1/2} \frac{\boldsymbol{X}\boldsymbol{X}^{\operatorname{H}}}{N}\phi_{M}}}{\frac{1}{N}\tr{\boldsymbol{R}\hat{\boldsymbol{Q}}(z_{2})\phi_{M}}}\\
				=& \mean{\left(\frac{1}{N}\tr{\boldsymbol{B}\hat{\boldsymbol{Q}}(z_{1})\boldsymbol{R} \hat{\boldsymbol{Q}}(z_{2}) \boldsymbol{R}^{1/2} \frac{\boldsymbol{X}\boldsymbol{X}^{\operatorname{H}}}{N}\phi_{M}}\right)^{(\circ)}\left(\frac{1}{N}\tr{\boldsymbol{R}\hat{\boldsymbol{Q}}(z_{2})\phi_{M}}\right)^{(\circ)}}\\
				=& \frac{1}{N}\mean{\tr{\boldsymbol{B}\hat{\boldsymbol{Q}}(z_{1})\boldsymbol{R} \hat{\boldsymbol{Q}}(z_{2}) \boldsymbol{R}^{1/2} \frac{\boldsymbol{X}\boldsymbol{X}^{\operatorname{H}}}{N}}\frac{1}{N}\tr{\boldsymbol{R}\hat{\boldsymbol{Q}}(z_{2})}\phi_{M}}\\
				&- \mean{\frac{1}{N}\tr{\boldsymbol{B}\hat{\boldsymbol{Q}}(z_{1})\boldsymbol{R} \hat{\boldsymbol{Q}}(z_{2}) \boldsymbol{R}^{1/2} \frac{\boldsymbol{X}\boldsymbol{X}^{\operatorname{H}}}{N}\phi_{M}}}\mean{\frac{1}{N}\tr{\boldsymbol{R}\hat{\boldsymbol{Q}}(z_{2})\phi_{M}}}.
			\end{aligned}
			\label{eq:proof:aux:lemma2:expectation2: intermediate_7}
		\end{equation}
		Next, we apply the Cauchy-Schwarz inequality to \eqref{eq:proof:aux:lemma2:expectation2: intermediate_7}
		\begin{equation}
			\begin{aligned}
				&\abs{\cov{\frac{1}{N}\tr{\boldsymbol{B}\hat{\boldsymbol{Q}}(z_{1})\boldsymbol{R} \hat{\boldsymbol{Q}}(z_{2}) \boldsymbol{R}^{1/2} \frac{\boldsymbol{X}\boldsymbol{X}^{\operatorname{H}}}{N}\phi_{M}}}{\frac{1}{N}\tr{\boldsymbol{R}\hat{\boldsymbol{Q}}(z_{2})\phi_{M}}}}^{2}\\
				\leq &\mean{\abs{\left(\frac{1}{N}\tr{\boldsymbol{B}\hat{\boldsymbol{Q}}(z_{1})\boldsymbol{R} \hat{\boldsymbol{Q}}(z_{2}) \boldsymbol{R}^{1/2} \frac{\boldsymbol{X}\boldsymbol{X}^{\operatorname{H}}}{N}\phi_{M}}\right)^{(\circ)}}^{2}} \mean{\abs{\left(\frac{1}{N}\tr{\boldsymbol{R}\hat{\boldsymbol{Q}}(z_{2})\phi_{M}}\right)^{(\circ)}}^{2}}\\
				=& \var{\frac{1}{N}\tr{\boldsymbol{B}\hat{\boldsymbol{Q}}(z_{1})\boldsymbol{R} \hat{\boldsymbol{Q}}(z_{2}) \boldsymbol{R}^{1/2} \frac{\boldsymbol{X}\boldsymbol{X}^{\operatorname{H}}}{N}\phi_{M}}}\var{\frac{1}{N}\tr{\boldsymbol{R}\hat{\boldsymbol{Q}}(z_{2})\phi_{M}}} = \mathcal{O}\left(N^{-4}\right),
			\end{aligned}
			\label{eq:proof:aux:lemma2:expectation2: intermediate_8}
		\end{equation}
		where we have used \eqref{eq:aux:lemma2: variance2} in Lemma \ref{aux:lemma2} and \eqref{eq:aux:lemma1: variance1} in Lemma \ref{aux:lemma1} to obtain
		\begin{equation*}
			\var{\frac{1}{N}\tr{\boldsymbol{B}\hat{\boldsymbol{Q}}(z_{1})\boldsymbol{R} \hat{\boldsymbol{Q}}(z_{2}) \boldsymbol{R}^{1/2} \frac{\boldsymbol{X}\boldsymbol{X}^{\operatorname{H}}}{N}\phi_{M}}} = \mathcal{O}\left(N^{-2}\right), \qquad \var{\frac{1}{N}\tr{\boldsymbol{R}\hat{\boldsymbol{Q}}(z_{2})\phi_{M}}}= \mathcal{O}\left(N^{-2}\right).
		\end{equation*}
		As a direct consequence of \eqref{eq:proof:aux:lemma2:expectation2: intermediate_8} the covariance in \eqref{eq:proof:aux:lemma2:expectation2: intermediate_7} decays with
		\begin{equation*}
			\cov{\frac{1}{N}\tr{\boldsymbol{B}\hat{\boldsymbol{Q}}(z_{1})\boldsymbol{R} \hat{\boldsymbol{Q}}(z_{2}) \boldsymbol{R}^{1/2} \frac{\boldsymbol{X}\boldsymbol{X}^{\operatorname{H}}}{N}\phi_{M}}}{\frac{1}{N}\tr{\boldsymbol{R}\hat{\boldsymbol{Q}}(z_{2})\phi_{M}}}= \mathcal{O}\left(N^{-2}\right),
		\end{equation*}
		and we can decorrelated the second term on the right hand side of \eqref{eq:proof:aux:lemma2:expectation2: intermediate_3} as 
		\begin{equation}
			\begin{aligned}
				&\frac{1}{N}\mean{\tr{\boldsymbol{B}\hat{\boldsymbol{Q}}(z_{1})\boldsymbol{R} \hat{\boldsymbol{Q}}(z_{2}) \boldsymbol{R}^{1/2} \frac{\boldsymbol{X}\boldsymbol{X}^{\operatorname{H}}}{N}}\frac{1}{N}\tr{\boldsymbol{R}\hat{\boldsymbol{Q}}(z_{2})}\phi_{M}}\\
				=&\frac{1}{N}\mean{\tr{\boldsymbol{B}\hat{\boldsymbol{Q}}(z_{1})\boldsymbol{R} \hat{\boldsymbol{Q}}(z_{2}) \boldsymbol{R}^{1/2} \frac{\boldsymbol{X}\boldsymbol{X}^{\operatorname{H}}}{N}\phi_{M}}}\frac{1}{N}\mean{\tr{\boldsymbol{R}\hat{\boldsymbol{Q}}(z_{2})\phi_{M}}} +  \mathcal{O}\left(N^{-2}\right).
			\end{aligned}
			\label{eq:proof:aux:lemma2:expectation2: intermediate_9}	
		\end{equation}
		Substituting the two decorrelated terms in \eqref{eq:proof:aux:lemma2:expectation2: intermediate_6} and \eqref{eq:proof:aux:lemma2:expectation2: intermediate_9}	 into \eqref{eq:proof:aux:lemma2:expectation2: intermediate_3} yields
		\begin{equation}
			\begin{aligned}
				\frac{1}{N}\mean{\tr{\boldsymbol{B}\hat{\boldsymbol{Q}}(z_{1})\boldsymbol{R}\hat{\boldsymbol{Q}}(z_{2})\boldsymbol{R}^{1/2}\frac{\boldsymbol{X}\boldsymbol{X}^{\operatorname{H}}}{N}\phi_{M}}} =& - \frac{1}{N}\mean{\tr{\boldsymbol{B}\hat{\boldsymbol{Q}}(z_{1}) \boldsymbol{R}^{1/2} \frac{\boldsymbol{X}\boldsymbol{X}^{\operatorname{H}}}{N}\phi_{M}}} \frac{1}{N} \mean{\tr{\boldsymbol{R}\hat{\boldsymbol{Q}}(z_{1})\boldsymbol{R}\hat{\boldsymbol{Q}}(z_{2})\phi_{M}}}\\
				&-\frac{1}{N}\mean{\tr{\boldsymbol{B}\hat{\boldsymbol{Q}}(z_{1})\boldsymbol{R} \hat{\boldsymbol{Q}}(z_{2}) \boldsymbol{R}^{1/2} \frac{\boldsymbol{X}\boldsymbol{X}^{\operatorname{H}}}{N}\phi_{M}}}\frac{1}{N}\mean{\tr{\boldsymbol{R}\hat{\boldsymbol{Q}}(z_{2})\phi_{M}}}\\
				&+\frac{1}{N}\mean{\tr{\boldsymbol{B}\hat{\boldsymbol{Q}}(z_{1})\boldsymbol{R}\hat{\boldsymbol{Q}}(z_{2})\boldsymbol{R}^{1/2}\phi_{M}}} + \mathcal{O}\left(N^{-2}\right).
			\end{aligned}
			\label{eq:proof:aux:lemma2:expectation2: intermediate_10}
		\end{equation}
		Using the previously derived expectation in \eqref{eq:aux:lemma1: expectation1} and adding $\frac{1}{N}\mean{\tr{\boldsymbol{B}\hat{\boldsymbol{Q}}(z_{1})\boldsymbol{R} \hat{\boldsymbol{Q}}(z_{2}) \boldsymbol{R}^{1/2} \frac{\boldsymbol{X}\boldsymbol{X}^{\operatorname{H}}}{N}\phi_{M}}}\frac{1}{N}\tr{\boldsymbol{R}\bar{\boldsymbol{Q}}(z_{2})}$ on both sides of \eqref{eq:proof:aux:lemma2:expectation2: intermediate_10} yields
		\begin{equation}
			\begin{aligned}
				&\frac{1}{N}\mean{\tr{\boldsymbol{B}\hat{\boldsymbol{Q}}(z_{1})\boldsymbol{R}\hat{\boldsymbol{Q}}(z_{2})\boldsymbol{R}^{1/2}\frac{\boldsymbol{X}\boldsymbol{X}^{\operatorname{H}}}{N}\phi_{M}}}\left( 1+\frac{1}{N}\tr{\boldsymbol{R}\bar{\boldsymbol{Q}}(z_{2})} \right) \\
				=&- \frac{1}{N}\mean{\tr{\boldsymbol{B}\hat{\boldsymbol{Q}}(z_{1}) \boldsymbol{R}^{1/2} \frac{\boldsymbol{X}\boldsymbol{X}^{\operatorname{H}}}{N}\phi_{M}}} \frac{1}{N} \mean{\tr{\boldsymbol{R}\hat{\boldsymbol{Q}}(z_{1})\boldsymbol{R}\hat{\boldsymbol{Q}}(z_{2})\phi_{M}}}\\			
				&+\frac{1}{N}\mean{\tr{\boldsymbol{B}\hat{\boldsymbol{Q}}(z_{1})\boldsymbol{R}\hat{\boldsymbol{Q}}(z_{2})\boldsymbol{R}^{1/2}\phi_{M}}} + \mathcal{O}(N^{-1}).
			\end{aligned}
			\label{eq:proof:aux:lemma2:expectation2: intermediate_11}
		\end{equation}
		Multiplying \eqref{eq:proof:aux:lemma2:expectation2: intermediate_11} by $\frac{z_{2}}{\omega(z_{2})}$ on both sides where $1+\frac{1}{N}\tr{\boldsymbol{R}\bar{\boldsymbol{Q}}(z_{2})} = \frac{\omega(z_{2})}{z_{2}}$ which follows from \eqref{eq: definition_omega_z} and using the expectation in \eqref{eq:aux:lemma1: expectation2} in Lemma \ref{aux:lemma1} we obtain
		\begin{equation}
			\begin{aligned}
				\frac{1}{N}\mean{\tr{\boldsymbol{B}\hat{\boldsymbol{Q}}(z_{1})\boldsymbol{R}\hat{\boldsymbol{Q}}(z_{2})\boldsymbol{R}^{1/2}\frac{\boldsymbol{X}\boldsymbol{X}^{\operatorname{H}}}{N}\phi_{M}}} 
				=& - \frac{1}{N}\frac{z_{1}}{\omega(z_{1})}\frac{z_{2}}{\omega(z_{2})}\tr{\boldsymbol{B}\bar{\boldsymbol{Q}}(z_{1}) \boldsymbol{R}^{1/2}} \frac{1}{N} \mean{\tr{\boldsymbol{R}\hat{\boldsymbol{Q}}(z_{1})\boldsymbol{R}\hat{\boldsymbol{Q}}(z_{2})\phi_{M}}}\\			
				&+\frac{1}{N}\frac{z_{2}}{\omega(z_{2})}\mean{\tr{\boldsymbol{B}\hat{\boldsymbol{Q}}(z_{1})\boldsymbol{R}\hat{\boldsymbol{Q}}(z_{2})\boldsymbol{R}^{1/2}\phi_{M}}} + \mathcal{O}(N^{-1}).
			\end{aligned}
			\label{eq:proof:aux:lemma2:expectation2: intermediate_12}
		\end{equation}
		Next, we replace $\boldsymbol{B}$ by $\boldsymbol{R}^{\operatorname{H}/2}\boldsymbol{B}$ in \eqref{eq:proof:aux:lemma2:expectation2: intermediate_12} such that
		\begin{equation}
			\begin{aligned}
				\frac{1}{N}\mean{\tr{\boldsymbol{R}^{\operatorname{H}/2}\boldsymbol{B}\hat{\boldsymbol{Q}}(z_{1})\boldsymbol{R}\hat{\boldsymbol{Q}}(z_{2})\boldsymbol{R}^{1/2}\frac{\boldsymbol{X}\boldsymbol{X}^{\operatorname{H}}}{N}\phi_{M}}} =& \frac{1}{N}\mean{\tr{\boldsymbol{B}\hat{\boldsymbol{Q}}(z_{1})\boldsymbol{R}\hat{\boldsymbol{Q}}(z_{2})\hat{\boldsymbol{R}}\phi_{M}}}\\
				=& - \frac{1}{N}\frac{z_{1}}{\omega(z_{1})}\frac{z_{2}}{\omega(z_{2})}\tr{\boldsymbol{B}\bar{\boldsymbol{Q}}(z_{1}) \boldsymbol{R}} \frac{1}{N} \mean{\tr{\boldsymbol{R}\hat{\boldsymbol{Q}}(z_{1})\boldsymbol{R}\hat{\boldsymbol{Q}}(z_{2})\phi_{M}}}\\			
				&+\frac{1}{N}\frac{z_{2}}{\omega(z_{2})}\mean{\tr{\boldsymbol{B}\hat{\boldsymbol{Q}}(z_{1})\boldsymbol{R}\hat{\boldsymbol{Q}}(z_{2})\boldsymbol{R}\phi_{M}}} + \mathcal{O}(N^{-1}).
			\end{aligned}
			\label{eq:proof:aux:lemma2:expectation2: intermediate_13}
		\end{equation}
		Inserting the previous expression in \eqref{eq:proof:aux:lemma2:expectation2: intermediate_13} into  \eqref{eq:proof:aux:lemma2:expectation1: intermediate_13} yields
		\begin{equation}
			\begin{aligned}
				\frac{1}{N}\mean{\tr{\boldsymbol{B}\hat{\boldsymbol{Q}}(z_{1})\boldsymbol{R}\hat{\boldsymbol{Q}}(z_{2})\phi_{M}}} 
				=&-z_{2}^{-1} \frac{z_{1}}{\omega(z_{1})}\frac{1}{N} \tr{\boldsymbol{B}\bar{\boldsymbol{Q}}(z_{1})\boldsymbol{R}}\frac{1}{N}\mean{\tr{\hat{\boldsymbol{Q}}(z_{1})\boldsymbol{R}\hat{\boldsymbol{Q}}(z_{2})\boldsymbol{R}\phi_{M}}}\left(1- \frac{z_{2}}{\omega(z_{2})}\frac{1}{N} \tr{\boldsymbol{R}\bar{\boldsymbol{Q}}(z_{2})}\right)\\
				&+z_{2}^{-1} \frac{1}{N} \mean{\tr{\boldsymbol{B}\hat{\boldsymbol{Q}}(z_{1})\boldsymbol{R}\hat{\boldsymbol{Q}}(z_{2})\boldsymbol{R}\phi_{M}}}\left(1- \frac{z_{2}}{\omega(z_{2})}\frac{1}{N} \tr{\boldsymbol{R}\bar{\boldsymbol{Q}}(z_{2})}\right)\\
				&- z_{2}^{-1} \frac{1}{N}\tr{\boldsymbol{B}\bar{\boldsymbol{Q}}(z_{1})\boldsymbol{R}} + \mathcal{O}(N^{-1}).
			\end{aligned}
			\label{eq:proof:aux:lemma2:expectation1: intermediate_14}
		\end{equation}
		Now using $\frac{z_{2}}{\omega(z_{2})} = 1- \frac{z_{2}}{\omega(z_{2})}\frac{1}{N} \tr{\boldsymbol{R}\bar{\boldsymbol{Q}}(z_{2})}$ which follows from \eqref{eq: definition_omega_z} we can simplify the expression in \eqref{eq:proof:aux:lemma2:expectation1: intermediate_14} as
		\begin{equation}
			\begin{aligned}
				\frac{1}{N}\mean{\tr{\boldsymbol{B}\hat{\boldsymbol{Q}}(z_{1})\boldsymbol{R}\hat{\boldsymbol{Q}}(z_{2})\phi_{M}}} =& - \frac{z_{1}}{\omega(z_{1})\omega(z_{2})}\frac{1}{N} \tr{\boldsymbol{B}\bar{\boldsymbol{Q}}(z_{1})\boldsymbol{R}}\frac{1}{N}\mean{\tr{\hat{\boldsymbol{Q}}(z_{1})\boldsymbol{R}\hat{\boldsymbol{Q}}(z_{2})\boldsymbol{R}\phi_{M}}}\\
				&+\frac{1}{\omega(z_{2})}\frac{1}{N} \mean{\tr{\boldsymbol{B}\hat{\boldsymbol{Q}}(z_{1})\boldsymbol{R}\hat{\boldsymbol{Q}}(z_{2})\boldsymbol{R}\phi_{M}}}- z_{2}^{-1} \frac{1}{N}\tr{\boldsymbol{B}\bar{\boldsymbol{Q}}(z_{1})\boldsymbol{R}} + \mathcal{O}(N^{-1}).
			\end{aligned}
			\label{eq:proof:aux:lemma2:expectation1: intermediate_15_0}
		\end{equation}
		Replacing $\boldsymbol{B}$ in \eqref{eq:proof:aux:lemma2:expectation1: intermediate_15_0} by $\bar{\boldsymbol{Q}}(z_{2})\boldsymbol{R}$ we obtain
		\begin{equation}
			\begin{aligned}
				\frac{1}{N}\mean{\tr{\bar{\boldsymbol{Q}}(z_{2})\boldsymbol{R}\hat{\boldsymbol{Q}}(z_{1})\boldsymbol{R}\hat{\boldsymbol{Q}}(z_{2})\phi_{M}}} =& - \frac{z_{1}}{\omega(z_{1})\omega(z_{2})}\frac{1}{N} \tr{\bar{\boldsymbol{Q}}(z_{2})\boldsymbol{R}\bar{\boldsymbol{Q}}(z_{1})\boldsymbol{R}}\frac{1}{N}\mean{\tr{\hat{\boldsymbol{Q}}(z_{1})\boldsymbol{R}\hat{\boldsymbol{Q}}(z_{2})\boldsymbol{R}\phi_{M}}}\\
				&+\frac{1}{\omega(z_{2})}\frac{1}{N} \mean{\tr{\bar{\boldsymbol{Q}}(z_{2})\boldsymbol{R}\hat{\boldsymbol{Q}}(z_{1})\boldsymbol{R}\hat{\boldsymbol{Q}}(z_{2})\boldsymbol{R}\phi_{M}}}- z_{2}^{-1} \frac{1}{N}\tr{\bar{\boldsymbol{Q}}(z_{2})\boldsymbol{R}\bar{\boldsymbol{Q}}(z_{1})\boldsymbol{R}} + \mathcal{O}(N^{-1}).
			\end{aligned}
			\label{eq:proof:aux:lemma2:expectation1: intermediate_15}
		\end{equation}
		By subtracting and adding $\omega(z_{2})\boldsymbol{I}_{M}$ inside the trace of the second term on the right hand side of \eqref{eq:proof:aux:lemma2:expectation1: intermediate_15}, we can equivalently express $\frac{1}{\omega(z_{2})}\frac{1}{N} \mean{\tr{\bar{\boldsymbol{Q}}(z_{2})\boldsymbol{R}\hat{\boldsymbol{Q}}(z_{1})\boldsymbol{R}\hat{\boldsymbol{Q}}(z_{2})\boldsymbol{R}\phi_{M}}}$ as
		\begin{equation}
			\begin{aligned}
				\frac{1}{\omega(z_{2})}\frac{1}{N} \mean{\tr{\bar{\boldsymbol{Q}}(z_{2})\boldsymbol{R}\hat{\boldsymbol{Q}}(z_{1})\boldsymbol{R}\hat{\boldsymbol{Q}}(z_{2})\boldsymbol{R}\phi_{M}}} =& \frac{1}{\omega(z_{2})}\frac{1}{N} \mean{\tr{\boldsymbol{R}\hat{\boldsymbol{Q}}(z_{1})\boldsymbol{R}\hat{\boldsymbol{Q}}(z_{2})\left(\boldsymbol{R}-\omega(z_{2})\boldsymbol{I}_{M} + \omega(z_{2})\boldsymbol{I}_{M}\right)\bar{\boldsymbol{Q}}(z_{2})\phi_{M}}}\\
				=&\frac{1}{\omega(z_{2})}\frac{1}{N} \mean{\tr{\boldsymbol{R}\hat{\boldsymbol{Q}}(z_{1})\boldsymbol{R}\hat{\boldsymbol{Q}}(z_{2})\left(\boldsymbol{R}-\omega(z_{2})\boldsymbol{I}_{M} \right)\bar{\boldsymbol{Q}}(z_{2})\phi_{M}}}\\
				&+ \frac{1}{N} \mean{\tr{\boldsymbol{R}\hat{\boldsymbol{Q}}(z_{1})\boldsymbol{R}\hat{\boldsymbol{Q}}(z_{2})\bar{\boldsymbol{Q}}(z_{2})\phi_{M}}}\\
				=&\frac{1}{\omega(z_{2})}\frac{\omega(z_{2})}{z_{2}}\frac{1}{N} \mean{\tr{\boldsymbol{R}\hat{\boldsymbol{Q}}(z_{1})\boldsymbol{R}\hat{\boldsymbol{Q}}(z_{2})\phi_{M}}}+ \frac{1}{N} \mean{\tr{\boldsymbol{R}\hat{\boldsymbol{Q}}(z_{1})\boldsymbol{R}\hat{\boldsymbol{Q}}(z_{2})\bar{\boldsymbol{Q}}(z_{2})\phi_{M}}}\\
				=&\frac{1}{z_{2}}\frac{1}{N} \mean{\tr{\boldsymbol{R}\hat{\boldsymbol{Q}}(z_{1})\boldsymbol{R}\hat{\boldsymbol{Q}}(z_{2})\phi_{M}}}+ \frac{1}{N} \mean{\tr{\boldsymbol{R}\hat{\boldsymbol{Q}}(z_{1})\boldsymbol{R}\hat{\boldsymbol{Q}}(z_{2})\bar{\boldsymbol{Q}}(z_{2})\phi_{M}}},
			\end{aligned}
			\label{eq:proof:aux:lemma2:expectation1: intermediate_16}
		\end{equation}
		 where we have used the resolvent $\bar{\boldsymbol{Q}}(z_{2}) = \frac{\omega(z_{2})}{z_{2}}(\boldsymbol{R}-\omega(z_{2})\boldsymbol{I}_{M})^{-1}$ in \eqref{eq: resolvent_and_deterministic_equivalent}. Rearranging the expression in \eqref{eq:proof:aux:lemma2:expectation1: intermediate_15} according to 
		\begin{equation}
			\begin{aligned}
				&\frac{z_{1}}{\omega(z_{1})\omega(z_{2})}\frac{1}{N} \tr{\bar{\boldsymbol{Q}}(z_{2})\boldsymbol{R}\bar{\boldsymbol{Q}}(z_{1})\boldsymbol{R}}\frac{1}{N}\mean{\tr{\hat{\boldsymbol{Q}}(z_{1})\boldsymbol{R}\hat{\boldsymbol{Q}}(z_{2})\boldsymbol{R}\phi_{M}}}\\
				 =& - \frac{1}{N}\mean{\tr{\bar{\boldsymbol{Q}}(z_{2})\boldsymbol{R}\hat{\boldsymbol{Q}}(z_{1})\boldsymbol{R}\hat{\boldsymbol{Q}}(z_{2})\phi_{M}}}\\
				&+\frac{1}{\omega(z_{2})}\frac{1}{N} \mean{\tr{\bar{\boldsymbol{Q}}(z_{2})\boldsymbol{R}\hat{\boldsymbol{Q}}(z_{1})\boldsymbol{R}\hat{\boldsymbol{Q}}(z_{2})\boldsymbol{R}\phi_{M}}}- z_{2}^{-1} \frac{1}{N}\tr{\bar{\boldsymbol{Q}}(z_{2})\boldsymbol{R}\bar{\boldsymbol{Q}}(z_{1})\boldsymbol{R}} + \mathcal{O}\left(N^{-1}\right),
			\end{aligned}
			\label{eq:proof:aux:lemma2:expectation1: intermediate_17}
		\end{equation}
		and substituting \eqref{eq:proof:aux:lemma2:expectation1: intermediate_16} into \eqref{eq:proof:aux:lemma2:expectation1: intermediate_17} and multiplying by $z_{2}$ on both sides yields
		 \begin{equation*}
		 	\begin{aligned}
		 		&\frac{z_{1}z_{2}}{\omega(z_{1})\omega(z_{2})}\frac{1}{N} \tr{\bar{\boldsymbol{Q}}(z_{2})\boldsymbol{R}\bar{\boldsymbol{Q}}(z_{1})\boldsymbol{R}}\frac{1}{N}\mean{\tr{\hat{\boldsymbol{Q}}(z_{1})\boldsymbol{R}\hat{\boldsymbol{Q}}(z_{2})\boldsymbol{R}\phi_{M}}}\\
		 		=& \frac{1}{N} \mean{\tr{\boldsymbol{R}\hat{\boldsymbol{Q}}(z_{1})\boldsymbol{R}\hat{\boldsymbol{Q}}(z_{2})\phi_{M}}}- \frac{1}{N}\tr{\bar{\boldsymbol{Q}}(z_{2})\boldsymbol{R}\bar{\boldsymbol{Q}}(z_{1})\boldsymbol{R}} + \mathcal{O}\left(N^{-1}\right),
		 	\end{aligned}
		\end{equation*}
		which can equivalently be expressed as
		\begin{equation}
			\frac{1}{N}\mean{\tr{\boldsymbol{R}\hat{\boldsymbol{Q}}(z_{1})\boldsymbol{R}\hat{\boldsymbol{Q}}(z_{2})\phi_{M}}} \left( \frac{z_{1}z_{2}}{\omega(z_{1})\omega(z_{2})} \frac{1}{N}\tr{\bar{\boldsymbol{Q}}(z_{1})\boldsymbol{R}\bar{\boldsymbol{Q}}(z_{2})\boldsymbol{R}} -1 \right) = - \frac{1}{N} \mean{\tr{\bar{\boldsymbol{Q}}(z_{2})\boldsymbol{R}\hat{\boldsymbol{Q}}(z_{1})\boldsymbol{R}\phi_{M}}}+ \mathcal{O}\left(N^{-1}\right).
			\label{eq:proof:aux:lemma2:expectation1: intermediate_18}
		\end{equation}
		Using the previously derived expectation in \eqref{eq:aux:lemma1: expectation1} in Lemma \ref{aux:lemma1} we obtain
		\begin{equation}
			\frac{1}{N}\mean{\tr{\boldsymbol{R}\hat{\boldsymbol{Q}}(z_{1})\boldsymbol{R}\hat{\boldsymbol{Q}}(z_{2})\phi_{M}}} = \frac{1}{1-\gamma(z_{1},z_{2})} \frac{1}{N} \tr{\bar{\boldsymbol{Q}}(z_{1})\boldsymbol{R}\bar{\boldsymbol{Q}}(z_{2})\boldsymbol{R}} + \mathcal{O}\left(N^{-1}\right),
			\label{eq:proof:aux:lemma2:expectation1: intermediate_19}
		\end{equation}
		where
		\begin{equation}
			\gamma(z_{1},z_{2}) = \frac{z_{1}z_{2}}{\omega(z_{1})\omega(z_{2})} \frac{1}{N} \tr{\boldsymbol{R}\bar{\boldsymbol{Q}}(z_{1})\boldsymbol{R}\bar{\boldsymbol{Q}}(z_{2})}.
			\label{eq: definition_gamma_proof_of_lemma}
		\end{equation}
		Replacing $\boldsymbol{B}$ by $\bar{\boldsymbol{Q}}(z_{2})\boldsymbol{B}$ in \eqref{eq:proof:aux:lemma2:expectation1: intermediate_15_0} yields
		\begin{equation}
			\begin{aligned}
				\frac{1}{N}\mean{\tr{\boldsymbol{B}\hat{\boldsymbol{Q}}(z_{1})\boldsymbol{R}\hat{\boldsymbol{Q}}(z_{2})\bar{\boldsymbol{Q}}(z_{2})\phi_{M}}} =&- \frac{z_{1}}{\omega(z_{1})\omega(z_{2})}\frac{1}{N}\tr{\boldsymbol{B}\bar{\boldsymbol{Q}}(z_{1})\boldsymbol{R}\bar{\boldsymbol{Q}}(z_{2})}\frac{1}{N}\mean{\tr{\hat{\boldsymbol{Q}}(z_{1})\boldsymbol{R}\hat{\boldsymbol{Q}}(z_{2})\boldsymbol{R}\phi_{M}}}\\
				&+ \frac{1}{\omega(z_{2})}\frac{1}{N}\mean{\tr{\boldsymbol{B}\hat{\boldsymbol{Q}}(z_{1})\boldsymbol{R}\hat{\boldsymbol{Q}}(z_{2})\boldsymbol{R}\bar{\boldsymbol{Q}}(z_{2})\phi_{M}}}-\frac{1}{z_{2}}\frac{1}{N}\mean{\tr{\boldsymbol{B}\hat{\boldsymbol{Q}}(z_{1})\boldsymbol{R}\bar{\boldsymbol{Q}}(z_{2})\phi_{M}}}+ \mathcal{O}(N^{-1}).
			\end{aligned}
			\label{eq:proof:aux:lemma2:expectation1: intermediate_20}
		\end{equation}
		Substituting \eqref{eq:proof:aux:lemma2:expectation1: intermediate_19} into \eqref{eq:proof:aux:lemma2:expectation1: intermediate_20} we obtain
		\begin{equation}
			\begin{aligned}
				\frac{1}{N}\mean{\tr{\boldsymbol{B}\hat{\boldsymbol{Q}}(z_{1})\boldsymbol{R}\hat{\boldsymbol{Q}}(z_{2})\bar{\boldsymbol{Q}}(z_{2})\phi_{M}}} =&- \frac{z_{1}}{\omega(z_{1})\omega(z_{2})}\frac{1}{N}\tr{\boldsymbol{B}\bar{\boldsymbol{Q}}(z_{1})\boldsymbol{R}\bar{\boldsymbol{Q}}(z_{2})}\frac{1}{1-\gamma(z_{1},z_{2})} \frac{1}{N} \tr{\bar{\boldsymbol{Q}}(z_{1})\boldsymbol{R}\bar{\boldsymbol{Q}}(z_{2})\boldsymbol{R}} \\
				&+ \frac{1}{\omega(z_{2})}\frac{1}{N}\mean{\tr{\boldsymbol{B}\hat{\boldsymbol{Q}}(z_{1})\boldsymbol{R}\hat{\boldsymbol{Q}}(z_{2})\boldsymbol{R}\bar{\boldsymbol{Q}}(z_{2})\phi_{M}}}-\frac{1}{z_{2}}\frac{1}{N}\mean{\tr{\boldsymbol{B}\hat{\boldsymbol{Q}}(z_{1})\boldsymbol{R}\bar{\boldsymbol{Q}}(z_{2})\phi_{M}}}+ \mathcal{O}(N^{-1}).
			\end{aligned}
			\label{eq:proof:aux:lemma2:expectation1: intermediate_21}
		\end{equation}
		By adding and subtracting $\omega(z_{2})\boldsymbol{I}_{M}$ inside the trace of the second term on the right hand side of \eqref{eq:proof:aux:lemma2:expectation1: intermediate_21} and using $\bar{\boldsymbol{Q}}(z_{2}) = \frac{\omega(z_{2})}{z_{2}}(\boldsymbol{R}-\omega(z_{2})\boldsymbol{I}_{M})^{-1}$ in \eqref{eq: resolvent_and_deterministic_equivalent}, we can equivalently express $\frac{1}{\omega(z_{2})}\frac{1}{N}\mean{\tr{\boldsymbol{B}\hat{\boldsymbol{Q}}(z_{1})\boldsymbol{R}\hat{\boldsymbol{Q}}(z_{2})\boldsymbol{R}\bar{\boldsymbol{Q}}(z_{2})\phi_{M}}}$ as
		\begin{equation}
			\begin{aligned}
				\frac{1}{\omega(z_{2})}\frac{1}{N}\mean{\tr{\boldsymbol{B}\hat{\boldsymbol{Q}}(z_{1})\boldsymbol{R}\hat{\boldsymbol{Q}}(z_{2})\boldsymbol{R}\bar{\boldsymbol{Q}}(z_{2})\phi_{M}}} =& \frac{1}{\omega(z_{2})}\frac{1}{N}\mean{\tr{\boldsymbol{B}\hat{\boldsymbol{Q}}(z_{1})\boldsymbol{R}\hat{\boldsymbol{Q}}(z_{2})\left(\boldsymbol{R}-\omega(z_{2})\boldsymbol{I}_{M}+\omega(z_{2})\boldsymbol{I}_{M}\right) \bar{\boldsymbol{Q}}(z_{2})\phi_{M}}}\\
				=&\frac{1}{z_{2}} \frac{1}{N} \mean{\tr{\boldsymbol{B}\hat{\boldsymbol{Q}}(z_{1})\boldsymbol{R}\hat{\boldsymbol{Q}}(z_{2})\phi_{M}}} + \frac{1}{N}\mean{\tr{\boldsymbol{B}\hat{\boldsymbol{Q}}(z_{1})\boldsymbol{R}\hat{\boldsymbol{Q}}(z_{2})\bar{\boldsymbol{Q}}(z_{2})\phi_{M}}}.
			\end{aligned}
			\label{eq:proof:aux:lemma2:expectation1: intermediate_22}
		\end{equation}
		Substituting \eqref{eq:proof:aux:lemma2:expectation1: intermediate_22} into \eqref{eq:proof:aux:lemma2:expectation1: intermediate_21} yields
		\begin{equation}
			\begin{aligned}
				\frac{1}{N}\mean{\tr{\boldsymbol{B}\hat{\boldsymbol{Q}}(z_{1})\boldsymbol{R}\hat{\boldsymbol{Q}}(z_{2})\bar{\boldsymbol{Q}}(z_{2})\phi_{M}}} =&- \frac{z_{1}}{\omega(z_{1})\omega(z_{2})}\frac{1}{N}\tr{\boldsymbol{B}\bar{\boldsymbol{Q}}(z_{1})\boldsymbol{R}\bar{\boldsymbol{Q}}(z_{2})}\frac{1}{1-\gamma(z_{1},z_{2})} \frac{1}{N} \tr{\bar{\boldsymbol{Q}}(z_{1})\boldsymbol{R}\bar{\boldsymbol{Q}}(z_{2})\boldsymbol{R}} \\
				&+ \frac{1}{z_{2}} \frac{1}{N} \mean{\tr{\boldsymbol{B}\hat{\boldsymbol{Q}}(z_{1})\boldsymbol{R}\hat{\boldsymbol{Q}}(z_{2})\phi_{M}}} + \frac{1}{N}\mean{\tr{\boldsymbol{B}\hat{\boldsymbol{Q}}(z_{1})\boldsymbol{R}\hat{\boldsymbol{Q}}(z_{2})\bar{\boldsymbol{Q}}(z_{2})\phi_{M}}}\\
				&-\frac{1}{z_{2}}\frac{1}{N}\mean{\tr{\boldsymbol{B}\hat{\boldsymbol{Q}}(z_{1})\boldsymbol{R}\bar{\boldsymbol{Q}}(z_{2})\phi_{M}}}+ \mathcal{O}(N^{-1}).
			\end{aligned}
			\label{eq:proof:aux:lemma2:expectation1: intermediate_23}
		\end{equation}

		Rearranging \eqref{eq:proof:aux:lemma2:expectation1: intermediate_23} and using the previously derived expectation in \eqref{eq:aux:lemma1: expectation1} in Lemma \ref{aux:lemma1} we obtain the expectation in \eqref{eq:aux:lemma2: expectation1} of Lemma \ref{aux:lemma2}
		\begin{equation}
			\begin{aligned}
				\frac{1}{N}\mean{\tr{\boldsymbol{B}\hat{\boldsymbol{Q}}(z_{1})\boldsymbol{R}\hat{\boldsymbol{Q}}(z_{2})\phi_{M}}}=& \frac{z_{1}z_{2}}{\omega(z_{1})\omega(z_{2})}\frac{1}{N}\tr{\boldsymbol{B}\bar{\boldsymbol{Q}}(z_{1})\boldsymbol{R}\bar{\boldsymbol{Q}}(z_{2})}\frac{1}{1-\gamma(z_{1},z_{2})} \frac{1}{N} \tr{\bar{\boldsymbol{Q}}(z_{1})\boldsymbol{R}\bar{\boldsymbol{Q}}(z_{2})\boldsymbol{R}} \\
				&+ \frac{1}{N}\mean{\tr{\boldsymbol{B}\hat{\boldsymbol{Q}}(z_{1})\boldsymbol{R}\bar{\boldsymbol{Q}}(z_{2})\phi_{M}}}+ \mathcal{O}(N^{-1})\\
				=&\frac{1}{1-\gamma(z_{1},z_{2})}\frac{1}{N}\tr{\boldsymbol{B}\bar{\boldsymbol{Q}}(z_{1})\boldsymbol{R}\bar{\boldsymbol{Q}}(z_{2})}+ \mathcal{O}(N^{-1}),
			\end{aligned}
		\end{equation}
		where we have used $\gamma(z_{1},z_{2})$ in \eqref{eq: definition_gamma_proof_of_lemma}. Using the previously derived expectation in \eqref{eq:aux:lemma2: expectation1} we can develop
		\begin{align}
			\frac{1}{N}\mean{\tr{\boldsymbol{B}\hat{\boldsymbol{Q}}(z_{1})\boldsymbol{R}\hat{\boldsymbol{Q}}(z_{2})\boldsymbol{R}^{1/2}\phi_{M}}} =& 	\frac{1}{1-\gamma(z_{1},z_{2})} \frac{1}{N} \tr{\boldsymbol{B}\bar{\boldsymbol{Q}}(z_{1})\boldsymbol{R}\bar{\boldsymbol{Q}}(z_{2})\boldsymbol{R}^{1/2}} 	+ \mathcal{O}\left(N^{-1}\right), \label{eq:proof:aux:lemma2:expectation2: intermediate_14} \\
			\frac{1}{N}\mean{\tr{\hat{\boldsymbol{Q}}(z_{1})\boldsymbol{R}\hat{\boldsymbol{Q}}(z_{2})\boldsymbol{R}\phi_{M}}} =& \frac{1}{1-\gamma(z_{1},z_{2})} \frac{1}{N} \tr{\bar{\boldsymbol{Q}}(z_{1})\boldsymbol{R}\bar{\boldsymbol{Q}}(z_{2})\boldsymbol{R}} 	+ \mathcal{O}\left(N^{-1}\right) \label{eq:proof:aux:lemma2:expectation2: intermediate_15}
		\end{align}		
		Substituting \eqref{eq:proof:aux:lemma2:expectation2: intermediate_14} and \eqref{eq:proof:aux:lemma2:expectation2: intermediate_15} into \eqref{eq:proof:aux:lemma2:expectation2: intermediate_12} we obtain the expectation in \eqref{eq:aux:lemma2: expectation2} in Lemma \ref{aux:lemma2}
		\begin{equation*}
			\begin{aligned}
				\frac{1}{N}\mean{\tr{\boldsymbol{B}\hat{\boldsymbol{Q}}(z_{1})\boldsymbol{R}\hat{\boldsymbol{Q}}(z_{2})\boldsymbol{R}^{1/2}\frac{\boldsymbol{X}\boldsymbol{X}^{\operatorname{H}}}{N}\phi_{M}}}=&-\frac{z_{1}}{\omega(z_{1})}\frac{z_{2}}{\omega(z_{2})}\frac{1}{N} \tr{\boldsymbol{B}\bar{\boldsymbol{Q}}(z_{1})\boldsymbol{R}^{1/2}} \frac{1}{1-\gamma(z_{1},z_{2})}\frac{1}{N} \tr{\bar{\boldsymbol{Q}}(z_{1})\boldsymbol{R}\bar{\boldsymbol{Q}}(z_{2})\boldsymbol{R}} \\
				&+\frac{z_{2}}{\omega(z_{2})}\frac{1}{1-\gamma(z_{1},z_{2})}\frac{1}{N} \tr{\boldsymbol{B}\bar{\boldsymbol{Q}}(z_{1})\boldsymbol{R}\bar{\boldsymbol{Q}}(z_{2})\boldsymbol{R}^{1/2}}+ \mathcal{O}(N^{-1})\\
				=&\frac{z_{2}}{\omega(z_{2})}\frac{1}{1-\gamma(z_{1},z_{2})}\frac{1}{N} \tr{\boldsymbol{B}\bar{\boldsymbol{Q}}(z_{1})\boldsymbol{R}\bar{\boldsymbol{Q}}(z_{2})\boldsymbol{R}^{1/2}}\\
				&-\frac{\gamma(z_{1},z_{2})}{1-\gamma(z_{1},z_{2})} \frac{1}{N}\tr{\boldsymbol{B}\bar{\boldsymbol{Q}}(z_{1})\boldsymbol{R}^{1/2}}+ \mathcal{O}(N^{-1}).
			\end{aligned}
		\end{equation*}
		
\section{Proof of Lemma \ref{proof:aux:lemma1_1}}
	\label{proof:aux:lemma1_1}
	
	\subsection{Variance}
		Using the Nash-Poincare inequality we can upper-bound the variance of 
	    \begin{equation*}
	        \var{\frac{1}{N} \tr{\hat{\boldsymbol{Q}}(z)^{2}\hat{\boldsymbol{R}}\phi_{M}}} \leq \frac{1}{N^{2}} \sum_{i=1}^{M}\sum_{j=1}^{N} \left(\mean{\abs{\frac{\partial}{\partial X_{ij}}\tr{\hat{\boldsymbol{Q}}(z)^{2}\hat{\boldsymbol{R}}\phi_{M}}}^{2}} + \mean{\abs{\frac{\partial}{\partial X_{ij}^{*}}\tr{\hat{\boldsymbol{Q}}(z)^{2}\hat{\boldsymbol{R}}\phi_{M}}}^{2}}\right),
	    \end{equation*}
	    where
	    \begin{align*}
	        &\frac{\partial}{\partial X_{ij}}\tr{\hat{\boldsymbol{Q}}(z)^{2}\hat{\boldsymbol{R}}\phi_{M}} = \frac{\partial}{\partial X_{ij}} \sum_{l=1}^{M} \boldsymbol{e}_{l}^{\operatorname{H}} \hat{\boldsymbol{Q}}(z) \hat{\boldsymbol{R}} \hat{\boldsymbol{Q}}(z) \boldsymbol{e}_{l} \\
	        =&  \sum_{l=1}^{M} \boldsymbol{e}_{l}^{\operatorname{H}} \left(\frac{\partial}{\partial X_{ij}}\hat{\boldsymbol{Q}}(z)\right) \hat{\boldsymbol{R}} \hat{\boldsymbol{Q}}(z) \boldsymbol{e}_{l}\phi_{M} +  \sum_{l=1}^{M} \boldsymbol{e}_{l}^{\operatorname{H}} \hat{\boldsymbol{Q}}(z) \left(\frac{\partial}{\partial X_{ij}}\hat{\boldsymbol{R}}\right) \hat{\boldsymbol{Q}}(z) \boldsymbol{e}_{l}\phi_{M}  + \sum_{l=1}^{M} \boldsymbol{e}_{l}^{\operatorname{H}} \hat{\boldsymbol{Q}}(z) \hat{\boldsymbol{R}} \left(\frac{\partial}{\partial X_{ij}}\hat{\boldsymbol{Q}}(z)\right) \boldsymbol{e}_{l}\phi_{M} \\
	        =&-\sum_{l=1}^{M} \boldsymbol{e}_{l}^{\operatorname{H}} \hat{\boldsymbol{Q}}(z) \boldsymbol{R}^{1/2} \frac{\boldsymbol{e}_{i}\boldsymbol{x}_{j}^{\operatorname{H}}}{N} \boldsymbol{R}^{\operatorname{H}/2} \hat{\boldsymbol{Q}}(z) \hat{\boldsymbol{R}} \hat{\boldsymbol{Q}}(z) \boldsymbol{e}_{l}\phi_{M} + \sum_{l=1}^{M} \boldsymbol{e}_{l}^{\operatorname{H}} \hat{\boldsymbol{Q}}(z) \boldsymbol{R}^{1/2} \frac{\boldsymbol{e}_{i}\boldsymbol{x}_{j}^{\operatorname{H}}}{N} \boldsymbol{R}^{\operatorname{H}/2} \hat{\boldsymbol{Q}}(z)\boldsymbol{e}_{l}\phi_{M}\\
	        & -\sum_{l=1}^{M} \boldsymbol{e}_{l}^{\operatorname{H}} \hat{\boldsymbol{Q}}(z)\hat{\boldsymbol{R}} \hat{\boldsymbol{Q}}(z) \boldsymbol{R}^{1/2} \frac{\boldsymbol{e}_{i}\boldsymbol{x}_{j}^{\operatorname{H}}}{N} \boldsymbol{R}^{\operatorname{H}/2} \hat{\boldsymbol{Q}}(z)\boldsymbol{e}_{l} \phi_{M} + \mathcal{O}\left(N^{-\mathbb{N}}\right)
	    \end{align*}
	    and
	    \begin{align*}
	        &\frac{\partial}{\partial X_{ij}^{*}}\tr{\hat{\boldsymbol{Q}}(z)^{2}\hat{\boldsymbol{R}}\phi_{M}} = \frac{\partial}{\partial X_{ij}} \sum_{l=1}^{M} \boldsymbol{e}_{l}^{\operatorname{H}} \hat{\boldsymbol{Q}}(z) \hat{\boldsymbol{R}} \hat{\boldsymbol{Q}}(z) \boldsymbol{e}_{l}\phi_{M} \\
	        =&  \sum_{l=1}^{M} \boldsymbol{e}_{l}^{\operatorname{H}} \left(\frac{\partial}{\partial X_{ij}^{*}}\hat{\boldsymbol{Q}}(z)\right) \hat{\boldsymbol{R}} \hat{\boldsymbol{Q}}(z) \boldsymbol{e}_{l}\phi_{M} +  \sum_{l=1}^{M} \boldsymbol{e}_{l}^{\operatorname{H}} \hat{\boldsymbol{Q}}(z) \left(\frac{\partial}{\partial X_{ij}^{*}}\hat{\boldsymbol{R}}\right) \hat{\boldsymbol{Q}}(z) \boldsymbol{e}_{l}\phi_{M}  + \sum_{l=1}^{M} \boldsymbol{e}_{l}^{\operatorname{H}} \hat{\boldsymbol{Q}}(z) \hat{\boldsymbol{R}} \left(\frac{\partial}{\partial X_{ij}^{*}}\hat{\boldsymbol{Q}}(z)\right) \boldsymbol{e}_{l}\phi_{M} \\
	        =& - \sum_{l=1}^{M} \boldsymbol{e}_{l}^{\operatorname{H}} \hat{\boldsymbol{Q}}(z) \boldsymbol{R}^{1/2} \frac{\boldsymbol{x}_{j}\boldsymbol{e}_{i}^{\operatorname{T}}}{N} \boldsymbol{R}^{\operatorname{H}/2} \hat{\boldsymbol{Q}}(z) \hat{\boldsymbol{R}} \hat{\boldsymbol{Q}}(z) \boldsymbol{e}_{l}\phi_{M} +  \sum_{l=1}^{M} \boldsymbol{e}_{l}^{\operatorname{H}} \hat{\boldsymbol{Q}}(z) \boldsymbol{R}^{1/2}\frac{\boldsymbol{x}_{j}\boldsymbol{e}_{i}^{\operatorname{T}}}{N}\boldsymbol{R}^{\operatorname{H}/2} \hat{\boldsymbol{Q}}(z) \boldsymbol{e}_{l}\phi_{M} \\
	        &- \sum_{l=1}^{M} \boldsymbol{e}_{l}^{\operatorname{H}} \hat{\boldsymbol{Q}}(z) \hat{\boldsymbol{R}} \hat{\boldsymbol{Q}}(z) \boldsymbol{R}^{1/2} \frac{\boldsymbol{x}_{j}\boldsymbol{e}_{i}^{\operatorname{T}}}{N} \boldsymbol{R}^{\operatorname{H}/2} \hat{\boldsymbol{Q}}(z) \boldsymbol{e}_{l}\phi_{M}+ \mathcal{O}\left(N^{-\mathbb{N}}\right).
	    \end{align*}
    	Using Jensen's inequality we can upper-bound both derivatives by
    	\begin{align*}
	        \abs{ \frac{\partial}{\partial X_{ij}}\tr{\hat{\boldsymbol{Q}}(z)^{2}\hat{\boldsymbol{R}}\phi_{M}}}^{2}\leq& 3 \abs{\sum_{l=1}^{M} \boldsymbol{e}_{l}^{\operatorname{H}} \hat{\boldsymbol{Q}}(z) \boldsymbol{R}^{1/2} \frac{\boldsymbol{e}_{i}\boldsymbol{x}_{j}^{\operatorname{H}}}{N} \boldsymbol{R}^{\operatorname{H}/2} \hat{\boldsymbol{Q}}(z) \hat{\boldsymbol{R}} \hat{\boldsymbol{Q}}(z) \boldsymbol{e}_{l}\phi_{M}}^{2} \\
	        &+3\abs{\sum_{l=1}^{M} \boldsymbol{e}_{l}^{\operatorname{H}} \hat{\boldsymbol{Q}}(z) \boldsymbol{R}^{1/2} \frac{\boldsymbol{e}_{i}\boldsymbol{x}_{j}^{\operatorname{H}}}{N} \boldsymbol{R}^{\operatorname{H}/2} \hat{\boldsymbol{Q}}(z)\boldsymbol{e}_{l}\phi_{M}}^{2} \\
	        &+3\abs{\sum_{l=1}^{M} \boldsymbol{e}_{l}^{\operatorname{H}} \hat{\boldsymbol{Q}}(z)\hat{\boldsymbol{R}} \hat{\boldsymbol{Q}}(z) \boldsymbol{R}^{1/2} \frac{\boldsymbol{e}_{i}\boldsymbol{x}_{j}^{\operatorname{H}}}{N} \boldsymbol{R}^{\operatorname{H}/2} \hat{\boldsymbol{Q}}(z)\boldsymbol{e}_{l}\phi_{M}}^{2} + \mathcal{O}\left(N^{-\mathbb{N}}\right),
	    \end{align*}
	    and
	    \begin{align*}
	        \abs{ \frac{\partial}{\partial X_{ij}^{*}}\tr{\hat{\boldsymbol{Q}}(z)^{2}\hat{\boldsymbol{R}}\phi_{M}}}^{2}\leq& 3 \abs{\sum_{l=1}^{M} \boldsymbol{e}_{l}^{\operatorname{H}} \hat{\boldsymbol{Q}}(z) \boldsymbol{R}^{1/2} \frac{\boldsymbol{x}_{j}\boldsymbol{e}_{i}^{\operatorname{T}}}{N} \boldsymbol{R}^{\operatorname{H}/2} \hat{\boldsymbol{Q}}(z) \hat{\boldsymbol{R}} \hat{\boldsymbol{Q}}(z) \boldsymbol{e}_{l}\phi_{M}}^{2} \\
	        &+3\abs{\sum_{l=1}^{M} \boldsymbol{e}_{l}^{\operatorname{H}} \hat{\boldsymbol{Q}}(z) \boldsymbol{R}^{1/2}\frac{\boldsymbol{x}_{j}\boldsymbol{e}_{i}^{\operatorname{T}}}{N}\boldsymbol{R}^{\operatorname{H}/2} \hat{\boldsymbol{Q}}(z) \boldsymbol{e}_{l}\phi_{M}}^{2} \\
	        &+3\abs{\sum_{l=1}^{M} \boldsymbol{e}_{l}^{\operatorname{H}} \hat{\boldsymbol{Q}}(z) \hat{\boldsymbol{R}} \hat{\boldsymbol{Q}}(z) \boldsymbol{R}^{1/2} \frac{\boldsymbol{x}_{j}\boldsymbol{e}_{i}^{\operatorname{T}}}{N} \boldsymbol{R}^{\operatorname{H}/2} \hat{\boldsymbol{Q}}(z) \boldsymbol{e}_{l}\phi_{M}}^{2}+ \mathcal{O}\left(N^{-\mathbb{N}}\right),
	    \end{align*}
	    such that
	    \begin{equation*}
	    	\mean{\abs{\frac{\partial}{\partial X_{ij}}\tr{\hat{\boldsymbol{Q}}(z)^{2}\hat{\boldsymbol{R}}\phi_{M}}}^{2}} = \mathcal{O}\left(\frac{M}{N}\right) \qquad \mean{\abs{\frac{\partial}{\partial X_{ij}^{*}}\tr{\hat{\boldsymbol{Q}}(z)^{2}\hat{\boldsymbol{R}}\phi_{M}}}^{2}} = \mathcal{O}\left(\frac{M}{N}\right)
	    \end{equation*}
	    and consequently the variance of the quantity of interest decays with
	    \begin{equation*}
	    	 \var{\frac{1}{N} \tr{\hat{\boldsymbol{Q}}(z)^{2}\hat{\boldsymbol{R}}\phi_{M}}} = \mathcal{O}\left(N^{-2}\right).
	    \end{equation*}

	\subsection{Expectation}
		The expectation can be derived by taking the derivative of \eqref{eq:aux:lemma1: expectation2} with respect to $z$ for $\boldsymbol{B}=\boldsymbol{I}_{M}$ which yields
		\begin{align*}
			\frac{\partial}{\partial z} \frac{1}{N}\mean{\hat{\boldsymbol{Q}}(z) \hat{\boldsymbol{R}} \phi_{M}} =& \frac{1}{N} \frac{\partial \omega(z)}{\partial z} \tr{\boldsymbol{R}(\boldsymbol{R}-\omega(z)\boldsymbol{I}_{M})^{-2}\phi_{M}} + \mathcal{O}\left(N^{-1}\right) \\
			=&\frac{1}{N} \omega^{\prime}(z) \tr{\boldsymbol{R}(\boldsymbol{R}-\omega(z)\boldsymbol{I}_{M})^{-2}\phi_{M}} + \mathcal{O}\left(N^{-1}\right) \\
			=& \frac{z^{2}}{\omega(z)^{2}} \frac{1}{N} \tr{\boldsymbol{R}\bar{\boldsymbol{Q}}(z)^{2}} \omega^{\prime}(z) + \mathcal{O}\left(N^{-1}\right).
		\end{align*}

\section{Proof of Lemma \ref{aux:lemma3}}
	\label{proof:aux:lemma3}
	
	\subsection{Variance}
		The variance $\var{\boldsymbol{s}_{1}^{\operatorname{H}}\hat{\boldsymbol{Q}}(z_{2})\hat{\boldsymbol{R}}\boldsymbol{s}_{2}\phi_{M}}$ can be upper-bounded using the Nash-Poincar\'e inequality which states that
		\begin{equation}
			\var{\boldsymbol{s}_{1}^{\operatorname{H}}\hat{\boldsymbol{Q}}(z_{2})\hat{\boldsymbol{R}}\boldsymbol{s}_{2}\phi_{M}} \leq \sum_{i=1}^{M}\sum_{j=1}^{N} \left( \mean{\left | \frac{\partial}{\partial X_{ij}}\boldsymbol{s}_{1}^{\operatorname{H}}\hat{\boldsymbol{Q}}(z_{2})\hat{\boldsymbol{R}}\boldsymbol{s}_{2} \phi_{M} \right|^2 } + \mean{\left | \frac{\partial}{\partial X_{ij}^{*}}\boldsymbol{s}_{1}^{\operatorname{H}}\hat{\boldsymbol{Q}}(z_{2})\hat{\boldsymbol{R}}\boldsymbol{s}_{2} \phi_{M} \right|^2 }\right),
			\label{eq:proof:aux:lemma3:variance1: intermediate_1}
		\end{equation}
		with first order derivative with respect to $X_{ij}$
		\begin{equation}
			\frac{\partial}{\partial X_{ij}} \left( \boldsymbol{s}_{1}^{\operatorname{H}}\hat{\boldsymbol{Q}}(z_{2})\hat{\boldsymbol{R}}\boldsymbol{s}_{2} \phi_{M} \right) = \boldsymbol{s}_{1}^{\operatorname{H}} \hat{\boldsymbol{Q}}(z_{2}) \boldsymbol{R}^{1/2} \frac{\boldsymbol{e}_{i}\boldsymbol{x}_{j}^{\operatorname{H}}}{N}\boldsymbol{R}^{\operatorname{H}/2} \boldsymbol{s}_{2} \phi_{M} - \boldsymbol{s}_{1}^{\operatorname{H}} \hat{\boldsymbol{Q}}(z_{2})\boldsymbol{R}^{1/2} \frac{\boldsymbol{e}_{i}\boldsymbol{x}_{j}^{\operatorname{H}}}{N}\boldsymbol{R}^{\operatorname{H}/2} \hat{\boldsymbol{Q}}(z_{2}) \hat{\boldsymbol{R}} \boldsymbol{s}_{2} \phi_{M} + \mathcal{O}\left(N^{-\mathbb{N}}\right),
			\label{eq:proof:aux:lemma3:variance1: intermediate_2a}
		\end{equation}
		and first order derivative with respect to $X_{ij}^{*}$
		\begin{equation}
			\frac{\partial}{\partial X_{ij}^{*}} \left( \boldsymbol{s}_{1}^{\operatorname{H}}\hat{\boldsymbol{Q}}(z_{2})\hat{\boldsymbol{R}}\boldsymbol{s}_{2} \phi_{M} \right)= \boldsymbol{s}_{1}^{\operatorname{H}} \hat{\boldsymbol{Q}}(z_{2}) \boldsymbol{R}^{1/2} \frac{\boldsymbol{x}_{j}\boldsymbol{e}_{i}^{\operatorname{T}}}{N}\boldsymbol{R}^{\operatorname{H}/2} \boldsymbol{s}_{2}\phi_{M} - \boldsymbol{s}_{1}^{\operatorname{H}} \hat{\boldsymbol{Q}}(z_{2})\boldsymbol{R}^{1/2} \frac{\boldsymbol{x}_{j}\boldsymbol{e}_{i}^{\operatorname{T}}}{N}\boldsymbol{R}^{\operatorname{H}/2} \hat{\boldsymbol{Q}}(z_{2}) \hat{\boldsymbol{R}} \boldsymbol{s}_{2}\phi_{M}+ \mathcal{O}\left(N^{-\mathbb{N}}\right).
			\label{eq:proof:aux:lemma3:variance1: intermediate_2b}
		\end{equation}
		Jensen's inequality allows to upper-bound the absolute value squared of the first order derivative with respect to $X_{ij}$ in \eqref{eq:proof:aux:lemma3:variance1: intermediate_2a} by
		\begin{equation}
			\abs{\frac{\partial}{\partial X_{ij}} \left( \boldsymbol{s}_{1}^{\operatorname{H}}\hat{\boldsymbol{Q}}(z_{2})\hat{\boldsymbol{R}}\boldsymbol{s}_{2} \phi_{M} \right)}^{2} \leq 2 \abs{\boldsymbol{s}_{1}^{\operatorname{H}} \hat{\boldsymbol{Q}}(z_{2}) \boldsymbol{R}^{1/2} \frac{\boldsymbol{e}_{i}\boldsymbol{x}_{j}^{\operatorname{H}}}{N}\boldsymbol{R}^{\operatorname{H}/2} \boldsymbol{s}_{2}\phi_{M}}^{2} + 2\abs{\boldsymbol{s}_{1}^{\operatorname{H}} \hat{\boldsymbol{Q}}(z_{2})\boldsymbol{R}^{1/2} \frac{\boldsymbol{e}_{i}\boldsymbol{x}_{j}^{\operatorname{H}}}{N}\boldsymbol{R}^{\operatorname{H}/2} \hat{\boldsymbol{Q}}(z_{2}) \hat{\boldsymbol{R}} \boldsymbol{s}_{2}\phi_{M}}^{2}+ \mathcal{O}\left(N^{-\mathbb{N}}\right),
			\label{eq:proof:aux:lemma3:variance1: intermediate_3a}
		\end{equation}
		as well as the first order derivative with respect to $X_{ij}^{*}$ in \eqref{eq:proof:aux:lemma3:variance1: intermediate_2b} by
		\begin{equation}
			\abs{\frac{\partial}{\partial X_{ij}^{*}} \left( \boldsymbol{s}_{1}^{\operatorname{H}}\hat{\boldsymbol{Q}}(z_{2})\hat{\boldsymbol{R}}\boldsymbol{s}_{2} \phi_{M} \right)}^{2} \leq 2 \abs{\boldsymbol{s}_{1}^{\operatorname{H}} \hat{\boldsymbol{Q}}(z_{2}) \boldsymbol{R}^{1/2} \frac{\boldsymbol{x}_{j}\boldsymbol{e}_{i}^{\operatorname{T}}}{N}\boldsymbol{R}^{\operatorname{H}/2} \boldsymbol{s}_{2}\phi_{M}}^{2} + 2\abs{\boldsymbol{s}_{1}^{\operatorname{H}} \hat{\boldsymbol{Q}}(z_{2})\boldsymbol{R}^{1/2} \frac{\boldsymbol{x}_{j}\boldsymbol{e}_{i}^{\operatorname{T}}}{N}\boldsymbol{R}^{\operatorname{H}/2} \hat{\boldsymbol{Q}}(z_{2}) \hat{\boldsymbol{R}} \boldsymbol{s}_{2}\phi_{M}}^{2}+ \mathcal{O}\left(N^{-\mathbb{N}}\right).
			\label{eq:proof:aux:lemma3:variance1: intermediate_3b}
		\end{equation}
		From \eqref{eq:proof:aux:lemma3:variance1: intermediate_3a} and \eqref{eq:proof:aux:lemma3:variance1: intermediate_3b} it follows that
		\begin{equation}
			\sum_{i=1}^{M}\sum_{j=1}^{M}\mean{\abs{\frac{\partial}{\partial X_{ij}} \left(\boldsymbol{s}_{1}^{\operatorname{H}}\hat{\boldsymbol{Q}}(z_{2})\hat{\boldsymbol{R}}\boldsymbol{s}_{2}\phi_{M}\right)}^{2}} \leq \mathcal{O}\left(N^{-1}\right), \qquad \sum_{i=1}^{M}\sum_{j=1}^{M}\mean{\abs{\frac{\partial}{\partial X_{ij}^{*}} \left( \boldsymbol{s}_{1}^{\operatorname{H}}\hat{\boldsymbol{Q}}(z_{2})\hat{\boldsymbol{R}}\boldsymbol{s}_{2}\phi_{M}\right)}^{2}} \leq \mathcal{O}\left(N^{-1}\right).
			\label{eq:proof:aux:lemma3:variance1: intermediate_4}
		\end{equation}
		Substituting \eqref{eq:proof:aux:lemma3:variance1: intermediate_4} into \eqref{eq:proof:aux:lemma3:variance1: intermediate_1} we obtain the expression in Lemma \ref{aux:lemma3} in \eqref{eq:aux:lemma3: variance1}
		\begin{equation*}
			\var{\boldsymbol{s}_{1}^{\operatorname{H}}\hat{\boldsymbol{Q}}(z_{2})\hat{\boldsymbol{R}}\boldsymbol{s}_{2}\phi_{M}} \leq \sum_{i=1}^{M}\sum_{j=1}^{N} \left( \mean{\left | \frac{\partial}{\partial X_{ij}}\boldsymbol{s}_{1}^{\operatorname{H}}\hat{\boldsymbol{Q}}(z_{2})\hat{\boldsymbol{R}}\boldsymbol{s}_{2}\phi_{M} \right|^2 } + \mean{\left | \frac{\partial}{\partial X_{ij}^{*}}\boldsymbol{s}_{1}^{\operatorname{H}}\hat{\boldsymbol{Q}}(z_{2})\hat{\boldsymbol{R}}\boldsymbol{s}_{2} \phi_{M} \right|^2 }\right)  \leq \mathcal{O}\left(N^{-1}\right).
		\end{equation*}

	\subsection{Expectation}	
		In order to derive the expectation in \eqref{eq:aux:lemma3: expectation1} we use the expression of the sample covariance matrix in \eqref{eq: sample_covariance_matrix_sum} and apply the integration by parts formula to develop
		\begin{equation}
			\begin{aligned}
				\mean{\boldsymbol{s}_{1}^{\operatorname{H}} \hat{\boldsymbol{Q}}(z_{2})\hat{\boldsymbol{R}} \boldsymbol{s}_{2}\phi_{M}} =& \sum_{i=1}^{M}\sum_{j=1}	^{N} \mean{X_{ij} \boldsymbol{s}_{1}^{\operatorname{H}} \hat{\boldsymbol{Q}}(z_{2}) \boldsymbol{R}^{1/2} \frac{\boldsymbol{e}_{i}\boldsymbol{x}_{j}^{\operatorname{H}}}{N}\boldsymbol{R}^{\operatorname{H}/2}\boldsymbol{s}_{2}\phi_{M}}\\
				=& \sum_{i=1}^{M}\sum_{j=1}	^{N} \mean{\frac{\partial}{\partial X_{ij}^{*}} \left( \boldsymbol{s}_{1}^{\operatorname{H}} \hat{\boldsymbol{Q}}(z_{2}) \boldsymbol{R}^{1/2} \frac{\boldsymbol{e}_{i}\boldsymbol{x}_{j}^{\operatorname{H}}}{N}\boldsymbol{R}^{\operatorname{H}/2}\boldsymbol{s}_{2}\phi_{M}\right)},
			\end{aligned}
			\label{eq:proof:aux:lemma3:expectation1: intermediate_1}
		\end{equation}
		where the first order derivative with respect to $X_{ij}^{*}$ is given by
		\begin{equation}
			\begin{aligned}
				\frac{\partial}{\partial X_{ij}^{*}} \left( \boldsymbol{s}_{1}^{\operatorname{H}} \hat{\boldsymbol{Q}}(z_{2}) \boldsymbol{R}^{1/2} \frac{\boldsymbol{e}_{i}\boldsymbol{x}_{j}^{\operatorname{H}}}{N}\boldsymbol{R}^{\operatorname{H}/2}\boldsymbol{s}_{2}\phi_{M} \right) =& \boldsymbol{s}_{1}^{\operatorname{H}} \hat{\boldsymbol{Q}}(z_{2}) \boldsymbol{R}^{1/2} \frac{\boldsymbol{e}_{i}\boldsymbol{e}_{i}^{\operatorname{T}}}{N}\boldsymbol{R}^{\operatorname{H}/2} \boldsymbol{s}_{2}\phi_{M}\\
				 &- \boldsymbol{s}_{1}^{\operatorname{H}} \hat{\boldsymbol{Q}}(z_{2}) \boldsymbol{R}^{1/2} \frac{\boldsymbol{x}_{j}\boldsymbol{e}_{i}^{\operatorname{H}}}{N} \boldsymbol{R}^{\operatorname{H}/2} \hat{\boldsymbol{Q}}(z_{2}) \boldsymbol{R}^{1/2} \frac{\boldsymbol{e}_{i}\boldsymbol{x}_{j}^{\operatorname{H}}}{N} \boldsymbol{R}^{\operatorname{H}/2} \boldsymbol{s}_{2}\phi_{M}+ \mathcal{O}\left(N^{-\mathbb{N}}\right).
			\end{aligned} 	
			\label{eq:proof:aux:lemma3:expectation1: intermediate_2}
		\end{equation}
		Substituting \eqref{eq:proof:aux:lemma3:expectation1: intermediate_2} into \eqref{eq:proof:aux:lemma3:expectation1: intermediate_1} yields 
		\begin{equation}
			\mean{\boldsymbol{s}_{1}^{\operatorname{H}} \hat{\boldsymbol{Q}}(z_{2})\hat{\boldsymbol{R}} \boldsymbol{s}_{2}\phi_{M}} = \mean{\boldsymbol{s}_{1}^{\operatorname{H}} \hat{\boldsymbol{Q}}(z_{2}) \boldsymbol{R} \boldsymbol{s}_{2}\phi_{M}} - \mean{\boldsymbol{s}_{1}^{\operatorname{H}} \hat{\boldsymbol{Q}}(z_{2}) \hat{\boldsymbol{R}} \boldsymbol{s}_{2} \phi_{M}\frac{1}{N}\tr{\hat{\boldsymbol{Q}}(z_{2})\boldsymbol{R}}}+ \mathcal{O}\left(N^{-\mathbb{N}}\right).
			\label{eq:proof:aux:lemma3:expectation1: intermediate_3}
		\end{equation}
		Furthermore, we can use the definition of $\alpha(z_{2})$ in \eqref{eq: alpha} to express 
		\begin{equation}
			\phi_{M}\frac{1}{N}\tr{\hat{\boldsymbol{Q}}(z_{2})\boldsymbol{R}} = \alpha(z_{2}) + \frac{1}{N}\tr{\boldsymbol{R}\bar{\boldsymbol{Q}}(z_{2})}.
			\label{eq:proof:aux:lemma3:expectation1: intermediate_4}
		\end{equation}
		Replacing $\frac{1}{N}\tr{\hat{\boldsymbol{Q}}(z_{2})\boldsymbol{R}}$ in the second term on the right hand side of \eqref{eq:proof:aux:lemma3:expectation1: intermediate_3} by \eqref{eq:proof:aux:lemma3:expectation1: intermediate_4} leads to
		\begin{equation}
			\mean{\boldsymbol{s}_{1}^{\operatorname{H}} \hat{\boldsymbol{Q}}(z_{2})\hat{\boldsymbol{R}} \boldsymbol{s}_{2}\phi_{M}} = \mean{\boldsymbol{s}_{1}^{\operatorname{H}} \hat{\boldsymbol{Q}}(z_{2}) \boldsymbol{R} \boldsymbol{s}_{2}\phi_{M}} - \mean{\boldsymbol{s}_{1}^{\operatorname{H}} \hat{\boldsymbol{Q}}(z_{2}) \hat{\boldsymbol{R}} \boldsymbol{s}_{2} \phi_{M} \alpha(z_{2})} - \mean{\boldsymbol{s}_{1}^{\operatorname{H}} \hat{\boldsymbol{Q}}(z_{2}) \hat{\boldsymbol{R}} \boldsymbol{s}_{2}\phi_{M}}\frac{1}{N}\tr{\boldsymbol{R}\bar{\boldsymbol{Q}}(z_{2})}+ \mathcal{O}\left(N^{-\mathbb{N}}\right).
			\label{eq:proof:aux:lemma3:expectation1: intermediate_5}
		\end{equation}
		By bringing the third term of the right hand side of \eqref{eq:proof:aux:lemma3:expectation1: intermediate_5} to the left hand side and using $\frac{\omega(z_{2})}{z_{2}} = 1 + \frac{1}{N} \tr{\boldsymbol{R}\bar{\boldsymbol{Q}}(z_{2})}$ which follows from \eqref{eq: definition_omega_z} we obtain
		\begin{equation}
			\mean{\boldsymbol{s}_{1}^{\operatorname{H}} \hat{\boldsymbol{Q}}(z_{2})\hat{\boldsymbol{R}} \boldsymbol{s}_{2}\phi_{M}} \frac{\omega(z_{2})}{z_{2}}\\
				=\mean{\boldsymbol{s}_{1}^{\operatorname{H}} \hat{\boldsymbol{Q}}(z_{2}) \boldsymbol{R} \boldsymbol{s}_{2}\phi_{M}} - \mean{\boldsymbol{s}_{1}^{\operatorname{H}} \hat{\boldsymbol{Q}}(z_{2}) \hat{\boldsymbol{R}} \boldsymbol{s}_{2} \phi_{M} \alpha(z_{2})}+ \mathcal{O}\left(N^{-\mathbb{N}}\right).
			\label{eq:proof:aux:lemma3:expectation1: intermediate_6}	
		\end{equation}
		Multiplying by $\frac{z_{2}}{\omega(z_{2})}$ on both sides of \eqref{eq:proof:aux:lemma3:expectation1: intermediate_6}	 yields 
		\begin{equation}
			\mean{\boldsymbol{s}_{1}^{\operatorname{H}} \hat{\boldsymbol{Q}}(z_{2})\hat{\boldsymbol{R}} \boldsymbol{s}_{2}\phi_{M}} = \frac{z_{2}}{\omega(z_{2})} \mean{\boldsymbol{s}_{1}^{\operatorname{H}} \hat{\boldsymbol{Q}}(z_{2}) \boldsymbol{R} \boldsymbol{s}_{2}\phi_{M}} - \frac{z_{2}}{\omega(z_{2})}\mean{\boldsymbol{s}_{1}^{\operatorname{H}} \hat{\boldsymbol{Q}}(z_{2}) \hat{\boldsymbol{R}} \boldsymbol{s}_{2} \phi_{M} \alpha(z_{2})}+ \mathcal{O}\left(N^{-\mathbb{N}}\right).
			\label{eq:proof:aux:lemma3:expectation1: intermediate_7}	
		\end{equation}
		According to Lemma \ref{aux:lemma1} the expectation and the variance of $\alpha(z_{2})$ in \eqref{eq: alpha} decay with $\mean{\alpha(z_{2})}= \mathcal{O}(N^{-1})$ and $\var{\alpha(z_{2})}= \mathcal{O}(N^{-1})$. Hence, $\mean{\abs{\alpha(z_{2})}^2 } = \var{\alpha(z_{2})} + \mean{\alpha(z_{2})}^{2} = \mathcal{O}(N^{-2})$ and we can upper-bound the second term on the right hand side of \eqref{eq:proof:aux:lemma3:expectation1: intermediate_7}	 using the Cauchy-Schwarz inequality
		\begin{equation*}
			\mean{\boldsymbol{s}_{1}^{\operatorname{H}} \hat{\boldsymbol{Q}}(z_{2}) \hat{\boldsymbol{R}} \boldsymbol{s}_{2} \phi_{M}\alpha(z_{2})} \leq \sqrt{ \mean{\abs{\boldsymbol{s}_{1}^{\operatorname{H}} \hat{\boldsymbol{Q}}(z_{2}) \hat{\boldsymbol{R}} \boldsymbol{s}_{2}\phi_{M} }^2} \mean{\abs{\alpha(z_{2})}^2 }} = \mathcal{O}\left(N^{-1}\right),
		\end{equation*}
		and \eqref{eq:proof:aux:lemma3:expectation1: intermediate_7}	 simplifies to 
		\begin{equation}
			\mean{\boldsymbol{s}_{1}^{\operatorname{H}} \hat{\boldsymbol{Q}}(z_{2})\hat{\boldsymbol{R}} \boldsymbol{s}_{2}\phi_{M}} = \frac{z_{2}}{\omega(z_{2})} \mean{\boldsymbol{s}_{1}^{\operatorname{H}} \hat{\boldsymbol{Q}}(z_{2}) \boldsymbol{R} \boldsymbol{s}_{2}\phi_{M}} +\mathcal{O}\left(N^{-1}\right).
			\label{eq:proof:aux:lemma3:expectation1: intermediate_8}	
		\end{equation}
		By definition the resolvent $\hat{\boldsymbol{Q}}(z_{2})$ converges to its deterministic equivalent $\bar{\boldsymbol{Q}}(z_{2})$ in \eqref{eq: resolvent_and_deterministic_equivalent}. Hence, the expectation in \eqref{eq:aux:lemma3: expectation1} in Lemma \ref{aux:lemma3} is obtained by replacing $\hat{\boldsymbol{Q}}(z_{2})$ with $\bar{\boldsymbol{Q}}(z_{2})$ in \eqref{eq:proof:aux:lemma3:expectation1: intermediate_8}	 according to
		\begin{equation*}
			\mean{\boldsymbol{s}_{1}^{\operatorname{H}} \hat{\boldsymbol{Q}}(z_{2})\hat{\boldsymbol{R}} \boldsymbol{s}_{2}\phi_{M}} = \frac{z_{2}}{\omega(z_{2})} \boldsymbol{s}_{1}^{\operatorname{H}} \bar{\boldsymbol{Q}}(z_{2}) \boldsymbol{R} \boldsymbol{s}_{2} +\mathcal{O}\left(N^{-1}\right),
		\end{equation*}
		since $\mean{\boldsymbol{s}_{1}^{\operatorname{H}}\hat{\boldsymbol{Q}}(z_{2}) \boldsymbol{R} \boldsymbol{s}_{2}\phi_{M}} = \boldsymbol{s}_{1}^{\operatorname{H}} \mean{\hat{\boldsymbol{Q}}(z_{2})\phi_{M}} \boldsymbol{R} \boldsymbol{s}_{2}$.

\section{Proof of Lemma \ref{aux:lemma4}}
	\label{proof:aux:lemma4}
	
	\subsection{Variance}
		The variance $\var{\boldsymbol{s}_{1}^{\operatorname{H}} \hat{\boldsymbol{Q}}(z_{1}) \boldsymbol{R} \hat{\boldsymbol{Q}}(z_{2}) \boldsymbol{s}_{2}\phi_{M}}$ can be upper-bounded by applying the Nash-Poincar\'e inequality
		\begin{equation}
			\var{\boldsymbol{s}_{1}^{\operatorname{H}} \hat{\boldsymbol{Q}}(z_{1}) \boldsymbol{R} \hat{\boldsymbol{Q}}(z_{2}) \boldsymbol{s}_{2} \phi_{M}} \leq \sum_{i=1}^{M}\sum_{j=1}^{N} \left( \mean{\abs{\frac{\partial }{\partial X_{ij}}\boldsymbol{s}_{1}^{\operatorname{H}} \hat{\boldsymbol{Q}}(z_{1}) \boldsymbol{R} \hat{\boldsymbol{Q}}(z_{2}) \boldsymbol{s}_{2}\phi_{M}}^2} + \mean{\abs{\frac{\partial}{\partial X_{ij}^{*}}\boldsymbol{s}_{1}^{\operatorname{H}} \hat{\boldsymbol{Q}}(z_{1}) \boldsymbol{R} \hat{\boldsymbol{Q}}(z_{2}) \boldsymbol{s}_{2}\phi_{M}}^2} \right),
			\label{eq:proof:aux:lemma4:variance1: intermediate_1}
		\end{equation}
		where the first order derivative with respect to $X_{ij}$ is given by
		\begin{equation}
			\begin{aligned}
				\frac{\partial }{\partial X_{ij}} \boldsymbol{s}_{1}^{\operatorname{H}} \hat{\boldsymbol{Q}}(z_{1}) \boldsymbol{R} \hat{\boldsymbol{Q}}(z_{2}) \boldsymbol{s}_{2}\phi_{M} =& -\boldsymbol{s}_{1}^{\operatorname{H}} \hat{\boldsymbol{Q}}(z_{1}) \boldsymbol{R}^{1/2} \frac{\boldsymbol{e}_{i}\boldsymbol{x}_{j}^{\operatorname{H}}}{N} \boldsymbol{R}^{\operatorname{H}/2} \hat{\boldsymbol{Q}}(z_{1}) \boldsymbol{R} \hat{\boldsymbol{Q}}(z_{2}) \boldsymbol{s}_{2}\phi_{M}\\
				&- \boldsymbol{s}_{1}^{\operatorname{H}} \hat{\boldsymbol{Q}}(z_{1}) \boldsymbol{R} \hat{\boldsymbol{Q}}(z_{2}) \boldsymbol{R}^{1/2} \frac{\boldsymbol{e}_{i}\boldsymbol{x}_{j}^{\operatorname{H}}}{N}\boldsymbol{R}^{\operatorname{H}/2} \hat{\boldsymbol{Q}}(z_{2}) \boldsymbol{s}_{2}\phi_{M} + \mathcal{O}\left(N^{-\mathbb{N}}\right),
			\end{aligned}
			\label{eq:proof:aux:lemma4:variance1: intermediate_2a}
		\end{equation}
		and the first order derivative with respect to $X_{ij}^{*}$ is given by
		\begin{equation}
			\begin{aligned}
				\frac{\partial }{\partial X_{ij}^{*}} \boldsymbol{s}_{1}^{\operatorname{H}} \hat{\boldsymbol{Q}}(z_{1}) \boldsymbol{R} \hat{\boldsymbol{Q}}(z_{2}) \boldsymbol{s}_{2}\phi_{M} =& -\boldsymbol{s}_{1}^{\operatorname{H}} \hat{\boldsymbol{Q}}(z_{1}) \boldsymbol{R}^{1/2} \frac{\boldsymbol{x}_{j}\boldsymbol{e}_{i}^{\operatorname{T}}}{N} \boldsymbol{R}^{\operatorname{H}/2} \hat{\boldsymbol{Q}}(z_{1}) \boldsymbol{R} \hat{\boldsymbol{Q}}(z_{2}) \boldsymbol{s}_{2}\phi_{M}\\
				 &- \boldsymbol{s}_{1}^{\operatorname{H}} \hat{\boldsymbol{Q}}(z_{1}) \boldsymbol{R} \hat{\boldsymbol{Q}}(z_{2}) \boldsymbol{R}^{1/2} \frac{\boldsymbol{x}_{j}\boldsymbol{e}_{i}^{\operatorname{T}}}{N} \boldsymbol{R}^{\operatorname{H}/2} \hat{\boldsymbol{Q}}(z_{2}) \boldsymbol{s}_{2}\phi_{M}+ \mathcal{O}\left(N^{-\mathbb{N}}\right).
			\end{aligned}
			\label{eq:proof:aux:lemma4:variance1: intermediate_2b}		
		\end{equation}
		Using Jensen's inequality we can upper-bound the squared absolute value of the first order derivative with respect to $X_{ij}$ in \eqref{eq:proof:aux:lemma4:variance1: intermediate_2a} by
		\begin{equation}
			\begin{aligned}
				\abs{\frac{\partial }{\partial X_{ij}} \boldsymbol{s}_{1}^{\operatorname{H}} \hat{\boldsymbol{Q}}(z_{1}) \boldsymbol{R} \hat{\boldsymbol{Q}}(z_{2}) \boldsymbol{s}_{2} \phi_{M}}^{2} \leq & 2 \abs{\boldsymbol{s}_{1}^{\operatorname{H}} \hat{\boldsymbol{Q}}(z_{1}) \boldsymbol{R}^{1/2} \frac{\boldsymbol{e}_{i}\boldsymbol{x}_{j}^{\operatorname{H}}}{N} \boldsymbol{R}^{\operatorname{H}/2} \hat{\boldsymbol{Q}}(z_{1}) \boldsymbol{R} \hat{\boldsymbol{Q}}(z_{2}) \boldsymbol{s}_{2}\phi_{M}}^{2} \\
				&+ 2 \abs{\boldsymbol{s}_{1}^{\operatorname{H}} \hat{\boldsymbol{Q}}(z_{1}) \boldsymbol{R} \hat{\boldsymbol{Q}}(z_{2}) \boldsymbol{R}^{1/2} \frac{\boldsymbol{e}_{i}\boldsymbol{x}_{j}^{\operatorname{H}}}{N}\boldsymbol{R}^{\operatorname{H}/2} \hat{\boldsymbol{Q}}(z_{2}) \boldsymbol{s}_{2}\phi_{M}}^{2}+ \mathcal{O}\left(N^{-\mathbb{N}}\right)
			\end{aligned}
			\label{eq:proof:aux:lemma4:variance1: intermediate_3a}
		\end{equation}
		as well as the squared absolute value of the first oder derivative with respect to $X_{ij}^{*}$ in \eqref{eq:proof:aux:lemma4:variance1: intermediate_2b} by
		\begin{equation}
			\begin{aligned}
				\abs{\frac{\partial }{\partial X_{ij}^{*}} \boldsymbol{s}_{1}^{\operatorname{H}} \hat{\boldsymbol{Q}}(z_{1}) \boldsymbol{R} \hat{\boldsymbol{Q}}(z_{2}) \boldsymbol{s}_{2}\phi_{M}}^{2} \leq & 2 \abs{\boldsymbol{s}_{1}^{\operatorname{H}} \hat{\boldsymbol{Q}}(z_{1}) \boldsymbol{R}^{1/2} \frac{\boldsymbol{x}_{j}\boldsymbol{e}_{i}^{\operatorname{T}}}{N} \boldsymbol{R}^{\operatorname{H}/2} \hat{\boldsymbol{Q}}(z_{1}) \boldsymbol{R} \hat{\boldsymbol{Q}}(z_{2}) \boldsymbol{s}_{2}\phi_{M}}^{2} \\
				&+ 2 \abs{\boldsymbol{s}_{1}^{\operatorname{H}} \hat{\boldsymbol{Q}}(z_{1}) \boldsymbol{R} \hat{\boldsymbol{Q}}(z_{2}) \boldsymbol{R}^{1/2} \frac{\boldsymbol{x}_{j}\boldsymbol{e}_{i}^{\operatorname{T}}}{N} \boldsymbol{R}^{\operatorname{H}/2} \hat{\boldsymbol{Q}}(z_{2}) \boldsymbol{s}_{2}\phi_{M}}^{2}+ \mathcal{O}\left(N^{-\mathbb{N}}\right).
			\end{aligned}
			\label{eq:proof:aux:lemma4:variance1: intermediate_3b}
		\end{equation}
		It can be observed in \eqref{eq:proof:aux:lemma4:variance1: intermediate_3a} and \eqref{eq:proof:aux:lemma4:variance1: intermediate_3b} that 	
		\begin{equation*}
			\sum_{i=1}^{M}\sum_{j=1}^{M} \mean{ \abs{\frac{\partial }{\partial X_{ij}} \boldsymbol{s}_{1}^{\operatorname{H}} \hat{\boldsymbol{Q}}(z_{1}) \boldsymbol{R} \hat{\boldsymbol{Q}}(z_{2}) \boldsymbol{s}_{2} \phi_{M}}^{2}} \leq \mathcal{O}\left(N^{-1}\right), \qquad \sum_{i=1}^{M}\sum_{j=1}^{M} \mean{\abs{\frac{\partial }{\partial X_{ij}^{*}} \boldsymbol{s}_{1}^{\operatorname{H}} \hat{\boldsymbol{Q}}(z_{1}) \boldsymbol{R} \hat{\boldsymbol{Q}}(z_{2}) \boldsymbol{s}_{2} \phi_{M}}^{2}} \leq \mathcal{O}\left(N^{-1}\right),
		\end{equation*}
		and as a direct consequence the variance in \eqref{eq:proof:aux:lemma4:variance1: intermediate_1} decays with $\var{\boldsymbol{s}_{1}^{\operatorname{H}} \hat{\boldsymbol{Q}}(z_{1}) \boldsymbol{R} \hat{\boldsymbol{Q}}(z_{2}) \boldsymbol{s}_{2}\phi_{M}}\leq \mathcal{O}(N^{-1})$ which coincides with the variance in \eqref{eq:aux:lemma4: variance1} in Lemma \ref{aux:lemma4}.

	\subsection{Expectations}
		In order to derive the expectation in \eqref{eq:aux:lemma4: expectation1} we use the expression of the sample covariance matrix $\hat{\boldsymbol{R}}$ in \eqref{eq: sample_covariance_matrix_sum} and apply the integration by parts formula
		\begin{equation}
			\begin{aligned}
				\mean{\boldsymbol{s}_{1}^{\operatorname{H}} \hat{\boldsymbol{Q}}(z_{1}) \boldsymbol{R} \hat{\boldsymbol{Q}}(z_{2}) \hat{\boldsymbol{R}}\boldsymbol{s}_{2}\phi_{M}} =& \sum_{i=1}^{M}\sum_{j=1}^{N}\mean{ X_{ij} \boldsymbol{s}_{1}^{\operatorname{H}} \hat{\boldsymbol{Q}}(z_{1}) \boldsymbol{R} \hat{\boldsymbol{Q}}(z_{2}) \boldsymbol{R}^{1/2} \frac{\boldsymbol{e}_{i}\boldsymbol{x}_{j}^{\operatorname{H}}}{N} \boldsymbol{R}^{\operatorname{H}/2} \boldsymbol{s}_{2}\phi_{M}}\\
				=& \sum_{i=1}^{M}\sum_{j=1}^{N}\mean{ \frac{\partial}{\partial X_{ij}^{*}} \boldsymbol{s}_{1}^{\operatorname{H}} \hat{\boldsymbol{Q}}(z_{1}) \boldsymbol{R} \hat{\boldsymbol{Q}}(z_{2}) \boldsymbol{R}^{1/2} \frac{\boldsymbol{e}_{i}\boldsymbol{x}_{j}^{\operatorname{H}}}{N} \boldsymbol{R}^{\operatorname{H}/2} \boldsymbol{s}_{2}\phi_{M}},
			\end{aligned}
			\label{eq:proof:aux:lemma4:expectation1: intermediate_1}
		\end{equation}
		where the first order derivative with respect to $X_{ij}^{*}$ is given by
		\begin{equation}
			\begin{aligned}
				\frac{\partial}{\partial X_{ij}^{*}} \boldsymbol{s}_{1}^{\operatorname{H}} \hat{\boldsymbol{Q}}(z_{1}) \boldsymbol{R} \hat{\boldsymbol{Q}}(z_{2}) \boldsymbol{R}^{1/2} \frac{\boldsymbol{e}_{i}\boldsymbol{x}_{j}^{\operatorname{H}}}{N} \boldsymbol{R}^{\operatorname{H}/2} \boldsymbol{s}_{2}\phi_{M} =& -\boldsymbol{s}_{1}^{\operatorname{H}} \hat{\boldsymbol{Q}}(z_{1}) \boldsymbol{R}^{1/2} \frac{\boldsymbol{x}_{j}\boldsymbol{e}_{i}^{\operatorname{T}}}{N}\boldsymbol{R}^{\operatorname{H}/2} \hat{\boldsymbol{Q}}(z_{1}) \boldsymbol{R} \hat{\boldsymbol{Q}}(z_{2}) \boldsymbol{R}^{1/2}\frac{\boldsymbol{e}_{i}\boldsymbol{x}_{j}^{\operatorname{H}}}{N}\boldsymbol{R}^{\operatorname{H}/2}\boldsymbol{s}_{2}\phi_{M} \\
				&- \boldsymbol{s}_{1}^{\operatorname{H}} \hat{\boldsymbol{Q}}(z_{1}) \boldsymbol{R} \hat{\boldsymbol{Q}}(z_{2}) \boldsymbol{R}^{1/2} \frac{\boldsymbol{x}_{j}\boldsymbol{e}_{i}^{\operatorname{H}}}{N}\boldsymbol{R}^{\operatorname{H}/2} \hat{\boldsymbol{Q}}(z_{2})\boldsymbol{R}^{1/2}\frac{\boldsymbol{e}_{i}\boldsymbol{x}_{j}^{\operatorname{H}}}{N}\boldsymbol{R}^{\operatorname{H}/2}\boldsymbol{s}_{2}\phi_{M} \\
				&+ \boldsymbol{s}_{1}^{\operatorname{H}} \hat{\boldsymbol{Q}}(z_{1}) \boldsymbol{R} \hat{\boldsymbol{Q}}(z_{2}) \boldsymbol{R}^{1/2} \frac{\boldsymbol{e}_{i}\boldsymbol{e}_{i}^{\operatorname{T}}}{N}\boldsymbol{R}^{\operatorname{H}/2} \boldsymbol{s}_{2}\phi_{M}+ \mathcal{O}\left(N^{-\mathbb{N}}\right).
			\end{aligned}
			\label{eq:proof:aux:lemma4:expectation1: intermediate_2}
		\end{equation}
		Substituting \eqref{eq:proof:aux:lemma4:expectation1: intermediate_2} into \eqref{eq:proof:aux:lemma4:expectation1: intermediate_1} and computing the double sum yields
		\begin{equation}
			\begin{aligned}
				\mean{\boldsymbol{s}_{1}^{\operatorname{H}} \hat{\boldsymbol{Q}}(z_{1}) \boldsymbol{R} \hat{\boldsymbol{Q}}(z_{2}) \hat{\boldsymbol{R}}\boldsymbol{s}_{2}\phi_{M}} =& -\mean{\boldsymbol{s}_{1}^{\operatorname{H}} \hat{\boldsymbol{Q}}(z_{1}) \hat{\boldsymbol{R}}\boldsymbol{s}_{2} \phi_{M} \frac{1}{N} \tr{\hat{\boldsymbol{Q}}(z_{1}) \boldsymbol{R} \hat{\boldsymbol{Q}}(z_{2}) \boldsymbol{R}}}  -\mean{\boldsymbol{s}_{1}^{\operatorname{H}} \hat{\boldsymbol{Q}}(z_{1}) \boldsymbol{R} \hat{\boldsymbol{Q}}(z_{2}) \hat{\boldsymbol{R}} \boldsymbol{s}_{2} \phi_{M} \frac{1}{N} \tr{\boldsymbol{R}\hat{\boldsymbol{Q}}(z_{2})}} \\
				& + \mean{\boldsymbol{s}_{1}^{\operatorname{H}} \hat{\boldsymbol{Q}}(z_{1}) \boldsymbol{R} \hat{\boldsymbol{Q}}(z_{2}) \boldsymbol{R} \boldsymbol{s}_{2}\phi_{M}}+ \mathcal{O}\left(N^{-\mathbb{N}}\right).
			\end{aligned}
			\label{eq:proof:aux:lemma4:expectation1: intermediate_3}
		\end{equation} 
		Using the definition of $\alpha(z_{2})$ in \eqref{eq: alpha} we know that 
		\begin{equation*}
			\phi_{M} \frac{1}{N}\tr{\boldsymbol{R}\hat{\boldsymbol{Q}}(z_{2})} = \alpha(z_{2}) + \frac{1}{N} \tr{\boldsymbol{R}\bar{\boldsymbol{Q}}(z_{2})},
		\end{equation*}
		which allows to express \eqref{eq:proof:aux:lemma4:expectation1: intermediate_3} as
		\begin{equation}
			\begin{aligned}
				\mean{\boldsymbol{s}_{1}^{\operatorname{H}} \hat{\boldsymbol{Q}}(z_{1}) \boldsymbol{R} \hat{\boldsymbol{Q}}(z_{2}) \hat{\boldsymbol{R}}\boldsymbol{s}_{2}\phi_{M}} =& -\mean{\boldsymbol{s}_{1}^{\operatorname{H}} \hat{\boldsymbol{Q}}(z_{1}) \hat{\boldsymbol{R}}\boldsymbol{s}_{2}\phi_{M} \frac{1}{N} \tr{\hat{\boldsymbol{Q}}(z_{1}) \boldsymbol{R} \hat{\boldsymbol{Q}}(z_{2}) \boldsymbol{R}}}  -\mean{\boldsymbol{s}_{1}^{\operatorname{H}} \hat{\boldsymbol{Q}}(z_{1}) \boldsymbol{R} \hat{\boldsymbol{Q}}(z_{2}) \hat{\boldsymbol{R}} \boldsymbol{s}_{2} \phi_{M} \alpha(z_{2})} \\
				& -\mean{\boldsymbol{s}_{1}^{\operatorname{H}} \hat{\boldsymbol{Q}}(z_{1}) \boldsymbol{R} \hat{\boldsymbol{Q}}(z_{2}) \hat{\boldsymbol{R}} \boldsymbol{s}_{2}\phi_{M} } \frac{1}{N} \tr{\boldsymbol{R}\bar{\boldsymbol{Q}}(z_{2})}  + \mean{\boldsymbol{s}_{1}^{\operatorname{H}} \hat{\boldsymbol{Q}}(z_{1}) \boldsymbol{R} \hat{\boldsymbol{Q}}(z_{2}) \boldsymbol{R} \boldsymbol{s}_{2}\phi_{M}}+ \mathcal{O}\left(N^{-\mathbb{N}}\right).
			\end{aligned}
			\label{eq:proof:aux:lemma4:expectation1: intermediate_4}
		\end{equation}
		Bringing the third term on the right hand side of \eqref{eq:proof:aux:lemma4:expectation1: intermediate_4} to the other side  and using $1+\frac{1}{N}\tr{\boldsymbol{R}\bar{\boldsymbol{Q}}(z_{2})} = \frac{\omega(z_{2})}{z_{2}}$ which follows from \eqref{eq: definition_omega_z} we obtain
		\begin{equation}
			\begin{aligned}
				\frac{\omega(z_{2})}{z_{2}}\mean{\boldsymbol{s}_{1}^{\operatorname{H}} \hat{\boldsymbol{Q}}(z_{1}) \boldsymbol{R} \hat{\boldsymbol{Q}}(z_{2}) \hat{\boldsymbol{R}}\boldsymbol{s}_{2}\phi_{M}}=& -\mean{\boldsymbol{s}_{1}^{\operatorname{H}} \hat{\boldsymbol{Q}}(z_{1}) \hat{\boldsymbol{R}}\boldsymbol{s}_{2} \phi_{M} \frac{1}{N} \tr{\hat{\boldsymbol{Q}}(z_{1}) \boldsymbol{R} \hat{\boldsymbol{Q}}(z_{2}) \boldsymbol{R}}} -\mean{\boldsymbol{s}_{1}^{\operatorname{H}} \hat{\boldsymbol{Q}}(z_{1}) \boldsymbol{R} \hat{\boldsymbol{Q}}(z_{2}) \hat{\boldsymbol{R}} \boldsymbol{s}_{2} \phi_{M}\alpha(z_{2})}\\
				& + \mean{\boldsymbol{s}_{1}^{\operatorname{H}} \hat{\boldsymbol{Q}}(z_{1}) \boldsymbol{R} \hat{\boldsymbol{Q}}(z_{2}) \boldsymbol{R} \boldsymbol{s}_{2}\phi_{M}}+ \mathcal{O}\left(N^{-\mathbb{N}}\right).
			\end{aligned}
			\label{eq:proof:aux:lemma4:expectation1: intermediate_5}
		\end{equation}
		By adding and subtracting $z_{2}\boldsymbol{I}_{M}$ to $\mean{\boldsymbol{s}_{1}^{\operatorname{H}} \hat{\boldsymbol{Q}}(z_{1}) \boldsymbol{R} \hat{\boldsymbol{Q}}(z_{2}) \hat{\boldsymbol{R}}\boldsymbol{s}_{2}\phi_{M}}$ and using the definition of the resolvent $\hat{\boldsymbol{Q}}(z_{2})$ in \eqref{eq: resolvent_and_deterministic_equivalent} we can equivalently express the left hand side of \eqref{eq:proof:aux:lemma4:expectation1: intermediate_5} as
		\begin{equation}
			\begin{aligned}
				\frac{\omega(z_{2})}{z_{2}}\mean{\boldsymbol{s}_{1}^{\operatorname{H}}\hat{\boldsymbol{Q}}(z_{1}) \boldsymbol{R}\hat{\boldsymbol{Q}}(z_{2})\hat{\boldsymbol{R}} \boldsymbol{s}_{2}\phi_{M}} =& 	\frac{\omega(z_{2})}{z_{2}} \mean{\boldsymbol{s}_{1}^{\operatorname{H}} \hat{\boldsymbol{Q}}(z_{1}) \boldsymbol{R} \hat{\boldsymbol{Q}}(z_{2}) \left(\hat{\boldsymbol{R}} - z_{2}\boldsymbol{I}_{M} + z_{2}\boldsymbol{I}_{M} \right) \boldsymbol{s}_{2}\phi_{M}}\\
				=&\frac{\omega(z_{2})}{z_{2}} \mean{\boldsymbol{s}_{1}^{\operatorname{H}} \hat{\boldsymbol{Q}}(z_{1}) \boldsymbol{R} \hat{\boldsymbol{Q}}(z_{2}) \left(\hat{\boldsymbol{R}} - z_{2}\boldsymbol{I}_{M} \right) \boldsymbol{s}_{2}\phi_{M}}+ \frac{\omega(z_{2})}{z_{2}} z_{2} \mean{\boldsymbol{s}_{1}^{\operatorname{H}} \hat{\boldsymbol{Q}}(z_{1}) \boldsymbol{R} \hat{\boldsymbol{Q}}(z_{2}) \boldsymbol{s}_{2}\phi_{M}}\\
				=&	\frac{\omega(z_{2})}{z_{2}} \mean{\boldsymbol{s}_{1}^{\operatorname{H}} \hat{\boldsymbol{Q}}(z_{1}) \boldsymbol{R} \boldsymbol{s}_{2}\phi_{M}}+ \omega(z_{2}) \mean{\boldsymbol{s}_{1}^{\operatorname{H}} \hat{\boldsymbol{Q}}(z_{1}) \boldsymbol{R} \hat{\boldsymbol{Q}}(z_{2}) \boldsymbol{s}_{2}\phi_{M}}.
			\end{aligned}
			\label{eq:proof:aux:lemma4:expectation1: intermediate_6}
		\end{equation}
		Replacing the left hand side of \eqref{eq:proof:aux:lemma4:expectation1: intermediate_5} with the expression obtained in \eqref{eq:proof:aux:lemma4:expectation1: intermediate_6} and subtracting $\omega(z_{2}) \mean{\boldsymbol{s}_{1}^{\operatorname{H}} \hat{\boldsymbol{Q}}(z_{1}) \boldsymbol{R} \hat{\boldsymbol{Q}}(z_{2}) \boldsymbol{s}_{2}\phi_{M}}$ on both sides yields
		\begin{equation}
			\begin{aligned}
				\frac{\omega(z_{2})}{z_{2}} \mean{\boldsymbol{s}_{1}^{\operatorname{H}} \hat{\boldsymbol{Q}}(z_{1}) \boldsymbol{R} \boldsymbol{s}_{2}\phi_{M}} =& -\mean{\boldsymbol{s}_{1}^{\operatorname{H}} \hat{\boldsymbol{Q}}(z_{1}) \hat{\boldsymbol{R}}\boldsymbol{s}_{2} \phi_{M} \frac{1}{N} \tr{\hat{\boldsymbol{Q}}(z_{1}) \boldsymbol{R} \hat{\boldsymbol{Q}}(z_{2}) \boldsymbol{R}}}  \\
				& -\mean{\boldsymbol{s}_{1}^{\operatorname{H}} \hat{\boldsymbol{Q}}(z_{1}) \boldsymbol{R} \hat{\boldsymbol{Q}}(z_{2}) \hat{\boldsymbol{R}} \boldsymbol{s}_{2} \phi_{M}\alpha(z_{2})} + \mean{\boldsymbol{s}_{1}^{\operatorname{H}} \hat{\boldsymbol{Q}}(z_{1}) \boldsymbol{R} \hat{\boldsymbol{Q}}(z_{2}) \left( \boldsymbol{R} - \omega(z_{2})\boldsymbol{I}_{M}\right) \boldsymbol{s}_{2}\phi_{M}}+ \mathcal{O}\left(N^{-\mathbb{N}}\right).
			\end{aligned}
			\label{eq:proof:aux:lemma4:expectation1: intermediate_7}
		\end{equation}
		Next, we decorrelate the first term on the right hand side of \eqref{eq:proof:aux:lemma4:expectation1: intermediate_7} by analyzing the covariance
		\begin{equation}
			\begin{aligned}
				\cov{\boldsymbol{s}_{1}^{\operatorname{H}} \hat{\boldsymbol{Q}}(z_{1}) \hat{\boldsymbol{R}} \boldsymbol{s}_{2}\phi_{M}}{\phi_{M}\frac{1}{N} \tr{\hat{\boldsymbol{Q}}(z_{1}) \boldsymbol{R} \hat{\boldsymbol{Q}}(z_{2}) \boldsymbol{R}}}=&\mean{\left(\boldsymbol{s}_{1}^{\operatorname{H}} \hat{\boldsymbol{Q}}(z_{1}) \hat{\boldsymbol{R}} \boldsymbol{s}_{2}\phi_{M}\right)^{(\circ)}\left(\phi_{M}\frac{1}{N} \tr{\hat{\boldsymbol{Q}}(z_{1}) \boldsymbol{R} \hat{\boldsymbol{Q}}(z_{2}) \boldsymbol{R}}\right)^{(\circ)} } \\
				=& \mean{\boldsymbol{s}_{1}^{\operatorname{H}} \hat{\boldsymbol{Q}}(z_{1}) \hat{\boldsymbol{R}}\boldsymbol{s}_{2} \phi_{M} \frac{1}{N} \tr{\hat{\boldsymbol{Q}}(z_{1}) \boldsymbol{R} \hat{\boldsymbol{Q}}(z_{2}) \boldsymbol{R}}} \\
				&- \mean{\boldsymbol{s}_{1}^{\operatorname{H}}\hat{\boldsymbol{Q}}(z_{1}) \hat{\boldsymbol{R}} \boldsymbol{s}_{2}\phi_{M}} \mean{\phi_{M}\frac{1}{N} \tr{\hat{\boldsymbol{Q}}(z_{1}) \boldsymbol{R} \hat{\boldsymbol{Q}}(z_{2}) \boldsymbol{R}}}.
			\end{aligned}
			\label{eq:proof:aux:lemma4:expectation1: intermediate_8}
		\end{equation}
		Using the Cauchy-Schwarz inequality we can upper-bound  
		\begin{equation}
			\begin{aligned}
				&\abs{\mean{\left(\boldsymbol{s}_{1}^{\operatorname{H}} \hat{\boldsymbol{Q}}(z_{1}) \hat{\boldsymbol{R}} \boldsymbol{s}_{2}\phi_{M}\right)^{(\circ)}\left(\phi_{M}\frac{1}{N} \tr{\hat{\boldsymbol{Q}}(z_{1}) \boldsymbol{R} \hat{\boldsymbol{Q}}(z_{2}) \boldsymbol{R}}\right)^{(\circ)} }}^2 \\
				\leq & \mean{\abs{\left(\boldsymbol{s}_{1}^{\operatorname{H}} \hat{\boldsymbol{Q}}(z_{1}) \hat{\boldsymbol{R}} \boldsymbol{s}_{2}\phi_{M}\right)^{(\circ)}}^2} \mean{\abs{\left(\phi_{M}\frac{1}{N} \tr{\hat{\boldsymbol{Q}}(z_{1}) \boldsymbol{R} \hat{\boldsymbol{Q}}(z_{2}) \boldsymbol{R}}\right)^{(\circ)}}^2}\\
				=& \var{\boldsymbol{s}_{1}^{\operatorname{H}} \hat{\boldsymbol{Q}}(z_{1}) \hat{\boldsymbol{R}} \boldsymbol{s}_{2}\phi_{M}} \var{\phi_{M}\frac{1}{N} \tr{\hat{\boldsymbol{Q}}(z_{1}) \boldsymbol{R} \hat{\boldsymbol{Q}}(z_{2}) \boldsymbol{R}}} = \mathcal{O}\left(N^{-3}\right),
			\end{aligned}
			\label{eq:proof:aux:lemma4:expectation1: intermediate_9}
		\end{equation}
		where we have used \eqref{eq:aux:lemma3: variance1} in Lemma \ref{aux:lemma3} and \eqref{eq:aux:lemma2: variance1} in Lemma \ref{aux:lemma2} to obtain 
		\begin{equation*}
			\var{\boldsymbol{s}_{1}^{\operatorname{H}} \hat{\boldsymbol{Q}}(z_{1}) \hat{\boldsymbol{R}} \boldsymbol{s}_{2}\phi_{M}} = \mathcal{O}(N^{-1}), \qquad \var{\phi_{M}\frac{1}{N} \tr{\hat{\boldsymbol{Q}}(z_{1}) \boldsymbol{R} \hat{\boldsymbol{Q}}(z_{2}) \boldsymbol{R}}}= \mathcal{O}(N^{-2}).
		\end{equation*}
		From \eqref{eq:proof:aux:lemma4:expectation1: intermediate_9} it follows that the covariance in \eqref{eq:proof:aux:lemma4:expectation1: intermediate_8} decays with
		\begin{equation*}
			\cov{\boldsymbol{s}_{1}^{\operatorname{H}} \hat{\boldsymbol{Q}}(z_{1}) \hat{\boldsymbol{R}} \boldsymbol{s}_{2}\phi_{M}}{\phi_{M}\frac{1}{N} \tr{\hat{\boldsymbol{Q}}(z_{1}) \boldsymbol{R} \hat{\boldsymbol{Q}}(z_{2}) \boldsymbol{R}}} \leq \mathcal{O}\left( N^{-3/2} \right),
		\end{equation*}
		and we can express the first term on the right hand side of \eqref{eq:proof:aux:lemma4:expectation1: intermediate_7} as
		\begin{equation}
			\begin{aligned}
				\mean{\boldsymbol{s}_{1}^{\operatorname{H}} \hat{\boldsymbol{Q}}(z_{1}) \hat{\boldsymbol{R}}\boldsymbol{s}_{2}\phi_{M} \frac{1}{N} \tr{\hat{\boldsymbol{Q}}(z_{1}) \boldsymbol{R} \hat{\boldsymbol{Q}}(z_{2}) \boldsymbol{R}}} =  \mean{\boldsymbol{s}_{1}^{\operatorname{H}}\hat{\boldsymbol{Q}}(z_{1}) \hat{\boldsymbol{R}} \boldsymbol{s}_{2}\phi_{M}} \mean{\phi_{M}\frac{1}{N} \tr{\hat{\boldsymbol{Q}}(z_{1}) \boldsymbol{R} \hat{\boldsymbol{Q}}(z_{2}) \boldsymbol{R}}} + \mathcal{O}\left( N^{-3/2} \right).
			\end{aligned}
			\label{eq:proof:aux:lemma4:expectation1: intermediate_10}
		\end{equation}
		Using the previously derived expectations in \eqref{eq:aux:lemma3: expectation1} in Lemma \ref{aux:lemma3} and in \eqref{eq:aux:lemma2: expectation1} in Lemma \ref{aux:lemma2} we can further develop \eqref{eq:proof:aux:lemma4:expectation1: intermediate_10} as
		\begin{equation}
			\mean{\boldsymbol{s}_{1}^{\operatorname{H}} \hat{\boldsymbol{Q}}(z_{1}) \hat{\boldsymbol{R}}\boldsymbol{s}_{2} \phi_{M} \frac{1}{N} \tr{\hat{\boldsymbol{Q}}(z_{1}) \boldsymbol{R} \hat{\boldsymbol{Q}}(z_{2}) \boldsymbol{R}}}= \left(\frac{z_{1}}{\omega(z_{1})} \boldsymbol{s}_{1}^{\operatorname{H}} \bar{\boldsymbol{Q}}(z_{1}) \boldsymbol{R} \boldsymbol{s}_{2}\right)\left(\frac{1}{1-\gamma(z_{1},z_{2})}\frac{1}{N} \tr{\bar{\boldsymbol{Q}}(z_{1}) \boldsymbol{R}\bar{\boldsymbol{Q}}(z_{2}) \boldsymbol{R}}\right) + \mathcal{O}\left( N^{-1}\right),
			\label{eq:proof:aux:lemma4:expectation1: intermediate_11}
		\end{equation}
		where $\gamma(z_{1},z_{2})$ is defined in \eqref{eq:aux:lemma2: gamma_definition_1}. Substituting \eqref{eq:proof:aux:lemma4:expectation1: intermediate_11} into \eqref{eq:proof:aux:lemma4:expectation1: intermediate_7} 	and multiplying with $\frac{z_{2}}{\omega(z_{2})}$ on both sides yields		
		\begin{equation}
			\begin{aligned}
				\mean{\boldsymbol{s}_{1}^{\operatorname{H}} \hat{\boldsymbol{Q}}(z_{1}) \boldsymbol{R} \boldsymbol{s}_{2}\phi_{M}} =& -\frac{z_{1}}{\omega(z_{1})} \frac{z_{2}}{\omega(z_{2})} \boldsymbol{s}_{1}^{\operatorname{H}} \bar{\boldsymbol{Q}}(z_{1}) \boldsymbol{R} \boldsymbol{s}_{2}\left(\frac{1}{1-\gamma(z_{1},z_{2})}\frac{1}{N} \tr{\bar{\boldsymbol{Q}}(z_{1}) \boldsymbol{R}\bar{\boldsymbol{Q}}(z_{2}) \boldsymbol{R}}\right)  -\frac{z_{2}}{\omega(z_{2})}\mean{\boldsymbol{s}_{1}^{\operatorname{H}} \hat{\boldsymbol{Q}}(z_{1}) \boldsymbol{R} \hat{\boldsymbol{Q}}(z_{2}) \hat{\boldsymbol{R}} \boldsymbol{s}_{2} \phi_{M}\alpha(z_{2})} \\
				&+ \frac{z_{2}}{\omega(z_{2})}\mean{\boldsymbol{s}_{1}^{\operatorname{H}} \hat{\boldsymbol{Q}}(z_{1}) \boldsymbol{R} \hat{\boldsymbol{Q}}(z_{2}) \left( \boldsymbol{R} - \omega(z_{2})\boldsymbol{I}_{M}\right) \boldsymbol{s}_{2}\phi_{M}} + \mathcal{O}\left( N^{-1}\right).
			\end{aligned}
			\label{eq:proof:aux:lemma4:expectation1: intermediate_12}
		\end{equation}
		Inserting $\bar{\boldsymbol{Q}}(z_{2})\boldsymbol{s}_{2}$ for $\boldsymbol{s}_{2}$ and using $\bar{\boldsymbol{Q}}(z_{2})=\frac{\omega(z_{2})}{z_{2}} (\boldsymbol{R} - \omega(z_{2})\boldsymbol{I}_{M})^{-1}$ we obtain 
		\begin{equation}
			\begin{aligned}
				\mean{\boldsymbol{s}_{1}^{\operatorname{H}} \hat{\boldsymbol{Q}}(z_{1}) \boldsymbol{R} \bar{\boldsymbol{Q}}(z_{2})\boldsymbol{s}_{2}\phi_{M}} =& -\frac{z_{1}}{\omega(z_{1})} \frac{z_{2}}{\omega(z_{2})} \boldsymbol{s}_{1}^{\operatorname{H}} \bar{\boldsymbol{Q}}(z_{1}) \boldsymbol{R} \bar{\boldsymbol{Q}}(z_{2})\boldsymbol{s}_{2}\left(\frac{1}{1-\gamma(z_{1},z_{2})}\frac{1}{N} \tr{\bar{\boldsymbol{Q}}(z_{1}) \boldsymbol{R}\bar{\boldsymbol{Q}}(z_{2}) \boldsymbol{R}}\right) \\
				 &-\frac{z_{2}}{\omega(z_{2})}\mean{\boldsymbol{s}_{1}^{\operatorname{H}} \hat{\boldsymbol{Q}}(z_{1}) \boldsymbol{R} \hat{\boldsymbol{Q}}(z_{2}) \hat{\boldsymbol{R}} \bar{\boldsymbol{Q}}(z_{2})\boldsymbol{s}_{2} \phi_{M} \alpha(z_{2})} + \mean{\boldsymbol{s}_{1}^{\operatorname{H}} \hat{\boldsymbol{Q}}(z_{1}) \boldsymbol{R} \hat{\boldsymbol{Q}}(z_{2})\boldsymbol{s}_{2}\phi_{M}} + \mathcal{O}\left( N^{-1}\right).
			\end{aligned}
			\label{eq:proof:aux:lemma4:expectation1: intermediate_13}
		\end{equation}
		By rearranging \eqref{eq:proof:aux:lemma4:expectation1: intermediate_13} the quantity of interest in \eqref{eq:aux:lemma4: expectation1} can be expressed as
		\begin{equation}
			\begin{aligned}
				\mean{\boldsymbol{s}_{1}^{\operatorname{H}} \hat{\boldsymbol{Q}}(z_{1}) \boldsymbol{R} \hat{\boldsymbol{Q}}(z_{2})\boldsymbol{s}_{2}\phi_{M}} =& \mean{\boldsymbol{s}_{1}^{\operatorname{H}} \hat{\boldsymbol{Q}}(z_{1}) \boldsymbol{R} \bar{\boldsymbol{Q}}(z_{2})\boldsymbol{s}_{2}\phi_{M}}  + \boldsymbol{s}_{1}^{\operatorname{H}} \bar{\boldsymbol{Q}}(z_{1}) \boldsymbol{R} \bar{\boldsymbol{Q}}(z_{2})\boldsymbol{s}_{2} \frac{\gamma(z_{1},z_{2})}{1-\gamma(z_{1},z_{2})} \\
				 &+\frac{z_{2}}{\omega(z_{2})}\mean{\boldsymbol{s}_{1}^{\operatorname{H}} \hat{\boldsymbol{Q}}(z_{1}) \boldsymbol{R} \hat{\boldsymbol{Q}}(z_{2}) \hat{\boldsymbol{R}} \bar{\boldsymbol{Q}}(z_{2})\boldsymbol{s}_{2} \phi_{M}\alpha(z_{2})}  + \mathcal{O}\left( N^{-1}\right),
			\end{aligned}
			\label{eq:proof:aux:lemma4:expectation1: intermediate_14}
		\end{equation}
		where we have used $\gamma(z_{1},\gamma_{2})$ in \eqref{eq:aux:lemma4: gamma_definition_2}.
		In the following we upper-bound the third expression on the right hand side of \eqref{eq:proof:aux:lemma4:expectation1: intermediate_14}. Using Cauchy-Schwarz inequality it can be seen that
		\begin{equation*}
			\abs{\mean{\boldsymbol{s}_{1}^{\operatorname{H}} \hat{\boldsymbol{Q}}(z_{1}) \boldsymbol{R} \hat{\boldsymbol{Q}}(z_{2}) \hat{\boldsymbol{R}} \bar{\boldsymbol{Q}}(z_{2})\boldsymbol{s}_{2} \phi_{M}\alpha(z_{2})}}^{2} \leq \mean{\abs{\boldsymbol{s}_{1}^{\operatorname{H}} \hat{\boldsymbol{Q}}(z_{1}) \boldsymbol{R} \hat{\boldsymbol{Q}}(z_{2}) \hat{\boldsymbol{R}} \bar{\boldsymbol{Q}}(z_{2})\boldsymbol{s}_{2}\phi_{M}}^{2}}\mean{\abs{\alpha(z_{2})}^{2}},
		\end{equation*}
		where $\mean{\abs{\alpha(z_{2})}^{2}}= \var{\alpha(z_{2})} + \mean{\alpha(z_{2})}^{2} = \mathcal{O}(N^{-2})$ which follows from Lemma \ref{aux:lemma1} and the definition of $\alpha(z_{2})$ in \eqref{eq: alpha}. Consequently, the third term on the right hand side of \eqref{eq:proof:aux:lemma4:expectation1: intermediate_14} decays as
		\begin{equation*}
			\frac{z_{2}}{\omega(z_{2})}\mean{\boldsymbol{s}_{1}^{\operatorname{H}} \hat{\boldsymbol{Q}}(z_{1}) \boldsymbol{R} \hat{\boldsymbol{Q}}(z_{2}) \hat{\boldsymbol{R}} \bar{\boldsymbol{Q}}(z_{2})\boldsymbol{s}_{2} \phi_{M} \alpha(z_{2})} \leq \mathcal{O}\left(N^{-1}\right),
		\end{equation*}
		and can therefore be neglected. Using the deterministic equivalent $\bar{\boldsymbol{Q}}(z_{2})$ in \eqref{eq: resolvent_and_deterministic_equivalent} of the resolvent $\hat{\boldsymbol{Q}}(z_{2})$ and $\mean{\hat{\boldsymbol{Q}}(z)\phi_{M}} = \bar{\boldsymbol{Q}}(z) + \mathcal{O}(N^{-1})$ we obtain the expectation in \eqref{eq:aux:lemma4: expectation1} in Lemma \ref{aux:lemma4}
		\begin{equation*}
			\begin{aligned}
				\mean{\boldsymbol{s}_{1}^{\operatorname{H}} \hat{\boldsymbol{Q}}(z_{1}) \boldsymbol{R} \hat{\boldsymbol{Q}}(z_{2})\boldsymbol{s}_{2}\phi_{M}} =& \boldsymbol{s}_{1}^{\operatorname{H}} \bar{\boldsymbol{Q}}(z_{1}) \boldsymbol{R} \bar{\boldsymbol{Q}}(z_{2})\boldsymbol{s}_{2}  + \boldsymbol{s}_{1}^{\operatorname{H}} \bar{\boldsymbol{Q}}(z_{1}) \boldsymbol{R} \bar{\boldsymbol{Q}}(z_{2})\boldsymbol{s}_{2} \frac{\gamma(z_{1},z_{2})}{1-\gamma(z_{1},z_{2})}  + \mathcal{O}\left( N^{-1}\right)\\
				=& \left(\frac{1}{1-\gamma(z_{1},z_{2})} \right) \boldsymbol{s}_{1}^{\operatorname{H}} \bar{\boldsymbol{Q}}(z_{1}) \boldsymbol{R} \bar{\boldsymbol{Q}}(z_{2})\boldsymbol{s}_{2}  + \mathcal{O}\left( N^{-1}\right).
			\end{aligned}
		\end{equation*}

%
%


\bibliographystyle{IEEEtran}
\input{ms.bbl}



%% file: figures/probability_of_resolution_vs_snr_M_15_N_10_K_2_doa_45_50_runs_10000.tex
%
%
\definecolor{mycolor1}{rgb}{0.00000,0.44700,0.74100}%
\definecolor{mycolor2}{rgb}{0.85000,0.32500,0.09800}%
\definecolor{mycolor3}{rgb}{0.92900,0.69400,0.12500}%
\definecolor{mycolor4}{rgb}{0.49400,0.18400,0.55600}%
\definecolor{mycolor5}{rgb}{0.46600,0.67400,0.18800}%
\definecolor{mycolor6}{rgb}{0,0,0}%
\definecolor{mycolor7}{rgb}{1,1,0}%
\begin{tikzpicture}

\begin{axis}[%
width=9in,
height=8in,
scale=0.29,
at={(1.011in,0.642in)},
scale only axis,
unbounded coords=jump,
xmin=-5,
xmax=30,
xlabel style={font=\color{white!15!black}},
xlabel={SNR in [dB]},
ymin=-0.04,
ymax=1.58,
ylabel style={font=\color{white!15!black}},
ylabel={Probability of Resolution},
label style={font=\footnotesize},
axis background/.style={fill=white},
xmajorgrids,
ymajorgrids,
ytick={0,  0.2, 0.4, 0.6, 0.8, 1},
legend style={at={(0,1)},anchor=north west, legend cell align=left, align=center, draw=white!15!black, legend columns = 2, column sep=12.7pt, font=\scriptsize},%
]
\addplot [color=mycolor1, line width=1.2pt]
  table[row sep=crcr]{%
-10	0.1213\\
-9.5	0.1125\\
-9	0.1088\\
-8.5	0.1108\\
-8	0.1093\\
-7.5	0.1074\\
-7	0.1213\\
-6.5	0.1342\\
-6	0.1522\\
-5.5	0.1614\\
-5	0.1819\\
-4.5	0.2185\\
-4	0.2443\\
-3.5	0.2733\\
-3	0.3062\\
-2.5	0.3448\\
-2	0.3842\\
-1.5	0.4179\\
-1	0.4495\\
-0.5	0.4854\\
0	0.5277\\
0.5	0.5488\\
1	0.5847\\
1.5	0.624\\
2	0.6423\\
2.5	0.6589\\
3	0.6986\\
3.5	0.7146\\
4	0.7445\\
4.5	0.7675\\
5	0.7888\\
5.5	0.827\\
6	0.8453\\
6.5	0.8641\\
7	0.8803\\
7.5	0.9013\\
8	0.9184\\
8.5	0.935\\
9	0.9473\\
9.5	0.9606\\
10	0.9712\\
10.5	0.9737\\
11	0.9824\\
11.5	0.9879\\
12	0.9934\\
12.5	0.9947\\
13	0.997\\
13.5	0.9972\\
14	0.9989\\
14.5	0.9996\\
15	0.9998\\
15.5	0.9999\\
16	0.9999\\
16.5	1\\
17	1\\
17.5	1\\
18	1\\
18.5	1\\
19	1\\
19.5	1\\
20	1\\
20.5	1\\
21	1\\
21.5	1\\
22	1\\
22.5	1\\
23	1\\
23.5	1\\
24	1\\
24.5	1\\
25	1\\
25.5	1\\
26	1\\
26.5	1\\
27	1\\
27.5	1\\
28	1\\
28.5	1\\
29	1\\
29.5	1\\
30	1\\
30.5	1\\
31	1\\
31.5	1\\
32	1\\
32.5	1\\
33	1\\
33.5	1\\
34	1\\
34.5	1\\
35	1\\
};
\addlegendentry{sim. g-MUSIC}

\addplot [color=mycolor1, line width=1.2pt, dashed]
  table[row sep=crcr]{%
-10	0.0525\\
-9.5	0.0494\\
-9	0.0414\\
-8.5	0.0377\\
-8	0.0324\\
-7.5	0.0305\\
-7	0.0245\\
-6.5	0.0279\\
-6	0.0249\\
-5.5	0.022\\
-5	0.021\\
-4.5	0.0235\\
-4	0.0205\\
-3.5	0.0228\\
-3	0.0267\\
-2.5	0.0256\\
-2	0.0241\\
-1.5	0.0285\\
-1	0.0286\\
-0.5	0.03\\
0	0.034\\
0.5	0.0319\\
1	0.037\\
1.5	0.0417\\
2	0.0495\\
2.5	0.0559\\
3	0.0657\\
3.5	0.0756\\
4	0.0855\\
4.5	0.1017\\
5	0.117\\
5.5	0.137\\
6	0.1604\\
6.5	0.1841\\
7	0.2111\\
7.5	0.2396\\
8	0.2741\\
8.5	0.3208\\
9	0.3452\\
9.5	0.395\\
10	0.4441\\
10.5	0.4802\\
11	0.5241\\
11.5	0.5584\\
12	0.5908\\
12.5	0.6267\\
13	0.6542\\
13.5	0.68\\
14	0.7097\\
14.5	0.7411\\
15	0.7689\\
15.5	0.7948\\
16	0.82\\
16.5	0.8447\\
17	0.866\\
17.5	0.8886\\
18	0.9094\\
18.5	0.9269\\
19	0.9355\\
19.5	0.9521\\
20	0.9606\\
20.5	0.9741\\
21	0.9785\\
21.5	0.9815\\
22	0.9872\\
22.5	0.9924\\
23	0.9945\\
23.5	0.9966\\
24	0.9977\\
24.5	0.999\\
25	0.9994\\
25.5	0.9997\\
26	0.9996\\
26.5	1\\
27	1\\
27.5	1\\
28	1\\
28.5	1\\
29	1\\
29.5	1\\
30	1\\
30.5	1\\
31	1\\
31.5	1\\
32	1\\
32.5	1\\
33	1\\
33.5	1\\
34	1\\
34.5	1\\
35	1\\
};
\addlegendentry{corr. sim. g-MUSIC}

\addplot [color=mycolor3, line width=1.2pt, dashed, mark=x, mark size=2.5pt, mark repeat = 4, mark phase = 1, mark options={solid, mycolor3}]
  table[row sep=crcr]{%
-10	nan\\
-9.5	nan\\
-9	nan\\
-8.5	nan\\
-8	nan\\
-7.5	nan\\
-7	nan\\
-6.5	nan\\
-6	nan\\
-5.5	nan\\
-5	nan\\
-4.5	nan\\
-4	nan\\
-3.5	nan\\
-3	nan\\
-2.5	nan\\
-2	nan\\
-1.5	nan\\
-1	nan\\
-0.5	nan\\
0	0.555035842646622\\
0.5	0.581711478958212\\
1	0.603922813356212\\
1.5	0.625523334501677\\
2	0.647343468757719\\
2.5	0.669680669386386\\
3	0.692633091932488\\
3.5	0.716188425464816\\
4	0.740254559359739\\
4.5	0.764672407127071\\
5	0.789223463045201\\
5.5	0.813637193998517\\
6	0.837600778370037\\
6.5	0.860772488819835\\
7	0.882799165654601\\
7.5	0.903337416434553\\
8	0.922077327521242\\
8.5	0.938766633075318\\
9	0.953232589113789\\
9.5	0.965398429559215\\
10	0.975291431412948\\
10.5	0.983040426290893\\
11	0.988862074634139\\
11.5	0.993037179797359\\
12	0.995880357541049\\
12.5	0.997707931908219\\
13	0.998809445854752\\
13.5	0.999427334148635\\
14	0.999747227660055\\
14.5	0.999898656810515\\
15	0.999963515212515\\
15.5	0.999988356479377\\
16	0.99999675341143\\
16.5	0.999999221813565\\
17	0.999999842554647\\
17.5	0.999999973656738\\
18	0.999999996437826\\
18.5	0.999999999620637\\
19	0.99999999996909\\
19.5	0.999999999998135\\
20	0.99999999999992\\
20.5	0.999999999999998\\
21	1\\
21.5	1\\
22	1\\
22.5	1\\
23	1\\
23.5	1\\
24	1\\
24.5	1\\
25	1\\
25.5	1\\
26	1\\
26.5	1\\
27	1\\
27.5	1\\
28	1\\
28.5	1\\
29	1\\
29.5	1\\
30	1\\
30.5	1\\
31	1\\
31.5	1\\
32	1\\
32.5	1\\
33	1\\
33.5	1\\
34	1\\
34.5	1\\
35	1\\
};
\addlegendentry{pre. g-MUSIC}

\addplot [color=mycolor3, line width=1.2pt, mark=x, mark size=2.5pt, mark repeat = 4, mark phase = 1, mark options={solid, mycolor3}]
  table[row sep=crcr]{%
-10	nan\\
-9.5	nan\\
-9	nan\\
-8.5	nan\\
-8	nan\\
-7.5	nan\\
-7	nan\\
-6.5	nan\\
-6	nan\\
-5.5	nan\\
-5	nan\\
-4.5	nan\\
-4	nan\\
-3.5	nan\\
-3	nan\\
-2.5	nan\\
-2	nan\\
-1.5	nan\\
-1	nan\\
-0.5	nan\\
0	nan\\
0.5	nan\\
1	nan\\
1.5	nan\\
2	nan\\
2.5	nan\\
3	nan\\
3.5	nan\\
4	nan\\
4.5	nan\\
5	nan\\
5.5	nan\\
6	nan\\
6.5	nan\\
7	nan\\
7.5	nan\\
8	nan\\
8.5	nan\\
9	nan\\
9.5	nan\\
10	nan\\
10.5	nan\\
11	nan\\
11.5	0.5739389232924\\
12	0.602743416620062\\
12.5	0.628577231051064\\
13	0.654112153600998\\
13.5	0.680032649848382\\
14	0.706544408032166\\
14.5	0.733638045979001\\
15	0.761164719026218\\
15.5	0.788864673390875\\
16	0.816383872742942\\
16.5	0.843291336856387\\
17	0.869102961541761\\
17.5	0.893314281177951\\
18	0.915441924294833\\
18.5	0.935070665039762\\
19	0.951900135748788\\
19.5	0.965783111317262\\
20	0.976746747338459\\
20.5	0.984990072678116\\
21	0.990855592485967\\
21.5	0.994779176531613\\
22	0.997228524806357\\
22.5	0.998643968866477\\
23	0.99939440641826\\
23.5	0.999755811363152\\
24	0.999912172002838\\
24.5	0.999972201894389\\
25	0.999992374215929\\
25.5	0.999998217355094\\
26	0.999999651589423\\
26.5	0.999999944268883\\
27	0.999999992876546\\
27.5	0.999999999291715\\
28	0.999999999946843\\
28.5	0.999999999997089\\
29	0.999999999999888\\
29.5	0.999999999999997\\
30	1\\
30.5	1\\
31	1\\
31.5	1\\
32	1\\
32.5	1\\
33	1\\
33.5	1\\
34	1\\
34.5	1\\
35	1\\
};
\addlegendentry{corr. pre. g-MUSIC}

\addplot [color=mycolor2, line width=1.2pt]
  table[row sep=crcr]{%
-10	0.0449\\
-9.5	0.0375\\
-9	0.0336\\
-8.5	0.0286\\
-8	0.0212\\
-7.5	0.02\\
-7	0.0145\\
-6.5	0.0135\\
-6	0.0135\\
-5.5	0.013\\
-5	0.0115\\
-4.5	0.0113\\
-4	0.0127\\
-3.5	0.0115\\
-3	0.0115\\
-2.5	0.0134\\
-2	0.0145\\
-1.5	0.0173\\
-1	0.02\\
-0.5	0.0222\\
0	0.0241\\
0.5	0.0283\\
1	0.032\\
1.5	0.0384\\
2	0.0489\\
2.5	0.0598\\
3	0.0721\\
3.5	0.0809\\
4	0.0962\\
4.5	0.1235\\
5	0.1431\\
5.5	0.1824\\
6	0.2092\\
6.5	0.2469\\
7	0.2943\\
7.5	0.341\\
8	0.4038\\
8.5	0.4593\\
9	0.5263\\
9.5	0.5976\\
10	0.6607\\
10.5	0.7203\\
11	0.7831\\
11.5	0.8277\\
12	0.8764\\
12.5	0.903\\
13	0.9424\\
13.5	0.9597\\
14	0.9754\\
14.5	0.9859\\
15	0.9932\\
15.5	0.9959\\
16	0.9973\\
16.5	0.9993\\
17	0.9995\\
17.5	0.9998\\
18	0.9999\\
18.5	1\\
19	1\\
19.5	1\\
20	1\\
20.5	1\\
21	1\\
21.5	1\\
22	1\\
22.5	1\\
23	1\\
23.5	1\\
24	1\\
24.5	1\\
25	1\\
25.5	1\\
26	1\\
26.5	1\\
27	1\\
27.5	1\\
28	1\\
28.5	1\\
29	1\\
29.5	1\\
30	1\\
30.5	1\\
31	1\\
31.5	1\\
32	1\\
32.5	1\\
33	1\\
33.5	1\\
34	1\\
34.5	1\\
35	1\\
};
\addlegendentry{sim. MUSIC}

\addplot [color=mycolor2, line width=1.2pt, dashed]
  table[row sep=crcr]{%
-10	0.0111\\
-9.5	0.0071\\
-9	0.0047\\
-8.5	0.0029\\
-8	0.0016\\
-7.5	0.0017\\
-7	0.0006\\
-6.5	0.0003\\
-6	0.0004\\
-5.5	0.0003\\
-5	0.0003\\
-4.5	0.0003\\
-4	0\\
-3.5	0\\
-3	0\\
-2.5	0\\
-2	0\\
-1.5	0\\
-1	0\\
-0.5	0\\
0	0\\
0.5	0\\
1	0\\
1.5	0\\
2	0\\
2.5	0\\
3	0\\
3.5	0\\
4	0\\
4.5	0\\
5	0\\
5.5	0\\
6	0\\
6.5	0\\
7	0\\
7.5	0\\
8	0\\
8.5	0\\
9	0\\
9.5	0\\
10	0.0001\\
10.5	0.0001\\
11	0.0005\\
11.5	0.0005\\
12	0.0012\\
12.5	0.0011\\
13	0.0025\\
13.5	0.0045\\
14	0.007\\
14.5	0.0137\\
15	0.0195\\
15.5	0.0325\\
16	0.0514\\
16.5	0.077\\
17	0.1057\\
17.5	0.1526\\
18	0.2023\\
18.5	0.2636\\
19	0.3263\\
19.5	0.3943\\
20	0.4933\\
20.5	0.5697\\
21	0.6434\\
21.5	0.7063\\
22	0.7815\\
22.5	0.8385\\
23	0.8817\\
23.5	0.9158\\
24	0.9474\\
24.5	0.9666\\
25	0.9744\\
25.5	0.9866\\
26	0.9912\\
26.5	0.997\\
27	0.9976\\
27.5	0.9988\\
28	0.9992\\
28.5	0.9997\\
29	1\\
29.5	1\\
30	0.9998\\
30.5	1\\
31	1\\
31.5	1\\
32	1\\
32.5	1\\
33	1\\
33.5	1\\
34	1\\
34.5	1\\
35	1\\
};
\addlegendentry{corr. sim. MUSIC}

\addplot [color=mycolor5, line width=1.2pt, dashed, mark=x, mark size=2.5pt, mark repeat = 4, mark phase = 2, mark options={solid, mycolor5}]
  table[row sep=crcr]{%
-10	nan\\
-9.5	nan\\
-9	nan\\
-8.5	nan\\
-8	nan\\
-7.5	nan\\
-7	nan\\
-6.5	nan\\
-6	nan\\
-5.5	nan\\
-5	nan\\
-4.5	nan\\
-4	nan\\
-3.5	nan\\
-3	nan\\
-2.5	nan\\
-2	nan\\
-1.5	nan\\
-1	nan\\
-0.5	nan\\
0	0.0674914949923989\\
0.5	0.0699252204393576\\
1	0.0737798808931439\\
1.5	0.0789452831796038\\
2	0.0854579131125801\\
2.5	0.0934412335495001\\
3	0.103088286857467\\
3.5	0.114657199654968\\
4	0.128471475072127\\
4.5	0.144921285714621\\
5	0.164462858744897\\
5.5	0.187612785070756\\
6	0.21493334276441\\
6.5	0.247004072705913\\
7	0.284374334191647\\
7.5	0.327492084281723\\
8	0.376606617858328\\
8.5	0.4316485384523\\
9	0.492099506799982\\
9.5	0.556876811995057\\
10	0.624270661393631\\
10.5	0.691979315628383\\
11	0.757280457789177\\
11.5	0.81734921971582\\
12	0.869684142767287\\
12.5	0.912545609591149\\
13	0.945274402470893\\
13.5	0.968371117600954\\
14	0.983291576053316\\
14.5	0.992023837108852\\
15	0.996601924498892\\
15.5	0.998725646059952\\
16	0.999585624936766\\
16.5	0.999885111246438\\
17	0.99997333862692\\
17.5	0.999994927241803\\
18	0.999999226666906\\
18.5	0.999999907937863\\
19	0.999999991684507\\
19.5	0.999999999448201\\
20	0.999999999974055\\
20.5	0.99999999999917\\
21	0.999999999999983\\
21.5	1\\
22	1\\
22.5	1\\
23	1\\
23.5	1\\
24	1\\
24.5	1\\
25	1\\
25.5	1\\
26	1\\
26.5	1\\
27	1\\
27.5	1\\
28	1\\
28.5	1\\
29	1\\
29.5	1\\
30	1\\
30.5	1\\
31	1\\
31.5	1\\
32	1\\
32.5	1\\
33	1\\
33.5	1\\
34	1\\
34.5	1\\
35	1\\
};
\addlegendentry{pre. MUSIC}

\addplot [color=mycolor5, line width=1.2pt, mark=x, mark size=2.5pt, mark repeat = 4, mark phase = 2, mark options={solid, mycolor5}]
  table[row sep=crcr]{%
-10	nan\\
-9.5	nan\\
-9	nan\\
-8.5	nan\\
-8	nan\\
-7.5	nan\\
-7	nan\\
-6.5	nan\\
-6	nan\\
-5.5	nan\\
-5	nan\\
-4.5	nan\\
-4	nan\\
-3.5	nan\\
-3	nan\\
-2.5	nan\\
-2	nan\\
-1.5	nan\\
-1	nan\\
-0.5	nan\\
0	nan\\
0.5	nan\\
1	nan\\
1.5	nan\\
2	nan\\
2.5	nan\\
3	nan\\
3.5	nan\\
4	nan\\
4.5	nan\\
5	nan\\
5.5	nan\\
6	nan\\
6.5	nan\\
7	nan\\
7.5	nan\\
8	nan\\
8.5	nan\\
9	nan\\
9.5	nan\\
10	nan\\
10.5	nan\\
11	nan\\
11.5	0.0310981035755694\\
12	0.0329048418120076\\
12.5	0.0356837944231935\\
13	0.0394416001105884\\
13.5	0.0442858874537839\\
14	0.0504019855785781\\
14.5	0.0580527570682136\\
15	0.0675877849424749\\
15.5	0.0794575501202225\\
16	0.0942298419951356\\
16.5	0.112604971312906\\
17	0.135424291069543\\
17.5	0.163663228087827\\
18	0.198395753918997\\
18.5	0.240712955943554\\
19	0.291576473204292\\
19.5	0.351592380396916\\
20	0.420708741992417\\
20.5	0.497876007833984\\
21	0.580762605858689\\
21.5	0.665671840237588\\
22	0.747822066046089\\
22.5	0.822079426668912\\
23	0.884046136312095\\
23.5	0.931163564243304\\
24	0.963345641244118\\
24.5	0.982787110242925\\
25	0.993002110921569\\
25.5	0.997586189250886\\
26	0.999308945747232\\
26.5	0.999839685090473\\
27	0.999970644153709\\
27.5	0.999995877332924\\
28	0.999999569796494\\
28.5	0.999999967789548\\
29	0.999999998335352\\
29.5	0.999999999943121\\
30	0.999999999998775\\
30.5	0.999999999999984\\
31	1\\
31.5	1\\
32	1\\
32.5	1\\
33	1\\
33.5	1\\
34	1\\
34.5	1\\
35	1\\
};
\addlegendentry{corr. pre. MUSIC}

\addplot [color=mycolor6, line width=1.2pt]
  table[row sep=crcr]{%
0	-0.05\\
0	1.05\\
};
\addlegendentry{sep. boundary}

\addplot [color=mycolor6, line width=1.2pt,dashed]
  table[row sep=crcr]{%
11.5	-0.05\\
11.5	1.05\\
};
\addlegendentry{corr. sep. boundary}

\end{axis}
\end{tikzpicture}%

%% file: figures/probability_of_resolution_vs_snr_M_15_N_15_doa_45_50_runs_10000.tex
%
%
\definecolor{mycolor1}{rgb}{0.00000,0.44700,0.74100}%
\definecolor{mycolor2}{rgb}{0.85000,0.32500,0.09800}%
\definecolor{mycolor3}{rgb}{0.92900,0.69400,0.12500}%
\definecolor{mycolor4}{rgb}{0.49400,0.18400,0.55600}%
\definecolor{mycolor5}{rgb}{0.46600,0.67400,0.18800}%
\definecolor{mycolor6}{rgb}{0,0,0}%
\definecolor{mycolor7}{rgb}{1,1,0}%
\begin{tikzpicture}

\begin{axis}[%
width=9in,
height=8in,
scale=0.29,
at={(1.011in,0.642in)},
scale only axis,
unbounded coords=jump,
xmin=-10,
xmax=25,
xlabel style={font=\color{white!15!black}},
xlabel={SNR in [dB]},
ymin=-0.04,
ymax=1.58,
ylabel style={font=\color{white!15!black}},
ylabel={Probability of Resolution},
label style={font=\footnotesize},
axis background/.style={fill=white},
xmajorgrids,
ymajorgrids,
ytick={0,  0.2, 0.4, 0.6, 0.8, 1},
legend style={at={(0,1)},anchor=north west, legend cell align=left, align=center, draw=white!15!black, legend columns = 2, column sep=12.7pt, font=\scriptsize},%
]

\addplot [color=mycolor1, line width=1.2pt]
  table[row sep=crcr]{%
-11	0.1102\\
-10	0.0944\\
-9	0.093\\
-8	0.1127\\
-7	0.147\\
-6	0.1952\\
-5	0.2533\\
-4	0.3347\\
-3	0.4215\\
-2	0.497\\
-1	0.5597\\
0	0.6353\\
1	0.6877\\
2	0.7356\\
3	0.7903\\
4	0.8372\\
5	0.8845\\
6	0.9194\\
7	0.9501\\
8	0.9733\\
9	0.9842\\
10	0.9939\\
11	0.9981\\
12	0.9996\\
13	0.9997\\
14	1\\
15	1\\
16	1\\
17	1\\
18	1\\
19	1\\
20	1\\
21	1\\
22	1\\
23	1\\
24	1\\
25	1\\
26	1\\
};
\addlegendentry{sim. g-\ac{MUSIC}}

\addplot [color=mycolor1, line width=1.2pt, dashed]
  table[row sep=crcr]{%
-11	0.0481\\
-10	0.0353\\
-9	0.0267\\
-8	0.0248\\
-7	0.0231\\
-6	0.0213\\
-5	0.0196\\
-4	0.022\\
-3	0.0245\\
-2	0.0268\\
-1	0.0331\\
0	0.039\\
1	0.0474\\
2	0.059\\
3	0.0878\\
4	0.1168\\
5	0.1619\\
6	0.2211\\
7	0.2992\\
8	0.3896\\
9	0.4752\\
10	0.5567\\
11	0.6354\\
12	0.7076\\
13	0.7566\\
14	0.8129\\
15	0.8602\\
16	0.9071\\
17	0.9404\\
18	0.9663\\
19	0.9814\\
20	0.9897\\
21	0.9957\\
22	0.9985\\
23	0.9992\\
24	1\\
25	1\\
26	1\\
};
\addlegendentry{corr. sim. g-\ac{MUSIC}}

\addplot [color=mycolor3, line width=1.2pt, dashed, mark=x, mark size=2.5pt, mark repeat = 2, mark phase = 1, mark options={solid, mycolor3}]
  table[row sep=crcr]{%
-11	nan\\
-10	nan\\
-9	nan\\
-8	nan\\
-7	nan\\
-6	nan\\
-5	nan\\
-4	nan\\
-3	nan\\
-2	nan\\
-1	0.589879982533668\\
0	0.636773921255219\\
1	0.684963657234872\\
2	0.735319380519261\\
3	0.78674332851671\\
4	0.837108653683629\\
5	0.88366972269466\\
6	0.92362472587335\\
7	0.954849073620917\\
8	0.976574515343769\\
9	0.989668845658973\\
10	0.996278655164755\\
11	0.998958733212313\\
12	0.999787517009104\\
13	0.999970795947159\\
14	0.999997554514384\\
15	0.999999890055585\\
16	0.999999997737254\\
17	0.999999999982562\\
18	0.999999999999961\\
19	1\\
20	1\\
21	1\\
22	1\\
23	1\\
24	1\\
25	1\\
26	1\\
};
\addlegendentry{pre. g-\ac{MUSIC}}

\addplot [color=mycolor3, line width=1.2pt, mark=x, mark size=2.5pt, mark repeat = 2, mark phase = 1, mark options={solid, mycolor3}]
  table[row sep=crcr]{%
-11	nan\\
-10	nan\\
-9	nan\\
-8	nan\\
-7	nan\\
-6	nan\\
-5	nan\\
-4	nan\\
-3	nan\\
-2	nan\\
-1	nan\\
0	nan\\
1	nan\\
2	nan\\
3	nan\\
4	nan\\
5	nan\\
6	nan\\
7	nan\\
8	nan\\
9	nan\\
10	0.581922460108689\\
11	0.638666907006614\\
12	0.694475872681534\\
13	0.752343909360717\\
14	0.810670285576663\\
15	0.866085616001574\\
16	0.914418083508184\\
17	0.952014864419234\\
18	0.977238389871363\\
19	0.991266279340165\\
20	0.997434898265004\\
21	0.99946138133773\\
22	0.999925705920791\\
23	0.999993942848576\\
24	0.999999744254483\\
25	0.999999995263474\\
26	0.999999999968777\\
};
\addlegendentry{corr. pre. g-\ac{MUSIC}}

\addplot [color=mycolor2, line width=1.2pt]
  table[row sep=crcr]{%
-11	0.0384\\
-10	0.0204\\
-9	0.0139\\
-8	0.0101\\
-7	0.0074\\
-6	0.0055\\
-5	0.0059\\
-4	0.0077\\
-3	0.0108\\
-2	0.013\\
-1	0.02\\
0	0.0297\\
1	0.0496\\
2	0.0788\\
3	0.1186\\
4	0.1909\\
5	0.2657\\
6	0.3811\\
7	0.5092\\
8	0.6647\\
9	0.7944\\
10	0.8894\\
11	0.9555\\
12	0.9826\\
13	0.9957\\
14	0.9989\\
15	0.9999\\
16	1\\
17	1\\
18	1\\
19	1\\
20	1\\
21	1\\
22	1\\
23	1\\
24	1\\
25	1\\
26	1\\
};
\addlegendentry{sim. \ac{MUSIC}}

\addplot [color=mycolor2, line width=1.2pt, dashed]
  table[row sep=crcr]{%
-11	0.0082\\
-10	0.0026\\
-9	0.0014\\
-8	0.0002\\
-7	0.0001\\
-6	0\\
-5	0\\
-4	0\\
-3	0\\
-2	0\\
-1	0\\
0	0\\
1	0\\
2	0\\
3	0\\
4	0\\
5	0\\
6	0\\
7	0\\
8	0\\
9	0.0001\\
10	0\\
11	0.0009\\
12	0.0032\\
13	0.0143\\
14	0.0338\\
15	0.078\\
16	0.1619\\
17	0.2945\\
18	0.4637\\
19	0.6454\\
20	0.7955\\
21	0.9028\\
22	0.9586\\
23	0.9859\\
24	0.9955\\
25	0.9995\\
26	0.9998\\
};
\addlegendentry{corr. sim. \ac{MUSIC}}

\addplot [color=mycolor5, line width=1.2pt, dashed, mark=x, mark size=2.5pt, mark repeat = 2, mark phase = 2, mark options={solid, mycolor5}]
  table[row sep=crcr]{%
-11	nan\\
-10	nan\\
-9	nan\\
-8	nan\\
-7	nan\\
-6	nan\\
-5	nan\\
-4	nan\\
-3	nan\\
-2	nan\\
-1	0.0558521964359782\\
0	0.0641473658750538\\
1	0.0789386455841701\\
2	0.102205525242957\\
3	0.13742224604801\\
4	0.189451916058214\\
5	0.263924681584646\\
6	0.365278787252307\\
7	0.492952964993895\\
8	0.636876864455033\\
9	0.776389350377302\\
10	0.887479059300237\\
11	0.956463958053839\\
12	0.987949637836633\\
13	0.997810776345245\\
14	0.999764774127027\\
15	0.999986846021998\\
16	0.99999967340238\\
17	0.999999997049264\\
18	0.999999999992448\\
19	0.999999999999996\\
20	1\\
21	1\\
22	1\\
23	1\\
24	1\\
25	1\\
26	1\\
};
\addlegendentry{pre. \ac{MUSIC}}

\addplot [color=mycolor5, line width=1.2pt, mark=x, mark size=2.5pt, mark repeat = 2, mark phase = 2, mark options={solid, mycolor5}]
  table[row sep=crcr]{%
-11	nan\\
-10	nan\\
-9	nan\\
-8	nan\\
-7	nan\\
-6	nan\\
-5	nan\\
-4	nan\\
-3	nan\\
-2	nan\\
-1	nan\\
0	nan\\
1	nan\\
2	nan\\
3	nan\\
4	nan\\
5	nan\\
6	nan\\
7	nan\\
8	nan\\
9	nan\\
10	0.0249238001271439\\
11	0.0281854189456899\\
12	0.0357908645573493\\
13	0.0493790733070403\\
14	0.0725248854282003\\
15	0.111343307576696\\
16	0.174991830509455\\
17	0.274420431368481\\
18	0.416567485739563\\
19	0.592844433801365\\
20	0.770277082932193\\
21	0.904265154430896\\
22	0.973595109254208\\
23	0.995766754225018\\
24	0.999661325526598\\
25	0.999988716025507\\
26	0.999999874079533\\
};
\addlegendentry{corr. pre. \ac{MUSIC}}

\addplot [color=mycolor6, line width=1.2pt]
  table[row sep=crcr]{%
-1	-0.05\\
-1	1.05\\
};
\addlegendentry{sep. boundary}

\addplot [color=mycolor6, line width=1.2pt,dashed]
  table[row sep=crcr]{%
10	-0.05\\
10	1.05\\
};
\addlegendentry{corr. sep. boundary}

\end{axis}
\end{tikzpicture}%

%% file: figures/probability_of_resolution_vs_snr_M_15_N_100_doa_45_50_runs_10000.tex
%
%
\definecolor{mycolor1}{rgb}{0.00000,0.44700,0.74100}%
\definecolor{mycolor2}{rgb}{0.85000,0.32500,0.09800}%
\definecolor{mycolor3}{rgb}{0.92900,0.69400,0.12500}%
\definecolor{mycolor4}{rgb}{0.49400,0.18400,0.55600}%
\definecolor{mycolor5}{rgb}{0.46600,0.67400,0.18800}%
\definecolor{mycolor6}{rgb}{0,0,0}%
\definecolor{mycolor7}{rgb}{1,1,0}%
\begin{tikzpicture}

\begin{axis}[%
width=9in,
height=8in,
scale=0.29,
at={(1.011in,0.642in)},
scale only axis,
unbounded coords=jump,
xmin=-15,
xmax=20,
xlabel style={font=\color{white!15!black}},
xlabel={SNR in [dB]},
ymin=-0.04,
ymax=1.58,
ylabel style={font=\color{white!15!black}},
ylabel={Probability of Resolution},
label style={font=\footnotesize},
axis background/.style={fill=white},
xmajorgrids,
ymajorgrids,
ytick={0,  0.2, 0.4, 0.6, 0.8, 1},
legend style={at={(0,1)},anchor=north west, legend cell align=left, align=center, draw=white!15!black, legend columns = 2, column sep=12.7pt, font=\scriptsize},%
]

\addplot [color=mycolor1, line width=1.2pt]
  table[row sep=crcr]{%
-16	0.0787\\
-15	0.0628\\
-14	0.061\\
-13	0.0744\\
-12	0.1047\\
-11	0.1547\\
-10	0.2364\\
-9	0.3362\\
-8	0.4585\\
-7	0.565\\
-6	0.656\\
-5	0.7364\\
-4	0.7974\\
-3	0.871\\
-2	0.9182\\
-1	0.954\\
0	0.9815\\
1	0.9944\\
2	0.9972\\
3	0.9989\\
4	0.9998\\
5	1\\
6	1\\
7	1\\
8	1\\
9	1\\
10	1\\
11	1\\
12	1\\
13	1\\
14	1\\
15	1\\
16	1\\
17	1\\
18	1\\
19	1\\
20	1\\
21	1\\
};
\addlegendentry{sim. g-\ac{MUSIC}}

\addplot [color=mycolor1, line width=1.2pt, dashed]
  table[row sep=crcr]{%
-16	0.0376\\
-15	0.0216\\
-14	0.0196\\
-13	0.0182\\
-12	0.0178\\
-11	0.0162\\
-10	0.0194\\
-9	0.02\\
-8	0.0199\\
-7	0.0241\\
-6	0.0247\\
-5	0.0304\\
-4	0.0362\\
-3	0.0481\\
-2	0.0617\\
-1	0.0879\\
0	0.1318\\
1	0.2038\\
2	0.2941\\
3	0.4145\\
4	0.5452\\
5	0.6419\\
6	0.7303\\
7	0.8106\\
8	0.8727\\
9	0.9274\\
10	0.9608\\
11	0.9812\\
12	0.9935\\
13	0.9976\\
14	0.9997\\
15	0.9999\\
16	1\\
17	1\\
18	1\\
19	1\\
20	1\\
21	1\\
};
\addlegendentry{corr. sim. g-\ac{MUSIC}}

\addplot [color=mycolor3, line width=1.2pt, dashed, mark=x, mark size=2.5pt, mark repeat = 2, mark phase = 1, mark options={solid, mycolor3}]
  table[row sep=crcr]{%
-16	nan\\
-15	nan\\
-14	nan\\
-13	nan\\
-12	nan\\
-11	nan\\
-10	nan\\
-9	nan\\
-8	nan\\
-7	0.588268437583334\\
-6	0.656186012021831\\
-5	0.728477794861612\\
-4	0.801178041798474\\
-3	0.867335374729772\\
-2	0.920815250584134\\
-1	0.95862423005879\\
0	0.981576477180211\\
1	0.993247892718785\\
2	0.998052989675745\\
3	0.999583237420348\\
4	0.999938532324916\\
5	0.999994318062667\\
6	0.999999708025989\\
7	0.999999992831759\\
8	0.999999999930535\\
9	0.999999999999791\\
10	1\\
11	1\\
12	1\\
13	1\\
14	1\\
15	1\\
16	1\\
17	1\\
18	1\\
19	1\\
20	1\\
21	1\\
};
\addlegendentry{pre. g-\ac{MUSIC}}

\addplot [color=mycolor3, line width=1.2pt, mark=x, mark size=2.5pt, mark repeat = 2, mark phase = 1, mark options={solid, mycolor3}]
  table[row sep=crcr]{%
-16	nan\\
-15	nan\\
-14	nan\\
-13	nan\\
-12	nan\\
-11	nan\\
-10	nan\\
-9	nan\\
-8	nan\\
-7	nan\\
-6	nan\\
-5	nan\\
-4	nan\\
-3	nan\\
-2	nan\\
-1	nan\\
0	nan\\
1	nan\\
2	nan\\
3	nan\\
4	0.564985338748985\\
5	0.644556609495485\\
6	0.723785881593881\\
7	0.804219761791732\\
8	0.877541921791839\\
9	0.934943412584981\\
10	0.971949357654958\\
11	0.990702519767317\\
12	0.997781218848072\\
13	0.999648116823208\\
14	0.999966418260304\\
15	0.999998296693416\\
16	0.999999960724757\\
17	0.999999999661926\\
18	0.999999999999153\\
19	1\\
20	1\\
21	1\\
};
\addlegendentry{corr. pre. g-\ac{MUSIC}}

\addplot [color=mycolor2, line width=1.2pt]
  table[row sep=crcr]{%
-16	0.0144\\
-15	0.0041\\
-14	0.001\\
-13	0.0004\\
-12	0.0002\\
-11	0.0004\\
-10	0.0003\\
-9	0.0005\\
-8	0.0012\\
-7	0.0039\\
-6	0.0108\\
-5	0.0292\\
-4	0.0682\\
-3	0.1414\\
-2	0.2715\\
-1	0.446\\
0	0.6328\\
1	0.8101\\
2	0.9294\\
3	0.9804\\
4	0.997\\
5	0.9995\\
6	1\\
7	1\\
8	1\\
9	1\\
10	1\\
11	1\\
12	1\\
13	1\\
14	1\\
15	1\\
16	1\\
17	1\\
18	1\\
19	1\\
20	1\\
21	1\\
};
\addlegendentry{sim. \ac{MUSIC}}

\addplot [color=mycolor2, line width=1.2pt, dashed]
  table[row sep=crcr]{%
-16	0.0014\\
-15	0.0001\\
-14	0\\
-13	0\\
-12	0\\
-11	0\\
-10	0\\
-9	0\\
-8	0\\
-7	0\\
-6	0\\
-5	0\\
-4	0\\
-3	0\\
-2	0\\
-1	0\\
0	0\\
1	0\\
2	0\\
3	0\\
4	0\\
5	0.0004\\
6	0.0031\\
7	0.0168\\
8	0.0582\\
9	0.1863\\
10	0.3808\\
11	0.6206\\
12	0.8118\\
13	0.9309\\
14	0.9804\\
15	0.9966\\
16	0.9996\\
17	1\\
18	1\\
19	1\\
20	1\\
21	1\\
};
\addlegendentry{corr. sim. \ac{MUSIC}}

\addplot [color=mycolor5, line width=1.2pt, dashed, mark=x, mark size=2.5pt, mark repeat = 2, mark phase = 2, mark options={solid, mycolor5}]
  table[row sep=crcr]{%
-16	nan\\
-15	nan\\
-14	nan\\
-13	nan\\
-12	nan\\
-11	nan\\
-10	nan\\
-9	nan\\
-8	nan\\
-7	0.0260236414860566\\
-6	0.0302179396942898\\
-5	0.0452616575335246\\
-4	0.0792306168520496\\
-3	0.14727520834914\\
-2	0.266641599193407\\
-1	0.440655356803\\
0	0.641228768832412\\
1	0.816820040337787\\
2	0.929597004033025\\
3	0.98083032016859\\
4	0.996557579789046\\
5	0.999627103052124\\
6	0.999978252942768\\
7	0.999999409678967\\
8	0.999999993806182\\
9	0.999999999980167\\
10	0.999999999999986\\
11	1\\
12	1\\
13	1\\
14	1\\
15	1\\
16	1\\
17	1\\
18	1\\
19	1\\
20	1\\
21	1\\
};
\addlegendentry{pre. \ac{MUSIC}}

\addplot [color=mycolor5, line width=1.2pt, mark=x, mark size=2.5pt, mark repeat = 2, mark phase = 2, mark options={solid, mycolor5}]
  table[row sep=crcr]{%
-16	nan\\
-15	nan\\
-14	nan\\
-13	nan\\
-12	nan\\
-11	nan\\
-10	nan\\
-9	nan\\
-8	nan\\
-7	nan\\
-6	nan\\
-5	nan\\
-4	nan\\
-3	nan\\
-2	nan\\
-1	nan\\
0	nan\\
1	nan\\
2	nan\\
3	nan\\
4	0.017174873488337\\
5	0.0167140890243378\\
6	0.0235461958018316\\
7	0.042395791503552\\
8	0.0871283245318204\\
9	0.182719306886428\\
10	0.353730105802887\\
11	0.587134039486672\\
12	0.808802812013109\\
13	0.942951685999684\\
14	0.990318712384808\\
15	0.999185091177123\\
16	0.999970982701578\\
17	0.999999639888142\\
18	0.999999998777313\\
19	0.999999999999164\\
20	1\\
21	1\\
};
\addlegendentry{corr. pre. \ac{MUSIC}}

\addplot [color=mycolor6, line width=1.2pt]
  table[row sep=crcr]{%
-7	-0.05\\
-7	1.05\\
};
\addlegendentry{sep. boundary}

\addplot [color=mycolor6, line width=1.2pt,dashed]
  table[row sep=crcr]{%
4	-0.05\\
4	1.05\\
};
\addlegendentry{corr. sep. boundary}

\end{axis}
\end{tikzpicture}%

%% file: figures/probability_of_resolution_vs_snapshots_snr_6dB_M_15_K_2_doa_45_50_rho_0_runs_10000.tex
%
%
\definecolor{mycolor1}{rgb}{0.00000,0.44700,0.74100}%
\definecolor{mycolor2}{rgb}{0.85000,0.32500,0.09800}%
\definecolor{mycolor3}{rgb}{0.92900,0.69400,0.12500}%
\definecolor{mycolor4}{rgb}{0.49400,0.18400,0.55600}%
\definecolor{mycolor5}{rgb}{0.46600,0.67400,0.18800}%
\definecolor{mycolor6}{rgb}{0,0,0}%
\definecolor{mycolor7}{rgb}{1,1,0}%
\begin{tikzpicture}

\begin{axis}[%
width=9in,
height=7in,
scale=0.29,
at={(1.011in,0.642in)},
scale only axis,
unbounded coords=jump,
xmin=0,
xmax=70,
xlabel style={font=\color{white!15!black}},
xlabel={Snapshots $N$},
ymin=-0.05,
ymax=1.05,
ylabel style={font=\color{white!15!black}},
ylabel={Probability of Resolution},
label style={font=\footnotesize},
axis background/.style={fill=white},
xmajorgrids,
ymajorgrids,
legend style={at={(1,0)},anchor=south east, legend cell align=left, align=left, draw=white!15!black, font=\scriptsize}
]
\addplot [color=mycolor1, line width=1.2pt]
  table[row sep=crcr]{%
2	0.0962\\
4	0.0821\\
7	0.1186\\
9	0.1826\\
11	0.2402\\
14	0.3498\\
16	0.4255\\
18	0.494\\
21	0.5974\\
23	0.6549\\
25	0.7059\\
28	0.7781\\
30	0.8201\\
32	0.8494\\
35	0.8897\\
37	0.9161\\
40	0.9362\\
42	0.9494\\
44	0.9604\\
47	0.9716\\
49	0.9746\\
51	0.9805\\
54	0.9892\\
56	0.9895\\
58	0.9929\\
61	0.9954\\
63	0.9964\\
65	0.9977\\
68	0.9985\\
70	0.9986\\
};
\addlegendentry{sim. g-MUSIC}

\addplot [color=mycolor3, line width=1.2pt, dashed, mark=x, mark size=2.5pt, mark repeat = 2, mark phase = 2, mark options={solid, mycolor3}]
  table[row sep=crcr]{%
2	nan\\
4	0.144919500014348\\
7	0.159463160045959\\
9	0.192752181515109\\
11	0.240292946696881\\
14	0.331380853513896\\
16	0.400182368286736\\
18	0.471243967138189\\
21	0.575216531894811\\
23	0.639341267698344\\
25	0.697564246050472\\
28	0.772420609436524\\
30	0.813914503646139\\
32	0.84907023411367\\
35	0.891192130186329\\
37	0.91317959965868\\
40	0.938714225333822\\
42	0.951685373796067\\
44	0.962060198608097\\
47	0.973768560711998\\
49	0.97956686656906\\
51	0.984125746161926\\
54	0.989177795865384\\
56	0.991638869291463\\
58	0.993552237705956\\
61	0.995647113943333\\
63	0.996656363413429\\
65	0.997435017503569\\
68	0.998280468840121\\
70	0.998684645491257\\
};
\addlegendentry{pre. g-MUSIC}

\addplot [color=mycolor2, line width=1.2pt]
  table[row sep=crcr]{%
2	0.0962\\
4	0.5632\\
7	0.7473\\
9	0.81\\
11	0.8571\\
14	0.9125\\
16	0.9374\\
18	0.9484\\
21	0.9677\\
23	0.9739\\
25	0.9777\\
28	0.988\\
30	0.9899\\
32	0.992\\
35	0.994\\
37	0.9966\\
40	0.9971\\
42	0.9979\\
44	0.9989\\
47	0.999\\
49	0.999\\
51	0.9995\\
54	0.9999\\
56	0.9999\\
58	0.9999\\
61	0.9999\\
63	0.9998\\
65	1\\
68	0.9999\\
70	1\\
};
\addlegendentry{sim. MUSIC}

\addplot [color= mycolor5, line width=1.2pt, dashed, mark=x, mark size=2.5pt, mark repeat = 2, mark phase = 1,mark options={solid, mycolor5}]
  table[row sep=crcr]{%
2	nan\\
4	0.583096463317217\\
7	0.74459731870589\\
9	0.811164273879223\\
11	0.860340842918876\\
14	0.911188231440112\\
16	0.934316018932252\\
18	0.951406604467824\\
21	0.969054563815555\\
23	0.977080229208701\\
25	0.983015479305785\\
28	0.989154338794479\\
30	0.991952126575913\\
32	0.994025072278333\\
35	0.996174223771102\\
37	0.997156155035738\\
40	0.998175862925289\\
42	0.998642537837728\\
44	0.998989404908947\\
47	0.999350363914326\\
49	0.999515901579355\\
51	0.999639130565624\\
54	0.999767594304807\\
56	0.999826612948034\\
58	0.999870605055304\\
61	0.99991653559709\\
63	0.999937668867464\\
65	0.999953439114659\\
68	0.999969925633745\\
70	0.999977521135402\\
};
\addlegendentry{pre. MUSIC}

\addplot [color=mycolor6, line width=1.2pt]
  table[row sep=crcr]{%
4	-0.05\\
4	1.05\\
};
\addlegendentry{sep. boundary}

\end{axis}
\end{tikzpicture}%

%% file: figures/probability_of_resolution_vs_angular_separation_M_15_N_15_doa_45_runs_10000.tex
%
%
\definecolor{mycolor1}{rgb}{0.00000,0.44700,0.74100}%
\definecolor{mycolor2}{rgb}{0.85000,0.32500,0.09800}%
\definecolor{mycolor3}{rgb}{0.92900,0.69400,0.12500}%
\definecolor{mycolor4}{rgb}{0.49400,0.18400,0.55600}%
\definecolor{mycolor5}{rgb}{0.46600,0.67400,0.18800}%
\definecolor{mycolor6}{rgb}{0,0,0}%
\definecolor{mycolor7}{rgb}{1,1,0}%
\begin{tikzpicture}

\begin{axis}[%
width=9in,
height=7in,
scale=0.29,
at={(1.011in,0.642in)},
scale only axis,
unbounded coords=jump,
xmin=0,
xmax=10,
xlabel style={font=\color{white!15!black}},
xlabel={Angular Separation $\Delta\vartheta$ in $\lbrack \text{deg} \rbrack$},
ymin=-0.05,
ymax=1.4,
ylabel style={font=\color{white!15!black}},
ylabel={Probability of Resolution},
label style={font=\footnotesize},
axis background/.style={fill=white},
xmajorgrids,
ymajorgrids,
ytick={0,  0.2, 0.4, 0.6, 0.8, 1},
legend style={at={(0,1)},anchor=north west, legend cell align=left, align=center, draw=white!15!black, legend columns = 2, column sep=18.1pt, font=\scriptsize},%
]
\addplot [color=mycolor1, line width=1.2pt]
  table[row sep=crcr]{%
0.2	0.0092\\
0.451282051282051	0.0128\\
0.702564102564103	0.0199\\
0.953846153846154	0.0322\\
1.20512820512821	0.0495\\
1.45641025641026	0.0799\\
1.70769230769231	0.1344\\
1.95897435897436	0.1811\\
2.21025641025641	0.2517\\
2.46153846153846	0.316\\
2.71282051282051	0.3714\\
2.96410256410256	0.4191\\
3.21538461538462	0.4692\\
3.46666666666667	0.5051\\
3.71794871794872	0.5461\\
3.96923076923077	0.5858\\
4.22051282051282	0.629\\
4.47179487179487	0.6635\\
4.72307692307692	0.6899\\
4.97435897435898	0.732\\
5.22564102564103	0.7751\\
5.47692307692308	0.8102\\
5.72820512820513	0.8409\\
5.97948717948718	0.8774\\
6.23076923076923	0.9018\\
6.48205128205128	0.9195\\
6.73333333333333	0.9456\\
6.98461538461539	0.9599\\
7.23589743589744	0.9675\\
7.48717948717949	0.9795\\
7.73846153846154	0.9856\\
7.98974358974359	0.9919\\
8.24102564102564	0.9939\\
8.49230769230769	0.9962\\
8.74358974358974	0.9978\\
8.99487179487179	0.9991\\
9.24615384615385	0.9988\\
9.4974358974359	0.9996\\
9.74871794871795	0.9993\\
10	0.9996\\
};
\addlegendentry{sim. g-\ac{MUSIC}}

\addplot [color=mycolor3, line width=1.2pt, dashed, mark=x, mark size=2.5pt, mark repeat = 2, mark phase = 2, mark options={solid, mycolor3}]
  table[row sep=crcr]{%
0.2	nan\\
0.451282051282051	nan\\
0.702564102564103	nan\\
0.953846153846154	nan\\
1.20512820512821	nan\\
1.45641025641026	nan\\
1.70769230769231	nan\\
1.95897435897436	nan\\
2.21025641025641	nan\\
2.46153846153846	nan\\
2.71282051282051	nan\\
2.96410256410256	nan\\
3.21538461538462	0.527462680662917\\
3.46666666666667	0.546412192325487\\
3.71794871794872	0.568507827646425\\
3.96923076923077	0.594466867045674\\
4.22051282051282	0.624295447854023\\
4.47179487179487	0.6576050457523\\
4.72307692307692	0.69365866153768\\
4.97435897435898	0.731418418957819\\
5.22564102564103	0.769635261402212\\
5.47692307692308	0.806980215740047\\
5.72820512820513	0.842195905952014\\
5.97948717948718	0.874237037676354\\
6.23076923076923	0.902370744951243\\
6.48205128205128	0.926219930180269\\
6.73333333333333	0.945748970969089\\
6.98461538461539	0.961204742848449\\
7.23589743589744	0.973032840817656\\
7.48717948717949	0.981788876734244\\
7.73846153846154	0.988060005461425\\
7.98974358974359	0.992405298710248\\
8.24102564102564	0.995317524550776\\
8.49230769230769	0.997204506767532\\
8.74358974358974	0.998385804496431\\
8.99487179487179	0.999099677588126\\
9.24615384615385	0.999515668886142\\
9.4974358974359	0.999749117153566\\
9.74871794871795	0.99987509386756\\
10	0.999940352556901\\
};
\addlegendentry{pre. g-\ac{MUSIC}}

\addplot [color=mycolor2, line width=1.2pt]
  table[row sep=crcr]{%
0.2	0\\
0.451282051282051	0\\
0.702564102564103	0\\
0.953846153846154	0\\
1.20512820512821	0\\
1.45641025641026	0\\
1.70769230769231	0\\
1.95897435897436	0\\
2.21025641025641	0\\
2.46153846153846	0\\
2.71282051282051	0\\
2.96410256410256	0.0001\\
3.21538461538462	0.0002\\
3.46666666666667	0.001\\
3.71794871794872	0.0028\\
3.96923076923077	0.0069\\
4.22051282051282	0.0116\\
4.47179487179487	0.0243\\
4.72307692307692	0.0454\\
4.97435897435898	0.0749\\
5.22564102564103	0.1244\\
5.47692307692308	0.176\\
5.72820512820513	0.243\\
5.97948717948718	0.3314\\
6.23076923076923	0.4256\\
6.48205128205128	0.5209\\
6.73333333333333	0.6009\\
6.98461538461539	0.7001\\
7.23589743589744	0.7765\\
7.48717948717949	0.8492\\
7.73846153846154	0.8913\\
7.98974358974359	0.9354\\
8.24102564102564	0.9561\\
8.49230769230769	0.9725\\
8.74358974358974	0.9858\\
8.99487179487179	0.9934\\
9.24615384615385	0.994\\
9.4974358974359	0.998\\
9.74871794871795	0.9987\\
10	0.999\\
};
\addlegendentry{sim. \ac{MUSIC}}

\addplot [color= mycolor5, line width=1.2pt, dashed, mark=x, mark size=2.5pt, mark repeat = 2, mark phase = 1,mark options={solid, mycolor5}]
  table[row sep=crcr]{%
0.2	nan\\
0.451282051282051	nan\\
0.702564102564103	nan\\
0.953846153846154	nan\\
1.20512820512821	nan\\
1.45641025641026	nan\\
1.70769230769231	nan\\
1.95897435897436	nan\\
2.21025641025641	nan\\
2.46153846153846	nan\\
2.71282051282051	nan\\
2.96410256410256	nan\\
3.21538461538462	0.0198189061430286\\
3.46666666666667	0.0217746028791452\\
3.71794871794872	0.0256170336866725\\
3.96923076923077	0.0316893383831539\\
4.22051282051282	0.040732213004268\\
4.47179487179487	0.0538733539640759\\
4.72307692307692	0.072653752598324\\
4.97435897435898	0.0990057427301254\\
5.22564102564103	0.135106024570307\\
5.47692307692308	0.183034120785847\\
5.72820512820513	0.244216761411509\\
5.97948717948718	0.318747909038275\\
6.23076923076923	0.404810582654594\\
6.48205128205128	0.498501253991708\\
6.73333333333333	0.594273652690436\\
6.98461538461539	0.685966703358439\\
7.23589743589744	0.768088897864486\\
7.48717948717949	0.836889300735799\\
7.73846153846154	0.890852991376988\\
7.98974358974359	0.930539989571242\\
8.24102564102564	0.957959435424672\\
8.49230769230769	0.975792875717115\\
8.74358974358974	0.986734258977487\\
8.99487179487179	0.99307867841282\\
9.24615384615385	0.996561203183546\\
9.4974358974359	0.998373053138507\\
9.74871794871795	0.999267259426615\\
10	0.999686058356731\\
};
\addlegendentry{pre. \ac{MUSIC}}

\addplot [color=mycolor6, line width=1.2pt]
  table[row sep=crcr]{%
3.21538461538462	-0.05\\
3.21538461538462	1.05\\
};
\addlegendentry{sep. boundary}

\end{axis}
\end{tikzpicture}%